\documentclass[12pt,onecolumn]{article}
\pdfoutput=1 
\usepackage{hyperref}
\hypersetup{colorlinks=true,linkcolor=blue,citecolor=blue}
\usepackage{amsmath,amssymb,amsfonts,amsthm}
\usepackage{caption,subcaption}
\usepackage[pdftex]{graphicx}
\usepackage{mathtools}
\usepackage{young}
\usepackage{comment}
\usepackage{mathrsfs}
\usepackage{titlesec}
\usepackage[normalem]{ulem}
\usepackage[affil-it]{authblk}
\usepackage{soul}

\def\be{\begin{equation}}
\def\ee{\end{equation}}

\def\arg{\operatorname{arg}}

\def\CF{{\mathcal F}}
\def\CT{{\mathcal T}}

\def\CP{{\mathcal P}}

\def\IS{{\mathbb{S}}}
\def\CW{{\mathcal W}}
\def\CB{{\mathcal B}}
\def\CH{{\mathcal H}}

\def\CT{{\mathcal T}}

\def\fb{{\mathfrak{b}}}
\def\det{\mathrm{det}}
\def\rank{{\mathrm{rank}}}

\def\CK{{\mathcal K}}
\def\CN{{\mathcal N}}

\def\CL{{\mathcal L}}
\def\CM{{\mathcal M}}
\def\CS{{\mathcal S}}

\def\IR{{\mathbb{R}}}
\def\IZ{{\mathbb{Z}}}
\def\IN{{\mathbb{N}}}

\def\IC{{\mathbb{C}}}

\def\x{{\times}}

\def\CF{{\cal F}}
\def\CC{{\cal C}}
\def\CT{{\cal T}}

\def\cl{{\rm cl}}

\def\diag{{\rm diag}}

\newtheorem{proposition}{Prop}

\titleclass{\subsubsubsection}{straight}[\subsection]

\newcounter{subsubsubsection}[subsubsection]
\renewcommand\thesubsubsubsection{\thesubsubsection.\arabic{subsubsubsection}}
\titleformat{\subsubsubsection}
  {\normalfont\normalsize\bfseries}{\thesubsubsubsection}{1em}{}
\titlespacing*{\subsubsubsection}
{0pt}{3.25ex plus 1ex minus .2ex}{1.5ex plus .2ex}

\makeatletter
\renewcommand\paragraph{\@startsection{paragraph}{5}{\z@}%
  {3.25ex \@plus1ex \@minus.2ex}%
  {-1em}%
  {\normalfont\normalsize\bfseries}}
\renewcommand\subparagraph{\@startsection{subparagraph}{6}{\parindent}%
  {3.25ex \@plus1ex \@minus .2ex}%
  {-1em}%
  {\normalfont\normalsize\bfseries}}
\def\toclevel@subsubsubsection{4}
\def\toclevel@paragraph{5}
\def\toclevel@paragraph{6}
\def\l@subsubsubsection{\@dottedtocline{4}{7em}{4em}}
\def\l@paragraph{\@dottedtocline{5}{10em}{5em}}
\def\l@subparagraph{\@dottedtocline{6}{14em}{6em}}
\makeatother

\newcommand{\weight}{\nu}
\newcommand{\weightB}{\tau}
\newcommand{\gA}{\mathrm{A}}
\newcommand{\gD}{\mathrm{D}}
\newcommand{\gE}{\mathrm{E}}
\newcommand{\gADE}{\mathfrak{g}_\mathrm{ADE}}

\newcommand{\cN}{\mathcal{N}}

\newcommand{\bC}{\mathbb{C}}

\newcommand{\fg}{\mathfrak{g}}

\newcommand{\ft}{\mathfrak{t}}

\newcommand{\SU}{\mathrm{SU}}
\newcommand{\SO}{\mathrm{SO}}

\newcommand{\dd}{\mathrm{d}}

\usepackage[usenames,dvipsnames]{xcolor}

\definecolor{darkergreen}{rgb}{0,0.667,0}
\definecolor{gold}{rgb}{1.,0.843137,0.}
\definecolor{deepskyblue}{rgb}{0.,0.74902,1.}
\definecolor{mediumorchid}{rgb}{0.729412,0.333333,0.827451}

\usepackage{tikz}
\usetikzlibrary{positioning}
\usetikzlibrary{decorations.pathmorphing}
\usetikzlibrary{decorations.markings}
\usetikzlibrary{shapes.misc}
\usetikzlibrary{intersections}

\usetikzlibrary{arrows}
\tikzset{>={angle 60}}

\tikzstyle{every picture}=[scale=1,baseline=(current bounding box.south)]

\tikzstyle{X}=[cross out, draw, scale = 0.75, thick]
\tikzstyle{Z}=[draw, circle, minimum size=1em, scale=1, inner sep=2pt]
\tikzstyle{W}= [circle, draw, minimum size=1em]
\tikzstyle{D}= [circle, minimum size=1em]

\usepackage{listings}
\newcommand{\foothref}[2]{\href{#1}{#2}\footnote{\href{#1}{\lstinline[basicstyle=\ttfamily]{#1}}}}

\newcommand{\ehz}[1]{\overline{{#1}}}

\newcommand{\ephz}[1]{{#1}}

\newcommand{\eh}[1]{\check{#1}}

\allowdisplaybreaks
\numberwithin{equation}{section}

\begin{document}

\title{ADE Spectral Networks}

\author[1]{{Pietro Longhi}\thanks{pietro.longhi@physics.uu.se}}
\author[2]{{Chan Y. Park}\thanks{chan@physics.rutgers.edu}}

\affil[1]{
{\normalsize Department of Physics and Astronomy, Uppsala University}
\protect\\
{\normalsize Box 516, SE-75120 Uppsala, Sweden} 
}

\affil[2]{
\normalsize{NHETC and Department of Physics and Astronomy,} 
\protect\\
{Rutgers University, NJ, USA}
}

\date{}
\maketitle

\begin{abstract}
\noindent We introduce a new perspective and a generalization of spectral networks for 4d $\mathcal{N}=2$ theories of class $\mathcal{S}$ associated to Lie algebras $\mathfrak{g} = \textrm{A}_n$, $\textrm{D}_n$, $\textrm{E}_{6}$, and $\textrm{E}_{7}$.
Spectral networks directly compute the BPS spectra of 2d theories on surface defects coupled to the 4d theories. A Lie algebraic interpretation of these spectra emerges naturally from our construction, leading to a new description of 2d-4d wall-crossing phenomena.
Our construction also provides an efficient framework for the study of BPS spectra of the 4d theories.
In addition, we consider novel types of surface defects associated with minuscule representations of $\mathfrak{g}$.

\end{abstract}

\clearpage

\tableofcontents

%%%%%%%%%%%%%%%%%%%%%%%%%%%%%%%%%%
\section{Introduction and summary}
%%%%%%%%%%%%%%%%%%%%%%%%%%%%%%%%%%

A distinguishing feature of 4d $\CN=2$ theories of class $\CS$ is their intimate relationship with Hitchin systems \cite{Gaiotto:2009hg}.
This connection establishes a unified picture capturing many interesting aspects of the 4d $\CN=2$ dynamics, such as the UV duality web \cite{Gaiotto:2009we, Tachikawa:2009rb}, the geometry of the Coulomb branch, its Seiberg-Witten description \cite{Gaiotto:2009hg, Seiberg:1994rs, Seiberg:1994aj}, the compactification to three dimensions \cite{Seiberg:1996nz}, the wall-crossing of the BPS spectrum \cite{Seiberg:1994rs, Seiberg:1994aj, Ferrari:1996sv, Gaiotto:2008cd, Gaiotto:2012rg}, and the insertion of line and surface defects \cite{Gukov:2006jk, Gaiotto:2010be, Gaiotto:2011tf, Gaiotto:2013sma}.

In this paper we focus on the BPS spectrum of 2d-4d coupled systems, which is encoded in the geometry of the Hitchin system, and can be studied using the spectral networks of \cite{Gaiotto:2012rg}. 
Previous studies of BPS spectra based on spectral networks \cite{Galakhov:2013oja, Maruyoshi:2013fwa} and a refinement of the original construction of \cite{Gaiotto:2012rg} to compute the spin of BPS states \cite{Galakhov:2014xba} have been restricted to class $\CS$ theories of $\fg=\gA_n$ type, following the original construction of spectral networks.
Here we introduce an extension and a new perspective on spectral networks for all class $\mathcal{S}$ theories of $\mathfrak{g} = \gA_n$, $\gD_n$, $\gE_6$, and $\gE_7$.  
The key observation on which the present work develops is that a spectral network carries an intrinsic Lie algebraic structure, which is inherited from the Hitchin system, as first suggested in \cite{Gaiotto:2012rg}. 
This leads to a new interpretation of several objects belonging to the class $\CS$ realm in terms of familiar Lie algebraic data.
From the viewpoint of BPS spectroscopy, the main result of this paper is an algorithmic approach to computing the BPS spectrum of class $\CS$ theories of the above $ \gA\gD\gE$ types, extending the framework proposed in \cite{Gaiotto:2012rg}. 
Readers mainly interested in this application can find a self-contained guide to this procedure in Section \ref{sec:examples}, where we also introduce \foothref{http://het-math2.physics.rutgers.edu/loom/}{\texttt{loom}}, a program to study spectral networks.

To put our results into some context, let us briefly recall the key aspects of the relation between a class $\mathcal{S}$ theory and a Hitchin system \cite{Gaiotto:2009hg}. 
Given a triplet of data, consisting of a simply-laced Lie algebra $\fg$, a punctured Riemann surface $C$, and certain data $D$ describing boundary conditions at punctures, a Hitchin system is defined by the following equations \cite{Hitchin},
\begin{align}
	F +  R^2 \,[\varphi,\overline\varphi]  = 0\,,\qquad {\partial_{\bar z}}\varphi + [A_{\bar z}, \varphi] = 0,
	\label{eq:hitchin-eqs}
\end{align}
where $R$ is a positive real number.\footnote{In the usual formulation of Hitchin's equations, $R$ does not appear. However, in the context of class $\CS$ theories, $R$ appears naturally as the compactification radius for the theory formulated on $\IR^3\times S^1$. More details can be found in \cite[Sec. 3.1.6]{Gaiotto:2009hg}}
The moduli space of solutions to this system of equations is a hyperk{\"a}hler manifold $\CM$, 
which carries the additional structure of an algebraic integrable system.
$\CM$ admits a fibration by Lagrangian tori $\CM\to \CB$, whose base is identified with the parameter space of $k$-differentials on $C$, the Casimirs of the Higgs field $\varphi$. 
From the class $\CS$ viewpoint, $\CB$ coincides with the Coulomb branch of the 4d theory.
Each point of $\CM$ defines a Riemann surface $\Sigma_\rho$, the spectral curve of the Hitchin system
\be
	\det_{\rho}  \left( \lambda\, \mathbb{I} - \varphi \right)  = 0 
	\label{eq:Sigma}
\ee	
for a choice of a representation $\rho$, together with a flat line bundle $\CL$ over $\Sigma_\rho$.
When $\rho$ is the vector representation\footnote{%
In this paper we shall call the \emph{vector representation} of an $\gA\gD\gE$ Lie algebra the first fundamental representation for $\gA_n$ and $\gD_{n\geq 3}$ as well as the  \textbf{27} of $\gE_6$ and the \textbf{56} of $\gE_7$.
} of $\fg$, the spectral curve is identified with the Seiberg-Witten curve of the class $\CS$ theory, while the tautological 1-form $\lambda$ is identified with the Seiberg-Witten differential.
Finally, the line bundle $\CL$ characterizes the vacuum expectation values of electric and magnetic Wilson lines of the IR theory compactified on a time-like circle \cite{Seiberg:1996nz, Gaiotto:2008cd}.
The Riemann surface $C$ also carries a physical interpretation from the viewpoint of the UV gauge theory, it is the parameter space of a certain class of \emph{canonical} surface defects $\IS_{z}$. 
The data $D$ encodes information about global symmetries of the 4d theory.
The Riemann surface $\Sigma_\rho$ is naturally presented as a ramified covering $\pi:\Sigma_\rho \to C$, and the discrete set of points in the fiber $\pi^{-1}(z)\in \Sigma_\rho$ is identified with the set of vacua of the low energy 2d $\CN=(2,2)$ theory carried by $\IS_{z}$.%
\footnote{While these interpretations of $C$ and $\Sigma$ can be stated in a 4d gauge theory language, their motivation and explanation are best understood in terms of M-theory, in the context of M2-M5 brane configurations \cite{Gaiotto:2009we, Gaiotto:2009hg, Gaiotto:2011tf}.}
The Hitchin geometry encodes more than the low-energy Seiberg-Witten description.
In particular, the spectrum of BPS states of the 4d gauge theory contributes
corrections to the hyperk{\"a}hler metric of $\CM$ \cite{Gaiotto:2008cd}, and can
thus be read off from the Hitchin geometry, at least in principle.
The task of extracting BPS spectra from Hitchin geometries is far from 
straightforward, nevertheless it is greatly simplified by the spectral networks framework of \cite{Gaiotto:2012rg}.
For a review of the relation between spectral networks and hyperk{\"a}hler 
geometry see \cite{FelixKlein, AndyPisa, AndyKITP}.

A new result in this paper is the reformulation of spectral networks data in a Lie-algebraic language. 
A spectral network $\CW$ consists of two pieces of data: geometric data encoded in a network of real 1-dimensional curves on $C$, each of which is called an $\mathcal{S}$-wall, and combinatorial topological data attached to an $\mathcal{S}$-wall, called soliton data.
The geometry of $\CW$ is fixed by a choice of $u\in\CB$ and a phase $\vartheta\in \IR/2\pi\IZ$, while soliton data is determined by the topology of $\CW$.\footnote{More precisely, the combinatorial data attached to $\CW$ is expressed in terms of topological data on  $\Sigma$.}
For generic $u$, there are branch points of the covering $\pi$ on $C$ where 
\be
	\langle\alpha,\varphi(z)\rangle = 0
\ee 
for one or more roots $\alpha$ of $\fg$. 
An $\CS$-wall that emanates from a branch point is labeled by the corresponding root $\alpha$, as shown in Figure \ref{fig:intro-network}.
The geometry of an $\CS$-wall $\CS_\alpha$ depends on $(u,\vartheta)$ through a differential equation
\be\label{eq:2d-BPS-condition}
	(\partial_t\,,\,\langle\alpha , \varphi\rangle) \in e^{i\vartheta}\IR^-\,,
\ee
where $t$ is a coordinate along the wall.
From a physical perspective, this equation may be viewed as the BPS condition for solitons of the 2d $\CN=(2,2)$ theory on $\IS_z$ \cite{Gaiotto:2011tf},
When $\CS$-walls intersect each other, a new $\CS$-wall may be produced, as shown in Figure \ref{fig:streets}.
When an $\mathcal{S}$-wall crosses a branch cut on $C$, its root-type may jump across the cut. 
In both cases the behavior of the network is determined entirely by Lie-algebraic data carried by $\CS$-walls and branch points, without any reference to the spectral cover $\Sigma_\rho$.
\begin{figure}[!ht]
\begin{center}
\includegraphics[width=.4\textwidth]{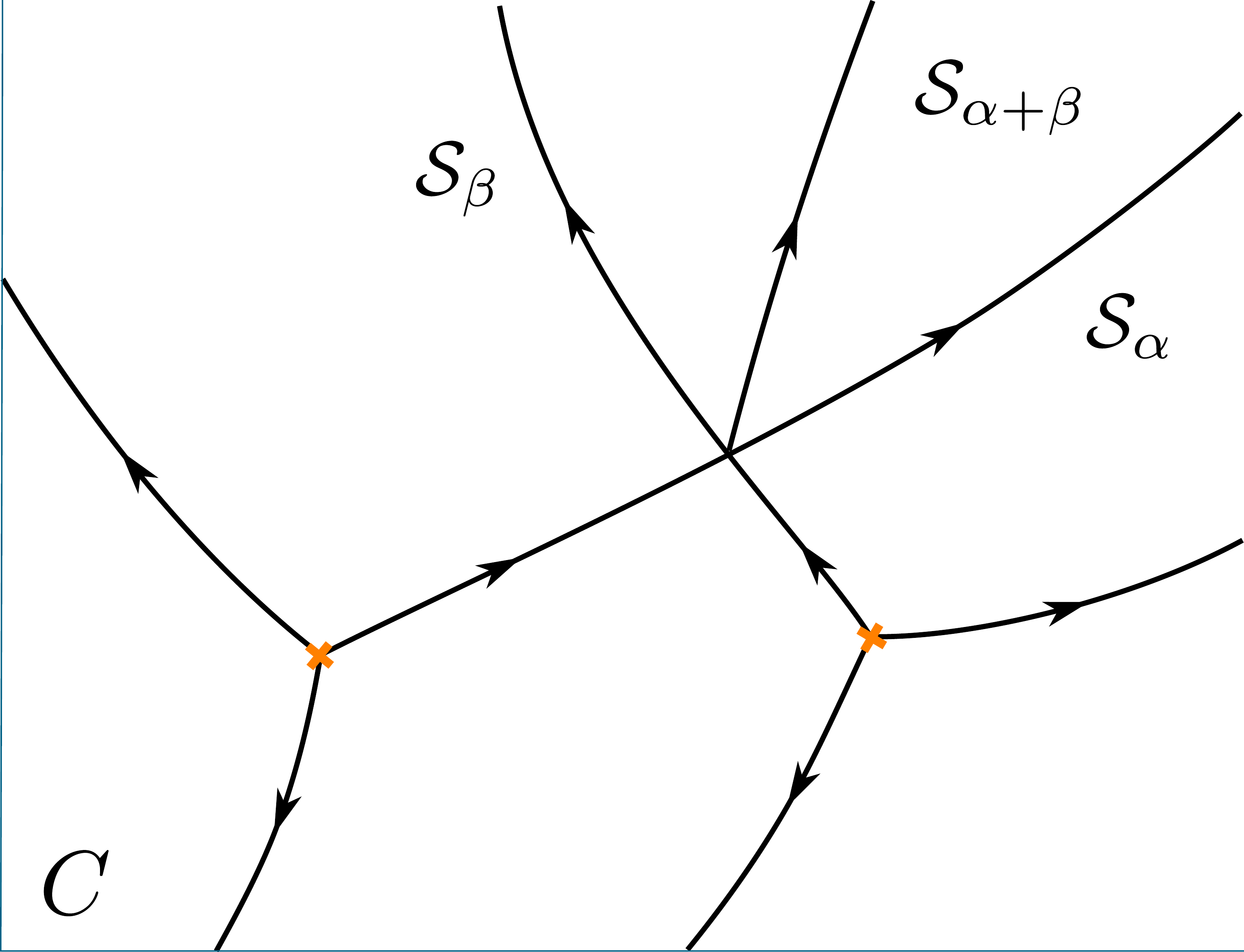}
\caption{Part of a spectral network. At the branch point on the left $\langle\alpha,\varphi\rangle = 0$, while that on the right is of type $\beta$.}
\label{fig:intro-network}
\end{center}
\end{figure}

The soliton data of $\CW$, on the other hand, depends on a choice of representation $\rho$, and is characterized by topological equivalence classes of open paths on $\Sigma_\rho$. 
The soliton data attached to each $\mathcal{S}$-wall is determined by the topology of $\CW$ according to two basic rules.\footnote{Both rules really descend from the single principle of twisted homotopy invariance for a certain \emph{formal parallel transport} on $C$, this viewpoint was advocated in \cite{Gaiotto:2012rg} and will play a central role in our construction too.}
The first rule fixes the soliton content of primary $\mathcal{S}$-walls, i.e. those which emanate directly from branch points. 
The second rule --- the 2d wall-crossing formula --- describes how the soliton data changes across intersections of $\CS$-walls. 
While the $\CS$-wall geometry is locally determined by the differential equation (\ref{eq:2d-BPS-condition}), 
the soliton data counts solutions which can be lifted \emph{globally} to $\Sigma_\rho$.
Physically, $\CW\subset C$ is a set of points in the parameter space of $\IS_z$, for which there are 2d BPS solitons with central charge of phase $\vartheta$.
Above $z\in C$, points of the fiber $\pi^{-1}(z)$ are identified with massive vacua of the 2d theory on $\IS_z$, and are labeled by weights $\weight$ of the representation $\rho$ since  $\pi^{-1}(z) = \{ \langle \weight, \varphi(z) \rangle, \weight \in \Lambda_\rho \}$.
Correspondingly, the soliton data of an $\CS$-wall $\CS_\alpha$ going through $z$ is classified by pairs of weights $\weight_i$, $\weight_j$ such that $\weight_j-\weight_i = \alpha$, as well as topological data of open paths on $\Sigma_\rho$.
The soliton data encodes the spectrum of 2d BPS solitons of $\IS_z$ \cite{Gaiotto:2011tf,Gaiotto:2012rg}, in fact the 2d wall-crossing formula of $\CS$-walls was found by \cite{Gaiotto:2012rg} to coincide with a twisted refinement of the Cecotti-Vafa wall-crossing formula \cite{Cecotti:1992rm}. 
Our framework  offers a natural interpretation of the 2d wall-crossing formula as a generalized ``Lie bracket" of certain generating functions $\Xi_\alpha$, $\Xi_\beta$ of 2d soliton spectra carried by intersecting $\CS$-walls $\CS_\alpha$, $\CS_\beta$.

Another new direction explored in this paper is the study of a spectral curve in a minuscule representation $\rho$. While there is a distinguished choice of $\rho$, the vector representation of $\mathfrak{g}$,  for which the spectral curve is identified with the Seiberg-Witten curve, from the viewpoint of the Hitchin system it makes perfect sense to explore other choices as well. From a physical viewpoint, we propose to identify a choice of $\rho$ with a choice of surface defect inserted in the 4d theory, which we denote as $\IS_{z,\rho}$.\footnote{
The M-theoretic description of such 2d defects and the 2d theories on the defects for $\mathfrak{g} = \gA_n$ is described in \cite{Hori:2013ewa}.
}
Using our definition of ADE spectral networks, we check this proposal through the physics of 2d-4d wall-crossing, which states that the 4d BPS spectrum is probed by the 2d BPS spectrum, in the sense that bound states of 2d BPS states can mix with 4d BPS states and vice versa. 
Doing so requires a careful identification of the physical lattice of 4d gauge and flavor charges as a sub-quotient of the homology lattice of $\Sigma_\rho$. 
We propose definitions for both 4d and 2d physical charges by making contact with work of Donagi on cameral covers \cite{Donagi:1993, Martinec:1995by}.
In this paper we focus on \emph{minuscule} representations of ADE-type Lie algebras.
Spectral networks then allow us to compute the 2d BPS spectrum carried by $\IS_{z,\rho}$ for a minuscule representation $\rho$, and to derive a generalization of the 2d-4d wall-crossing formula of \cite{Gaiotto:2011tf,Gaiotto:2012rg}, which relates the spectrum of 2d BPS solitons to the 4d BPS spectrum through 2d-4d wall-crossing.
We test our formulas against several nontrivial examples.

While the physics of the 4d gauge theories should be independent of the choice of $\rho$, this affects significantly the physics on the surface defects. 
It is natural to ask how this is compatible with the 2d-4d wall-crossing picture, which relates the 2d and 4d BPS spectra.
We find a solution to this puzzle by noting that, for $\rho$ other than the first fundamental representation of $\mathfrak{g} = \gA_n$, the 2d soliton spectra enjoy a high degree of symmetry. 
Although 2d spectra can grow very large with different choices of $\rho$, the actual amount of information they contain is always tamed by a large 2d soliton symmetry.
Using spectral networks, we derive the existence of this symmetry for all minuscule defects of ADE class $\CS$ theories. As a consistency check, we find that it plays a crucial role in the derivation of the 2d-4d wall-crossing formula.

\subsubsection*{Organization of the paper}

The paper is organized as follows.
Section \ref{sec:covers} is devoted to the study of spectral curves of Hitchin systems, in various representations. This lays out the foundations for the definition of ADE spectral networks, and contains our proposal for the definition of the physical lattice of gauge and flavor charges.
Section \ref{sec:spectral-networks} contains the definition of ADE spectral networks. Here we define the geometry of $\CW$ and derive the Lie-algebraic description of the Cecotti-Vafa wall-crossing formula, and argue that 2d BPS spectra enjoy of minuscule defects exhibit a certain discrete symmetry.
In Section \ref{subsec:k-wall-jumps} we study the 2d-4d wall-crossing phenomenon through spectral networks. By computing the jump of framed 2d-4d degeneracies at \emph{$\mathcal{K}$-walls}, we derive the generalization of the 2d-4d wall-crossing formula of \cite{Gaiotto:2011tf, Gaiotto:2012rg}.
Section \ref{sec:examples} contains several examples that illustrate our definitions and serve as nontrivial checks.

%%%%%%%%%%%%%%%%%%%%%%%%%%%%%%%%%%
\section{Spectral covers in class \texorpdfstring{$\CS$}{S} theories of ADE types}
\label{sec:covers}
%%%%%%%%%%%%%%%%%%%%%%%%%%%%%%%%%%

%%%%%%%%%%%%%%%%%%%%%%%%%%%%%%%%%%
\subsection{Trivializing spectral covers}\label{subsec:trivialization}
%%%%%%%%%%%%%%%%%%%%%%%%%%%%%%%%%%

The geometry of the Hitchin spectral curve encodes 
the BPS spectrum of the class $\CS$ theory, a useful tool for studying the geometry of these covers
is the use of a trivialization.
While generally there is no canonical choice of trivialization, and the physics is expected to be independent of 
such a choice,  it turns out that all trivializations exhibit certain universal features
for the class of systems we are going to study.
Here we present some of these common features,
which will play a key role in the definition of ADE spectral networks.

Let $\CB^\textrm{sing}$ be singular loci on the Coulomb branch of a class $\CS$ theory, and $\CB^*=\CB\setminus \CB^\textrm{sing}$ its complement. 
A point $u \in \CB^*$ determines a holomorphic section $\varphi$ of $(K\otimes \ft)/W\to C$, where $W$ is the Weyl 
group of $\fg$.
A choice of $d$-dimensional representation $\rho$, which we assume to be irreducible without loss of generality, determines a family of spectral curves fibered over $\CB^*$,
\begin{align}
	\Sigma_\rho \coloneqq \left\{  \lambda\, \big|\, \det \big( \lambda\, \mathbb{I}_d - \rho(\varphi) \big) = 0 \right\} \subset T^*C\,.
\end{align}
For each such curve, there is a natural projection map $\pi:\Sigma_{\rho}\rightarrow C$ that presents $\Sigma_\rho$ as a ramified $d$-sheeted covering of $C$.
Denoting the weights of $\rho$ by $\weight_j$ $(j = 1, \ldots, d)$, the sheets above a generic $z \in C$ are
\begin{align}
	\pi^{-1}(z) & = 
		\Big\{ \lambda_{z}\in T^{*}_{z}C \,\Big|\,
			\prod_{j = 1}^{d} \left(\lambda_z - x_j(z)\,\dd z \right) = 0 
	 	\Big\},\\
		& x_j(z)\,\mathrm{d} z= \langle \weight_j, \varphi(z) \rangle \in \bC,
\end{align}
where $\langle\cdot,\cdot\rangle$ denotes the natural pairing of $\ft^{*}$ and $\ft$. 
Sheets of the cover therefore correspond to weights of $\Lambda_{\rho}$, the weight system of the representation $\rho$. However, the identification of each sheet with some weight $\weight$ can be made only locally on $C$ until a choice of trivialization is made.
Specifying a trivialization of a spectral cover consists of two pieces of data: a choice of branch cuts, and the assignment of a weight of $\rho$ to each sheet.\footnote{
Such assignments are not arbitrary in general. For a detailed discussion of the compatibility conditions, see Appendix \ref{app:sheet-weight-identification}.}.
Here we show that, after a choice of a trivialization of a spectral cover which we call a standard trivialization, we can identify the branch points of the cover with Weyl reflections associated with simple roots.

%%%%%%%%%%%%%%%%%%%%%%%%%%%%%%%%%%%%%
\subsubsection*{Weyl branching structure}
%%%%%%%%%%%%%%%%%%%%%%%%%%%%%%%%%%%%%

Let us assume a choice of trivialization has been made. Then at a branch point $z_{*}\in C$ of the covering map $\pi$, we have two (or more) sheets colliding,
\be
x_{i}=x_{j}\qquad\Leftrightarrow \qquad \langle \weight_{i}-\weight_{j},\varphi(z_{*})\rangle = 0,
\ee
where $\weight_{i}-\weight_{j}$ is an element of the root lattice $\Lambda_{\textrm{root}}$, and not necessarily a root. 
On the other hand, the collision of two (or more) sheets is only part of the definition of a branch point,
as it does not imply the occurrence of actual sheet monodromy around $z_*$. 
Whenever there is a sheet monodromy, it always corresponds to a Weyl group action, which we call the \emph{Weyl branching property}.
A simple proof of this goes as follows. Consider a loop based at some $z_{0}\in C$, as in Figure \ref{fig:branch-monodromy}. 
\begin{figure}[!ht]
\begin{center}
\includegraphics[width=0.35\textwidth]{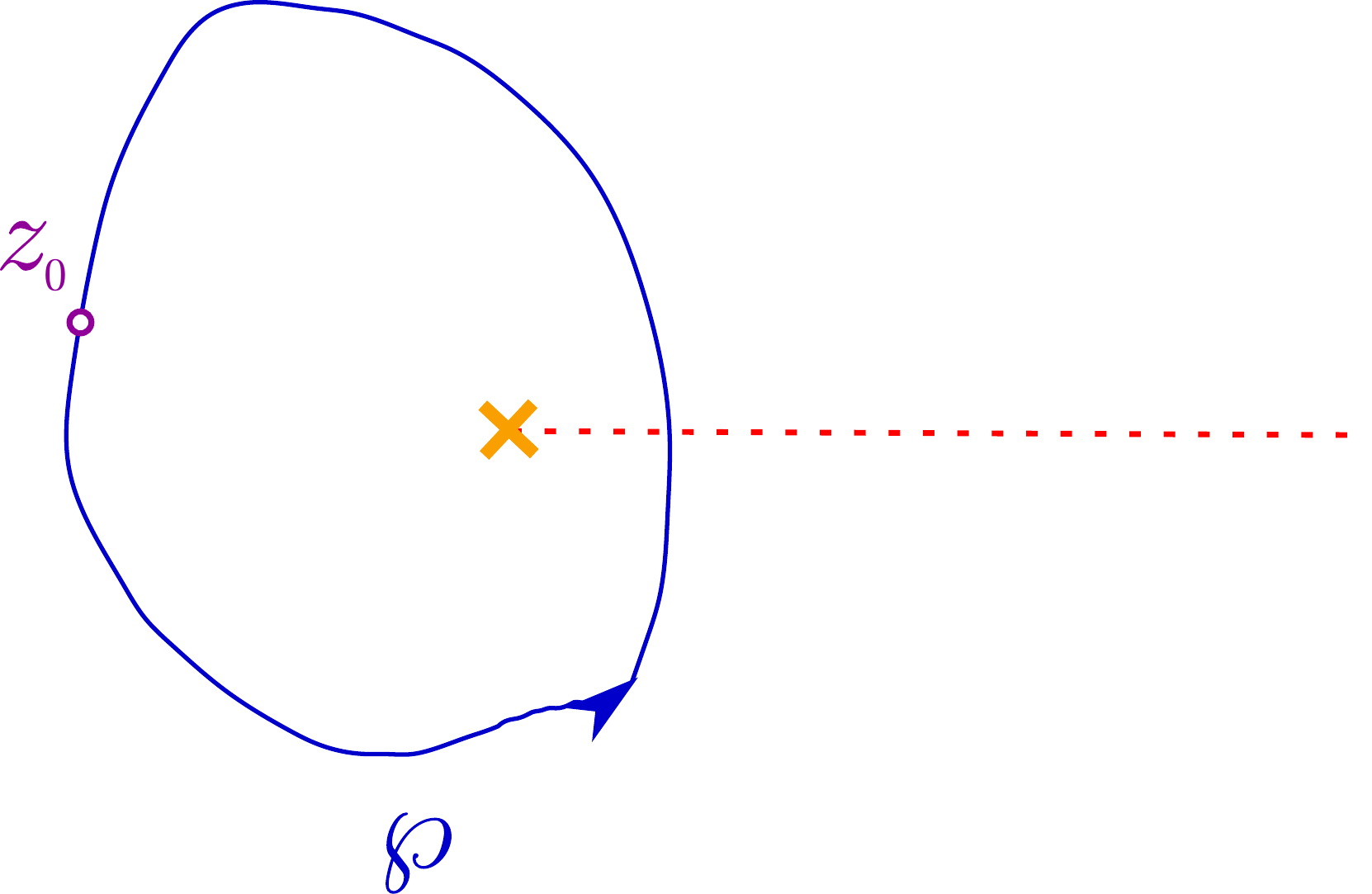}
\caption{A path around a branch point. Sheets undergo a Weyl-type monodromy along $\wp$. }
\label{fig:branch-monodromy}
\end{center}
\end{figure} 
Then consider the variation of $\varphi(z(t))$ as we vary $t$ from $0$ to $1$ ($z_{0}=z(t=0)=z(t=1)$), where $z_0$ is away from the branch cut, and we choose a representative for $\varphi(z_0)$ valued in $\mathfrak{t}$ instead of $\mathfrak{t}/W$. 
Due to the monodromy along $\wp$, we expect $\varphi(z(t=1))\neq \varphi(z(t=0))$. But their invariant polynomials must coincide, which means that they must be in the same conjugacy class, and conjugate elements of $\mathfrak{t}$ are related by a Weyl transformation, by definition. 

The Weyl branching property has a number of implications. First of all, any two choices of trivializations involving the same choice of cuts on $C$ must differ by a \emph{global} Weyl transformation. Let us consider two such trivializations that differ by the assignment of weights $\{\weight_i\}_{i=1}^{d}=\Lambda_\rho$ to the sheets of $\Sigma_\rho$. 
Concretely, let $z$ be any point on $C$ away from branch points, and let $\{x_i(z)\}_{i=1}^{d}$ be the fiber coordinates above $z$. Then in one trivialization we have a \emph{global} assignment 
\be
	\textit{Triv}\,:\ 	x_i \quad \rightarrow \weight_{f(i)}\,,
\ee
while in the other trivialization we have
\be
	\textit{Triv}'\,:\ 	x_i \quad \rightarrow \weight_{f'(i)}\,.
\ee
Then there must be a unique $w\in W$ such that
\be
	\weight_{f'(i)} = w \,\cdot\,\weight_{f(i)}\qquad \forall i=1,\dots, d\,.
\ee
This fact follows directly from the compatibility constraints on the assignments of weights to sheets.
A thorough discussion of such assignments, for all minuscule representations of ADE Lie algebras, can be found in Appendix \ref{app:sheet-weight-identification}.

Next let us consider changing the choice of \textit{Triv} by deforming branch cuts.
As long as a branch cut does not hit a branch point during the deformation, the global assignment of weights to the sheets is still determined by \textit{Triv}: above a patch of $C$ swept by a branch cut corresponding to some $w\in W$, the weight-sheet identification will simply jump from $x_i \to \weight_{f(i)}$ to $x_i \to w^{\pm_1}\cdot \weight_{f(i)}$.
On the other hand, if a branch cut of type $w$ sweeps across a branch point $w'$, then the ramification type of the latter will change by a conjugation $w'\to w'' = w^{-1} w' w$. See Figure \ref{fig:branch-cut-move} for an example of this deformation.
In either case, the weight-sheet assignment changes \emph{locally} by a Weyl transformation.

We therefore learn the following: fix any $z\in C$ away from branch points and punctures, then for any two choices of trivializations, the corresponding weight-sheet assignments above $z$ will differ by a Weyl transformation on the weights.
Theferore the notion of whether two sheets $x_i(z)$, $x_j(z)$ ``differ by a root" i.e.\ whether 
\be
	\weight_{f(i)} - \weight_{f(j)} \quad\text{is a root}
\ee
is actually independent of the choice of trivialization.
Given this fact, we can state the following claim, which will be proved below: a branch points with a sheet monodromy corresponds to a ramification points of the sheets $x_i(z)$, $x_j(z)$ such that $\weight_{i}-\weight_{j}=\alpha$ is a root, while there is no ramification at other $z_*$, although sheets may nevertheless collide.\footnote{We expect this claim to hold only for generic $u\in\CB$.}
Note that the  occurrence of ramification above a certain $z$ doesn't depend on the choice of trivialization, as neither does our characterization.

%%%%%%%%%%%%%%%%%%%%%%%%%%%%%%%%%%%%%
\subsubsection*{Minuscule representations}
%%%%%%%%%%%%%%%%%%%%%%%%%%%%%%%%%%%%%

The weight system of a representation $\rho$, $\Lambda_{\rho}$, is closed under the action of $W$ on $\ft^{*}$, and will in general comprise several Weyl orbits:
\be
	\Lambda_{\rho} \quad\rightarrow\quad \big\{{[\weight]}\big\}_{\weight\in\Lambda_{\rho}}
\ee
where $[\weight]$ denotes a Weyl orbit, understood as an equivalence class on $\Lambda_{\rho}$. 
Since a sheet monodromy corresponds to a Weyl transformation, $\Sigma_{\rho}$ factorizes into sub-covers
\be
	\Sigma_{\rho} \quad\rightarrow\quad \big\{\Sigma_{[\weight]}\big\}_{[\weight]\subset\Lambda_{\rho}}\,.
\ee
Motivated by this observation, from now on we will focus exclusively on minuscule representations of $\fg$, whose $\Lambda_{\rho}$ is a single Weyl orbit.
There is a finite number of such representations, which we list in Table \ref{tab:minuscule}.
\begin{table}[ht]
\caption{Minuscule representations of $\mathfrak{g}$.}
\begin{center}
\begin{tabular}{cl}
%	$\fg$ &  \\
%	\hline
	\hline
	A$_{n}$ :& all fundamental representations \\
	D$_{n}$ :& the vector and the two spinors \\
	E$_{6}$ :& the ${\bf 27, \overline{27}}$ \\
	E$_{7}$ :& the ${\bf 56}$ \\
	\hline
\end{tabular}
\end{center}
\label{tab:minuscule}
\end{table}%

Notably, that E$_{8}$ does not have any minuscule representation, although the adjoint representation ${\bf 248}$  is quasi-minuscule, i.e.\ its non-zero weights are in a single Weyl orbit. Defining spectral networks for covers in quasi-minuscule representations is an interesting problem, because it would enable us to study class $\mathcal{S}$ theories of all simple Lie algebra. In this paper we will not attempt this generalization.

%%%%%%%%%%%%%%%%%%%%%%%%%%%%%%%%%%%%%
\subsubsection*{Square-root branch points are labeled by simple roots}
%%%%%%%%%%%%%%%%%%%%%%%%%%%%%%%%%%%%%

By a genericity assumption\footnote{I.e. by studying $\Sigma_{\rho}(u)$ for generic $u\in\CB$. More precisely, a generic choice of $u$ is assumed to imply that the dual of $\varphi(z)$ never crosses the \emph{intersection} of two or more Weyl-reflection hyperplanes in $\ft^{*}$.} 
we can focus on covers with branch points of square-root type only, for the following reason.
$\ft^{*}$ is divided into disjoint Weyl chambers, each  of these is  delimited by a number of faces, each corresponds to a hyperplane $\CH_{\alpha}$ orthogonal to some root $\alpha$. 
When the dual of $\varphi(z_*)$ lies on a {generic} point of $\CH_{\alpha}$, 
\be
  \begin{split}%
      \langle \beta,\varphi(z_*)\rangle =0 & \quad\text{ iff }\quad\beta=\alpha 
  \end{split}
\ee
this in turn implies that the branch point at $z_{*}$ is of square-root type, in the sense that the square of the sheet permutation monodromy is trivial.
To see this, consider the orthogonal decomposition of $\ft^{*}\simeq \alpha\IR\oplus \CH_{\alpha}$ induced by $\alpha$, and denote the corresponding components of a weight by $\weight_{i} = \weight_{i}^{\parallel}+\weight_{i}^{\perp}$.
Then the fiber coordinates of sheets are
\be
	x_{i}(z_{*}) \, \dd z\,=\, \langle \weight_{i},\varphi(z_{*})\rangle \,=\, \langle \weight_{i}^{\perp},\varphi(z_{*})\rangle\ \in T^{*}C\big|_{z_{*}}.
\ee
Therefore sheets corresponding to weights with the same orthogonal component $\weight^{\perp}$ come together above $z = z_*$. The sheets group either in singles or in pairs: $x_{i}=x_{j}$ entails that $\weight_{i}$, $\weight_{j}$ lie on the same affine line parallel to $\alpha$, passing through $\weight_{i}^{\perp}=\weight_{j}^{\perp}$, but $\Lambda_{\rho}$ must lie on a hypersphere in $\ft^{*}$ since it's a $W$-orbit and $W$ preserves norms, and the line can intersect a sphere in at most two points. 
Single sheets, which don't ramify, are those with $\weight_{} \perp \alpha$, while the others must arrange in pairs. An illustration of this statement is given in Figure \ref{fig:weyl-orbit split}. 
Therefore above a branch point at $z=z_*$ there will be a number $k_{\rho}$ of ramification points, which depends only on the representation $\rho$, with sheets colliding pairwise. This proves that the branch point is of square-root type.

The converse is also true: a square-root branch point is always labeled by a root, we provide a proof in Appendix \ref{app:root-bp}. More precisely, the sheet monodromy around a branch point of square-root type corresponds to a Weyl reflection $w_{\alpha}:\,\weight_{i}\mapsto \weight_{i} - (\weight_{i},\, \alpha^\vee)\, \alpha$ under the local sheet-weight identification.
Moreover, in the neighborhood of a square-root type branch-point one can always choose a local coordinate $z$ such that $z_{*}=0$ and for any pair of colliding sheets
\be
\begin{split}
	& \langle \weight_{i},\varphi(z)\rangle \sim (x_{0} +\sqrt{z})dz \\
	& \langle \weight_{j},\varphi(z)\rangle \sim (x_{0} -\sqrt{z})dz \\
\end{split}
\label{eq:branch-point-cover}
\ee
where $x_{0} = \langle \weight_{i},\varphi(0)\rangle=\langle \weight_{j},\varphi(0)\rangle$ and  $\weight^{\perp}_{i} =  \weight_{j}^{\perp}$.

If the dual of $\varphi(z_*)$ lies on an intersection of multiple hyperplanes $\{\CH_{\beta_{i}}\}_{i}$, we have
\be
  \langle \beta_i,\varphi(z_*)\rangle = 0 ,
\ee
and there will be a higher-index branch point at $z = z_*$. The sheet monodromy will then be a product of the $\{w_{\beta_{i}}\}_{i}$. Without loss of generality, we will often restrict for simplicity to covers whose branch-points are only of the square-root type. This only involves a mild genericity assumption, because the cases with higher-order branch points can be thought of as certain limits of the generic cases.

%%%%%%%%%%%%%%%%%%%%%%%%%%%%%%%%%%%%%
\subsubsection*{Standard trivializations and Weyl chambers}
%%%%%%%%%%%%%%%%%%%%%%%%%%%%%%%%%%%%%

Having discussed the branching structure of a minuscule cover, we now turn to its trivialization. 
A choice of trivialization involves first choosing cuts on $C$, then identifying each point $x_{i}$ in the fiber $\pi^{-1}(z)$ with a weight $\weight_{i}$, for every $z\in C$ away from the branch cuts. 
The choice of a trivialization is not unique and there is generally no canonical one, therefore the physics should not depend on it.
Here we argue that there is always a choice of trivialization, for a minuscule cover at generic $u \in \mathcal{B}$, such that every branch point is associated with a simple root, which we will call a standard trivialization.

For every $z\in C$ away from the branch points, $\varphi(z)$ can be conjugated \emph{globally} into a unique choice of Cartan $\ft\subset\fg$.
After choosing cuts, we can do better: at each point (the dual of) $\varphi(z)$ can be conjugated into  the \emph{fundamental Weyl chamber} $\CC_{0}\subset\ft^{*}$.\footnote{If $\varphi$ is conjugated into some other Weyl chamber, we can perform a global Weyl transformation (by conjugating $\varphi(z)\mapsto w^{-1}\varphi(z) w\,,\ \forall z$ that brings $\varphi*(z)$ into $\CC_{0}$.} 
To see this, suppose that $\varphi(z_{0})$ lies in the dual of $\CC_{0}$ and $\varphi(z_{1})$ does not, and consider any path $\wp:z_{0}\to z_{1}$ that does not cross any branch cut.
By continuity of $\varphi$, at some point along the path $\varphi(z_{*})$ must lie on a face of $\CC_{0}$. But if this happens, then $z_{*}$ is a branch point, contradicting the assumptions.

The fact that a choice of cuts restricts $\varphi$ to be valued in the dual of $\CC_{0}$ implies that we can associate square-root branch points on $C$ with simple roots, because the interior of the fundamental Weyl chamber is spanned by non-negative linear combinations of fundamental weights, and therefore the chamber is bounded by hyperplanes orthogonal to simple coroots $\alpha_i^\vee$, which are the same as simple roots for a simply-laced Lie algebra. 
{For higher-order branch cuts, which may emanate from irregular singularities, a similar argument implies that they should correspond to Weyl transformations associated with {edges} of $\CC_{0}$.}
So we learn that a standard trivialization always exists. Figure \ref{fig:canonical-trivialization} shows an example of a standard trivialization. We hasten to stress that such trivializations may not be unique, and there is no canonical one among them.

\begin{figure}[ht]
\begin{center}
\includegraphics[width=0.5\textwidth]{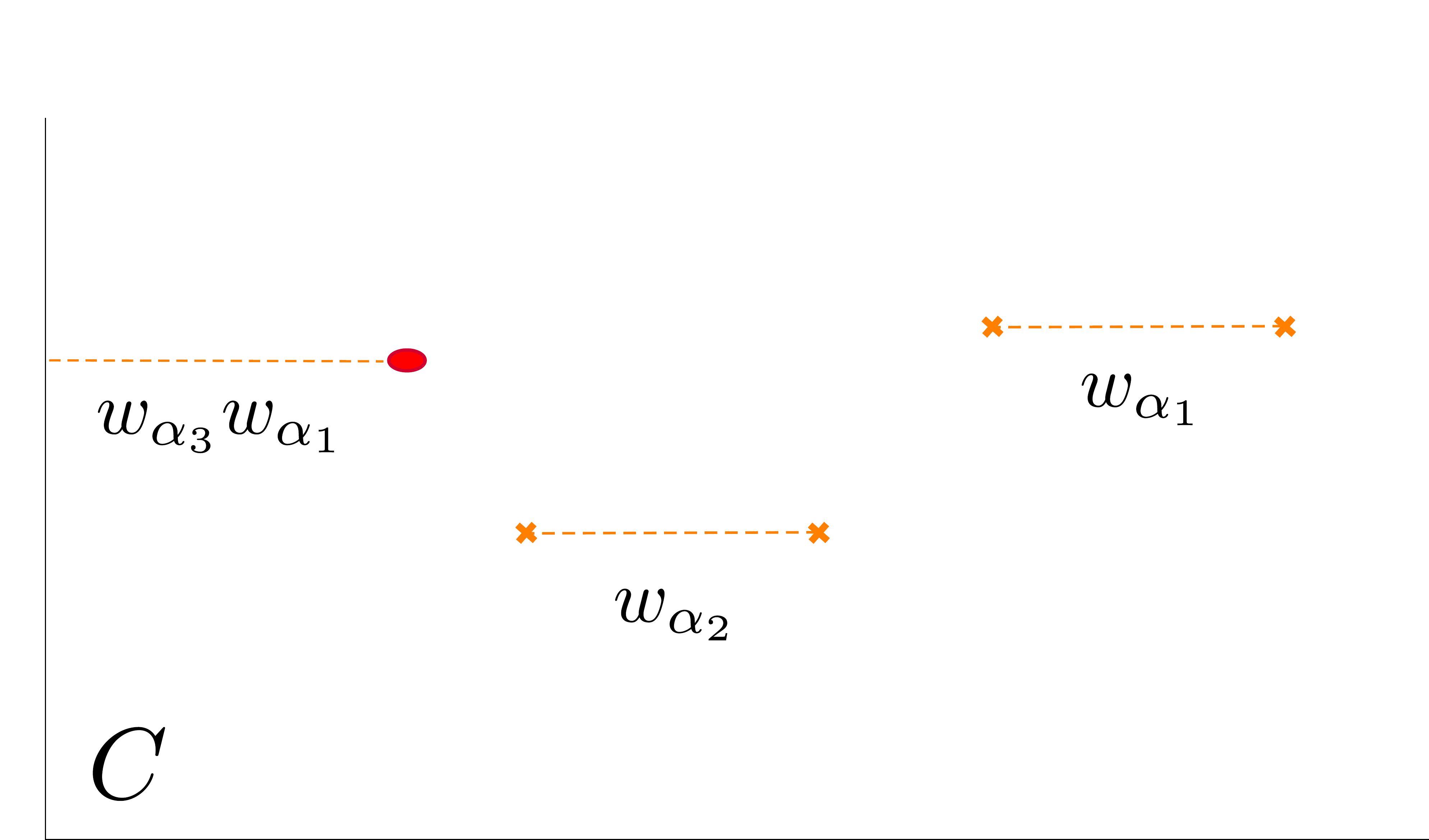}
\caption{With a choice of a standard trivialization, all square-root cuts are labeled by simple roots.
The irregular singularities shown here exhibits higher-order branching, e.g. where $\CH_{\alpha_{1}}\cap \CH_{\alpha_{3}}$ and $\langle\alpha_{i},\varphi\rangle=0$ for both $i=1,3$.}
\label{fig:canonical-trivialization}
\end{center}
\end{figure}

Finally, observe that connectedness of $\Sigma_{\rho}$, which is granted when $\rho$ is minuscule, puts further constraints on the types of cuts that must appear.
Pick any two points ${x_{\weight},\ x'_{\weight'} \in \Sigma_\rho}$, lying above $z$, $z'\in C$, such that $\pi(x_{\weight}) = z\,,\ \pi(x'_{\weight'})=z'$. 
In particular, $x_{\weight}$ is on the sheet labeled by the weight $\weight$, and similarly for $x'_{\weight'}$.
Since $\Sigma_\rho$ is connected there must be a path from $x_{v}$ to $x'_{\weight'}$: choose any such path $\gamma_{\weight\weight'}$ and project it down to ${C \setminus \{\text{branch points and punctures}\}}$. 
The path $\wp = \pi(\gamma_{\weight\weight'})$ may go through various cuts, starting from $z$ and ending at $z'$.
At each branch cut crossed by $\wp$, its preimage $\gamma_{\weight\weight'}$ crosses from one sheet $\weight_{i}$ to another sheet $w_{}(\weight_{i})$ for a certain $w\in W$. 
Taking into account all the branch cuts crosed by $\pi(\gamma_{\weight\weight'})$ we simply recover the relation $\weight' = w_{1}w_{2}\dots w_{k} (\weight)$, where $w_{i}$ are the Weyl transformations induced by crossing the cuts. 
Since $\rho$ is a Weyl orbit, the weights $\weight$, $\weight'$ will be related by a generic element of $W$ (for generic choices of the weights), and therefore
\begin{center}
	\emph{(Weyl elements associated with)  branch cuts must generate $W$}.
\end{center}
If all cuts are of the square-root-type, this means in particular that 
all simple roots must appear on the branch cuts. This latter requirement is lifted if there are higher-order branch cuts, such as those typically associated with irregular singularities. There is however a large class of interesting theories, namely when $C$ only involves regular punctures, which are subject to this property.

%%%%%%%%%%%%%%%%%%%%%%%%%%%%%%%%%%
\subsection{Physical charge lattice and Cameral covers}\label{subsec:prym}
%%%%%%%%%%%%%%%%%%%%%%%%%%%%%%%%%%

In this section we describe the construction of the physical charge lattice $\hat\Gamma$ of gauge and flavor charges in a 4d class $\mathcal{S}$ theory, through its relation to the homology lattice $H_{1}(\Sigma_{\rho},\IZ)$ of the spectral curve. \footnote{The notation of this paper differs slightly from \cite{Gaiotto:2012rg}. What we call $\hat\Gamma$ was denoted there as $\Gamma$. We will reserve the notation $\Gamma$ for another object that will be introduced later.}
By physical arguments, $\hat\Gamma$ is expected to be an extension of the lattice of gauge charges by flavor 
\be
	1\ \rightarrow \ \hat\Gamma_\textrm{f}\ \rightarrow \ \hat\Gamma\ \rightarrow \ \hat\Gamma_\textrm{g}\ \rightarrow \ 1
\ee
and a sub-quotient (the quotient of a sub-lattice) of $H_{1}(\Sigma_\rho,\IZ)$ \cite{Seiberg:1994aj, Gaiotto:2009hg}.
Much of this section is devoted to describing in some detail both the projection and the quotient, the content is somewhat technical but crucial for the definition of ADE spectral networks. 
Readers who are not interested in the details may safely skip this section on a first reading, 
as essential concepts will be captured by an example presented in Section \ref{subsec:explicit-trivialization}.

Recall that  $H_{1}(\Sigma_{\rho},\IZ)$ should really be thought of as a lattice fibration over $\CB$, with nontrivial monodromy around singular codimension-one loci. Physical charges are then sections of this fibration.
There is a distinguished sub-lattice which fibers trivially over $\CB$, which is the radical of the intersection pairing $\langle \cdot,\cdot \rangle$. It is generated by the punctures on $\Sigma_{\rho}$, and will be denoted henceforth $H_{1}^{\textrm{punc.}}(\Sigma_{\rho},\IZ)\simeq \IZ^{\oplus n_{p}}$.
The quotient of $H_{1}(\Sigma_{\rho},\IZ)/H^{\textrm{punc.}}_{1}(\Sigma_{\rho},\IZ)\simeq H_{1}(\overline\Sigma_{\rho},\IZ)$ can be thought as the homology lattice of $\Sigma_{\rho}$ after all punctures are filled in.
Here we describe a linear map $P$ from $H_{1}(\Sigma_{\rho},\IZ)$ to a sub-lattice, and define
\begin{align}
	\hat{\Gamma} \coloneqq P\big(H_{1}(\Sigma_{\rho},\IZ)\big)\,\big/\,\ker(Z),
\end{align}
where $Z$ is the central charge map 
\be
	Z:H_{1}(\Sigma_{\rho},\IZ)\to \IC\,,\qquad \eh{\gamma} \mapsto \oint_{\eh{\gamma}}\lambda\,. 
\ee
and the notation $\check{}$ is used to denote standard homology classes in $H_1(\Sigma,\IZ)$.
The kernel of $Z$ is a sublattice of $H_1(\Sigma_\rho,\IZ)$ (since $Z$ acts linearly) whose action naturally restricts to $ P(H_{1}(\Sigma_{\rho},\IZ))$.\footnote{The fact that $\hat\Gamma$ is a lattice follows from the fact that the quotient is by a normal subgroup (sublattice) of $ P(H_{1}(\Sigma_{\rho},\IZ))$, namely by $\ker(Z)\cap P(H_{1}(\Sigma_{\rho},\IZ))$.}
The lattice of flavor charges $\hat\Gamma_\textrm{f}$ will then descend from $H_{1}^{\textrm{punc.}}(\Sigma_{\rho},\IZ)$, while $\hat\Gamma_\textrm{g}$ will descend from $H_{1}^{}(\overline \Sigma_{\rho},\IZ)$. 
We will now give a description of how these two are constructed.

%%%%%%%%%%%%%%%%%%%%%%%%%%%%%%%%%%
\subsubsection*{Flavor charge lattice $\hat{\Gamma}_\textrm{f}$}
%%%%%%%%%%%%%%%%%%%%%%%%%%%%%%%%%%
Identifying the physical sub-lattice of flavor charges has already discussed in the literature, 
for instance in \cite{Gaiotto:2009hg}, where $\gA_{1}$ spectral networks were first introduced in terms of dual triangulations. In that setting we work with $\gA_{1}$ covers in the fundamental representation, therefore a regular puncture $p$ on $C$ has two lifts $p_{\pm}$ on the two sheets of $\Sigma$.
Taking $c_{\pm}$ to be counterclockwise circles around $p_\pm$, one choice of 
physical combination is the anti-invariant $c_{p}=c_{+} - c_{-}$,
while the orthogonal combination is $c_{n}=c_{+} + c_{-}$.
Denoting the residue of $\varphi(z)$ 
at the puncture  by $\diag(m/2,-m/2)$, 
we see that $Z_{c_{p}}=m$ while $Z_{c_{n}}=0$.

We propose the following generalization for spectral covers of minuscule representations of ADE Lie algebras.
For each regular puncture $p$ on $C$ we consider $\{c_{\weight}\}_{\weight\in\Lambda_{\rho}}$, counterclockwise cycles around the lifts of $p$ to sheets, and denote the lattice generated by these cycles as $\hat\Gamma^{(p)}$.
For each simple root $\alpha_{i}$ we define a linear combination
\be
	c_{\alpha_{i}} := \sum_{\weight\in\Lambda_{\rho}}(\weight\cdot\alpha_{i}) \, c_{\weight}\,,
\ee
the collection of these spans a sub-lattice of $\hat\Gamma^{(p)}$.
In this way we associate a sub-lattice to each puncture of $C$, the sum of which defines $P(H_1^\textrm{punc.}(\Sigma_\rho,\IZ))$, then taking a quotient by $\ker Z$ produces $\hat\Gamma_\textrm{f}\subset H_1^{\textrm{punc.}}$. 

Together with the definition of the sub-lattice, let us consider an explicit operator $P_\textrm{f}$ on $H_1^\textrm{punc.}(\Sigma_\rho,\IZ)$. To this end, let us once again focus on a single puncture, and define the action of $P_\textrm{f}$ in the following way 
\be
	P_\textrm{f}:\,c_{\weight} \,\mapsto \,\sum_{\alpha_{i},\alpha_j}C^{-1}_{\alpha_i, \alpha_j}\,(\weight\cdot\alpha_j)\,c_{\alpha_i}\,,
\ee
where the sum runs over the simple roots of $\fg$, and $C^{-1}_{\alpha_i,\alpha_j}$ are matrix elements of the inverse of the Cartan matrix.

The image of $P_\textrm{f}$ is the sub-lattice $P(H_1^\textrm{punc.}(\Sigma_\rho,\IZ))$, but $P_\textrm{f}$ is not quite a projector because it is not idempotent but satisfies 
\be\
	P_\textrm{f}^2 = k_\rho\, P_\textrm{f},
\ee
where $k_\rho$ is a certain integer which depends on the choice of $\rho$.
More precisely, when $\rho$ is a minuscule representation, $\tilde C_{\alpha_i,\alpha_j} =  \sum_\weight (\alpha_i\cdot\weight)(\alpha_j\cdot\weight)$ is a multiple of the Cartan matrix, and $k_\rho$ is defined as the multiplicity constant in 
\be\label{eq:k-rho}
	\tilde C_{\alpha_i,\alpha_j} = k_\rho C_{\alpha_i,\alpha_j}\,.
\ee
A simple proof of this fact is given in Appendix \ref{sec:k_rho_cartan}, together with an interpretation of $k_\rho$ that will be used below in the construction of spectral networks.
Note that in defining $P_\textrm{f}$ we have some freedom to rescale it by an overall number. In doing so, one must generally choose among idempotency, fixing a certain normalization for $k_\rho$, or having integer entries in $P_\textrm{f}$. 
For reasons that will become clear in the rest of this section, in our setting it is natural to leave the normalization of $P_\textrm{f}$ as currently defined. 

Although the orthogonal complement $\ker(P_\textrm{f})\cap \hat{\Gamma}^{(p)} \subseteq \hat{\Gamma}^{(p)}$ happens to be a sub-lattice of $\ker(Z)$, in general it does not span the whole kernel, which is the case when one has a nonabelian flavor symmetry at the puncture with two or more eigenvalues of the residue of {$\varphi$} becoming equal, for example. Therefore it is meaningful to take a quotient by $\ker(Z)$ after the projection.

One important caveat in this construction is that it only applies to regular punctures. A generalization to irregular ones is certainly desirable but beyond the scope of this work. Nevertheless, we will study cases involving irregular punctures below in Section \ref{sec:examples}, and deal with them case-by-case.

%%%%%%%%%%%%%%%%%%%%%%%%%%%%%%%%%%
\subsubsection*{Gauge charge lattice $\hat{\Gamma}_\textrm{g}$ and the distinguished Prym of $\overline \Sigma_{\rho}$}
%%%%%%%%%%%%%%%%%%%%%%%%%%%%%%%%%%

We now turn to the description of $P(H_{1}(\overline \Sigma_{\rho},\IZ))$, where $\overline \Sigma_{\rho}$ is the normalization of a spectral cover $\Sigma_{\rho}$ and therefore is a compact curve. On the one hand, the 4d physics is expected to be independent of the choice of $\rho$. In particular 
the rank of the lattice of gauge charges is fixed by the complex dimension of the Coulomb branch
\be
	\rank(\hat\Gamma_\textrm{g}) = 2\,\dim_{\IC}\CB = 2r\,.
\ee
On the other hand, the first homology lattice of $\overline\Sigma_{\rho}$ depends on the choice of $\rho$ through the ramification structure of the covering map $\pi:\overline\Sigma_{\rho}\to C$. Recall that a point in the base of the Hitchin system $u\in\CB$ fixes the complex geometry of $\overline\Sigma_{\rho}(u)$, 
while the fiber $\theta\in\CM_{u}$ parametrizes holomorphic line bundles on $\overline\Sigma_{\rho}(u)$. 
The space of all holomorphic line bundles
on $\overline\Sigma_\rho$(u) is $J\big(\overline\Sigma_{\rho}(u)\big)$, whose complex dimension $g$, the genus of $\overline\Sigma_\rho$, is in general greater than that of the Hitchin fiber,
\be
	\dim_{\IC}\CM_u = r \leq g = \dim_{\IC}J\big(\overline\Sigma_\rho(u)\big)\,,
\ee 
and the Hitchin fiber $\CM_{u}$ maps to a sub-variety of the Jacobian.\footnote{%
The space of degree-0 line bundles on $\overline\Sigma_{\rho}(u)$ can be identified with the space of degree-0 divisors on $\overline\Sigma_{\rho}(u)$ up to linear equivalence, $\textrm{Pic}^{0}(\overline\Sigma_{\rho}(u))$. This in turn is identified by Abel's theorem (combined with Jacobi inversion) with the Jacobian variety of $\overline\Sigma_{\rho}(u)$, $J(\overline\Sigma_{\rho}(u))\simeq\IC^g/\Lambda$, where $\Lambda$ is the period lattice generated by a basis of holomorphic one-forms, i.e.\ a basis for $H^{0}(\overline\Sigma_{\rho}(u),K_{\overline\Sigma})$.
Denoting the basis of differentials by $\omega_{i}$,  
for each cycle there is a vector in $\IC^g$ defined by $\zeta_{\eh{\gamma}} = \left(\int_{\eh{\gamma}}\omega_{1},\dots, \int_{\eh{\gamma}}\omega_{g}\right)$.
Given a Darboux basis $a_{i},b_{i},\,1\leq i\leq g$ of 1-cycles on $\Sigma$, the lattice $\Lambda$ is generated by $\zeta_{a_{i}}$ and $\zeta_{b_{i}}$, and can be shown to be non-degenerate. 
The quotient $\IC^{g}/\Lambda$ is a $2g$-torus, the Hitchin fiber on the other hand is a $2r$-torus and maps to a distinguished sub-variety in $\IC^g/\Lambda$.
} 

In particular, this means that for any choice of $\rho$ the Jacobian  $J(\overline\Sigma_{\rho})$ always contains a distinguished sub-variety, common to all representations, which is identified with the Hitchin fiber $\CM_u$ in the sense that they both parametrize holomorphic line bundles on $\overline\Sigma_\rho(u)$.
The problem of identifying the common sub-variety of $J(\overline\Sigma_{\rho})$ has been studied in the literature on integrable systems, starting with the seminal work of Adler \& van Moerbeke \cite{Adler1980267,ADLER1980318}. 
Different approaches were developed for Toda systems by several authors \cite{mcdaniel1988, Kanev:1989, Kanev2, mcdaniel1992, mcdaniel1997,mcdaniel1998}, and were further generalized by Donagi in \cite{Donagi:1993}.
The latter approach consists of identifying a distinguished sub-variety of $J(\overline\Sigma_{\rho})$ by realizing (a desingularization of) the curve as the quotient of a {universal} object known as the \emph{Cameral cover} $\tilde\Sigma$.
This is a $W$-Galois cover, whose sheets are identified with different Weyl chambers, and it carries a natural $W$-action.
Roughly speaking, the spectral cover in representation $\rho$ can be obtained from $\tilde\Sigma$ as a quotient $\overline\Sigma_{\rho}\sim \tilde\Sigma / W_{P}$ by the stabilizer $W_{P}\subset W$ of the highest weight of $\rho$. 

The $W$ action on $\tilde\Sigma$ induces a corresponding action on its Jacobian via the regular representation of $W$. This action is in fact reducible, and decomposes $J(\tilde\Sigma)$ into sub-varieties corresponding to irreducible representations of $W$. 
The quotient curve $\overline\Sigma_{\rho}$ does not carry a $W$ action, as neither does $H_{1}(\overline\Sigma_{\rho},\IZ)$ nor $J(\overline\Sigma_{\rho})$.
However, \cite{Donagi:1993} shows that there is always a sub-variety of $J(\overline\Sigma_{\rho})$, 
which descends from a subvariety of $J(\tilde\Sigma)$ associated with the reflection representation of the $W$-action on $\tilde\Sigma$.
This sub-variety goes under the name of \emph{distinguished Prym}, 
and is the one which is identified with the Hitchin fiber $\CM_{u}$.
Following \cite{Martinec:1995by}, we will take the lattice of physical gauge charges $\hat{\Gamma}_\textrm{g}$ to be \emph{defined} as the sub-lattice of $H_1(\overline\Sigma_\rho,\IZ)$ that generates the distinguished Prym.

%%%%%%%%%%%%%%%%%%%%%%%%%%%%%%%%%%
\subsubsection*{Generators of $\hat{\Gamma}_\textrm{g}$}
%%%%%%%%%%%%%%%%%%%%%%%%%%%%%%%%%%

From a practical viewpoint, we wish to identify a sub-lattice $\hat\Gamma_\textrm{g}\subseteq H_{1}(\overline\Sigma_{\rho},\IZ)$: 
it must be a symplectic, rank $2r$ lattice, whose periods characterize the distinguished Prym.\footnote{More precisely, to each generator $\gamma$ for $\hat{\Gamma}_\textrm{g}$, one can associate a vector of periods in $\IC^g$, computed in a basis of holomorphic differentials on $\overline\Sigma_{\rho}$. The collection of period vectors for all generators of $\hat{\Gamma}_\textrm{g}$ then characterizes the distinguished Prym.}
A few explicit examples of how this task is carried out are available in the literature \cite{Martinec:1995by, Hollowood:1997pp}. While all these examples focus on Toda systems, which may be viewed as Hitchin systems for $C={\IC^\times}$, their characterization of the distinguished sub-lattice is to a large extent local, in the sense that the global topology of $C$ plays a secondary role. 
Building on this observation, together with previously discussed facts about trivializations of spectral covers, we can extrapolate the construction of  \cite{Martinec:1995by} to other types of Riemann surfaces. 
 
As discussed in Section \ref{subsec:trivialization}, trivializations of a generic spectral cover can be brought into a standard form: square-root branch cuts of simple-root type, and higher-order cuts (from irregular singularities) corresponding to edges of the fundamental Weyl chamber $\CC_0$.
Consider a square-root cut with sheet monodromy $w_\alpha$, the Weyl reflection by the root $\alpha$, as depicted in Figure \ref{fig:single-cut}. 
Above the cut, several sheets are glued pairwise: for any pair of weights related by the Weyl reflection $\weight_{j} = w_{\alpha}\cdot\weight_{i}$, the corresponding sheets will be glued, while sheets corresponding to weights $\weight_{k}$ fixed by $w_{\alpha}$ do not ramify.
\begin{figure}[t]
\begin{center}
\includegraphics[width=0.3\textwidth]{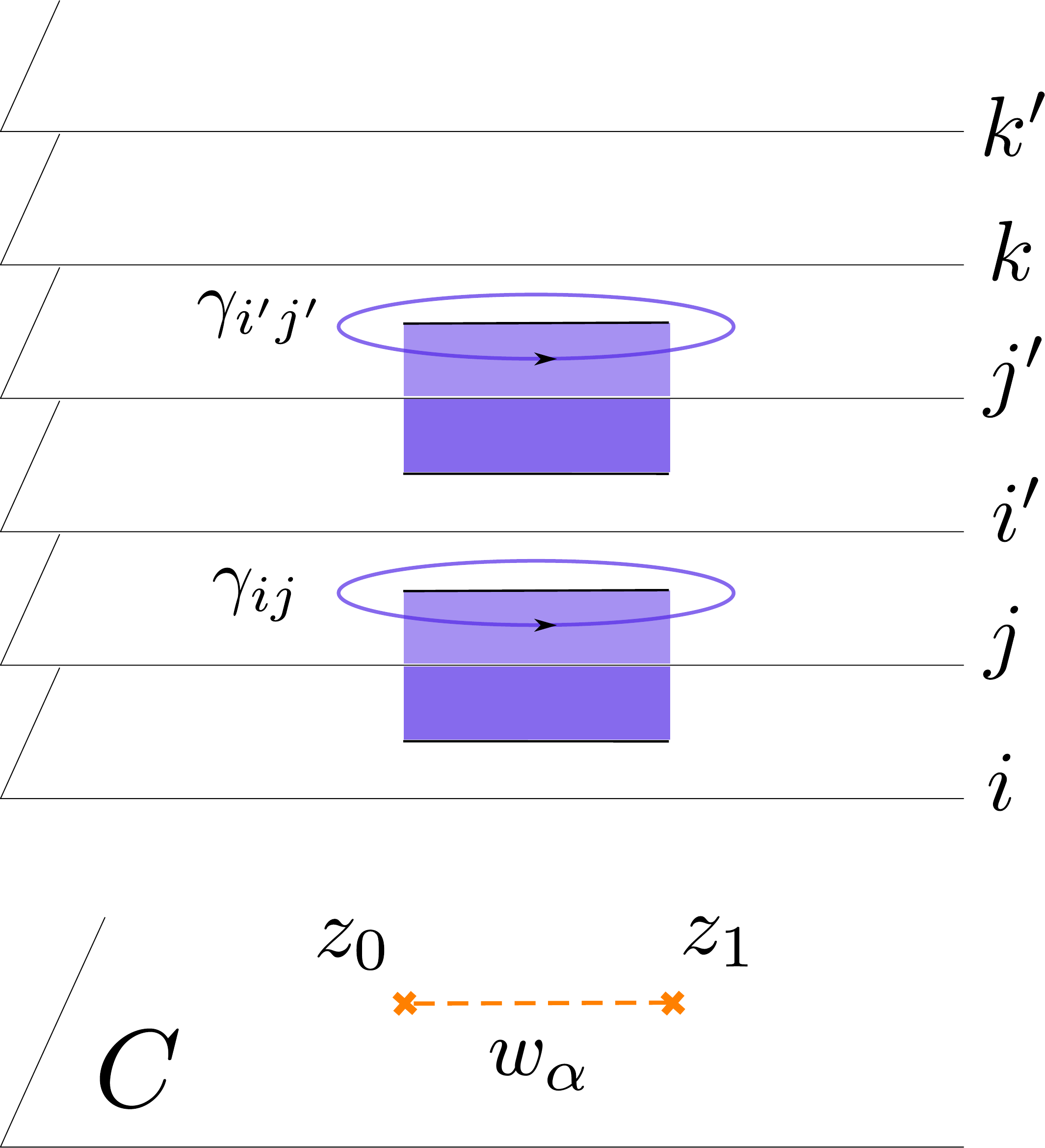}
\caption{Structure of the cover above a square root cut. Sheet pairs $(i,j)$, $(i',j')$ such that $\weight_j-\weight_i \sim \weight_{j'} - \weight_{i'} \sim \alpha$ are permuted by the sheet monodromy. Ramifying sheets are glued together above the branch cut, as indicated in blue. 
Other sheets $k,k'$ such that $\weight_k, \weight_{k'} \perp\alpha$ instead do not ramify. 
To each pair of ramifying sheets $(i,j)$ we assign a homology cycle $\eh{\gamma}_{ij}$. }
\label{fig:single-cut}
\end{center}
\end{figure}
There is a natural sub-lattice $\hat\Gamma^{(\alpha)}\subset H_1(\overline\Sigma_\rho,\IZ)$ associated to the cut that is generated by cycles $\gamma_{ij}$ wrapping around the gluing fixtures between sheets $x_i$ and $x_j$, which is illustrated in Figure \ref{fig:single-cut}.
The orientation of $\eh{\gamma}_{ij}$ is fixed to be counter-clockwise on the sheet $x_j$, which is the sheet whose corresponding weight  has positive Killing pairing with the root that is associated with the cut, i.e.\ $\weight_{j}\cdot\alpha = -\weight_{i}\cdot\alpha > 0$. 

Note that the central charges of the $\eh{\gamma}_{ij}$ for all pairs $(i, j)$ are equal,\footnote{Strictly speaking, the central charge is defined on $H_1(\Sigma,\IZ)$ but not on $H_1(\bar\Sigma,\IZ)$. Here it is understood that the statement holds after restoring the punctures on $\Sigma$, the contours are chosen ``close enough" to the plumbing fixtures that they don't include any puncture. Integration from $z_{0}$ to $z_{1}$ is understood to run below the cut.}
\be\label{eq:ij-period}
	Z_{\eh{\gamma}_{ij}} = \oint_{\eh{\gamma}_{ij}}\lambda_{} = \int_{z_{0}}^{z_{1}} (x_{{j}} - x_{{i}})\,dz =  \int_{z_{0}}^{z_{1}} \langle \weight_{j} - \weight_{i}, \varphi(z) \rangle = \int_{z_{0}}^{z_{1}} \langle \alpha, \varphi(z) \rangle \equiv Z{}^{(\alpha)}.
\ee
$Z_{\eh{\gamma}_{ij}}$ depends only on $\alpha$, and depends on it linearly. 
In $\hat\Gamma^{(\alpha)}$ there is one generator of the distinguished Prym \footnote{As noted by \cite{Martinec:1995by, Hollowood:1997pp}, the rationale behind this definition of $\eh{\gamma}_\alpha$ is that it manifestly descends from the part of $J(\Sigma)$ transforming in the reflection representation of $W$, given its linear dependence on $\alpha$.}
\be
	\eh{\gamma}_\alpha := \sum_{(i, j)}\eh{\gamma}_{ij} \,,
\ee
with central charge 
\be
	Z_{\eh{\gamma}_\alpha}  = \tilde k_{\rho}\,Z^{{(\alpha)}}\,
\ee
where $\tilde k_{\rho}$ is the number of pairs of weights such that $\nu_i-\nu_j = n \alpha$ for $n\in\IN$. $\tilde k_{\rho}$ does not depend on a particular root $\alpha$.
In fact, when $\rho$ is minuscule, we have 
\be\label{eq:k-rho-tilde}
	\tilde k_\rho = k_\rho,
\ee
with $k_\rho$ defined in (\ref{eq:k-rho}), a proof of this can be found in Appendix \ref{sec:k_rho_cartan}. Consider an operator $P_\textrm{g}\in \mathrm{End}(\hat\Gamma^{(\alpha)})$ acting as
\be\label{eq:projection-gauge}
	P_\textrm{g} : \eh{\gamma}_{ij}\,\mapsto \, \eh{\gamma}_\alpha\,
\ee
on all generators of $\hat\Gamma^{(\alpha)}$. 
The image of this operator is the rank-1 sub-lattice generated by $\eh{\gamma}_\alpha$, but this is not a projection because it is not idempotent, $\hat P^2 = \tilde k_\rho \hat P$.

There is a manifest property of $P_\textrm{g}$ that will be important in the rest of the paper: its kernel is a sub-lattice of $\ker(Z)$, in an appropriate sense.\footnote{The map $Z$ is defined on $H_1(\Sigma_\rho,\IZ)$, while $P_\textrm{g}$ is defined on $\hat\Gamma^{(\alpha)}\subset H_1(\overline\Sigma_\rho,\IZ)$. As we will explain shortly, we can choose an embedding of $\hat\Gamma^{(\alpha)}$ in $H_1(\Sigma_\rho,\IZ)$ by choosing a ``splitting'' of $H_1(\Sigma_\rho,\IZ)$. This choice however can  be made only locally on $\CB$.}
This can be easily seen by considering $P_\textrm{g}$ acting on $\hat\Gamma^{(\alpha)}$ in the basis $\{(1,0,\dots,0),\ \ldots,\ (0,\dots,0,1)\}$ of generators $\{\eh{\gamma}_{ij}\}$. Then $\ker(P_\textrm{g})$ is the sub-lattice of vectors $(a_1,\dots, a_{k_\rho}),\, a_i\in\IZ$ such that $\sum_i a_i = 0$, which clearly have vanishing central charge.

Applying this construction to all the square-root branch cuts of a given standard trivialization, we produce an isotropic sub-lattice of $\hat{\Gamma}_\textrm{g}$, and we make an assumption that the construction will give us a a Lagrangian sub-lattice of $\hat{\Gamma}_\textrm{g}$. 
As a matter of fact, this prescription produces a Lagrangian sub-lattice of $\hat{\Gamma}_\textrm{g}$ in all the examples we consider. 
More generally, however, it is not clear if this will be true for every case, as it is not obvious that our prescription would capture all of the Lagrangian sub-lattice of $\hat{\Gamma}_\textrm{g}$ characterized by Donagi's approach \cite{Donagi:1993}. While this question is not crucial for the construction of spectral networks, which is our main goal in this paper, it would nevertheless be important to clarify this point.

When maximality holds, $B$-cycles generating the complement of $A$-cycles in $\hat{\Gamma}_\textrm{g}$ are then obtained by choosing elements of $H_1(\overline\Sigma_\rho,\IZ)$ with suitable intersection pairings $\langle A_i, B_j\rangle =\tilde k_\rho\delta_{ij}$. 
We may choose such B-cycles, but they will generate not $\hat{\Gamma}_\mathrm{g}$ but a larger lattice that contains it. For example, even in the pure $\SU(2)$ case, it's $A$ and $-A + 2 B$ that generate $\hat{\Gamma}_\mathrm{g}$.
We will assume the existence of an operator $P_\textrm{g}$ defined on $H_1(\Sigma_\rho,\IZ)$ that satisfies $P_\textrm{g}^2 = k_\rho P_\textrm{g}$ and $\ker(P_\textrm{g})\subseteq \ker(Z)$ in an appropriate sense.
In support of this assumption, we note that while the above construction is far from being fully general, for all applicable cases we find such $P_\textrm{g}$ that maps $H_1(\Sigma_\rho,\IZ)$ to the distinguished Prym of \cite{Donagi:1993}, which is constructed on a fully general framework. The nontrivial examples worked out in Section \ref{sec:examples} offer further support to the validity of this assumption.

Finally, it should be noted that the intersection pairing on $H_1(\Sigma_\rho,\IZ)$ descends onto $\overline\Sigma_\rho$ and can be naturally restricted to $\hat{\Gamma}_\textrm{g}$.

%%%%%%%%%%%%%%%%%%%%%%%%%%%%%%%%%%
\subsubsection*{The full lattice of physical charges $\hat{\Gamma}$}
%%%%%%%%%%%%%%%%%%%%%%%%%%%%%%%%%%

So far we have treated pure flavor charges and pure gauge charges separately, defining sub-lattices $\hat{\Gamma}_\textrm{f}$ and $\hat{\Gamma}_\textrm{g}$.
To complete our description of the lattice of 4d charges $\hat{\Gamma}$ as a sub-quotient of $H_1(\Sigma_\rho, \IZ)$, we still have to explain how these are pieced together to form $\hat\Gamma$.
For simplicity, we choose to work \emph{locally} on $\CB^*=\CB\setminus \CB^\textrm{sing}$, i.e.\ we will not consider global issues due to monodromies of charges on the Coulomb branch.
Choosing to work on some contractible patch of $\CB^*$ allows us to treat $\hat\Gamma$ as a lattice, rather than a lattice fibration, and in particular it allows to choose a \emph{splitting} of the following short exact sequence
\be
	1\ \mathop{\rightarrow}^{} \ H_1^{\textrm{punc}}(\Sigma_\rho,\IZ) \ \mathop{\rightarrow}^{\iota} \ H_1(\Sigma_\rho,\IZ) \ \mathop{\rightarrow}^{\pi} \ H_1(\overline \Sigma_\rho,\IZ)\ \rightarrow \ 1
\ee
where $\iota$ is the natural inclusion map, and $\pi$ is the map induced by deleting punctures on $\Sigma_\rho$.
A splitting is a section $s: H_1(\overline \Sigma_\rho,\IZ)\to H_1(\Sigma_\rho,\IZ)$, which is also a homomorphism. Its practical purpose is that it allows\footnote{For general group extensions, existence of a splitting is not guaranteed. However, in the case at hand it is simple to show that there is always one.} us to write each $\eh{\gamma}\in H_1(\Sigma_\rho,\IZ)$ uniquely as
\be
	\eh{\gamma} = \iota(\eh{\gamma}_f) + s(\eh{\gamma}_g)\,,
\ee
for some $(\eh{\gamma}_f,\eh{\gamma}_g)\in H_1^{\textrm{punc}}(\Sigma_\rho,\IZ)\times H_1(\overline \Sigma_\rho,\IZ)$.
After choosing a splitting, we define $P\in \mathrm{End}(H_1(\Sigma_\rho,\IZ))$ as
\be
	P(\eh{\gamma}) := \iota(P_\textrm{f} (\eh{\gamma}_f)) + s(P_\textrm{g} (\eh{\gamma}_g) )\,.
\ee
It is easy to check that $P$ is a linear operator and that $P^2 = k_\rho \, P$. 
From the definitions of $P_\textrm{g}$ and $P_\textrm{f}$ given above, it follows that $\ker(P) \subseteq \ker(Z)$ on $H_1(\Sigma_\rho,\IZ)$.
Then we claim that there is an isomorphism
\be\label{eq:sub-quotient-iso}
	H_1(\Sigma_\rho,\IZ) \,/\, \ker(Z) \simeq \hat\Gamma,
\ee
that is, for every equivalence class $\ehz{\gamma}$ on the LHS there is a natural representative on the RHS.\footnote{More precisely, given a $\ehz{\gamma}$ on the LHS, its representative on the RHS (which also includes a quotient by $\ker(Z)$) will fall in the equivalence class of $k_\rho\,\ehz{\gamma}$. In other words, its ``central charge'', the period of $\lambda$ over the representative, is $k_\rho$ times that of $\ehz{\gamma}$. We will illustrate this point below with examples.}

Having defined physical charges, it remains to identify a suitable definition of the DSZ pairing. 
Let us denote the intersection pairing on $H_1(\Sigma_\rho,\IZ)$ by $\langle \cdot, \cdot \rangle$. 
Its entries must be homology cycles, in particular the intersection pairing is not well-defined on $\ker(Z)$-equivalence classes.
The physical DSZ pairing is denoted by $\langle \cdot , \cdot \rangle_\text{DSZ}$, and its entries are physical 4d charges: they could be either elements of $H_1(\Sigma_\rho,\IZ)/\ker(Z)$ or, by a mild abuse of notation, of $\hat\Gamma$, because the two are isomorphic as stated just above in (\ref{eq:sub-quotient-iso}).
The DSZ pairing can be defined in terms of the intersection pairing via the following relation. 
Given a physical charge ${\gamma}\in\hat\Gamma$, choose any representative in $\mathrm{Im}(P) \subseteq H_1(\Sigma,\IZ)$. Call this $\eh{\gamma}$, then 
\begin{align}
	\langle \eh{\gamma}_1, \eh{\gamma}_2 \rangle = k_\rho \langle \gamma_1, \gamma_2 \rangle_\text{DSZ}.
	\label{eq:intersection-DSZ}
\end{align}
The choice of representative may not be unique because $\ker(Z) \supseteq \ker(P)$, so there may be two representatives for $\gamma$, both in $\mathrm{Im}(P)$, which differ by $\eh{\gamma}_0\in \ker(Z) / \ker(P)$.
However, the DSZ pairing is well-defined provided that $\ker(Z)\cap \mathrm{Im}(s) = \ker(P)\cap \mathrm{Im}(s)$, since the intersection pairing is only affected by the gauge charge content, not by flavor charges. We don't have a rigorous proof that this condition is generally satisfied, but we will take it as a working assumption.

Finally, we hasten to stress that our characterization of $\hat\Gamma$ and the definition of $P$ are strictly local on $\CB^*$. The global extension is an interesting problem which we leave to a future work.

%%%%%%%%%%%%%%%%%%%%%%%%%%%%%%%%%%
\subsection{Example: \texorpdfstring{$\SO(6)$}{SO(6)} SYM}\label{subsec:explicit-trivialization}
%%%%%%%%%%%%%%%%%%%%%%%%%%%%%%%%%%
Let us illustrate the construction detailed above with an example. For simplicity we choose a theory with no flavor symmetry, the Hitchin system of 4d $\CN=2$ pure $\SO(6)$ gauge theory with a spectral cover in the vector representation.
Conventions for $\gD_{3}$ are collected in Appendix \ref{app:so6}.

The base curve is $C=\IC^{\times}$, and the equation of the spectral cover $\Sigma_\rho$ corresponding to the Weyl orbit $W\cdot\omega_{1}$ of the first fundamental weight is
\be
	x^{6} + u_{2} x^{4} + \left(z+u_{4}+\frac{1}{z}\right) x^{2} +u_{3}^{2} = 0.
\ee
Setting $y=x^{2}$ we find the discriminant of $y^{3}+ a y^{2}+ by + c$ to be
\be
	\delta = a^{2}b^{2} -4b^{3}+18 abc - 27 c^{2},
\ee
where $b$ is the only parameter carrying $z$-dependence,
\be
	b = z+u_{4}+\frac{1}{z}.
\ee
$\delta$ is cubic in $b$, hence there will be six branch points on $C$. Furthermore, for general values of $\{u_i\}$, they are first order zeros of $\delta$, therefore all six branch points are of square-root type, with two values of  $y$ colliding. 
Switching back to $x$-coordinates, this means there are four sheets colliding pairwise above each branch point.
The number of ramification points above a square root branch point is thus $k_\rho=2$.

At $z=0,\ \infty$ we have irregular singularities, around which the asymptotic forms of the spectral cover are 
\be
\begin{split}
	x^{2}(x^{4} +z) \sim 0 & \qquad \textrm{as}\ z\to\infty, \\
	x^{2}(x^{4} + z^{-1}) \sim  0 & \qquad \textrm{as}\ z \to 0 \,.
\end{split}
\ee
Around each of them there is a higher-order branch cut with a partition structure $(4)(2)$, meaning that four sheets will be permuted among themselves disjointly from the other two sheets.
From the ramification structure, the genus of the cover is therefore obtained to be  $g_{\Sigma} = 5$, and the corresponding homology lattice is rank $10$, which is larger than $2\cdot \rank (D_3) = 6$, the expected rank of $\hat{\Gamma}_\textrm{g}$.

The curve can be easily trivialized by means of simple numerics. In Appendix \ref{app:canonical-example} we give a detailed description of how such a trivialization is obtained. We present a schematic result in Figure \ref{fig:trivialization-nice}.
\begin{figure}[t]
\begin{center}
\includegraphics[width=0.75\textwidth]{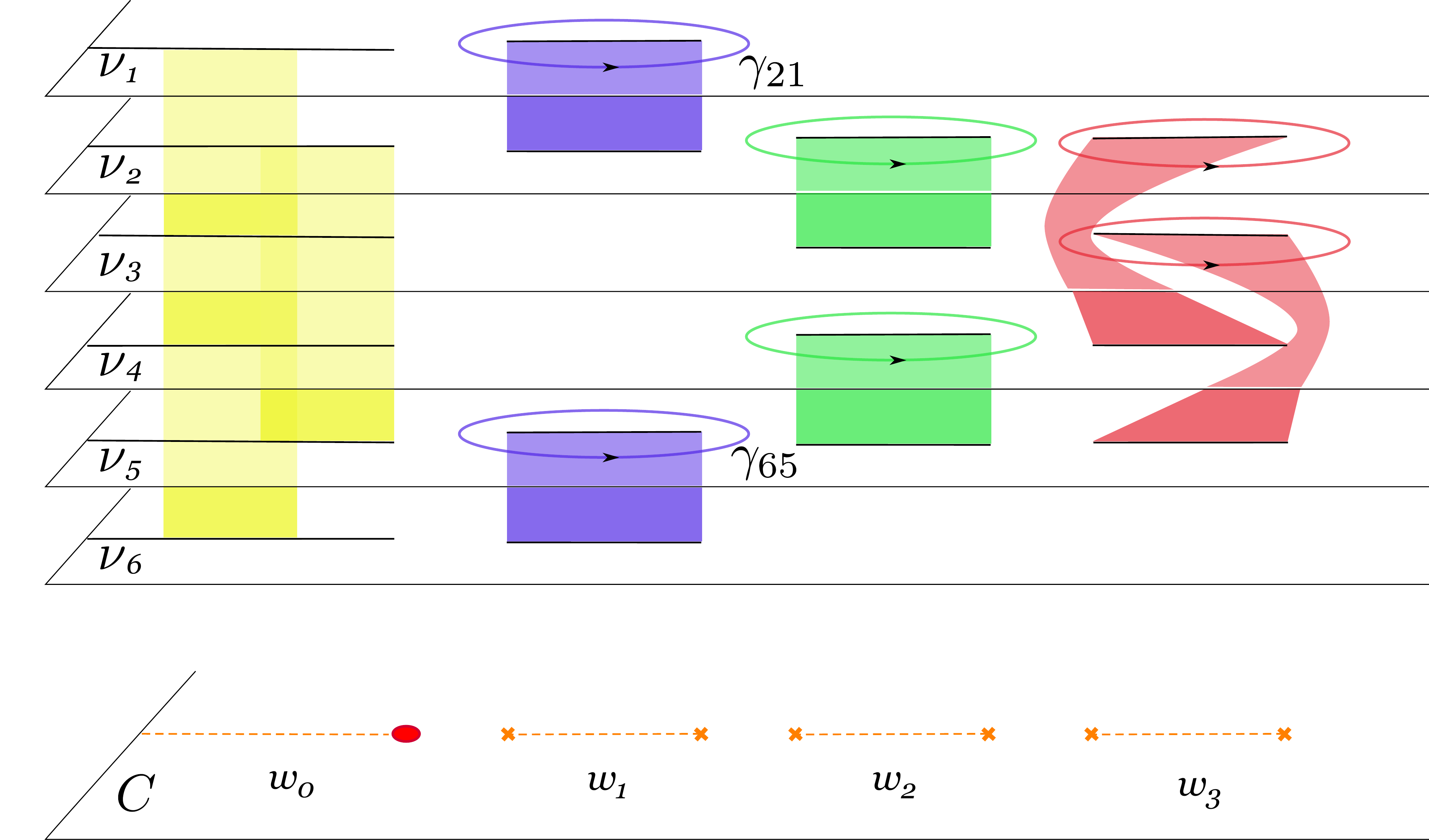}
\caption{A standard choice of trivialization for the spectral cover of $\SO(6)$ SYM theory in the vector representation.}
\label{fig:trivialization-nice}
\end{center}
\end{figure}
There are three branch cuts of square-root type, with sheet monodromy given by simple Weyl reflections $w_{i}$ $(i=1,2,3)$ corresponding to the simple roots $\alpha_{i}$.
The other cut has higher-degree branching and extends to infinity, its counter-clockwise monodromy is a Coxeter element $w_{0}=w_{2}w_{1}w_{3}$. 
The sheets are permuted by $w_{0}$ in the same way as the weights, i.e.\ $(\weight_1, \weight_3, \weight_6, \weight_4)(\weight_2,\weight_5)$. Thus the cut with monodromy $w_0$ can be associated with the vertex of the fundamental Weyl chamber, the origin of $\ft^{*}$. 

Remember that we find generators of $\hat{\Gamma}_\mathrm{g}$ by identifying $\hat{\Gamma}^{(\alpha)}$ for each square-root branch cut. Figure \ref{fig:trivialization-nice} shows $\{\eh{\gamma}_{ij}\}$, generators of $\hat{\Gamma}^{(\alpha_k)}$, above each square-root branch cut. For all three cuts $k=1$, $2$, $3$ we have $\rank (\hat{\Gamma}^{(\alpha_{k})})= k_{\rho} = 2$. The projection $P$ then singles out a combination $A_{1} = \eh{\gamma}_{21}+\eh{\gamma}_{65}$ from $\hat{\Gamma}^{(\alpha_{k})}$ and kills $\eh{\gamma}_{21}-\eh{\gamma}_{65}$, and similarly pick out $A_{2}$, $A_{3}$ for the other two cuts. 
The three distinguished cycles $A_{1}$, $A_{2}$, $A_{3}$ generate a rank-$3$ Lagrangian sub-lattice $\hat{\Gamma}_{A} \subset \hat{\Gamma}_\mathrm{g}$. Their dual cycles $\{B_i\}$ also admit a simple description in this case \cite{Martinec:1995by}. Noting that  $(w_0 w_i)^3=1$ for $i=1$, $2$, $3$, we can choose a path connecting two branch points of type $\alpha_i$, winding three times around the cut $w_0$ and through the cut $w_{\alpha_i}$, this then lifts to a closed cycle $B_i$ that satisfies $\langle A_i,B_j\rangle = k_\rho \delta_{ij} = 2\delta_{ij}$.

Because there is no flavor symmetry, we expect $\{\ephz{A_i}, \ephz{-A_i + 2 B_i}\}$ to generate the 4d charge lattice $\hat{\Gamma}$. The 4d charges are related to $\{ \ehz{\gamma}_{\mathrm{e}, i}, \ehz{\gamma}_{\mathrm{m}, i} \}$, where
\begin{align}
	\ehz{\gamma}_{\mathrm{e}, i} = [\eh{\gamma}_{21}]_{\ker(Z)} \in H_1(\Sigma_\rho, \mathbb{Z}) \big/ \ker(Z),
\end{align}
and similarly for the other charges. 
The isomorphism (\ref{eq:sub-quotient-iso}) provides complementary descriptions of the physical charges.
On the one hand, the physical central charge of a 4d charge $\ephz{\gamma}_{\mathrm{e}, i}\in\hat\Gamma$ is given by evaluating $Z$ on its representative $\ehz{\gamma}_{\mathrm{e}, i} \in H_1(\Sigma_\rho, \mathbb{Z})/\ker(Z)$, i.e.\ $Z(\ehz{\gamma}_{\mathrm{e}, i}) \equiv  Z({\eh{\gamma}_{21}})$. 
On the other hand, to get the physical DSZ pairing, one employs the intersection pairing of the distinguished representative in $\hat\Gamma$ which lies in the  sub-lattice $P\big( H_1(\Sigma_\rho, \mathbb{Z}) \big)$, i.e.\ $A_i$.\footnote{The careful reader will notice that one could as well get the physical central charge from the representative in $\hat\Gamma$, after a suitable rescaling by $k_\rho$. In fact, a better motivation for why we need to consider $H_1 / \ker(Z)$ will become apparent below, when studying 2d-4d wall-crossing.}

%%%%%%%%%%%%%%%%%%%%%%%%%%%%%%%%%%
\section{Spectral networks for minuscule covers}\label{sec:spectral-networks}
%%%%%%%%%%%%%%%%%%%%%%%%%%%%%%%%%%

A spectral network $\CW$ consists of two pieces of data: \emph{geometric} data encoded into a network of $\CS$-walls on $C$, and \emph{combinatorial} topological data individually attached to each $\CS$-wall, called soliton data.
The geometric data is obtained as a natural generalization of that in \cite{Gaiotto:2012rg} by rephrasing the latter in a Lie-algebraic language.
The generalization of soliton data, on the other hand, is much less trivial and hinges on specific properties of the spectral cover $\Sigma_{\rho}$ to which $\CW$ is associated.
The combinatorics of minuscule representations turns out to be particularly tractable, which is another reason for focussing on these.

This section is somewhat long and technical. Let us give an overview of how it is developed, and summarize the main points.
In Section \ref{sec:S-walls_and_solitons} we introduce the definition of $\CS$-walls and their soliton data, in particular we discuss the topological classification of soliton charges and its relation to the Lie algebra $\fg$.
An $\CS$-wall ending on a branch point, which we call a primary $\CS$-wall, is labeled by a root $\alpha$ and denoted by $\CS_\alpha$. 
Its geometry is determined by $\alpha$ through the differential equation (\ref{eq:geodesic-eq}), its soliton data is classified by pairs of weights differing by $\alpha$, together with topological data on $\Sigma_\rho$.
In Section \ref{subsec:parallel-transport} we introduce the \emph{formal parallel transport} on $C$, a formal generating series $F(\wp)$ associated to a path $\wp$ on $C$, which depends directly on the $\CS$-wall soliton data. 
The generic expression for $F(\wp)$ is given in (\ref{eq:pushforward-rel-hom}), while its relation to soliton data is described in (\ref{eq:rule-1}) and (\ref{eq:rule-2}).
%In Sections \ref{sec:primary-solitons}-\ref{sec:soliton-symmetry} we study the relation between these two, showing how a twisted version of homotopy invariance of $F(\wp)$ determines soliton data on all $\CS$-walls.
In Sections \ref{sec:primary-solitons}-\ref{sec:soliton-symmetry} we study how a twisted version of homotopy invariance of $F(\wp)$ determines soliton data on all $\CS$-walls.

The study of twisted homotopy invariance is divided into several parts. 
In Section \ref{sec:primary-solitons} we derive the soliton content of primary $\CS$-walls by requiring flatness of $F(\wp)$ as $\wp$ is deformed across branch points on $C$. 
In Section \ref{sec:joints-algebra} we analyze the constraint of twisted homotopy invariance applied to intersections of $\CS$-walls, or joints, and derive the equations that determine the soliton content of outgoing $\CS$-walls in terms of ingoing ones.
The joint equations factorize in a way that is reminiscent of branching rules for representations of the Lie algebra $\fg$, and are given in (\ref{eq:final_ansatz}).
In Section \ref{sec:primary-joints} we solve the joint equations for intersections of primary $\CS$-walls, for which we know the soliton content. 
The result is closely analogous to the Lie bracket for roots of $\fg$. 
From the joint of primary $\CS$-walls $\CS_\alpha,\CS_\beta$, a new $\CS$-wall $\CS_{\alpha+\beta}$ will be born if $\alpha+\beta$ is also a root.
Moreover, the soliton data of the three $\CS$-walls are encoded in generating functions $\Xi_{\alpha}, \Xi_{\beta}, \Xi_{\alpha+\beta}$ related by
\be
	\Xi_{\alpha+\beta} = [\Xi_\alpha,\Xi_\beta]\,.
\ee
In Section \ref{sec:generic-joints} we study joints of generic $\CS$-walls, not necessarily ending on branch points. By induction we are able to prove the the Lie bracket property extends to all joints of the network. This allows us to determine recursively the soliton data on all $\CS$-walls, in terms of primary $\CS$-wall data and the combinatorics of joints.

A fundamental ingredient in the  derivation of soliton data from homotopy invariance is the existence of a \emph{soliton symmetry}, relating different solitons carried by each $\CS$-wall. 
In a nutshell, the soliton content of a $\CS$-wall, which is classified by pairs in $\CP_\alpha$ and topological data on $\Sigma_\rho$, is symmetric under permutations of the pairs in $\CP_\alpha$.
A more complete formulation is stated in Proposition \ref{prop:soliton-symmetry}. 
In Section \ref{sec:soliton-symmetry} we prove the existence of this symmetry by first observing that it is respected by the soliton data of primary $\CS$-walls and then showing that it is a symmetry of the joint equations, thereby extending the symmetry to all descendant walls.

%%%%%%%%%%%%%%%%%%%%%%%%%%%%%%%%%%%
\subsection{\texorpdfstring{$\CS$}{S}-walls and soliton data}
\label{sec:S-walls_and_solitons}
%%%%%%%%%%%%%%%%%%%%%%%%%%%%%%%%%%%
%
From this point onwards, we assume that a choice of trivialization has been made for the covering $\Sigma_\rho\to C$. 
The choice is to a large extent free, and not necessarily within the class of trivializations described in Section \ref{subsec:trivialization}.

The $\CS$-walls of a spectral network $\CW$ are sourced by branch points or by intersections of other $\CS$-walls, called \emph{joints}. Their evolution is regulated by the differential equation (\ref{eq:geodesic-eq}). 
$\CS$-walls coming from branch points will be denoted as \emph{primary}, whereas others will be called \emph{descendants}. 
A joint among $\CS$-walls induces a splitting of these into a number of \emph{streets}, as shown in Figure \ref{fig:streets}.
To each street we associate individual soliton data, which differs from one street to another even along the same $\CS$-wall.
By a mild abuse of notation, we will sometimes refer to the soliton data, or content, of an $\CS$-wall, whenever it is clear from the context which particular street we are talking about.

\begin{figure}[!ht]
\begin{center}
\includegraphics[width=.3\textwidth]{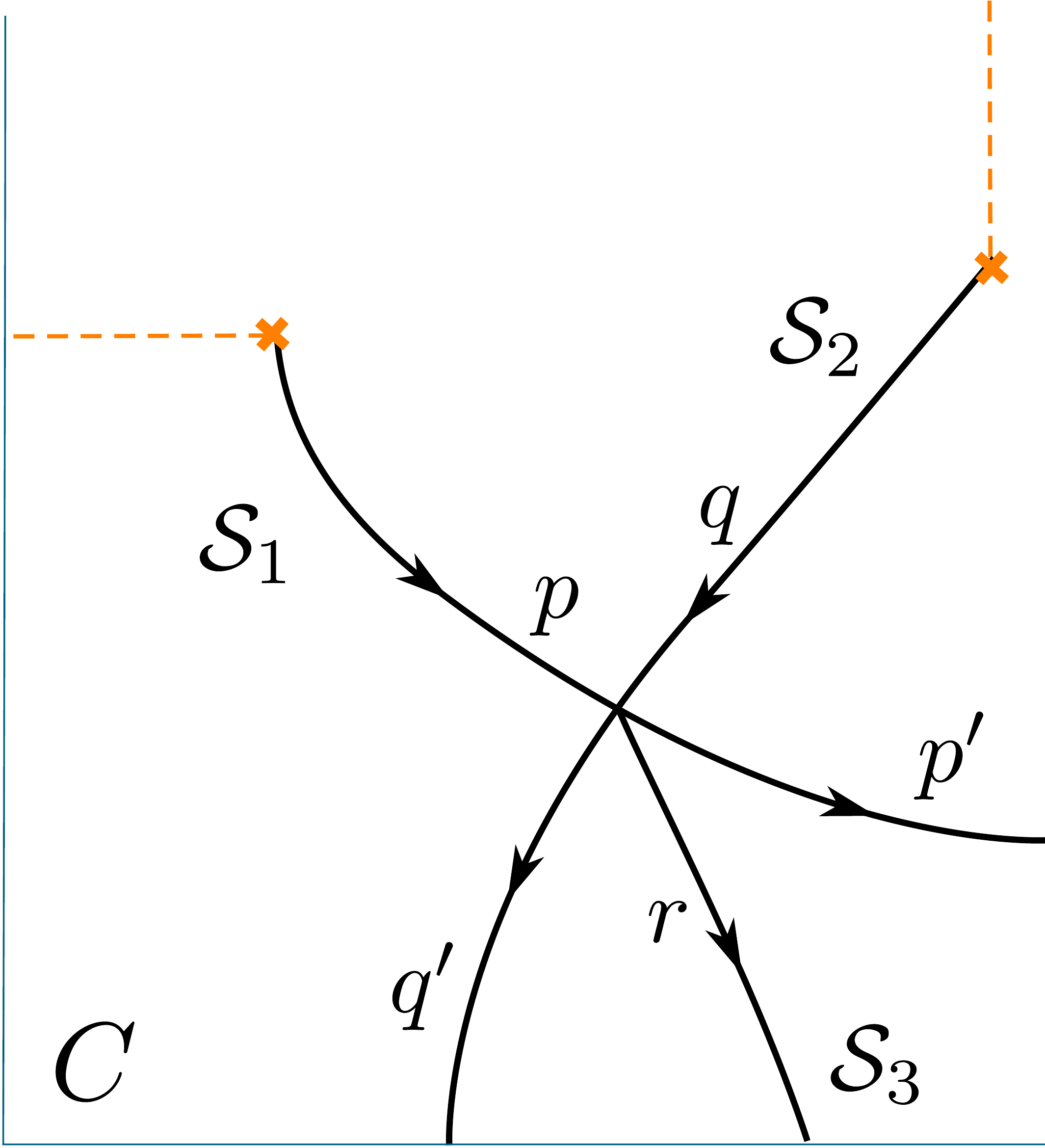}
\caption{Primary $\mathcal{S}$-walls $\CS_1, \CS_2$ are sourced at branch points, while descendant $\mathcal{S}$-walls such as $\CS_3$ are sourced by joints of other walls. Each wall carries several streets, for example $p, p'$ are streets for $\CS_1$.}
\label{fig:streets}
\end{center}
\end{figure}

%%%%%%%%%%%%%%%%%%%%%%%%%%%%%%%%%%%%%
\subsubsection*{Geometry of $\CS$-walls}
%%%%%%%%%%%%%%%%%%%%%%%%%%%%%%%%%%%%%
At generic $u\in\CB$, branch points will be of square-root type, and therefore labeled by positive roots $\alpha\in \Phi^{+}$.\footnote{In fact, as we saw in Section \ref{subsec:trivialization}, with a suitable choice of trivialization they are labeled by simple roots. 
We will however relax the constraint on the choice of trivialization, and work in greater generality, allowing for branch points of generic root types.
More precisely, a branch point is labeled by the hyperplane orthogonal to a root in $\ft^*$, which does not distinguish between $\alpha$ and $-\alpha$. 
Assuming a choice of positive roots is made, we adopt the convention of labeling branch points by positive roots from now on.
}
A primary $\mathcal{S}$-wall emanating from a branch-point of type $\alpha$ is labeled by a root $\CS_{\pm\alpha}$, its evolution is described by the following equation
\be\label{eq:geodesic-eq}
	(\partial_{t},\langle\alpha,\varphi\rangle) \in e^{i\vartheta}\IR^{+}\,.
\ee
In the neighborhood of the square-root branch-point at $z = z_0$, $\mathcal{S}$-walls are described by
\be
	z = z_0 +  t\, e^{i {\frac{2}{3}}( \vartheta + 2\pi  k)}\,, \ k\in\IZ/3\IZ,\ t\in\IR^+
\ee
Figure \ref{fig:branch-point} shows such $\mathcal{S}$-walls around a branch point labeled by $\alpha$.

\begin{figure}[t]
\begin{center}
\includegraphics[width=0.40\textwidth]{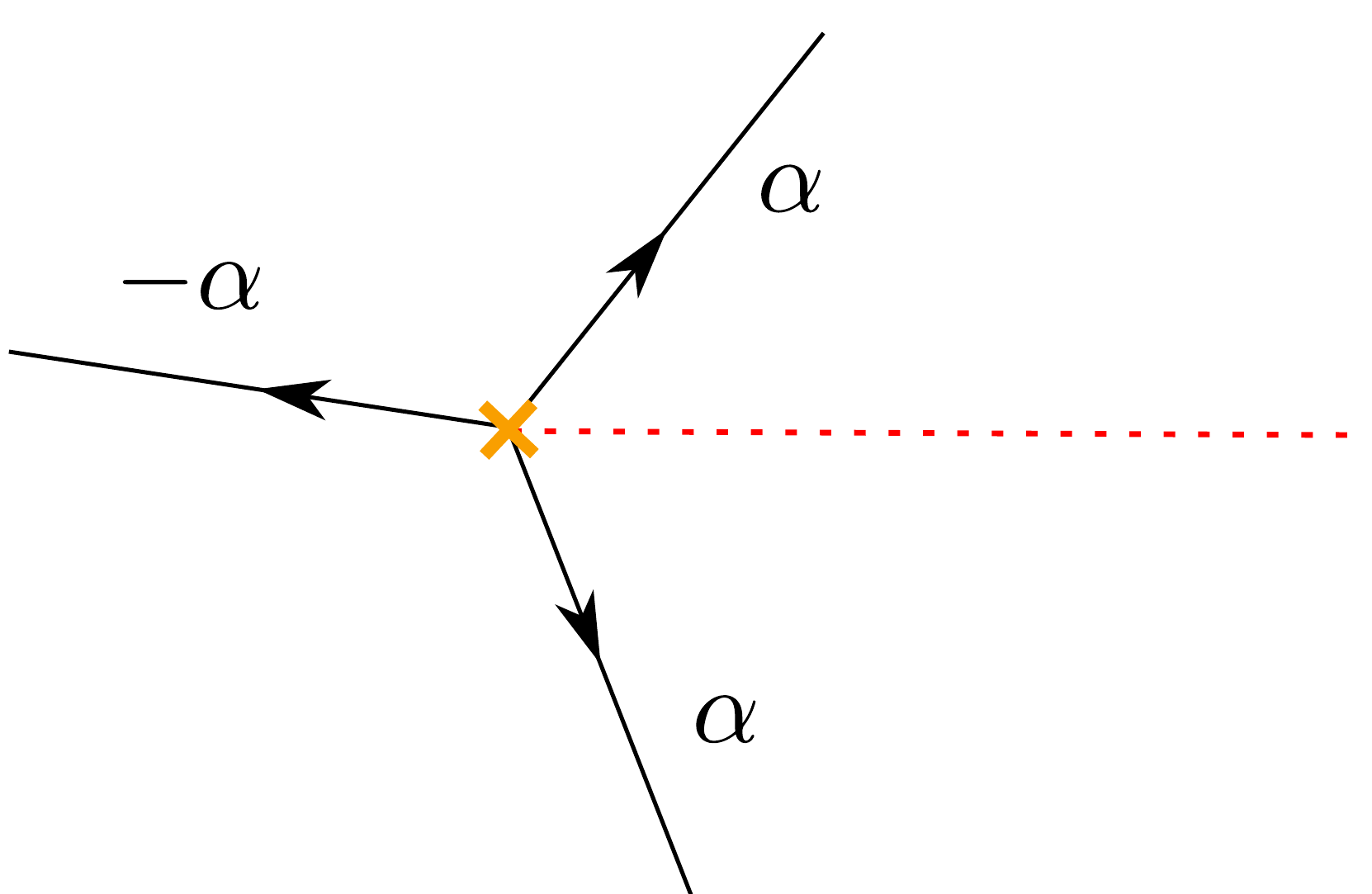}
\caption{$\CS$-walls emanating from a square-root branch point of type $\alpha$.}
\label{fig:branch-point}
\end{center}
\end{figure}

At this stage we cannot say anything general about descendant $\CS$-walls. However, we will show later that they are also labeled by roots. Therefore readers are advised to keep in mind that the current and forthcoming considerations will eventually apply to all streets of a spectral network.

%%%%%%%%%%%%%%%%%%%%%%%%%%%%%%%%%%%%%
\subsubsection*{Soliton data}
%%%%%%%%%%%%%%%%%%%%%%%%%%%%%%%%%%%%%

To each root $\alpha$, we associate a set of {ordered pairs} $\CP_{\alpha}$ defined by
\be\label{eq:soliton-charge-pairs}
	\CP_{\alpha}:= \big\{   (i,j) \,\big|  \, \weight_{j}-\weight_{i} = n \alpha  \,,\  \weight_{i},\weight_{j}\in \Lambda_{\rho}\big\}\,,\quad n\in\IN.
\ee
For minuscule representations one has $n=1$ for any pair and any choice of root.\footnote{In Appendix \ref{app:no-hexagons} this is shown by an explicit analysis of all minuscule representations, see in particular (\ref{eq:W2-orbits-A-type}) for A-type, (\ref{eq:W2-orbits-D-type-vector}) and (\ref{eq:W2-orbits-D-type-spinor}) for D-type.}
It is useful to further distinguish 
\be\label{eq:soliton-charge-pairs-pm}
\begin{split}
	\CP_{\alpha}^{-} :=  \big\{   i \,\big|  (i,j)\in\CP_{\alpha}  \big\}\,, \qquad
	\CP_{\alpha}^{+} :=  \big\{  j \,\big|  (i,j)\in\CP_{\alpha}  \big\}\,.
\end{split}
\ee
In minuscule representations these are always disjoint sets $\CP_{\alpha}^{+}\cap \CP_{\alpha}^{-} = \emptyset$, which we prove in Section \ref{sec:primary-solitons} (see in particular (\ref{eq:disjoint-property})).
Finally, note that 
\be\label{eq:k-rho-counts}
	|\CP_\alpha| = |\CP_\alpha^-| = |\CP_\alpha^+| = k_\rho
\ee 
which was defined in (\ref{eq:k-rho}), by virtue of (\ref{eq:k-rho-tilde}).

The above pairs classify the soliton data of the streets of a network, which we now introduce. Let $p$ be a street on a wall of type $\alpha$, and $z\in p$ be any point on the street.
The lift of $z$ to the sheet corresponding to a weight $\weight_{i}$ will be denoted $x_{i}(z)$. 
If $i\in\CP_{\alpha}^{\pm}$, we can further assign a tangent direction (a unit vector in $T_{x_i}\Sigma_\rho$)
by choosing the lift of $\pm$(the tangent direction) of $\CS_{\alpha}$ at $z$.
Then there is a canonical lift of $x_{i}$ to a point $\tilde x_{i}$ in the circle bundle $\tilde\pi:\tilde\Sigma_{\rho}\to\Sigma_{\rho}$.\footnote{See e.g. \cite[Sec. 3.5]{Gaiotto:2012rg} for the physical motivations for considering the lift of paths to $\tilde\Sigma_{\rho}$, also see \cite[Sec. 2.1.3]{Galakhov:2013oja} for further details on the lifting map by tangent framing.}
We then consider the set of relative homology classes of open paths on $\tilde\Sigma_{\rho}$ that start from $\tilde x_{i}(z)$ and end at $\tilde x_{j}(z)$,
\be
	H_{1}^\textrm{rel}(\tilde \Sigma_{\rho},\IZ; (\tilde x_{i}, \tilde x_{j})) / 2H,
\ee
where $H$ is the distinguished class in $H_{1}(\tilde\Sigma,\IZ)$ represented by a cycle winding once around a generic fiber.\footnote{$H_{1}^\textrm{rel}$ is a torsor for $H_{1}$, i.e.\ it carries an action by the latter. The quotient is understood in this sense.} This is a $\IZ_{2}$-extension of the more familiar 
$H_{1}^\textrm{rel}(\Sigma_{\rho}; (x_{i},  x_{j}))$, with grading given by the tangential winding number modulo $2$.
For each pair $(i,j)\in\CP_\alpha$, we define a set of $(i,j)$ soliton charges\footnote{The central charge $Z$ is defined on $H_1(\Sigma_\rho,\IZ)$, so the definition of $\Gamma_{ij}$ is understood to involve a choice of section $\sigma : H_1(\Sigma_\rho,\IZ) \to H_1(\tilde \Sigma_\rho,\IZ) / 2 H$, since $H_{1}^\textrm{rel}(\tilde \Sigma_{\rho},\IZ; (\tilde x_{i}, \tilde x_{j})) / 2H$ is a torsor for $H_1(\tilde \Sigma_\rho,\IZ) / 2 H$. Moreover $\sigma$ must be a homomorphism, so that $\sigma(\ker(Z))$ is again a lattice. Therefore to be precise the quotient should be by $\sigma(\ker(Z))$, but we use a sloppier notation of omitting $\sigma$ in the following.} 
\be\label{eq:soliton_ij_lattice}
	\Gamma_{ij}(z) := \Big( H_{1}^\textrm{rel}(\tilde \Sigma_{\rho},\IZ; (\tilde x_{i}, \tilde x_{j}))  / 2H\Big)\,\big/ \, \ker(Z)\,.
\ee
The quotient by $\ker (Z)$ in (\ref{eq:soliton_ij_lattice}) will play an important role below, and is related to the the definition of the lattice of physical charges $\hat\Gamma$ by a sub-quotient procedure that is discussed in Section \ref{subsec:prym}.
Denoting by $\Gamma$ the natural lift of $\hat\Gamma$ to $\tilde\Sigma_\rho$ (modulo $2 H$), $\Gamma_{ij}(z)$ will be a torsor for $\Gamma$.\footnote{This expectation follows from the physical picture of 2d-4d wall-crossing \cite{Gaiotto:2012rg, Gaiotto:2011tf}.}
There is a natural action of $H_{1}(\widetilde\Sigma_\rho,\IZ)$ on $H_{1}^\textrm{rel}(\widetilde\Sigma_\rho,\IZ;(\widetilde x_{i},\widetilde x_{j}))$, and we claim that this descends to an action of $\Gamma$ on $\Gamma_{ij}(z)$, which follows from the existence of the isomorphism (\ref{eq:sub-quotient-iso}).\footnote{
Both $\Gamma, \Gamma_{ij}(z)$ are obtained by the same quotient by $\ker(Z)$, so given $\eh{\gamma}\in H_1(\widetilde\Sigma_\rho,\IZ)$, $\eh{a}_{ij}\in H_1^{\textrm{rel}}(\widetilde\Sigma_\rho,\IZ)$  we have  $[\eh{\gamma}]+[\eh{a}]=[\eh{\gamma}+\eh{a}]$ if $Z_{[\eh{\gamma}]+[\eh{a}]}=Z_{[\eh{\gamma}+\eh{a}]}$, which does hold.
}
Introducing 
\be
	\Gamma(z) :=\bigcup_{(i,j)\in\CP_\alpha} \Gamma_{ij}(z)\,,
\ee
the soliton data of a street $p$ is the set of pairs
\be
	\big\{(a,\mu(a))\,|\, a\in \Gamma(p)\,,\ \mu(a)\in\IZ \big\} 
\ee
of soliton charges $a$ together with integers $\mu(a)$, known as soliton degeneracies. 
The latter obey the identity
\be
	\mu(a) / \mu(a') = (-1)^{w(a,a')}
\ee
for any pair $a$, $a'$ that differ by a winding number $w(a,\, a')$ around a fiber of $\tilde{\Sigma}_\rho$. 
This definition is closely related to the original one from \cite{Gaiotto:2012rg}, the main new ingredient being the classification by pairs $\CP_\alpha$. 
For each $(i,j)\in\CP_{\alpha}$ there is a class of solutions to the equation (\ref{eq:geodesic-eq}) that lifts to $\tilde\Sigma_\rho$, with endpoints $\tilde x_{i}$, $\tilde x_{j}$. 
The classification of these solutions by relative homology follows from the physical interpretation of $C$, $\Sigma_\rho$ and the points $z$, $\{x_i\}$ in terms of 2d-4d coupled systems \cite{Gaiotto:2011tf, Gaiotto:2013sma, Shifman:2014lba, Longhi:2016bte}.

%%%%%%%%%%%%%%%%%%%%%%%%%%%%%%%%%%%%%
\subsubsection*{Central charges from soliton trees}
%%%%%%%%%%%%%%%%%%%%%%%%%%%%%%%%%%%%%

It is useful to establish a simple operational criterion to determine whether two relative homology classes $\eh{a},\eh{a}'\in H_{1}^\textrm{rel}(\widetilde\Sigma_\rho,\IZ;(\widetilde x_{i},\widetilde x_{j}))/2H$ are identified by the quotient by $\ker (Z)$ in the definition of $\Gamma_{ij}(z)$. 
Here we provide such a criterion, by explaining how the central charge of a soliton path $a\in\Gamma_{ij}(z)$ is encoded in certain topological data of $\CW$.

The quotient by $\ker(Z)$ induces the following identification in $\Gamma_{ij}(p)$ 
\be\label{eq:soliton-equivalence}
	[\eh{a}]_{\ker(Z)} = a \sim a' = [\eh{a}']_{\ker(Z)} \qquad \text{if}\qquad Z_{a} = Z_{a'} \,,
\ee
where the central charge of solitons supported on $p$ is 
\be
	Z_{a} = \int_{\eh{a}}\lambda_{}\,.
\ee
\begin{figure}[t]
\begin{center}
\includegraphics[width=0.4\textwidth]{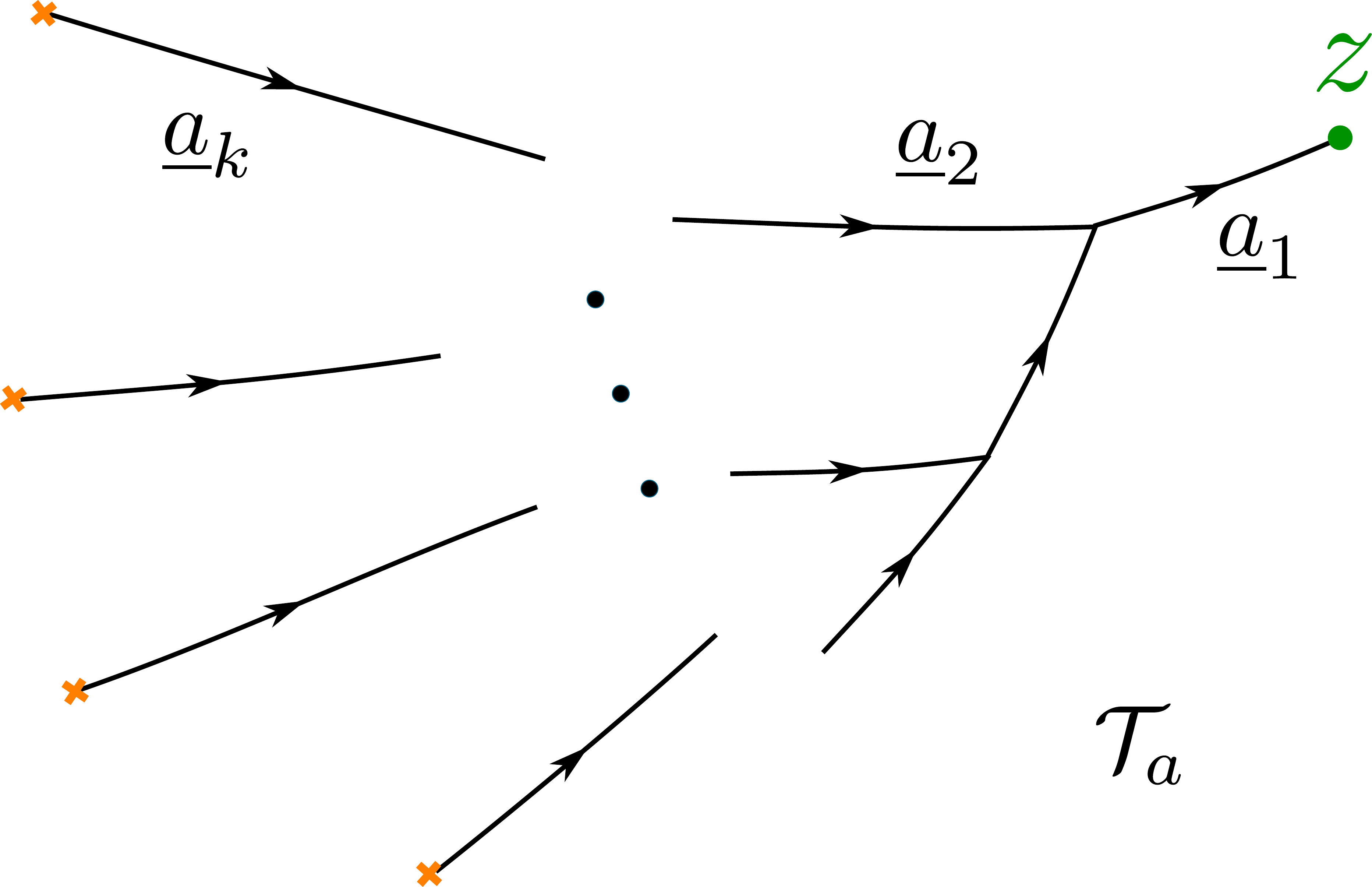}
\caption{A soliton tree $\CT_{a}$.}
\label{fig:soliton-tree}
\end{center}
\end{figure}

To any soliton charge $a$ we associate a \emph{decorated tree} $\CT_a$: its edges consist of streets of $\CW$, its nodes correspond to joints (intersections of $\CS$-walls), its leaves are branch points, and its root is $z$.
The central charge $Z_a$  is then entirely determined by the data of $\CT_a$.
To explain how decorated trees are defined, we need to state a few general properties of soliton data, which will be derived later. 

For each soliton charge $a\in\Gamma(z)$ with $\mu(a)\neq 0$ there is a canonical representative path $\mathfrak{a}$ (i.e.\ an \emph{actual} path) on $\tilde\Sigma_\rho$, whose projection on $C$, denoted by $T_{\mathfrak{a}}$, lies entirely on the network $\CW$. 
$T_{\mathfrak{a}}$ is topologically a tree, its edges $\{\underline a_i\}_{i=1}^{k}$ are streets of $\CW$. 
Above a street labeled by a root $\alpha$, the path $\mathfrak{a}$ runs on sheets $\tilde{x}_i$, $\tilde{x}_j$ of $\tilde\Sigma_\rho$ for some pair $(i,j)\in\CP_\alpha$.
Let $n_i$ be half the number of times $\mathfrak{a}$ runs above the edge $\underline a_i$ (by construction $a$ always runs twice over a street)\footnote{The same kind of counting already appeared, in the context of closed cycles, in \cite[App. B.2]{Galakhov:2014xba}.}, and $\alpha_i$ be the root that labels the underlying street.
Then $\CT_a$ is the decoration of $T_{\mathfrak{a}}$ obtained by associating $(\alpha_i, n_i)$ to each edge $\underline a_i$. See Figure \ref{fig:soliton-tree} for an example of a soliton tree.
The central charge of $T_{\mathfrak{a}}$ can then be expressed entirely in terms of the tree data,
\be
	Z_{a} = \sum_{i=1}^{k} n_{i} \int_{\underline{a}_{i}}\langle\alpha_{i},\varphi(z)\rangle = \sum_{i=1}^{k} n_{i}\,Z_{\underline{a}_{i}}\,,
\ee
where the orientation to be used for the integral is the one shown in Figure \ref{fig:soliton-tree}.\footnote{We tacitly relied on the fact, already mentioned below (\ref{eq:soliton-charge-pairs}), that $\weight_{j}-\weight_{i} = m\alpha$ with $m$  always $1$, since we restrict our analysis to minuscule representations.}
A useful criterion for distinguishing whether two soliton charges  $a,a'\in\Gamma_{ij}(p)$ coincide is then to compare their decorated trees:
\be
	\CT_{a}=\CT_{a'}\qquad\Rightarrow\qquad Z_{a}=Z_{a'}
\ee
and {soliton charges with the same decorated tree are equivalent}.

%%%%%%%%%%%%%%%%%%%%%%%%%%%%%%%%%%%
\subsection{Formal parallel transport}\label{subsec:parallel-transport}
%%%%%%%%%%%%%%%%%%%%%%%%%%%%%%%%%%%

Given a spectral network $\CW$, the associated formal parallel transport is characterized by defining a formal generating function $F(\wp)$ for any open path $\wp$ in $C$ from $z_{1}$ to $z_{2}$.
The definition will make use of the data of $\CW$, but we have not yet specified how to fix the soliton content.
It makes nevertheless good sense to give the formal definition in terms of unspecified soliton data. In fact the latter will ultimately be determined by imposing the flatness condition on $F(\wp)$. 
The definition is similar to that of \cite{Gaiotto:2012rg}, to which we refer for more details.

The first step is to introduce $\Gamma_{ij}(\tilde z_{1},\tilde z_{2})$, a cousin of $\Gamma_{ij}(z)$, which instead of classifying paths from lifts $\tilde x_{i}(z)$ to $\tilde x_{j}(z)$ lying above the same $z$, will classify paths from $\tilde x_{i}(z_{1})$ to $\tilde x_{j}(z_{2})$  up to relative homology,\footnote{The tangential directions encoded within $\tilde x_{i},\tilde x_{j}$ are understood to be the tangents at endpoints of $\wp$.} 
\be
	\Gamma_{ij}(\tilde z_{1},\tilde z_{2}) = \left[H_{1}^\textrm{rel}\big(\tilde\Sigma_{\rho};(\tilde x_{i}(z_{1}), \tilde x_{j}(z_{2}))\big)/2H\right]\Big/\ker(Z) \,.
\ee
We also define
\be
	\Gamma(\tilde z_{1}, \tilde z_{2}) = \bigcup_{ij}\Gamma_{ij}(\tilde z_{1}, \tilde z_{2})\,.
\ee
where the union runs over all values of $i$, $j=1,\, \ldots ,\, d = \dim(\rho)$.

Next we introduce a certain ring of formal variables: to each $a\in\Gamma_{ij}(\tilde z_{1},\tilde z_{2})$ we associate a formal variable $X_{a}$, which obey the following product rule\footnote{Notice that the concatenation operation is well-defined on the equivalence classes, but it's not injective.} 
\be
	X_{a}X_{b} = \left\{\begin{array}{ll} X_{a+b}\qquad  & \text{if $a$ ends where $b$ begins,} \\ 0 & \text{otherwise.}\end{array}\right.
\ee
In addition, given any $a,\, a' \in \Gamma(\tilde z_{1},\tilde z_{2})$ 
\be
	X_{a} / X_{a'} = (-1)^{w(a,a')} 
\ee
where $w(a,a')\in\IZ/2\IZ$ is the relative winding between $a,a'$, modulo $2H$. 

Given any path $\wp$ on $C$, the formal parallel transport $F(\wp,\CW)$ is a generating series
\be\label{eq:formal-transport-generic}
	F(\wp,\CW) = \sum_{a\in\Gamma(\tilde z_{1},\tilde z_{2})} \overline{\underline{\Omega}}(\wp,\CW,a) \, X_{a}
\ee
where the coefficients $\overline{\underline{\Omega}}(\wp,\CW,a)$ are $\IZ$-valued\footnote{$\overline{\underline{\Omega}}(\wp,\CW,a)$ is the \emph{framed} 2d-4d BPS degeneracy, introduced in \cite{Gaiotto:2011tf}, for a supersymmetric line interface characterized by $\wp$ and a framed BPS state of charge $a$.} with the following properties
\be
\begin{split}
	\overline{\underline{\Omega}}(\wp,\CW,a) / \overline{\underline{\Omega}}(\wp,\CW,a') = (-1)^{w(a,a')}\qquad& \text{if }a,\, a'\in\Gamma_{ij}(\tilde z_{1},\tilde z_{2}) \\
	\overline{\underline{\Omega}}(\wp,\CW,a) / \overline{\underline{\Omega}}(\wp',\CW,a) = (-1)^{w(\wp,\wp')}\qquad & \text{if }\wp,\, \wp' \text{ have equal endpoints } \\
& \text{and equal tangents at endpoints\,.}
\end{split}
\ee
These properties ensure that each term of the formal generating series depends only on the relative homology class $\ehz{a} = \tilde\pi_{*}(a) \in H_{1}^\textrm{rel}(\Sigma_{\rho}, (x_{i}, x_{j}))/\ker(Z)$ 
and not on the choice of representative $a$. 
With this understood, we can manifestly express the generating series as a sum over these homology classes
\be\label{eq:pushforward-rel-hom}
		F(\wp,\CW) = \sum_{\ehz{a}\in\Gamma( z_{1}, z_{2})} \overline{\underline{\Omega}}(\wp,\CW,a) \, X_{a},
\ee
where $\Gamma(z_{1}, z_{2}) = \bigcup_{ij}H_{1}^\textrm{rel}(\Sigma_{\rho}, (x_{i}, x_{j}))$.

The coefficients of the formal series (\ref{eq:formal-transport-generic}) are fixed by two rules. 
First, if $\wp\cap\CW = \emptyset$ 
\be\label{eq:rule-1}
	F(\wp,\CW) = D(\wp) := \sum_{i}X_{\wp^{(i)}} 
\ee
where $\wp^{(i)}$ denote the canonical lifts of $\wp$ to sheets of $\tilde\Sigma_{\rho}$ with tangent framing.
Second, if $\wp$ intersects $\CW$ on a street $p$ of type $\alpha$ at $z \in C$, the above formula is modified as
\be\label{eq:rule-2}
	F(\wp,\CW) = D(\wp_{+}) \left( 1 + \sum_{(i,j)\in\CP_{\alpha}}\sum_{a\in\Gamma_{ij}(z)} \mu(a) X_{a} \right) D(\wp_{-}), 
\ee
where $\wp_{\pm}$ are shown in Figure \ref{fig:detour-rule}. This is the standard ``detour rule'' of \cite{Gaiotto:2012rg}.

\begin{figure}[h!]
\begin{center}
\includegraphics[width=0.4\textwidth]{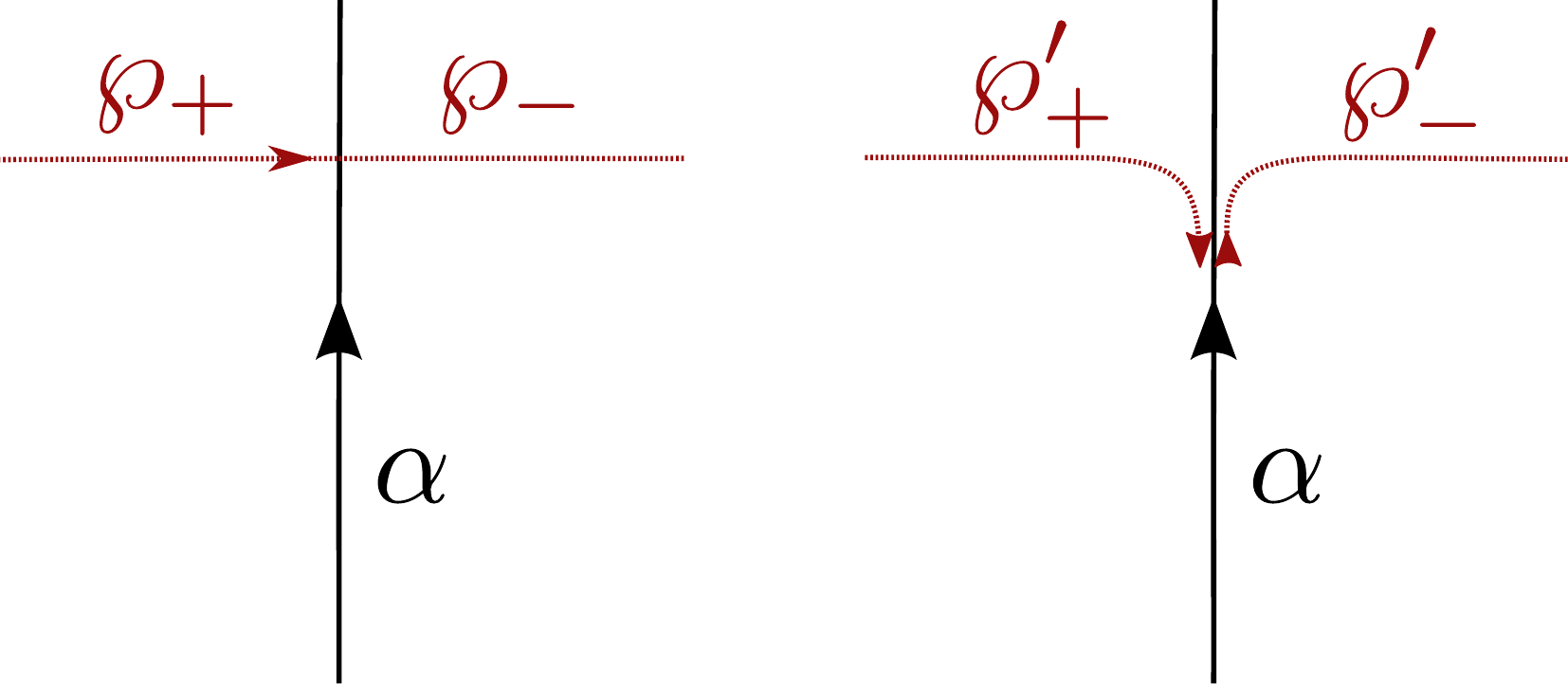}
\caption{When $\wp$ crosses s street $p$ of the network, the concatenation of $\wp^{(i)}_{\pm}$ with solitons on $p$ is understood to involve a slight deformation to match the tangents.}
\label{fig:detour-rule}
\end{center}
\end{figure}

These rules completely determine $F(\wp,\CW)$ in terms of the soliton data on $\CW$. 
We now turn to study the constraints that \emph{flatness}, i.e.\ the condition that $F(\wp,\CW)$ only depends on the (twisted) homotopy class of $\wp$\footnote{With tangents at the endpoints fixed.}, imposes on the soliton data.

%%%%%%%%%%%%%%%%%%%%%%%%%%%%%%%%%%%
\subsection{Primary \texorpdfstring{$\CS$}{S}-walls and their soliton data} \label{sec:primary-solitons}
%%%%%%%%%%%%%%%%%%%%%%%%%%%%%%%%%%%

\begin{figure}[ht]
\begin{center}
\includegraphics[width=0.5\textwidth]{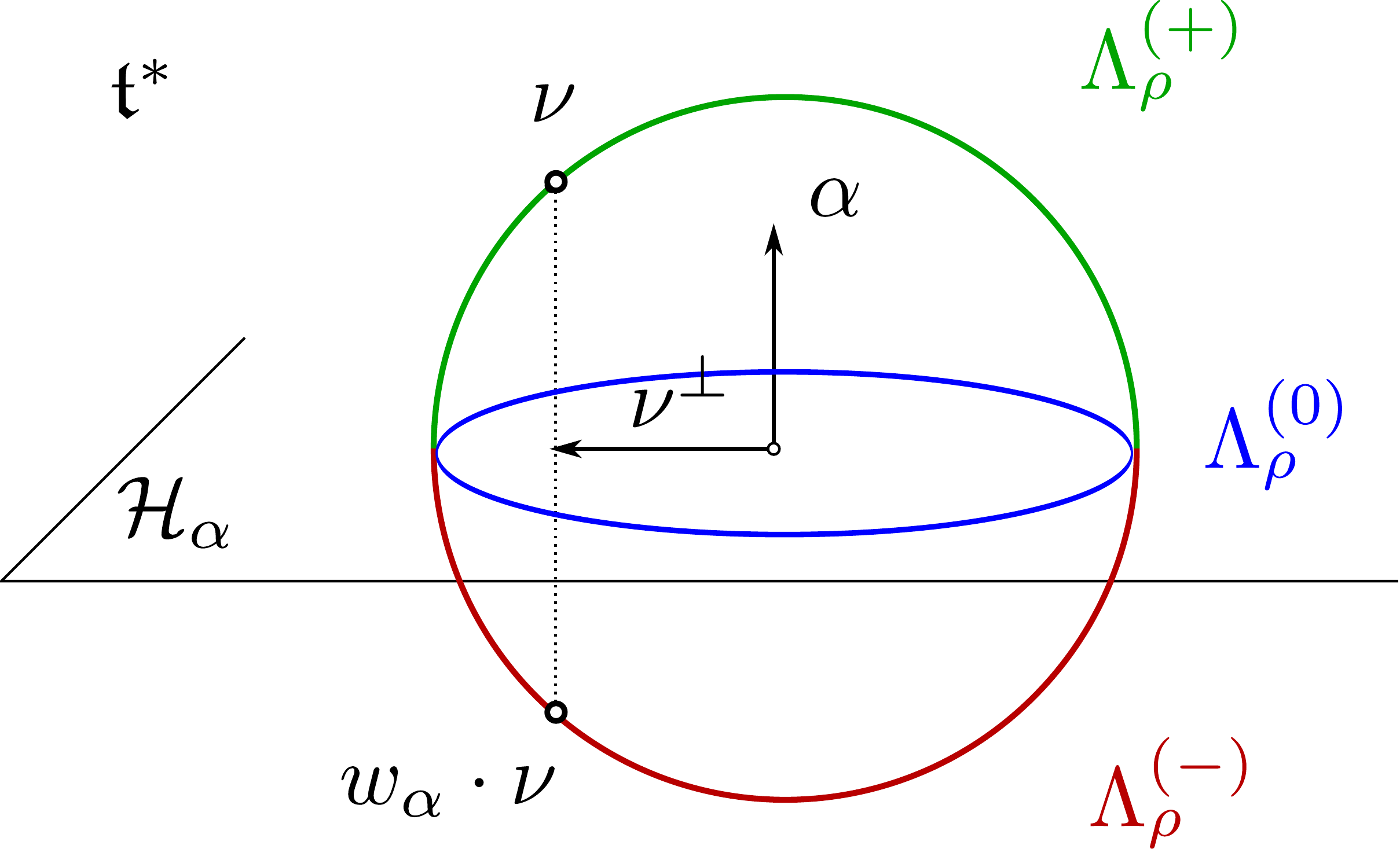}
\caption{The splitting of the weight system $\Lambda_{\rho}$ induced by the root $\alpha$. Weights of a minuscule representation lie on a hypersphere in $\ft^*$.}
\label{fig:weyl-orbit split}
\end{center}
\end{figure}

All weights $\weight\in\Lambda_{\rho}$ of a minuscule representation lie on a hypersphere $S^{r-1} \subset \ft^{*}$, which results in a simple structure of $\CP_{\alpha}$. Consider the partition of $\Lambda_{\rho}$ induced by $\alpha$
\be\label{eq:disjoint-property}
	\Lambda_{\rho} = \Lambda_{\rho}^{(0)}\,\sqcup\, \Lambda_{\rho}^{(+)}\,\sqcup\, \Lambda_{\rho}^{(-)}
\ee
where 
\begin{align*}
	\Lambda_{\rho}^{(0)} &= \{\weight\in\Lambda_{\rho}\,|\, \weight\cdot\alpha=0\},\\ 
	\Lambda_{\rho}^{(+)} &= \{\weight\in\Lambda_{\rho}\,|\, \weight\cdot\alpha>0\},\\ 
	\Lambda_{\rho}^{(-)} &= \{\weight\in\Lambda_{\rho}\,|\, \weight\cdot\alpha<0\}.
\end{align*}
The reflection $w_{\alpha}$ fixes each element of $\Lambda_{\rho}^{(0)}$ and maps $\Lambda_{\rho}^{{(+)}}$, $\Lambda_{\rho}^{(-)}$ into each other. This makes it manifest that for minuscule representations $\CP_{\alpha}^{-}\cap\CP^{+}_{\alpha}=\emptyset$ for such $\CS$-walls. 

At a square-root branch point of type $\alpha$, a sheet $x_i$ corresponding to $\weight_{i}\in\Lambda_{\rho}^{(-)}$ will come together with a sheet $x_j$ corresponding to $\weight_{j}\equiv w_{\alpha}\cdot\weight_{i}\in\Lambda_{\rho}^{(+)}$. 
Hence the soliton types carried by a primary $\mathcal{S}$-wall $\CS_\alpha$ is
\be\label{eq:weyl-paired}
	\CP_{\alpha}=\{(i,j) \,|\, \weight_{i}\in\Lambda_{\rho}^{(-)}\,;\weight_{j}=w_{\alpha}\cdot\weight_{i}  \}.
\ee
Consider ordering of the weights as
\be
	\big(\weight_{i_{1}},\weight_{\bar i_{1}},\dots, \weight_{i_{m}},\weight_{\bar i_{m}},\weight_{k_{1}},\dots,\weight_{k_{r}}   \big)
\ee
with ${i_{s}}\in \CP_{\alpha}^{-}$ and $\bar i_{s}$ the corresponding $w_{\alpha}$-mirror image, and the remaining $\weight_{k_{s}}\in\Lambda_{\rho}^{(0)}$.
Then the Stokes matrix capturing the detour rule for $\CS_{\alpha}$ has the block-diagonal form 
\be
\mathbb{I} + %
\left(\begin{array}{ccccc|ccc}
  0& \star& & & & & & \\
  0& 0& & & & & & \\
  & & \ddots& & & & & \\
  & & & 0 & \star & & & \\
  & & & 0 & 0 & & & \\
  \hline
  & & & & & 0 & & \\
  & & & & & & \ddots & \\
  & & & & & & & 0
\end{array}\right)
\ee
Likewise, the Stokes matrix of detours across a wall $\CS_{-\alpha}$ will have a shape corresponding to the transpose of this matrix.
When a parallel transport along $\wp$ crosses a branch cut of type $w_{\alpha}$, this will naturally be represented by the insertion of a matrix of the form
\be
\left(\begin{array}{ccccc|ccc}
  0& 1& & & & & & \\
  1& 0& & & & & & \\
  & & \ddots& & & & & \\
  & & & 0 & 1 & & & \\
  & & & 1 & 0 & & & \\
  \hline
  & & & & & 1 & & \\
  & & & & & & \ddots & \\
  & & & & & & & 1
\end{array}\right)\,,
\ee
reflecting the gluing of sheets of $\tilde\Sigma$ across the cut.

\begin{figure}[ht]
\begin{center}
\includegraphics[width=0.40\textwidth]{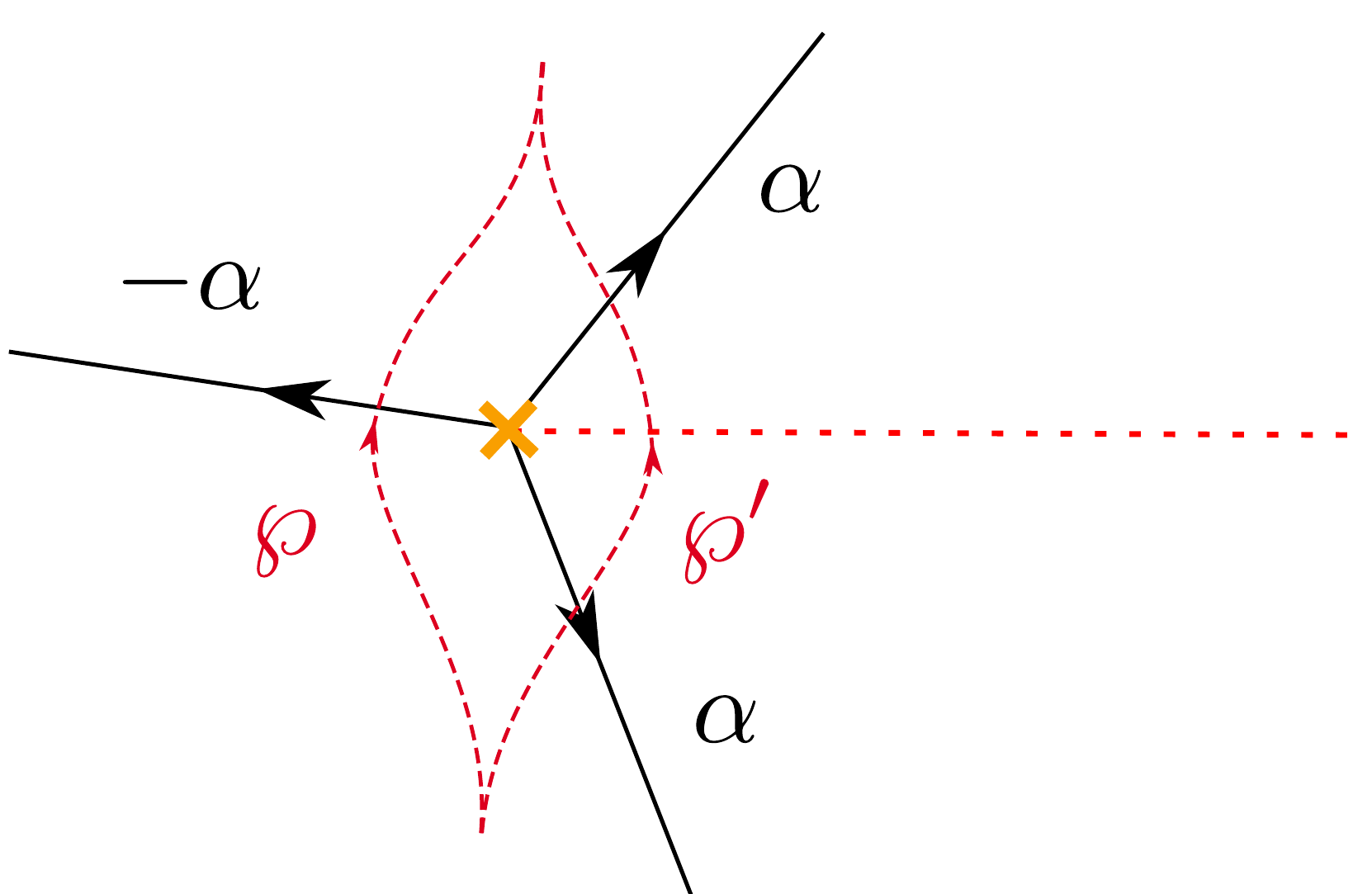}
\caption{
Homotopic paths $\wp, \wp'$ across a branch point. Demanding homotopy invariance $F(\wp,\CW)= F(\wp',\CW)$ fixes the soliton content of the primary $\mathcal{S}$-walls emanating form the branch point.
}
\label{fig:branch-cut-paths}
\end{center}
\end{figure}

Requiring homotopy invariance of $F(\wp,\CW)$ as $\wp$ is deformed across a branch-point (see Figure \ref{fig:branch-cut-paths}) then results in independent equations corresponding to blocks on the diagonal. 
The equation for each (nontrivial) block corresponds exactly to the well-understood case of the $\textrm{A}_{n}$ network in the first fundamental representation \cite{Gaiotto:2012rg}, and gives the following soliton content: a primary $\mathcal{S}$-wall $\CS_{\alpha}$ (resp.\ $\CS_{-\alpha}$) will carry exactly one  \emph{simpleton} for each $(i,j)\in \CP_{\alpha}$ (resp.\ $\CP_{-\alpha}$) with degeneracy $\mu=1$.

In Appendix \ref{app:simpletons} we provide an alternative derivation of the soliton content of primary $\mathcal{S}$-walls, where a detailed computation is carried out entirely by imposing the parallel transport rules and enforcing homotopy invariance.

%%%%%%%%%%%%%%%%%%%%%%%%%%%%%%%%%%%
\subsection{Joints and factorization of homotopy identities}
\label{sec:joints-algebra}
%%%%%%%%%%%%%%%%%%%%%%%%%%%%%%%%%%%

Having determined the soliton data of primary $\mathcal{S}$-walls, or more properly of streets ending on branch points, the next step is to consider  joints,  i.e.\ intersections of primary and descendant $\mathcal{S}$-walls.
For a generic network $\CW$ there are three distinct types of joints, depicted in Figure \ref{fig:joint-types}. 
4-way joints are trivial, in the sense that the soliton data of outgoing streets is the same as that of ingoing streets.
5-way and 6-way joints are instead nontrivial: in the former case a new street is sourced from the joint, in the latter the soliton data of incoming streets may change across the joint. 

\begin{figure}[ht]
\begin{center}
\begin{subfigure}{0.3\textwidth}
	\includegraphics[width=\textwidth]{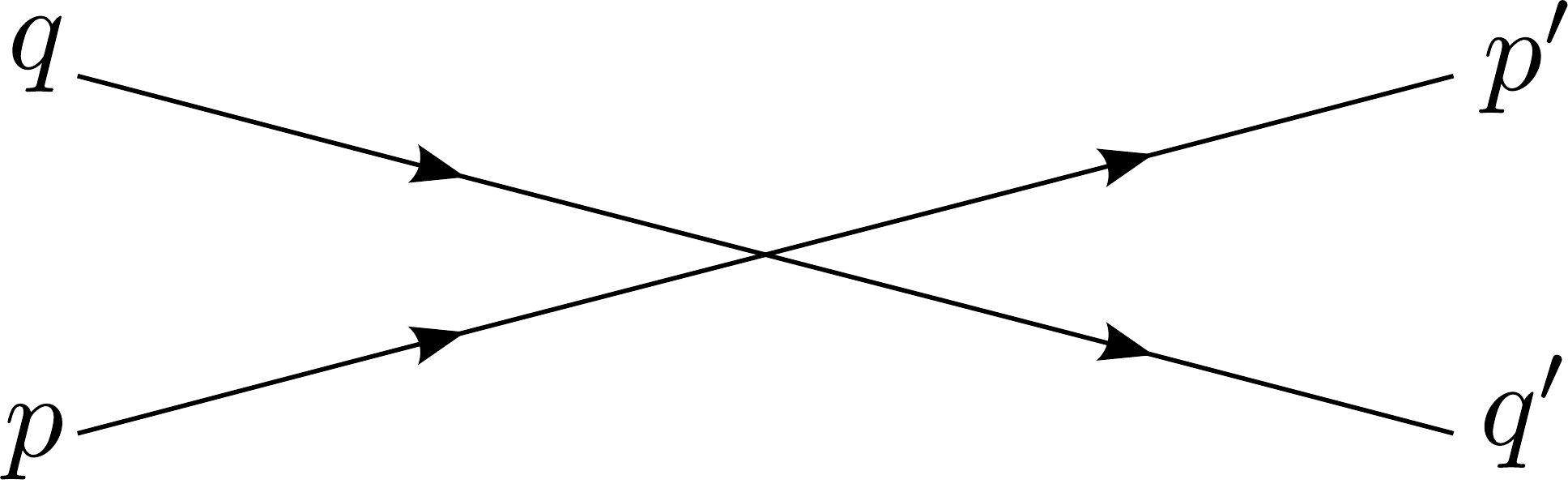}
	\caption{4-way}
\end{subfigure}
\hspace{0.03\textwidth}
\begin{subfigure}{0.3\textwidth}
	\includegraphics[width=\textwidth]{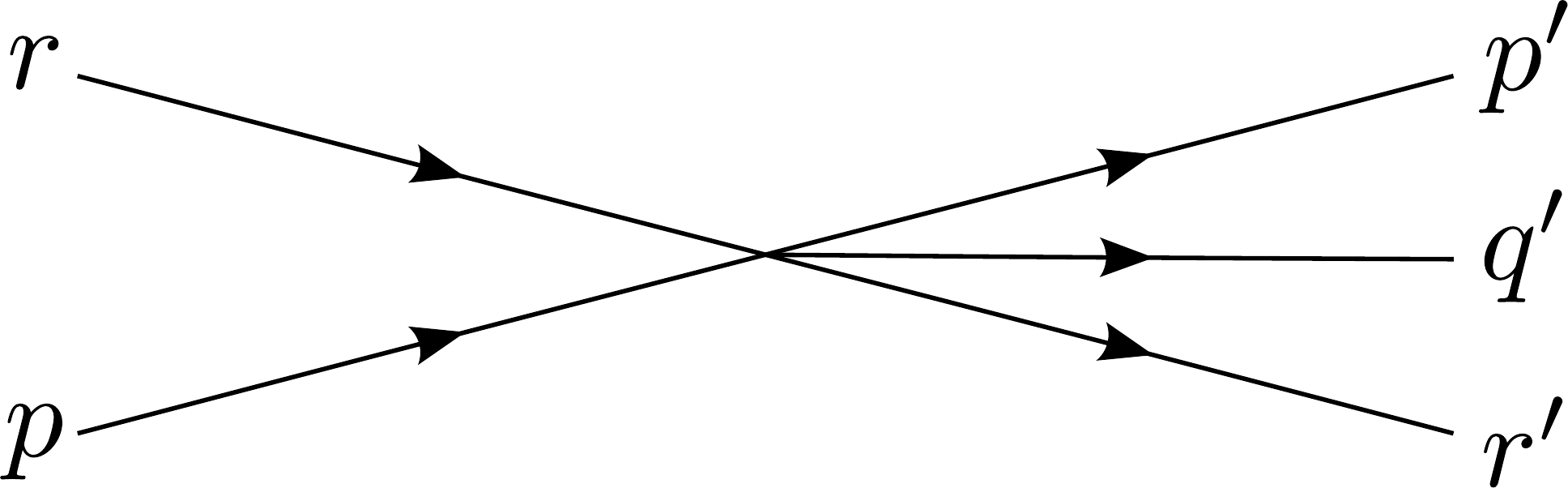}
	\caption{5-way}
	\label{fig:joint-types-5-way}	
\end{subfigure}
\hspace{0.03\textwidth}
\begin{subfigure}{0.3\textwidth}
	\includegraphics[width=\textwidth]{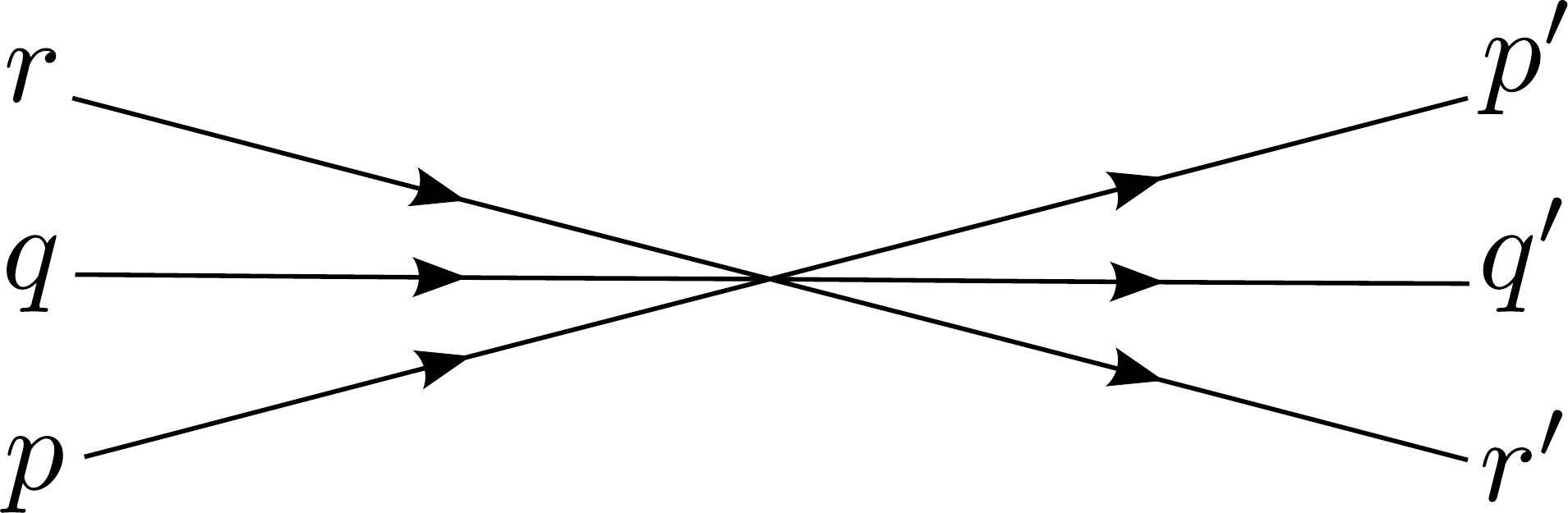}
	\caption{6-way}
\end{subfigure}
\caption{Types of joints.}
\label{fig:joint-types}
\end{center}
\end{figure}

We will begin by studying the joints of primary walls. Then we will proceed to discuss general properties of the soliton content of descendant walls, and argue that joints of descendant walls preserve such properties. 
The analysis is rather long and technical, stretching across the rest of Sections \ref{sec:spectral-networks}. For readers' convenience here we summarize the key results:

\begin{itemize}
\item Two intersecting primary walls $\CS_{\alpha},\CS_{\beta}$ will form a non-trivial joint if and only if $\alpha+\beta$ is a root. Otherwise the joint will be trivial.

\item Joints of descendant streets preserve this property: if all ingoing streets are of root-type, then all outgoing streets will be of root-type. In particular two intersecting descendant walls $\CS_{\alpha}$, $\CS_{\beta}$ will form a non-trivial joint if and only if $\alpha+\beta$ is a root. 

\item The flow of the soliton content across joints is determined by the combinatorics of concatenations of the corresponding soliton charges. No matter how rich the soliton content of a street may be, the flow ``factorizes'' in a way dictated by branching rules of standard representation theory. Concretely, given two roots $\alpha$ and $\beta$, a Weyl subgroup $W_2$ generated by $w_\alpha$ and $w_\beta$ induces an equivalence relation on $\Lambda_\rho$, each equivalence class being an orbit of $W_2$. 
Solitons supported on one street may concatenate with solitons supported on another street only if they fall in the same orbit.

\item For a nontrivial joint the twisted Cecotti-Vafa wall-crossing formula regulates the jump of soliton content (the 2d wall-crossing) across the joint.

\end{itemize}

In the remainder of this section we introduce some useful tools for proving the above claims.

%%%%%%%%%%%%%%%%%%%%%%%%%%%%%%%%%%%%%
\subsubsubsection*{Soliton diagrams}
%%%%%%%%%%%%%%%%%%%%%%%%%%%%%%%%%%%%%

For the purpose of studying joints, it helps to think schematically of the soliton types carried by a wall $\CS_{\alpha}$ in terms of {diagrams} in $\ft^*$. 
Such diagrams involve root-type vectors (solitons) connecting pairs of weight-type vectors (pairs of sheets connected by a soliton path), some examples are displayed in Figure {\ref{fig:fundamental-soliton-diagrams}}.

\begin{figure}[ht]
\begin{center}
\begin{subfigure}[b]{0.24\textwidth}
	\includegraphics[width=\textwidth]{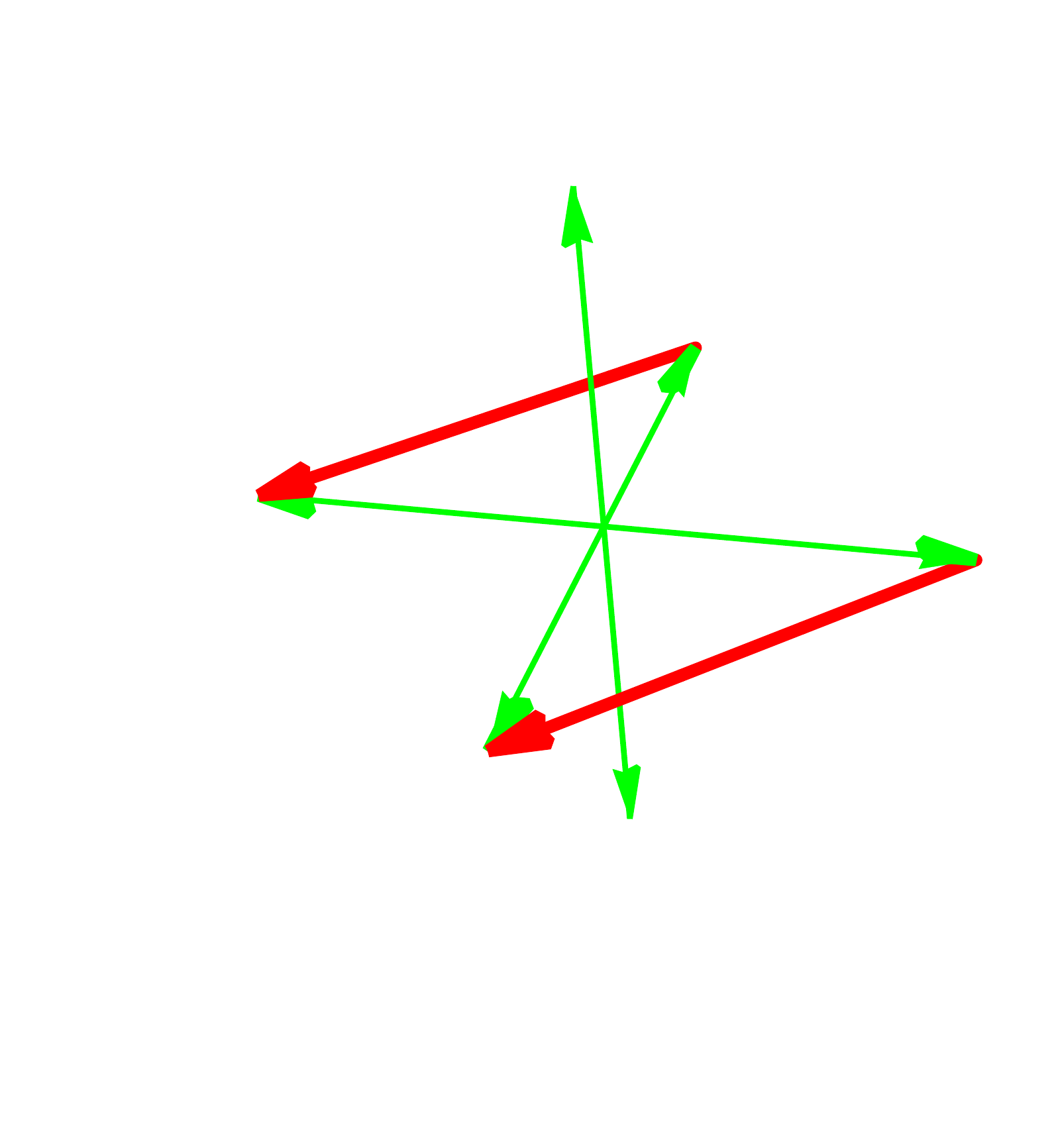}
	\caption{$\CS_{\alpha_{1}}$}
\end{subfigure}
\begin{subfigure}[b]{0.24\textwidth}
	\includegraphics[width=\textwidth]{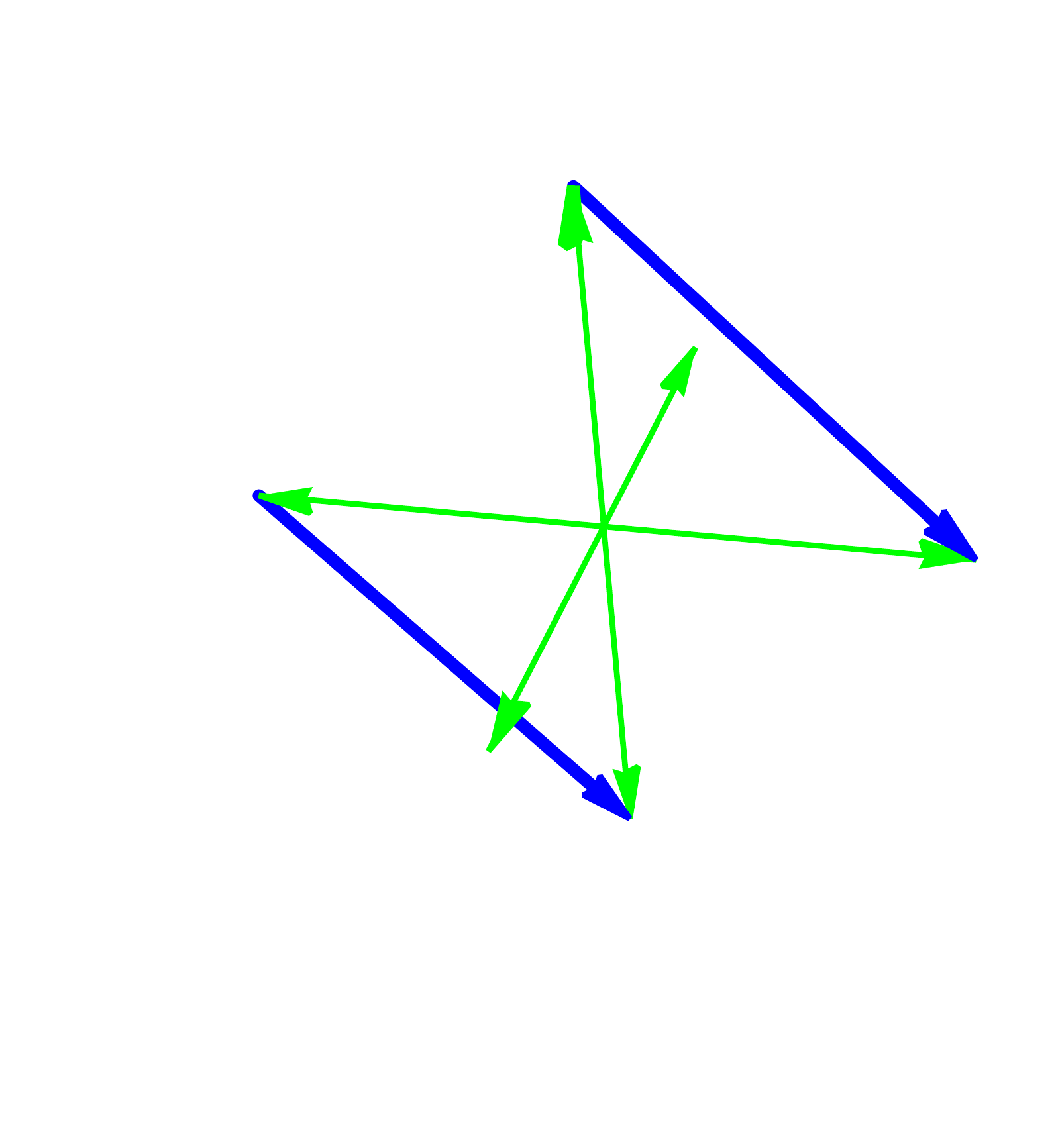}
	\caption{$\CS_{\alpha_{2}}$}
\end{subfigure}
\begin{subfigure}[b]{0.24\textwidth}
	\includegraphics[width=\textwidth]{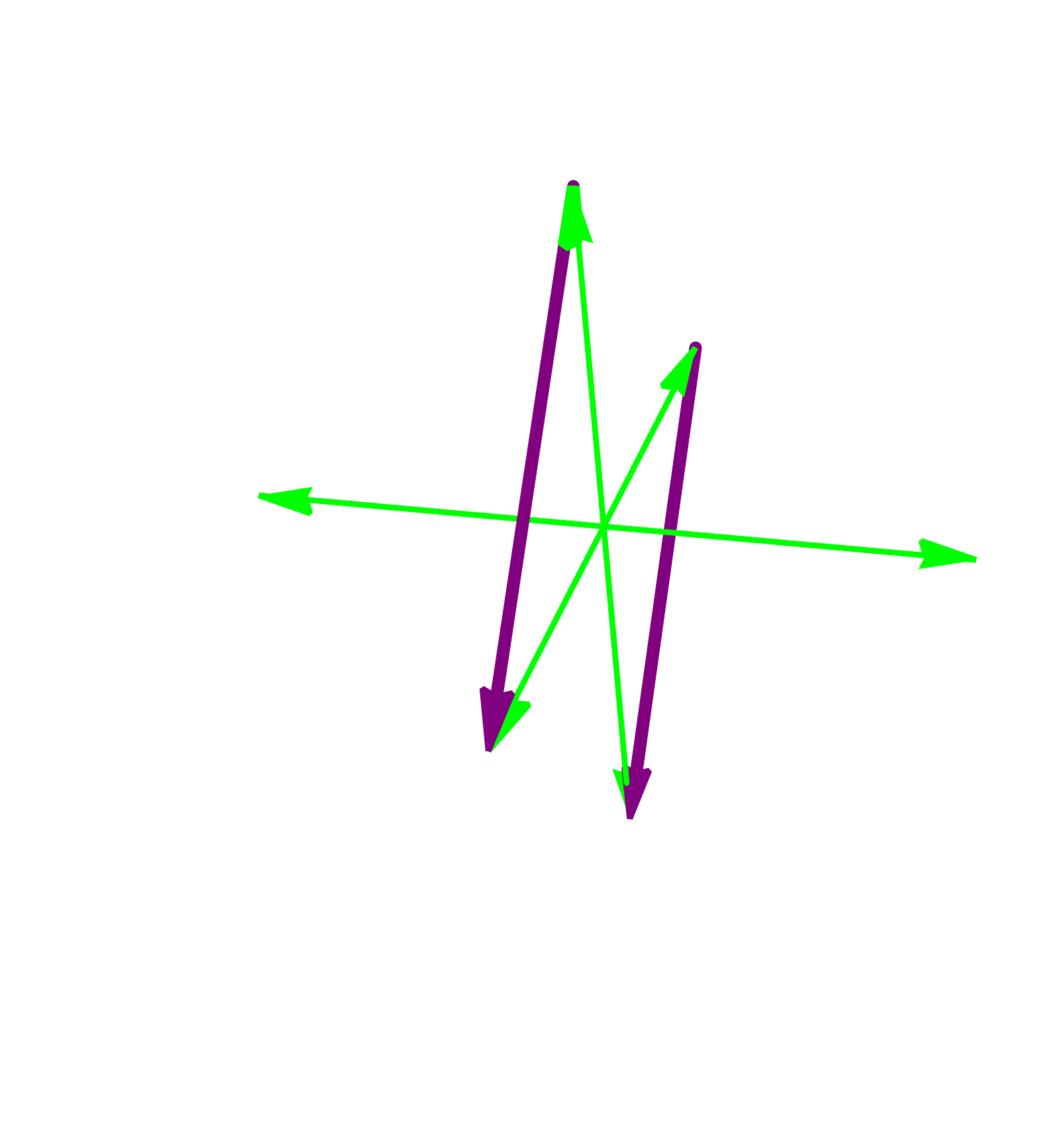}
	\caption{$\CS_{\alpha_{1}+\alpha_{2}}$}
\end{subfigure}
\begin{subfigure}[b]{0.24\textwidth}
	\includegraphics[width=\textwidth]{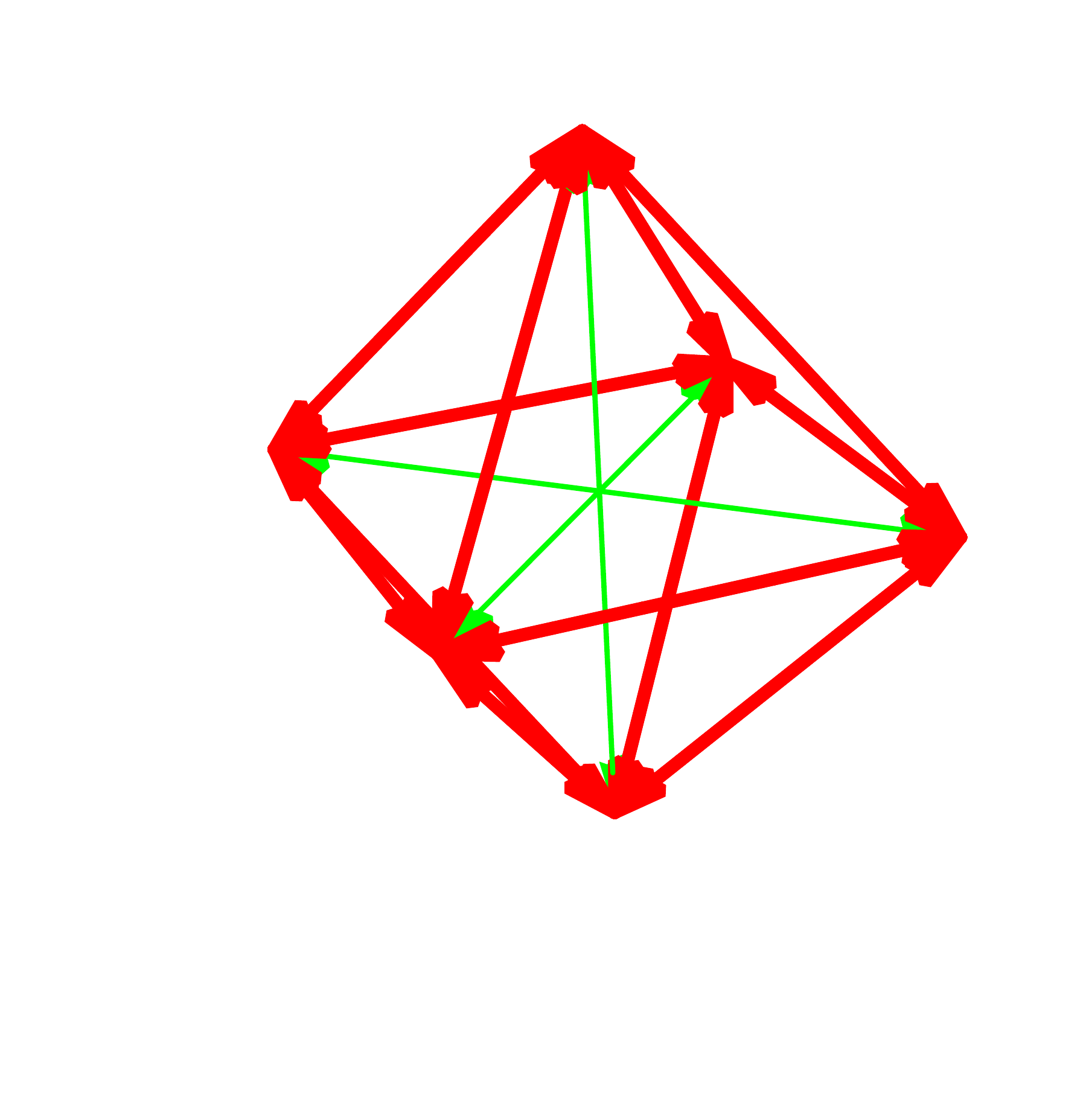}
	\caption{All soliton types}
\end{subfigure}
\caption{Left three frames: Soliton diagrams for the cover of $D_{3}$ in the vector representation. 
Green arrows denote weights (see also Figure \ref{fig:fundamental}). Arrows stretching between 
two weights $v_{i}\to v_{j}$ correspond to a type $(i,j)\in\CP_{\alpha}$ of solitons connecting the corresponding sheets.
}
\label{fig:fundamental-soliton-diagrams}
\end{center}
\end{figure}

Let $\alpha,\, \beta\in\Phi_\fg$ be roots with $\alpha\neq\pm\beta$, and consider walls $\CS_{\alpha},\,\CS_{\beta}$ labeled by roots (primary $\mathcal{S}$-walls would be of this type, but $\CS_\alpha,\, \CS_\beta$ need not be primary).
The choice of $\alpha,\, \beta$ then splits $\Lambda_{\rho}$ into subsets of weights arranged on several affine 2-planes, linearly generated by $\alpha$ and $\beta$. 
Recall that any two roots of a simply-laced Lie algebra must form one of these angles 
\be
\alpha\measuredangle\beta =\frac{\pi}{3},\, \frac{\pi}{2},\, \frac{2\pi}{3},\, \pi\,.
\ee
If the angle is $\pi/3$ or $2\pi/3$, each affine 2-plane contains a subset of the weights in $\Lambda_\rho$ arranged into  Weyl orbits of $\gA_{2}$\footnote{This follows from the fact that $E_{\alpha},\, E_{\beta}\in\fg$ generate an $\gA_{2}$ subalgebra.}, see Figure \ref{fig:joint-diagram-A2} for an illustration of an example.
On the other hand, if the angle is $\pi/2$ each affine plane will contain Weyl orbits of $\gD_{2}$, as in Figure \ref{fig:joint-diagram-D2}. 
\begin{figure}[!ht]
\begin{center}
\begin{subfigure}[t]{0.4\textwidth}
	\begin{center}
	\includegraphics[width=0.6\textwidth]{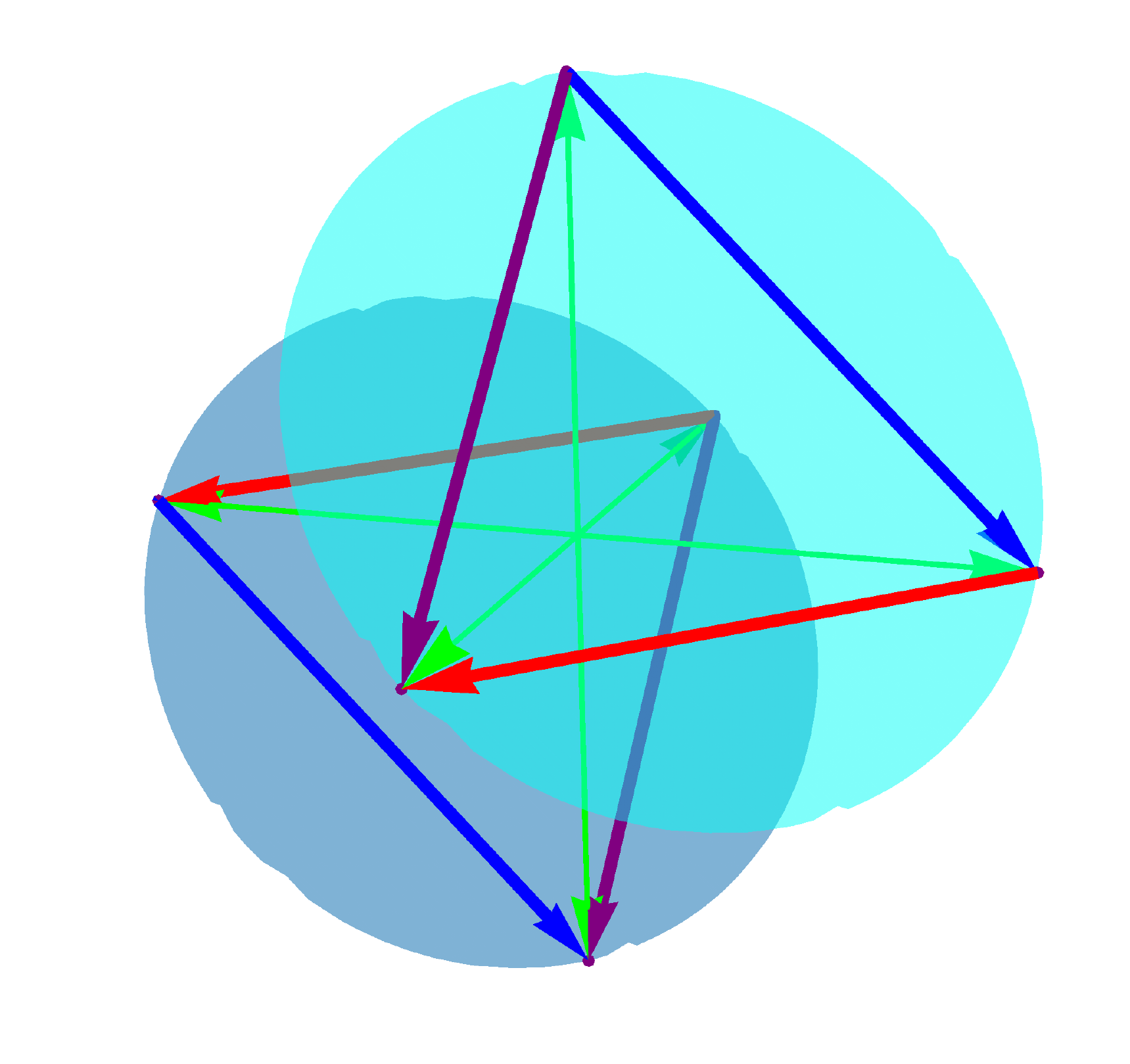}
	\caption{$\Lambda_{\rho}$ splits into two Weyl orbits when sliced parallel to the $\alpha_{1}\IR\oplus\alpha_{2}\IR$ plane. The red arrow is $\alpha_{1}$, the blue arrow is $\alpha_{2}$, and the purple is their sum, also a root.}
	\label{fig:joint-diagram-A2}
	\end{center}
\end{subfigure}
\hspace{1em}
\begin{subfigure}[t]{0.4\textwidth}
	\begin{center}
	\includegraphics[width=0.6\textwidth]{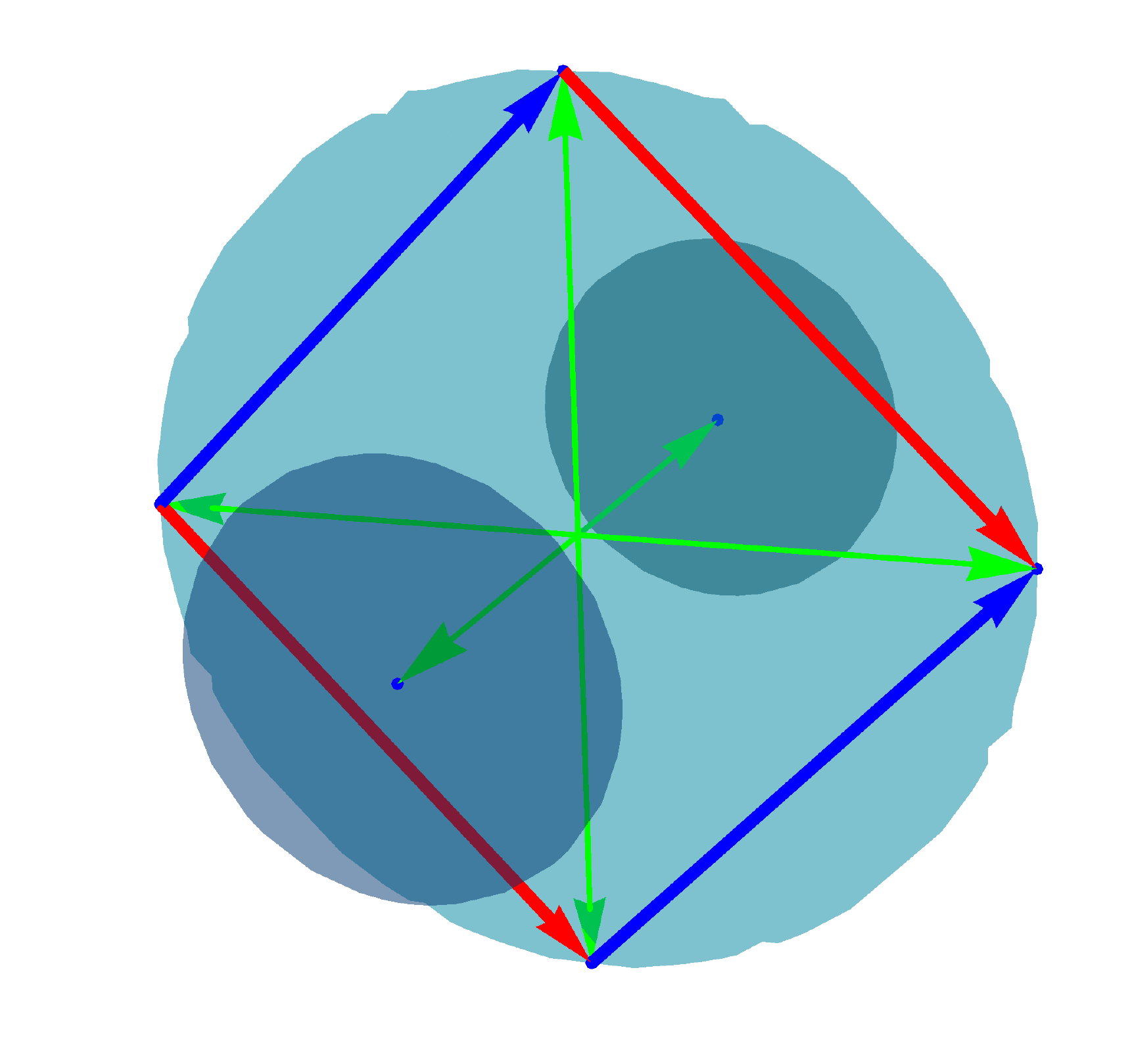}
	\caption{The slicing parallel to the $\alpha_{2}\IR\oplus\alpha_{3}\IR$ plane gives instead 3 Weyl orbits of $\gD_{2}$. Here $\alpha_{2}$ is red and $\alpha_{3}$ is blue.}
	\label{fig:joint-diagram-D2}	
	\end{center}
\end{subfigure}
\caption{Weyl orbits of the vector representation of D$_3$}
\label{fig:joint-diagrams}
\end{center}
\end{figure}
This decomposition of $\Lambda_\rho$ into Weyl orbits of subgroups of $W$ should be familiar, as it corresponds to the well-known {branching rules} of standard representation theory. 
For example, as a representation of the sub-algebra generated by $\alpha_{1},\alpha_{2}$, the vector representation of $\gD_3$ is reducible and branches into ${\bf 6\to 3\oplus \overline 3 }$,
hence the appearance of two triangles in Figure \ref{fig:joint-diagram-A2}. 
Similarly, $\alpha_{2},\, \alpha_{3}$ split $\Lambda_{\rho}$ into ${\bf 1\oplus 4 \oplus 1}$ of $\gD_{2}$, as displayed in Figure \ref{fig:joint-diagram-D2}.

%%%%%%%%%%%%%%%%%%%%%%%%%%%%%%%%%%%%%
\subsubsubsection*{Factorization of homotopy identities}
%%%%%%%%%%%%%%%%%%%%%%%%%%%%%%%%%%%%%

The schematic picture of soliton digrams can be put to good use: the appearance of several disjoint affine planes hints to a factorization property for the concatenations of solitons at a joint. 
We will now make this precise, and derive such a property.
For simplicity we discuss the case of 5-way joints, but it's straightforward to repeat the argument for 6-way joints.

Let $\alpha,\, \beta\in\Phi$ and consider the joint where $\CS$-walls of root-types $\CS_{\alpha},\,  \CS_{\beta}$ intersect. 
By a mild abuse of notation, we define for each street the following quantities
\be\label{eq:s-wall-generic-factor}
\begin{split}
	& \Xi_{\alpha}(p):=\sum_{(i,j)\in\CP_{\alpha}}\sum_{ \ehz{a}\in\Gamma_{ij}(p)}\mu(a)X_{a} \,,\\
	& \CS_{\alpha}(p):=1 + \Xi_{\alpha}(p)\,,
\end{split}
\ee
where $\Gamma_{ij}(p)$ is understood to be $\Gamma_{ij}(z)$ for any $z\in p$.
We consider paths $\wp,\, \wp'$ across the joint of $\CS_\alpha$ with $\CS_\beta$ as depicted in Figure \ref{fig:homotopy-invariance}, and study the the implications of homotopy invariance of the formal parallel transport
\be
	F(\wp,\CW) = F(\wp',\CW)\,.
\ee
\begin{figure}[h]
\begin{center}
\includegraphics[width=0.45\textwidth]{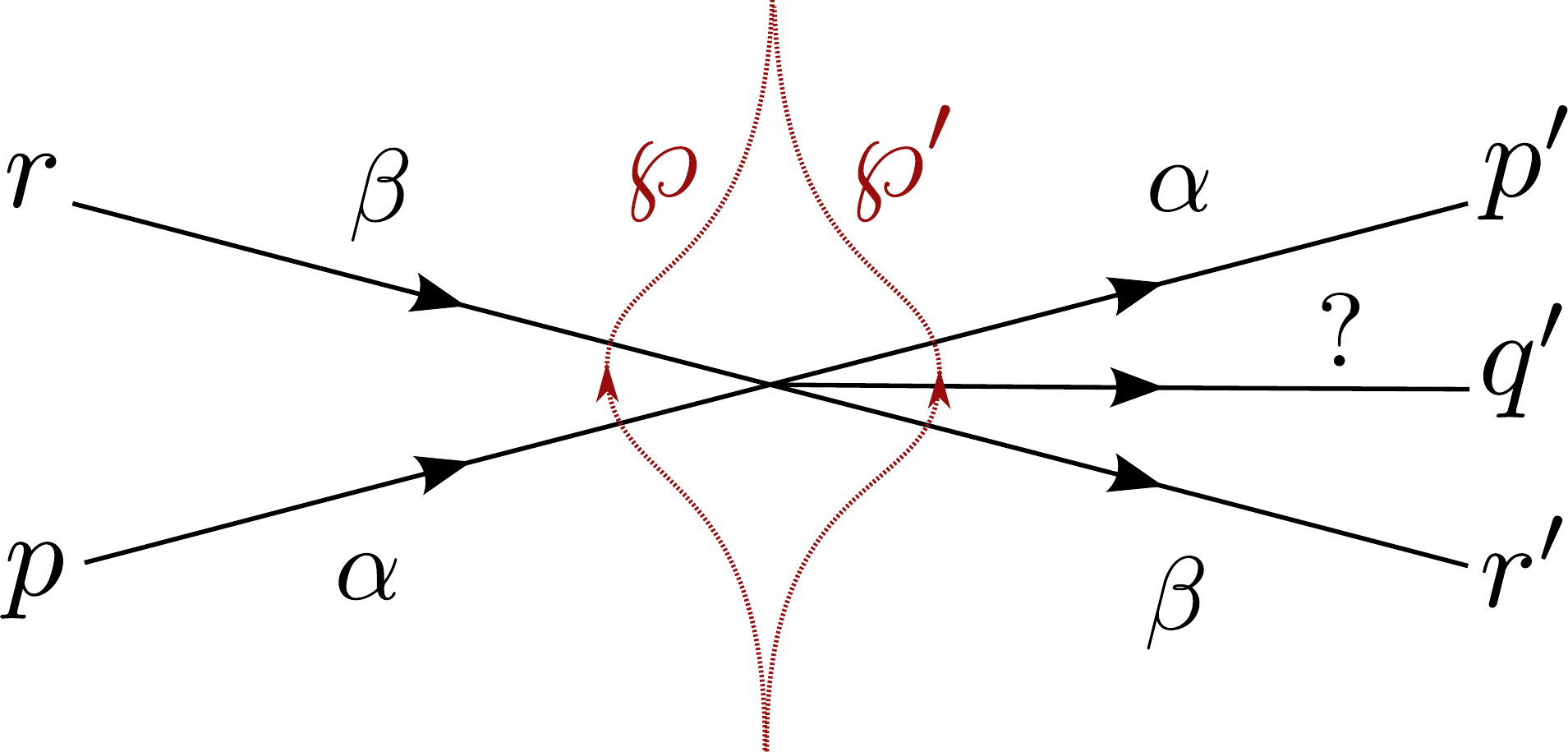}
\caption{Paths crossing the network on different sides of a joint. Homotopy invariance imposes relations between the ingoing and the outgoing soliton content. A priori, there may be more than three outgoing streets, and they might not be of root-type.}
\label{fig:homotopy-invariance}
\end{center}
\end{figure}
Rewriting the above equation in terms of detour rules yields
\be\label{eq:joint-problem}
	\CS_{\alpha}(p)\,\CS_{\beta}(r) \,=\, \CS_{\beta}(r') \, \CT_{(\alpha,\beta)} \, \CS_{\alpha}(p')\,
\ee
where $\CT_{(\alpha,\beta)}$ denotes the possible presence of several $\CS$-walls sourced by the joint (not necessarily walls of root-type).
Solving the constraint of homotopy invariance reduces to the problem of finding an expression for $\CT_{(\alpha,\beta)}$.

The picture of soliton diagrams suggests that we should consider an {orthogonal} decomposition 
\be
	\ft^{*}\simeq (\alpha\IR\oplus\beta\IR)\oplus\ft^{*,\perp} =\ft^{*,||}\oplus\ft^{*,\perp}\,,
\ee
and the corresponding decomposition of weights as 
$\weight =  \weight^{||} +  \weight^{\perp}$. 
Then all weights with the same orthogonal component $\weight^{\perp}$ belong to the same affine 2-plane $\IR^{2}_{\weight^{\perp}}$.
By definition of $\CP_{\alpha}$, a soliton of $\CS_{\alpha}$ runs from a sheet $i$ to another sheet $j$ such that $ \weight_{j}- \weight_{i} = \alpha$. Therefore $ \weight_{i}^{\perp} =  \weight_{j}^{\perp}$ and it makes sense to split solitons accordingly, e.g.
\be
	\Xi_{\alpha}(p) = \sum_{ \weight^{\perp}}\Xi_{\alpha}(p, \weight^{\perp}),
\ee
with the sum running over equivalence classes in $\Lambda_{\rho}$.
Adopting this splitting, the product rules of $X_a$ variables imply that
\be
	\Xi_{\alpha}(p, \weight^{\perp})\Xi_{\beta}(r, \weightB^{\perp}) = \delta_{ \weight^{\perp},\weightB^{\perp}}\Xi_{\alpha}(p, \weight^{\perp})\Xi_{\beta}(r, \weight^{\perp}) \,,
\ee
and therefore
\be
	\Xi_{\alpha}(p)\Xi_{\beta}(r) = %
	\sum_{ \weight^{\perp}}\sum_{\weightB^{\perp}}\Xi_{\alpha}( p,\weight^{\perp})\,\Xi_{\beta}(r,\weightB^{\perp}) %
	=  \sum_{ \weight^{\perp}}\Xi_{\alpha}(p, \weight^{\perp})\,\Xi_{\beta}(r, \weight^{\perp}) \,.
\ee
Moreover, noting that
\be
	\Xi_{\alpha}(p, \weight^{\perp}) \Xi_{\alpha}(r, \weightB^{\perp}) = 0 \qquad \text{if } \weight^{\perp}\neq \weightB^{\perp} \,
\ee
we can express
\be
	\CS_{\alpha}(p) %
	= \prod_{ \weight^{\perp}} \big(1+\Xi_{\alpha}( p, \weight^{\perp}) \big) %
	=: \prod_{ \weight^{\perp}}\CS_{\alpha}(p, \weight^{\perp})\,.
\ee
A similar expression holds for $\CS_\beta(r)$, allowing us to recast the LHS of the homotopy identity in the suggestive form
\be
	\CS_{\alpha}(p)\CS_{\beta}(r) = \prod_{ \weight^{\perp}} \Big[ \CS_{\alpha}( \weight^{\perp},p)\,\CS_{\beta}( \weight^{\perp}, r)  \Big]\,.
\ee
We then make an \emph{ansatz} for (\ref{eq:joint-problem}):
\begin{align}
	& \CT_{(\alpha,\beta)} = \prod_{ \weight^{\perp}}\CT_{(\alpha,\beta)}( \weight^{\perp}), \quad \text{with} \\
	& \CT_{(\alpha,\beta)}( \weight^{\perp}) \CS_{\alpha}(p',\mu^{\perp}) = \CS_{\alpha}(p', \mu^{\perp})  \CT_{(\alpha,\beta)}( \weight^{\perp}) \quad\text{if }\mu^{\perp}\neq \weight^{\perp},\\
	& \CT_{(\alpha,\beta)}( \weight^{\perp}) \CS_{\beta}(r', \mu^{\perp}) = \CS_{\beta}(r', \mu^{\perp})  \CT_{(\alpha,\beta)}( \weight^{\perp}) \quad\text{if }\mu^{\perp}\neq \weight^{\perp}.
\end{align}
This allows us to recast the RHS as
\be
	\CS_{\beta}(r')\CT_{(\alpha,\beta)}\CS_{\alpha}(p') = \prod_{ \weight^{\perp}} \Big[ \CS_{\beta}(r', \weight^{\perp})\,\CT_{(\alpha,\beta)}( \weight^{\perp})\, \CS_{\alpha}( p', \weight^{\perp})  \Big]
\ee
Thus the original homotopy equation (\ref{eq:joint-problem}) factorizes into a set of independent, more tractable equations:
\be
	\CS_{\alpha}( p, \weight^{\perp})\,\CS_{\beta}( r, \weight^{\perp}) = \CS_{\beta}( r', \weight^{\perp})\,\CT_{(\alpha,\beta)}( \weight^{\perp})\, \CS_{\alpha}(p', \weight^{\perp})\,.
\ee
There is one equation for each class of weights lying on the same affine 2-plane $\IR^2 _{\weight^{\perp}}$. 
Through the weight-sheet correspondence, each equation describes the combinatorics of concatenations for solitons which lie on the same plane $\IR_{\weight^\perp}^2$ in the soliton diagrams, where to a class of $(i,j)$ solitons we associate a vector, see Figure \ref{fig:joint-diagrams}.

In fact, these equations can be further factorized.
Let $\Lambda_{\rho}\big|_{ \weight^{\perp}}:= \IR^{2}_{ \weight^{\perp}} \cap \Lambda_{\rho}$ be the weight sub-system lying on $\IR^2_{\weight^{\perp}}$. This must be invariant under Weyl reflections generated by $\alpha$ and $\beta$. Therefore, $\Lambda_{\rho}\big|_{ \weight^{\perp}}$ consists of one or more Weyl orbits of $\gA_{2}$ (resp.\ $\gD_{2}$) when $\alpha\measuredangle\beta=\pi/3,\ 2\pi/3$ (resp.\ $\pi/2$). 
%\end{itemize}
A priori, $\Lambda_{\rho}\big|_{ \weight^{\perp}}$ may contain a single orbit of the subgroup $W_{2}\subset W$ generated by $w_\alpha$ and $w_\beta$, or several ones. 
In either case, it is clear that any $(i,j)$ soliton of $\CS_{\alpha}$ (or $\CS_\beta$) must connect a pair of sheets whose corresponding weights belong to the \emph{same} $W_{2}$ orbit.
This is because if $(i,j)\in\CP_\alpha$ then necessarily $\weight_j = w_\alpha \cdot\weight_i$, see (\ref{eq:weyl-paired}).
But now since an $\alpha$-type soliton $a\in\Gamma_{ij}$ can be concatenated with a $\beta$-type soliton $b\in\Gamma_{kl}$ only when $\weight_j=\weight_k$, this implies both that $\weight_{i}^{\perp} = \weight_{j}^{\perp} =  \weight_{l}^{\perp} $ and that
\be
	  \weight_{i},\, \weight_{j},\, \weight_{l} \text{ belong to the same $W_{2}$ orbit }o\,.
\ee
We then consider splitting  further
\be
	\Xi_{\alpha}( p, \weight^{\perp})=\sum_{o}\Xi_{\alpha}(p, \weight^{\perp},o)
\ee
(and similarly for other streets) into a sum over solitons connecting sheets/weights on different $W_{2}$-orbits. 
Then noting that
\be
	\Xi_{\alpha}(p, \weight^{\perp},o) \,\Xi_{\alpha}( p, \weight^{\perp},o') = 0 \qquad \text{if }o\neq o'\,,
\ee
we can rewrite 
\be
	\CS_{\alpha}(p) = \prod_{ \weight^{\perp},o}\big( 1 + \Xi_{\alpha}(p, \weight^{\perp},o) \big) = : \prod_{ \weight^{\perp},o}\CS_{\alpha}(p, \weight^{\perp},o) \,.
\ee
Correspondingly, the ansatz can be refined in the following fashion
\be
\begin{split}
	& \CT_{(\alpha,\beta)}( \weight^{\perp}) = \prod_{o}\CT_{(\alpha,\beta)}( \weight^{\perp},o)\,,
\end{split}
\ee
eventually reducing equation (\ref{eq:joint-problem}) to the following set of independent equations
\be \label{eq:final_ansatz}
	\CS_{\alpha}( p, \weight^{\perp},o)\,\CS_{\beta}( r, \weight^{\perp},o) = \CS_{\beta}(r',  \weight^{\perp},o)\,\CT_{(\alpha,\beta)}( \weight^{\perp},o)\, \CS_{\alpha}( p', \weight^{\perp},o)\,.
\ee
There is a distinct equation for each equivalence class of weights $( \weight^{\perp},o)$ in $\Lambda_{\rho}$, and we can focus on solving each of these independently.

%%%%%%%%%%%%%%%%%%%%%%%%%%%%%%%%%%%
\subsection{Joints of primary \texorpdfstring{$\CS$}{S}-walls}
\label{sec:primary-joints}
%%%%%%%%%%%%%%%%%%%%%%%%%%%%%%%%%%%

We now focus on joints of primary $\mathcal{S}$-walls, i.e.\ when the ingoing streets $p$ and $r$ of Figure \ref{fig:homotopy-invariance} are sourced by branch points of types $\alpha$ and $\beta$, respectively.
For primary $\mathcal{S}$-walls, the factorization property of the previous section holds since they are of root-type. 
Strictly speaking, this is only true for the LHS of the homotopy identity, which involves primary $\mathcal{S}$-walls, but doesn't need to be true for the RHS.
We will make the ansatz that the RHS factorizes as well, and show that the factorized homotopy identity (\ref{eq:final_ansatz}) indeed admits solutions.
We proceed by considering three different cases separately: $\alpha\measuredangle\beta=\pi/3$, $2\pi/3$, and $\pi/2$.

%%%%%%%%%%%%%%%%%%%%%%%%%%%%%%%%%%%%%%%%%%%%%%%%%%%%%%%%%%%%%%%%%%%%%%%%%%%%%%%%
\subsubsubsection*{Joints for $\alpha\measuredangle\beta = 2\pi/3$}
%%%%%%%%%%%%%%%%%%%%%%%%%%%%%%%%%%%%%%%%%%%%%%%%%%%%%%%%%%%%%%%%%%%%%%%%%%%%%%%%

Weyl orbits of $\gA_{2}$ may include $1,3$ or $6$ weights. 
The trivial orbit containing a single weight $\weight$ occurs on the plane $\IR_{\weight^{\perp}}$ if $ \weight\perp\alpha$ and $ \weight\perp\beta$, since in this case $ \weight^{||}=0$.
There will be $3$ weights when either $ \weight\perp\alpha$ \emph{or} $ \weight\perp\beta$. 
Otherwise, $\weight$ will fall into a (possibly squashed) hexagonal orbit on $\IR^2_{\weight^\perp}$. 

If $o$ is the trivial orbit, there is are no solitons and (\ref{eq:final_ansatz}) has the trivial solution
\be
	\CS_{\alpha}( p,\weight^{\perp}, o) = \CS_{\beta}( r,\weight^{\perp}, o)  = \CS_{\alpha}( p',\weight^{\perp}, o) = \CS_{\beta}( r',\weight^{\perp}, o) =\CT_{(\alpha,\beta)}( \weight^{\perp}, o) = 1\,.
\ee

If $o$ is a triangle, then it must be equilateral since $\mathfrak{g}$ is simply-laced, there are two cases shown in Figure \ref{fig:A2-joint}.

\begin{figure}[h!]
\begin{center}
\begin{subfigure}{0.20\textwidth}
\includegraphics[width=\textwidth]{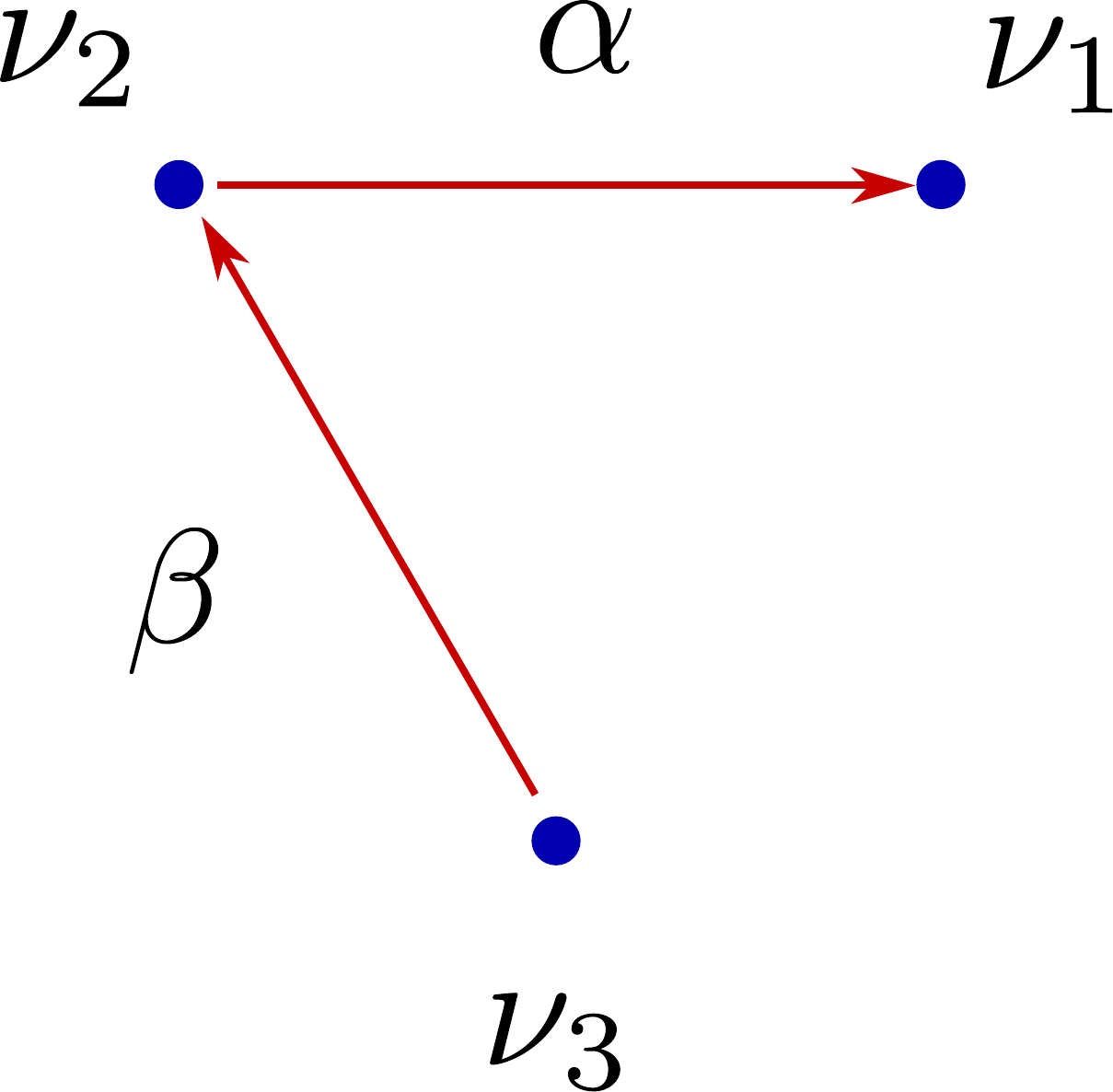}
\caption*{Case 1}
\end{subfigure}
\hspace{0.1\textwidth}
\begin{subfigure}{0.20\textwidth}
\includegraphics[width=\textwidth]{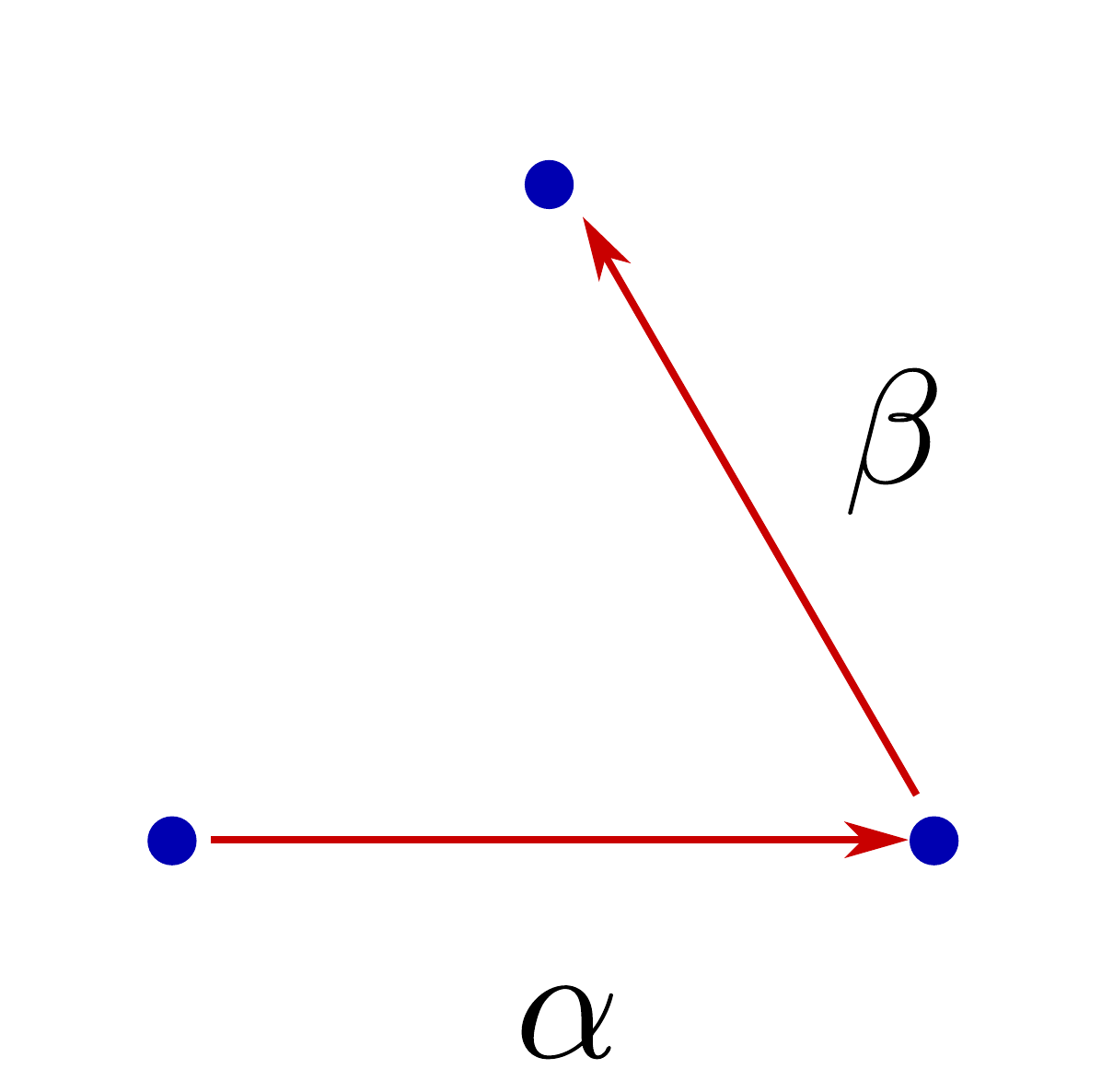}
\caption*{Case 2}
\end{subfigure}
\caption{Soliton diagrams for $\alpha\measuredangle\beta=2\pi/3$. 
}
\label{fig:A2-joint}
\end{center}
\end{figure}
\noindent In either case there is a simpleton $a\in\Gamma_{21}(p)$ with $\mu(a)=1$ supported on $\CS_{\alpha}$, as well as a simpleton $b\in\Gamma_{32}(r)$ with $\mu(b)=1$ supported on $\CS_{\beta}$. 
Both cases are familiar from fundamental $\gA$-type networks. For the first case we know that there will be a newborn wall of type $\alpha+\beta$ carrying solitons of type $(3,1)$, and that streets $p'$ and $r'$  will have the same soliton content as $p$ and $r$, respectively.
The two sides of the homotopy identity are thus
\be\label{eq:nontrivial-joint-soliton-flow}
\begin{split}
	\text{LHS}\,:\qquad & \CS_{\alpha}( p,\weight^{\perp}, o) \CS_{\beta}( r,\weight^{\perp}, o)  = 1+ X_{a} + X_{b} \\
	\text{RHS}\,:\qquad & \CS_{\beta}( r',\weight^{\perp}, o)\Big(1+\sum_{c\in\Gamma_{31}(q')}\mu(c)X_{c}\Big)\CS_{\alpha}( p',\weight^{\perp}, o) \\
	&\qquad\qquad\qquad\qquad\qquad = 1+ X_{a} + X_{b} +X_{ba} +\sum_{c\in\Gamma_{31}}\mu(c)X_{c} 
\end{split}
\ee
homotopy invariance demands that, given $c$ in the same class of $\Gamma_{31}(z)$
\be
	\mu(c) = - (-1)^{w(c,ab)}\,
\ee
and all other vanish. It turns our that the winding number is odd, therefore $\mu(c)=1$.
Unsurprisingly, we recover the (twisted) Cecotti-Vafa wall-crossing formula which appeared in the context of $\gA_{}$-type networks \cite{Gaiotto:2012rg}. The second case works out in a similar way.

Finally, if $o$ is a hexagon the situation is qualitatively different. However, {we show in Appendix \ref{app:no-hexagons} that this never occurs in minuscule representations of simply-laced Lie algebras.

%%%%%%%%%%%%%%%%%%%%%%%%%%%%%%%%%%%
\subsubsubsection*{Joints for  $\alpha\measuredangle\beta = \pi/3$}
%%%%%%%%%%%%%%%%%%%%%%%%%%%%%%%%%%%

In this case the various $\Lambda_{\rho}\big|_{ \weight^{\perp}}$ are of the same types as for the case of $\alpha\measuredangle\beta = 2\pi/3$, and $o$ is either the ${\bf 1, 3}$ or $\bf{\overline 3}$ of $A_{2}$. 
However, since $\alpha\measuredangle\beta=\pi/3$, there is no possible concatenation of solitons, which is evident in the soliton diagrams of Figure \ref{fig:angle-60}. In particular, given any
\be
\begin{split}
	& a\in\Gamma_{ij}(p)\,, \ b\in\Gamma_{kl}(r)\quad \text{for}\ \  (i,j)\in\CP_{\alpha},\ (k,\,l)\in\CP_{\beta},
\end{split}	
\ee
we always have $X_{a}X_{b}=0$. Similarly $X_{b'}X_{a'}=0$ for all solitons supported on $r'$ and $p'$, respectively.
The homotopy identity is then solved by taking $\CT_{(\alpha,\beta)}=1$, and by taking streets $p'$ and $r'$ to have the same soliton content as $p$ and $r$, respectively, which gives
\be\label{eq:trivial-joint-soliton-flow-A2}
	\CS_{\alpha}( p,\weight^{\perp}, o) \CS_{\beta}( r,\weight^{\perp}, o)   =\CS_{\beta}( r',\weight^{\perp}, o) \CS_{\alpha}( p',\weight^{\perp}, o).
\ee
Therefore these streets always form a 4-way joint.

\begin{figure}[h!]
\begin{center}
\begin{subfigure}{0.20\textwidth}
\includegraphics[width=\textwidth]{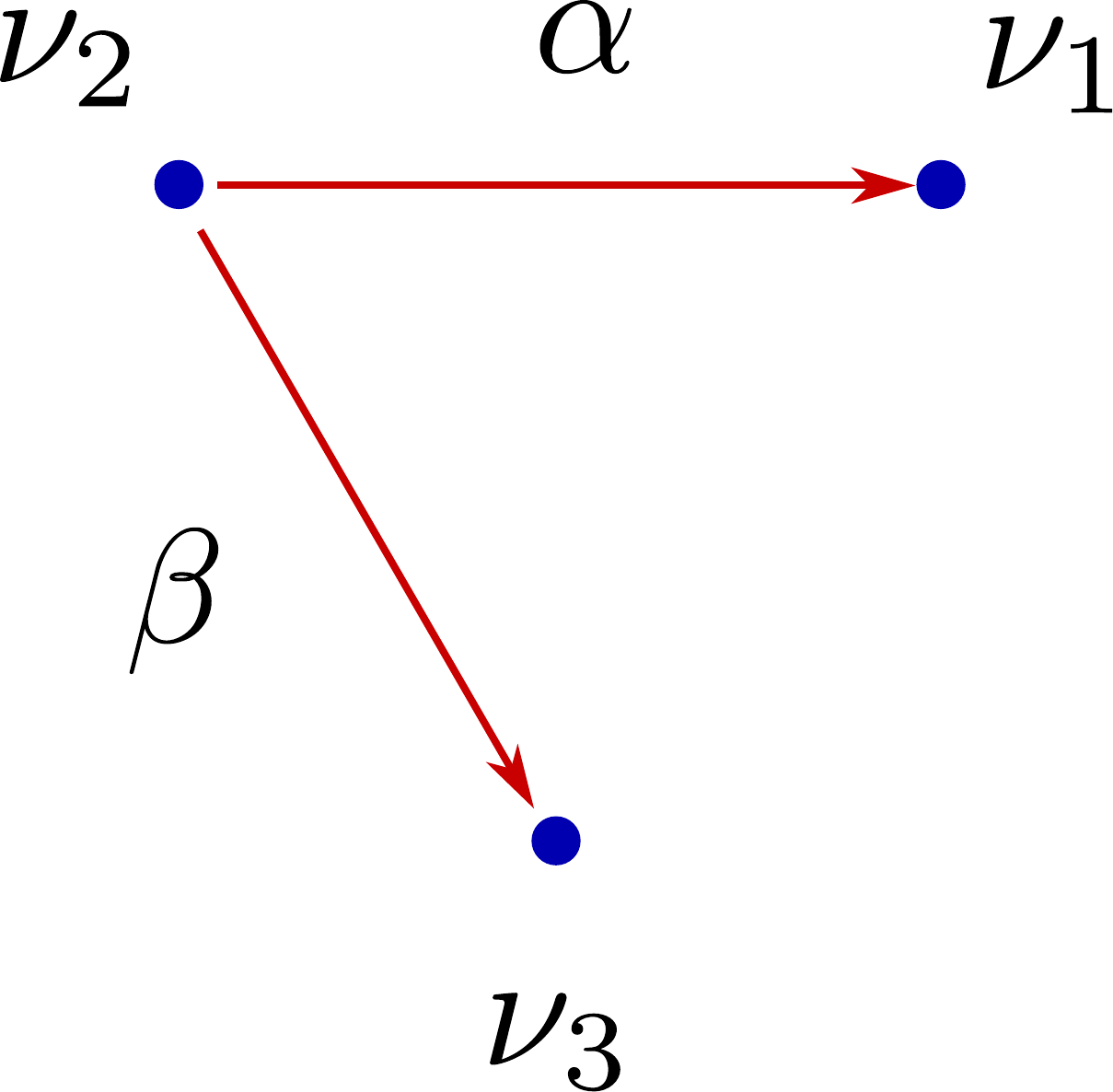}
\caption*{Case 1}
\end{subfigure}
\hspace{0.1\textwidth}
\begin{subfigure}{0.20\textwidth}
\includegraphics[width=\textwidth]{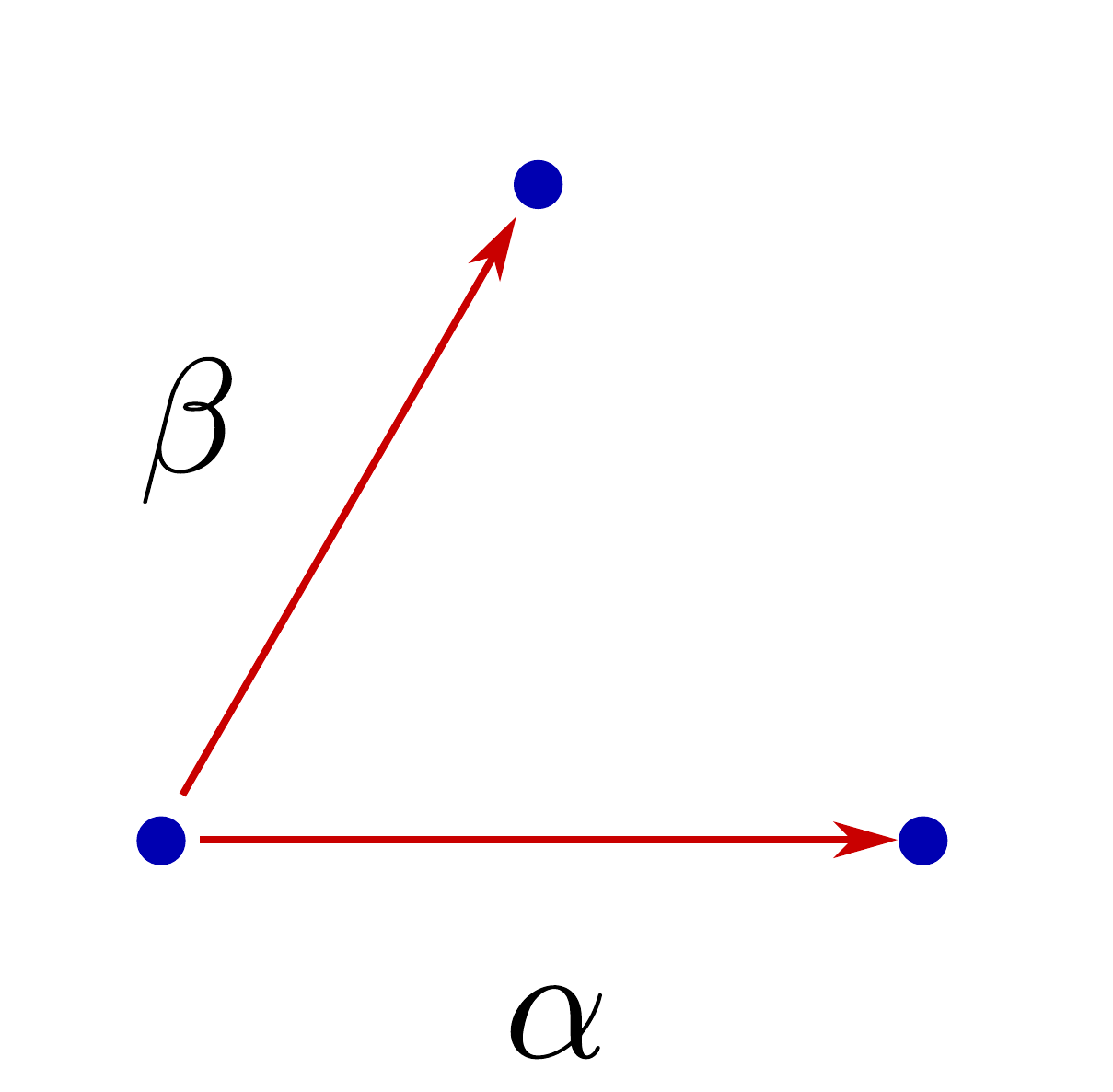}
\caption*{Case 2}
\end{subfigure}
\caption{Soliton diagrams for $\alpha\measuredangle\beta=\pi/3$. 
}
\label{fig:angle-60}
\end{center}
\end{figure}

%%%%%%%%%%%%%%%%%%%%%%%%%%%%%%%%%%%%%%%%%%%%%%%%%%%%%%%%%%%%%%%%%%%%%%%%%%%%%%%%
\subsubsubsection*{Joints for  $\alpha\measuredangle\beta = \pi/2$}
%%%%%%%%%%%%%%%%%%%%%%%%%%%%%%%%%%%%%%%%%%%%%%%%%%%%%%%%%%%%%%%%%%%%%%%%%%%%%%%%

The Weyl group of $D_{2}$ is $\IZ_{2}\times\IZ_{2}$, and its Weyl orbits may contain $1,2$ or $4$ weights.
There are five possible distinct cases which are displayed in Figure \ref{fig:angle-90}.
\begin{figure}[h!]
\begin{center}
\includegraphics[width=0.95\textwidth]{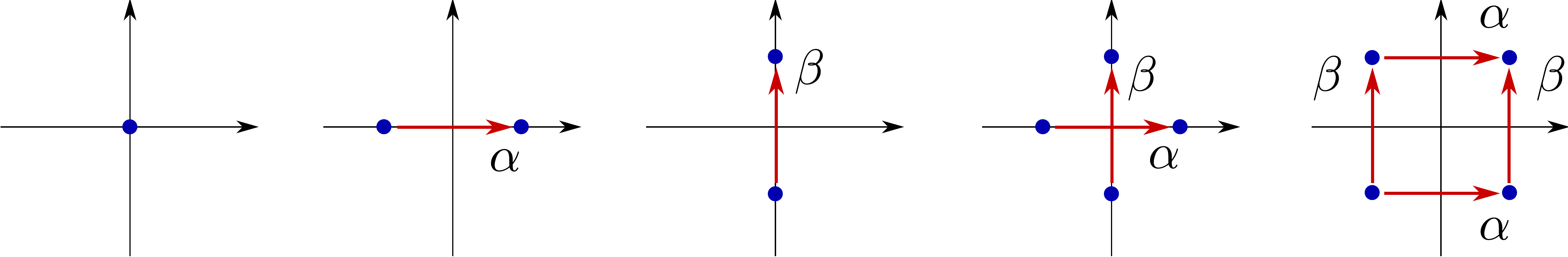}
\caption{Soliton diagrams for $\alpha\measuredangle\beta=\pi/2$.}
\label{fig:angle-90}
\end{center}
\end{figure}
From the figure it is evident that in any of the first four cases there is no possibility to concatenate solitons, much as in the case of $\alpha\measuredangle\beta=\pi/3$. 
Therefore we conclude that 
\be\label{eq:trivial-joint-soliton-flow-D2}
	\CS_{\alpha}( p,\weight^{\perp}, o) \CS_{\beta}( r,\weight^{\perp}, o)   =\CS_{\beta}( r',\weight^{\perp}, o) \CS_{\alpha}( p',\weight^{\perp}, o) \,,
\ee
for the first four cases.
The situation is only slightly more involved in the last case: since we are considering primary $\mathcal{S}$-walls, the soliton data on streets $p$, $r$ is simply
\be
\begin{split}
	& \CS_{\alpha}( p, \weight^{\perp},o) = 1+ X_{a_{1}} + X_{a_{2}},\	\CS_{\beta}( r, \weight^{\perp},o) = 1+ X_{b_{1}} + X_{b_{2}},
\end{split}
\ee
where
\be
\begin{split}
	& a_{1} \in \Gamma_{34}(z),\, a_{2}\in \Gamma_{21}(z),\,
	b_{1} \in \Gamma_{41}(z),\, b_{2}\in \Gamma_{32}(z).
\end{split}
\ee
This fixes the LHS of the homotopy identity. For the RHS, let us make the ansatz that streets $p'$ and $r'$ have the same soliton content as $p$ and $r$, respectively, and that $\CT_{(\alpha,\beta)}=1$. Then the two sides of the identity read
\be
\begin{split}
	\text{LHS}:\quad& \CS_{\alpha}( p, \weight^{\perp},o)  \CS_{\beta}( r, \weight^{\perp},o) = 1 + X_{a_{1}} + X_{a_{2}} + X_{b_{1}} + X_{b_{2}} + X_{a_{1}b_{1}}\\
	\text{RHS}:\quad& \CS_{\beta}(r',  \weight^{\perp},o)  \CS_{\alpha}( p', \weight^{\perp},o) = 1 + X_{a_{1}} + X_{a_{2}} + X_{b_{1}} + X_{b_{2}} + X_{b_{2}a_{2}}\,.
\end{split}
\ee
Now using the equivalence relation (\ref{eq:soliton-equivalence}), it is easy to see that
\be
	a_{1}b_{1}=:c_{1} \simeq c_{2} :=b_{2}a_{2}
\ee
since the soliton trees of simpletons are by definition identical, $\CT_{c_{1}} = \CT_{c_{2}}$. We thus find that (\ref{eq:trivial-joint-soliton-flow-D2}) holds in the last case as well.

%%%%%%%%%%%%%%%%%%%%%%%%%%%%%%%%%%%
\subsection{Descendant \texorpdfstring{$\mathcal{S}$}{S}-walls and generic joints}
\label{sec:generic-joints}
%%%%%%%%%%%%%%%%%%%%%%%%%%%%%%%%%%%
Having determined the soliton data of primary $\mathcal{S}$-walls across their mutual joints, we now move on to discuss descendant $\mathcal{S}$-walls and generic joints. While the soliton data of primary $\mathcal{S}$-walls includes just simpletons, for descendant $\mathcal{S}$-walls we have to work with a fully general soliton data.

The basic strategy is the following. We saw that, at a joint of primary $\mathcal{S}$-walls, all outgoing walls are labeled by roots. Here we consider a joint of generic $\CS$-walls, where all incoming streets are labeled by roots. 
We will see that for these joints all outgoing streets must also be of root-type, thus proving that all streets of the network are labeled by roots.
By adopting the factorization property, which is possible because we are working with root-type $\CS$-walls, the goal is to solve the factorized homotopy identity (\ref{eq:final_ansatz}). Once again it helps to consider different cases classified by $\alpha\measuredangle\beta$ one by one.

The analysis of the joints for $\alpha\measuredangle\beta=\pi/3$, $2\pi/3$ is a straightforward generalization of previous computations for joins of primary $\mathcal{S}$-walls.
In the case of $\alpha\measuredangle\beta=2\pi/3$, the equations for the parallel transport are very similar to the former ones, where the simpleton monomials $X_a, X_b$ in (\ref{eq:nontrivial-joint-soliton-flow}) are replaced by generic sums $\Xi_{\alpha}, \Xi_{\beta}$ as defined in (\ref{eq:s-wall-generic-factor}). The homotopy identity then expresses the outgoing solitons of type $\alpha+\beta$ in terms of concatenated solitons of of $\alpha,\beta$ types.
The corresponding analysis for $\alpha\measuredangle\beta=\pi/3$ easily yields the result that the joint must be trivial (i.e.\ a 4-way joint), beacuse no concatenations of soliton paths are actually possible for the incoming $\CS$-walls.

The case $\alpha\measuredangle\beta=\pi/2$ is considerably more involved, and requires a detailed analysis.
Again, the only nontrivial situation is that of the last frame in Figure \ref{fig:angle-90}, for which we would like to prove 
\be
	\Xi_{\alpha}(p, \weight^{\perp},o)\Xi_{\beta}( r, \weight^{\perp},o) = \Xi_{\beta}( r', \weight^{\perp},o)\Xi_{\alpha}( p', \weight^{\perp},o)\,.
\ee
This means that the soliton data of $p$ and $r$ are equal to that of $p'$ and $r'$, respectively. 
Then, by a small abuse of notation, the above can be recast into the suggestive form 
\be
	\big[\Xi_{\alpha}( \weight^{\perp},o)\,,\,\,\Xi_{\beta}( \weight^{\perp},o) \big]\,=\, 0\,.
\ee
After taking into account cancellations due to the product rule of formal $X$-variables, this amounts to proving that
\be
	\sum_{a\in\Gamma_{34}(p)}\sum_{b\in\Gamma_{41}(r)}\mu(a)\mu(b)\, X_{ab} = \sum_{b'\in\Gamma_{32}(r')}\sum_{a'\in\Gamma_{21}(p')} \mu(a')\mu(b')\, X_{b'a'}\,.
\ee
We claim that this is true, because the following holds: \\

\emph{For each $a\in\Gamma_{34}(p)$ there is an $a'\in\Gamma_{21}(p)$ such that $Z_{a'}=Z_{a}$ and $\mu(a)=\mu(a')$; likewise for each $b\in\Gamma_{41}(r)$ there is a $b'\in\Gamma_{32}(r)$ such that $Z_{b'}=Z_{b}$ and $\mu(b')=\mu(b)$. This ensures that each soliton $c=ab$ on the LHS is {equivalent} to a soliton $c'=b'a'$ on the RHS, with $\mu(c')=\mu(c)$, therefore establishing the equality}.\footnote{In comparing concatenations of solitons on $p$, $r$ with solitons on $p'$, $r'$ it is understood that we consider the lattices $\Gamma_{ij}(z)$ where $z$ is the location of the joint. }\\

To show this, we state and prove a general property of the soliton data of $\CS$-walls in the following.

%%%%%%%%%%%%%%%%%%%%%%%%%%%%%%%%%%%%%%%%%%%%%%%%%%%%%%%%%%%%%%%%%%%%%%%%%%%%%%%%
\subsection{Symmetries of soliton spectra}
\label{sec:soliton-symmetry}
%%%%%%%%%%%%%%%%%%%%%%%%%%%%%%%%%%%%%%%%%%%%%%%%%%%%%%%%%%%%%%%%%%%%%%%%%%%%%%%%

To prove our main result on the symmetry of soliton data, we will need a preliminary result, which we now state.

\begin{proposition}\label{prop:orthogonality}
Let ${\alpha}\in\Phi$ be any root vector of $\fg = \gA_n$, $\gD_n$, or $\gE_n$.
For any minuscule representation $\rho$ of $\fg$, any two weights $\weight_{i}$, $\weight_{i'}\in\Lambda_{\rho}$ with $i$, $i'\in\CP_{\alpha}^{-}$ satisfy 
\be
	\gamma = \weight_{i'}-\weight_{i}\perp\alpha
\ee
\end{proposition}

\begin{proof}
Let $i$, $i'\in\CP_{\alpha}^{-}$ and consider any $\beta\in\Phi, \,\beta\neq\pm\alpha$.
Then consider splitting of $\Lambda_\rho$ into orbits of $W_2$ as we have done previously.  
Denote by $o$ and $o'$ the orbits that include $\weight_i$ and $\weight_{i'}$, respectively.
We already know that each orbit is 
\begin{itemize}
\item a triangle or a point if $\alpha\measuredangle\beta=\pi/3$ or $2\pi/3$,
\item a square, a segment or a point if $\alpha\measuredangle\beta=\pi/2$.
\end{itemize}
Suppose $\alpha\measuredangle\beta=\pi/3$ or $2\pi/3$. If the orbit is a point, then $\CP_{\alpha}^{-}$ doesn't contain any weight from that orbit.
Then both $o$ and $o'$ must be triangles. There are two distinct cases, shown in Figure \ref{fig:perp-weights}.

In case (A), it's clear that $\gamma=\weight_{i'}-\weight_{i}=\weight^{\perp}_{i'}-\weight^{\perp}_{i}$, which is normal to $\alpha$ as claimed.\footnote{We used the fact that we always have $n=n'=1$, this is shown in Appendix \ref{app:no-hexagons} by an explicit analysis of all minuscule representations, see in particular (\ref{eq:W2-orbits-A-type}) for A-type, (\ref{eq:W2-orbits-D-type-vector}), (\ref{eq:W2-orbits-D-type-spinor}) and for D-type.}
In case (B) we have $\weight_{i'}^{\parallel}=(-1,-1/\sqrt{3})$ and $\weight_{i}^{\parallel}=(0,-2/\sqrt{3})$ according to our results from Appendix \ref{app:no-hexagons}, then since  $\alpha=(1,\sqrt{3})$ (as a vector in $\IR^2_{\weight^\perp}$) we find $\alpha\perp\gamma$.
Note that we haven't specified whether $\alpha\measuredangle\beta=2\pi/3$ or $\pi/3$, the argument applies to both cases.

\begin{figure}[ht]
\begin{center}
\begin{subfigure}{0.380\textwidth}
\includegraphics[width=\textwidth]{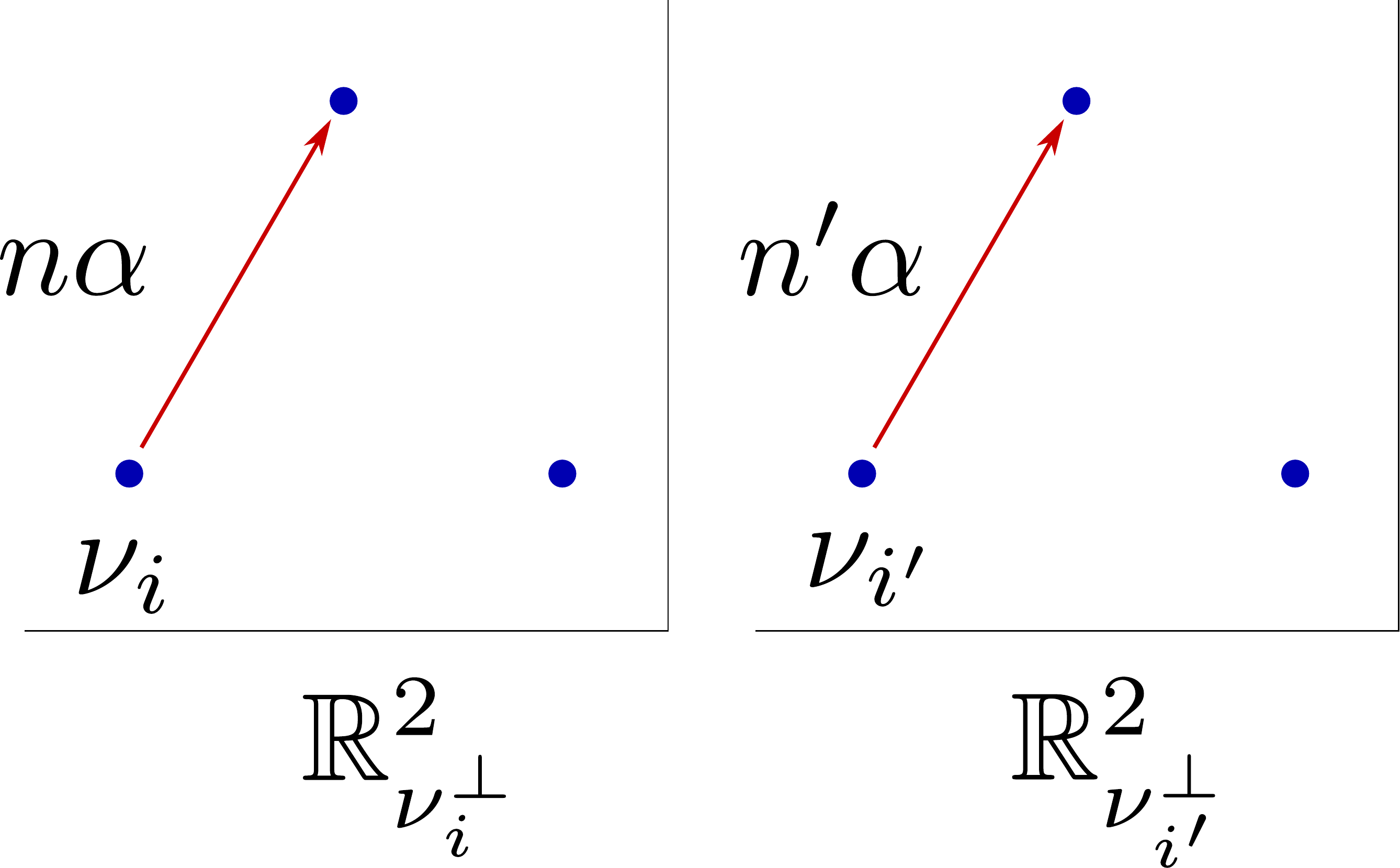}
\caption*{Case A}
\end{subfigure}
\hspace{0.1\textwidth}
\begin{subfigure}{0.380\textwidth}
\includegraphics[width=\textwidth]{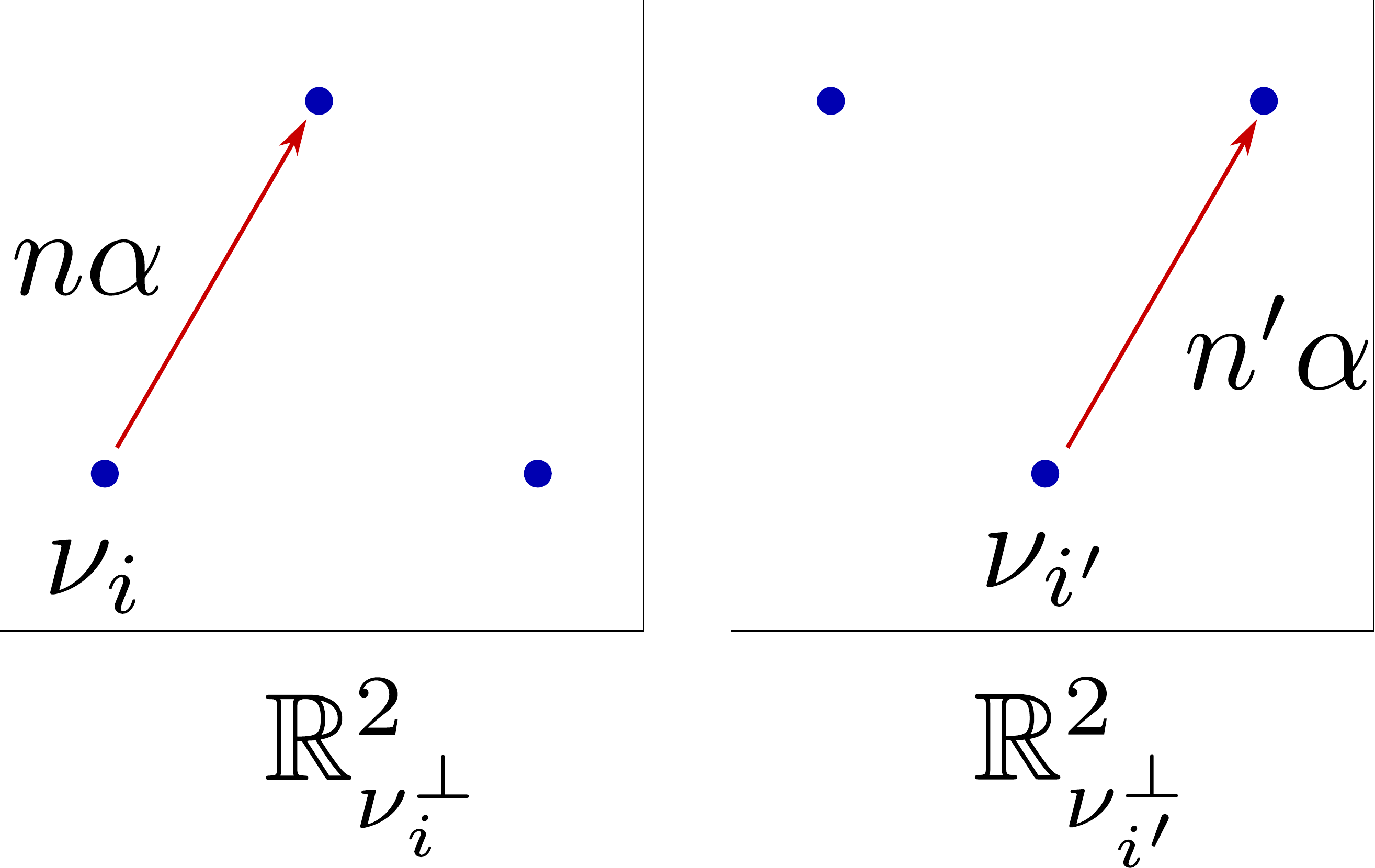}
\caption*{Case B}
\end{subfigure}
\caption{Two possible cases when $i,i'\in\CP_{\alpha}^{-}$ with $\alpha\measuredangle\beta=2\pi/3$ or $\pi/3$. 
}
\label{fig:perp-weights}
\end{center}
\end{figure}

The case $\alpha\measuredangle\beta=\pi/2$ is even simpler: if both $o$ and $o'$ are either points or segments, then $\weight_{i}$, $\weight_{i'}$ are always separated by $\gamma\perp\alpha$. 
The only nontrivial case is when one of $o$ and $o'$, or both of them, is a square. 
But in this case we must have 
\be
	\gamma\cdot\alpha = (\weight_{i'}^{\parallel} - \weight_{i}^{\parallel})\cdot\alpha = \#\beta\cdot\alpha
\ee
where $\#$ is a factor of $0, \pm1/2$ or $\pm 1 $.
Since $\beta\perp\alpha$, we find that our claim holds true.
\end{proof}

We are now ready prove the symmetries of the soliton data of $\CS$-walls.

\begin{proposition}\label{prop:soliton-symmetry}
Let $p$ be any street on  a root-type $\CS$-wall  $\CS_{\alpha}$, and let $(i,j), (i',\,j')\in\CP_{\alpha}$ be two distinct pairs. Then, for any soliton $a\in\Gamma_{ij}(p)$, there is an \underline{inequivalent} soliton $a'\in\Gamma_{i'j'}(p)$ (in the sense of (\ref{eq:soliton-equivalence})), with $Z_{a'}=Z_{a}$ and $\mu(a')=\mu(a)$.
\end{proposition}

\begin{proof}
The property is easily proved for primary $\CS$-walls. Their soliton data include only simpletons, which all have the same degeneracy. 
Moreover, if a simpleton $a$ stretches between sheets $ x_{i}$ and $ x_{j}$ with 
\be
	 \lambda_{j}- \lambda_{i}=\langle\weight_{j}-\weight_{i},\varphi\rangle=n\langle\alpha,\varphi\rangle	
\ee
we likewise have 
\be
	 \lambda_{j'}- \lambda_{i'}=\langle\weight_{j'}-\weight_{i'},\varphi\rangle = n' \,\langle\alpha,\varphi\rangle\,.
\ee
Then, from Appendix \ref{app:no-hexagons} we know that $n=n'=1$, which implies $\CT_{a}=\CT_{a'}$ and therefore $Z_{a}=Z_{a'}$, as desired. 
Note that $a$ and $a'$ are not equivalent because they do not belong to the same charge lattice, and this is in line with the statement of our proposition. 

To show the property for descendant $\CS$-walls, we need to consider what happens at joints. As we have seen, all $4$-way and $5$-way joints preserve the soliton spectrum of the incoming $\CS$-walls as they come out of the joint. So we only need to worry about the newborn $\CS$-wall at the $5$-way joint. 
Let us therefore consider a joint between $\CS_{\alpha}$ and $\CS_{\beta}$ with $\alpha\measuredangle\beta=2\pi/3$, producing a newborn wall $\CS_{\alpha+\beta}$, see e.g. Figure \ref{fig:joint-types-5-way}. We want to prove the property above for the soliton content of $\CS_{\alpha+\beta}$. 
The soliton content of the newborn wall $\CS_{\alpha+\beta}$ is encoded in the generating function $\Xi_{\alpha+\beta}(q')$, which is computed by the twisted Cecotti-Vafa wall-crossing formula\footnote{We suppressed the labeling by streets, to avoid cluttering notation. In equation (\ref{eq:twisted-CV}) we implicitly made use of the fact that the soliton data of streets $p$ and $r$ are the same as those of $p'$ and $r'$, respectively, when we write the commutator $[\Xi_{\alpha},\Xi_{\beta}]$.}
\be\label{eq:twisted-CV}
	\Xi_{\alpha+\beta} = [\Xi_{\alpha},\Xi_{\beta}] = \Xi_{\alpha}\Xi_{\beta} - \Xi_{\beta}\Xi_{\alpha}\,.
\ee
So let $c\in\Gamma_{ik}(q')$ be any soliton of $\CS_{\alpha+\beta}$, obtained by composing $c=ab$ with 
\be
	a\in\Gamma_{ij}(p),\ (i,j)\in\CP_{\alpha}\,;\ b\in\Gamma_{jk}(r),\ (j,k)\in\CP_{\beta}\,.
\ee
What we wish to prove is that, given any other pair $(i',k')\in\CP_{\alpha+\beta}$, there is a $c'\in\Gamma_{i'k'}(q')$ such that $\mu(c')=\mu(c)$ and $Z_{c'}=Z_{c}$. 
As an inductive hypothesis, we shall assume this property to hold for the soliton data of $\CS_{\alpha}$, $ \CS_{\beta}$.

As we saw several times previously, if $\alpha\measuredangle \beta =2\pi/3$ then $W_2$ orbits must be either triangles or points.
If $(i',k')\in\CP_{\alpha+\beta}$, then neither of $\weight_{i'}$, $\weight_{k'}$ can belong to an point-like orbit: if this were the case, it would imply that $\weight_{i'}\perp\alpha,\,\beta$ and therefore $\weight_{i'}\perp\alpha+\beta$, which contradicts the fact that $i'\in\CP_{\alpha+\beta}^{-}$ (a similar argument applies to $\weight_{k'}$).
This proves that $\nu_{i'}$ and $\nu_{k'}$ must belong to a triangle orbit $o'$, and there are only two possible ways to accommodate this, depicted in Figure \ref{fig:scenarios}:
\begin{enumerate}
\item[(a)] $i'\in\CP_{\alpha}^{-}\cap\CP_{\beta}^{0}, \,k'\in\CP_{\beta}^{+}\cap\CP_{\alpha}^{0}$, then $\Lambda_{\rho}$ contains a weight $\weight_{j'}\in\CP_{\alpha}^{+}\cap\CP_{\beta}^{-}$
\item[(b)] $i'\in\CP_{\beta}^{-}\cap\CP_{\alpha}^{0}, \,k'\in\CP_{\alpha}^{+}\cap\CP_{\beta}^{0}$ then $\Lambda_{\rho}$ contains a weight $\weight_{j'}\in\CP_{\beta}^{+}\cap\CP_{\alpha}^{-}$
\end{enumerate}
By the inductive hypothesis, in case (a) we there will be solitons $a'\in\Gamma_{i'j'}(p),\,b'\in\Gamma_{j'k'}(q)$ with $Z_{a'}=Z_{a}$ and $Z_{b'}=Z_{b}$ as well as $\mu(a)=\mu(a')$ and $\mu(b)=\mu(b')$. 
Then their concatenation $c'=a'b'$ would be part of the soliton content of the street $q'$, according to the analysis of joints from Section \ref{sec:generic-joints}.\footnote{Note that this is a statement about joints for $\alpha\measuredangle\beta=2\pi/3$, which does not depend on the property we are currently proving.}
Explicitly, the homotopy identity includes the following term
\be
\begin{split}
	& \mu(c')X_{c'} = \mu(a')\mu(b')X_{a'}X_{b'}\subset \Xi_{a}(p)\Xi_{b}(r)\subset\Xi_{\alpha+\beta}(q')
\end{split}
\ee
with $c'=a'b'$. Since 
\be
\begin{split}
	& \mu(c') = \mu(a')\mu(b') = \mu(a)\mu(b) = \mu(c) \\
	& Z_{c'} = Z_{a'}+Z_{b'} = Z_{a}+Z_{b}=Z_{c} 
\end{split}
\ee
the soliton $c'$ is precisely the one we were looking for.

In case (b) the roles are simply reversed, and the extra sign from equation (\ref{eq:twisted-CV}) accounts for the correct extra winding in concatenating $c'=b'a'$.

\end{proof}

\begin{figure}[ht]
\begin{center}
\begin{subfigure}{0.380\textwidth}
\includegraphics[width=\textwidth]{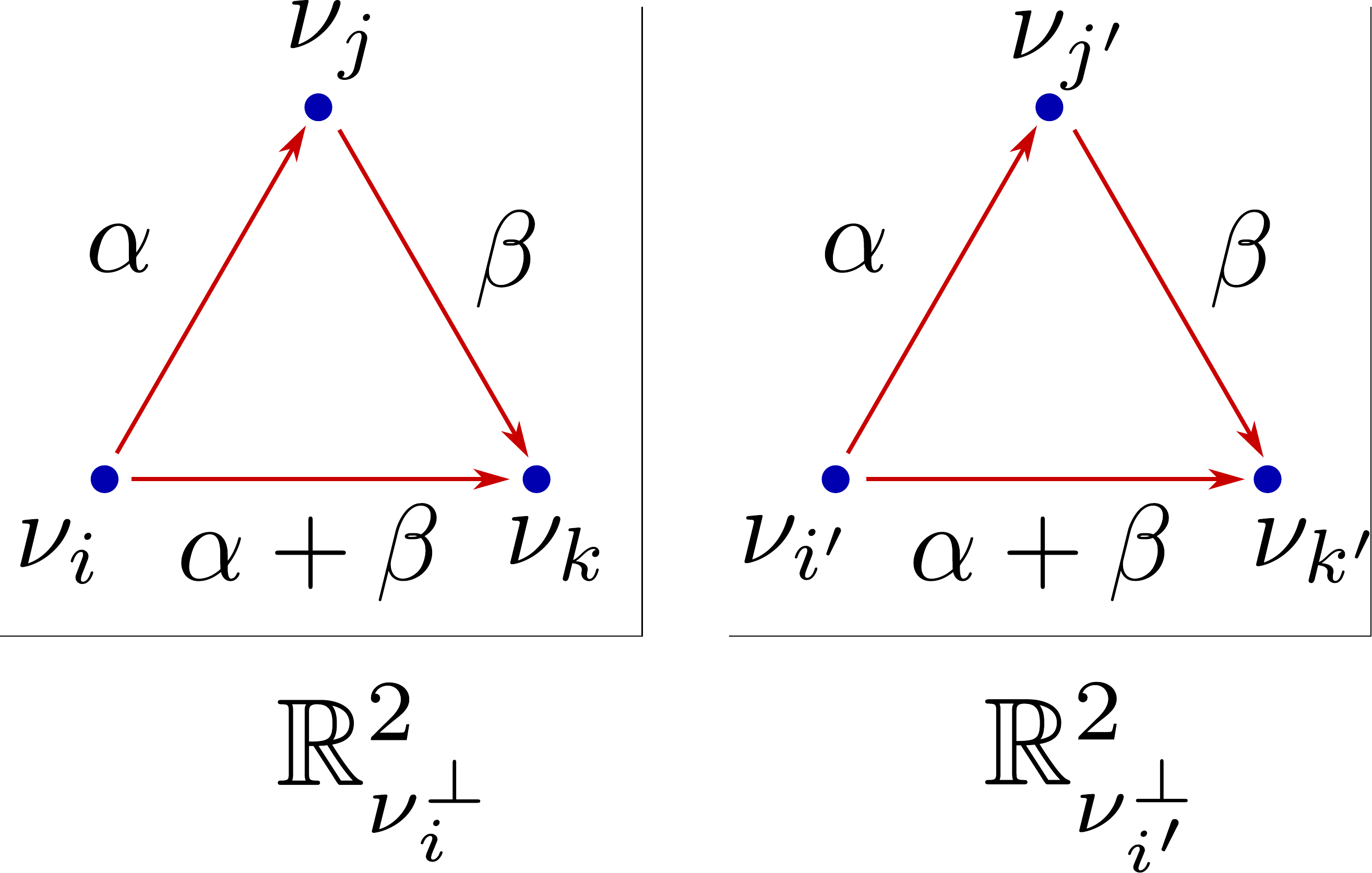}
\caption*{case (a)}
\end{subfigure}
\hspace{0.1\textwidth}
\begin{subfigure}{0.380\textwidth}
\includegraphics[width=\textwidth]{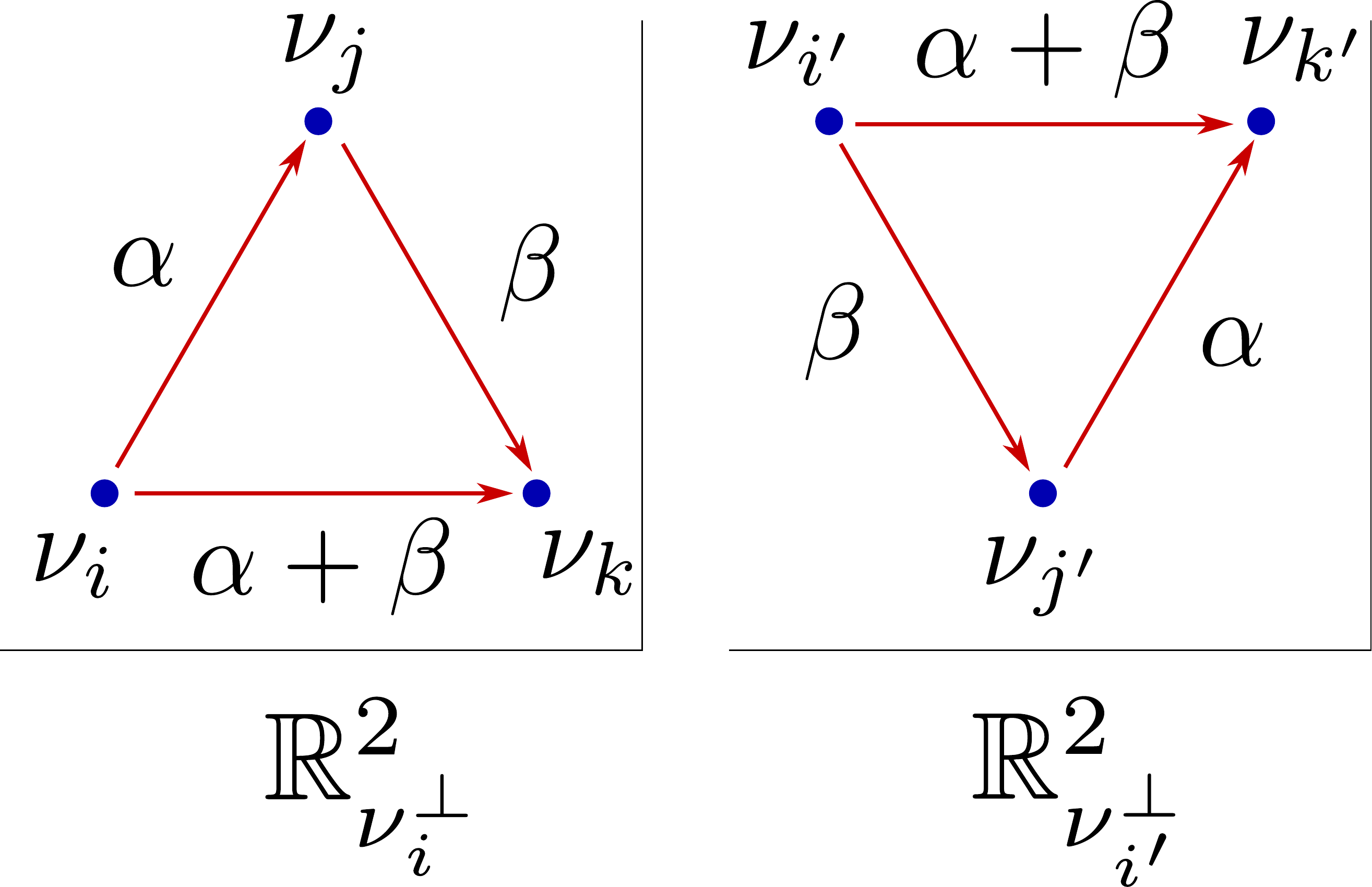}
\caption*{case (b)}
\end{subfigure}
\caption{The two possible cases: we show the different weight sub-systems at $\IR^{2}_{\weight_{i}^{\perp}}$ and at $\IR^{2}_{{\weight'_{i}}{}^{\perp}}$. In both cases $(ik), (i'k')\in\CP_{\alpha+\beta}$. 
}
\label{fig:scenarios}
\end{center}
\end{figure}

\smallskip

The symmetry of solitons supported by $\CS$-walls simplifies the problem of computing the propagation of soliton data across a network $\CW$ and plays an important role in Section \ref{subsec:k-wall-jumps}. 
We can rephrase the symmetry property in terms of the detour rule: for the Stokes factor of an $\CS$-wall,
\be\label{eq:stokes-factor-def}
	\CS_{\alpha} = 1 +\Xi_{\alpha} = 1+\sum_{(i,j)\in\CP_{\alpha}}\Xi_{ij},
\ee
the symmetry establishes a relation among all the $\Xi_{ij}$'s.
In particular, for $\Xi_{ij}$ and $\Xi_{i'j'}$ that are two formal series counting solitons charges $a\in\Gamma_{ij}$, $a'\in\Gamma_{i'j'}$ in different homology classes, the symmetry says that each series has solitons that have the same trees $\CT_{a} = \CT_{a'}$, the same central charges $Z_{a}=Z_{a'}$, and the same degeneracies $\mu(a)=\mu(a')$ as solitons from other series.

%%%%%%%%%%%%%%%%%%%%%%%%%%%%%%%%%%
\section{\texorpdfstring{$\CK$}{K}-wall jumps and 4d BPS states}\label{subsec:k-wall-jumps}
%%%%%%%%%%%%%%%%%%%%%%%%%%%%%%%%%%

\begin{figure}[ht]
\begin{center}
	\begin{subfigure}[b]{.31\textwidth}
		\includegraphics[width=\textwidth]{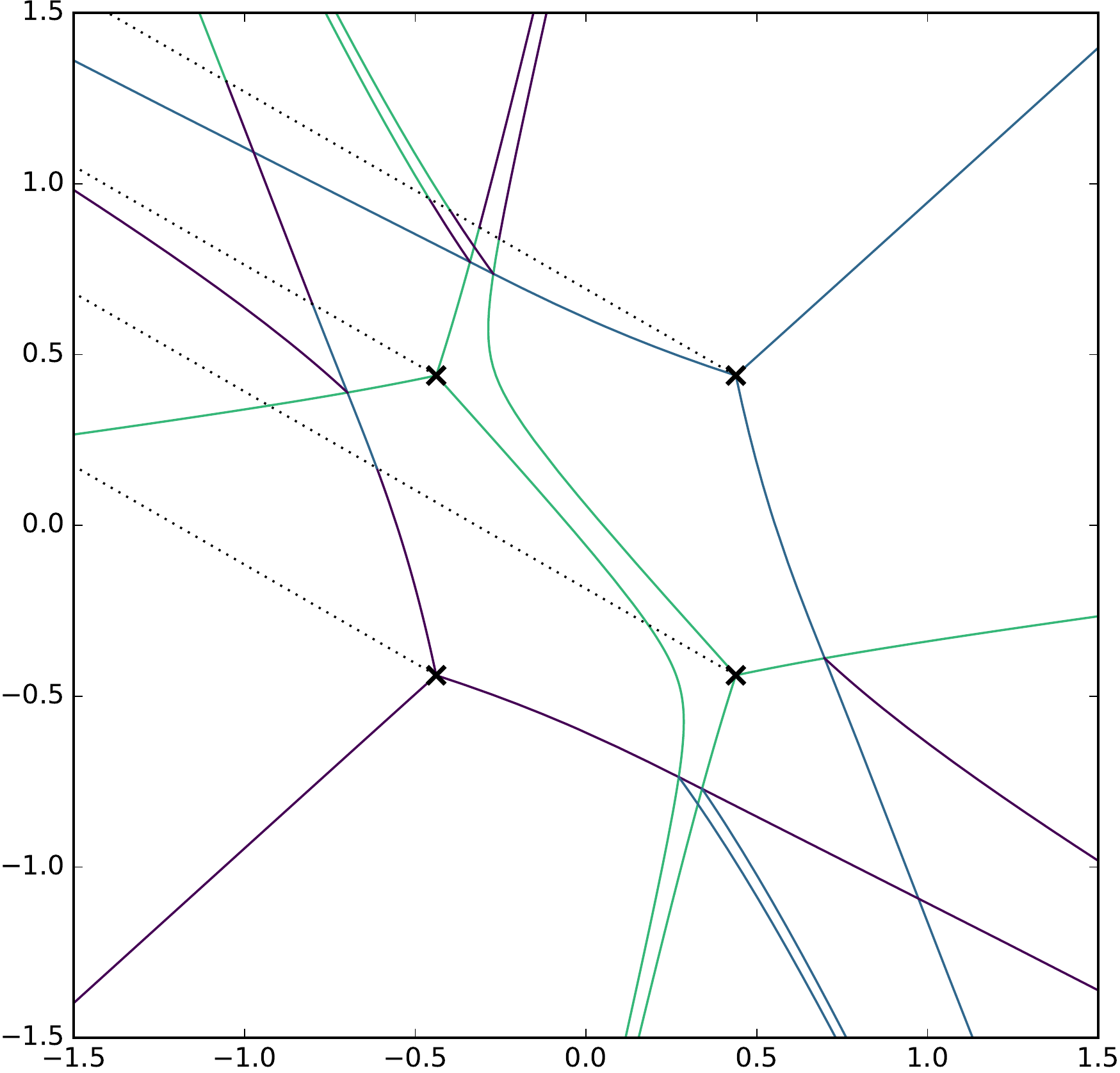}
		\caption{$\vartheta < \vartheta_\mathrm{c}$}	
		\label{fig:W_before_theta_c}
	\end{subfigure}
	\begin{subfigure}[b]{.31\textwidth}
		\includegraphics[width=\textwidth]{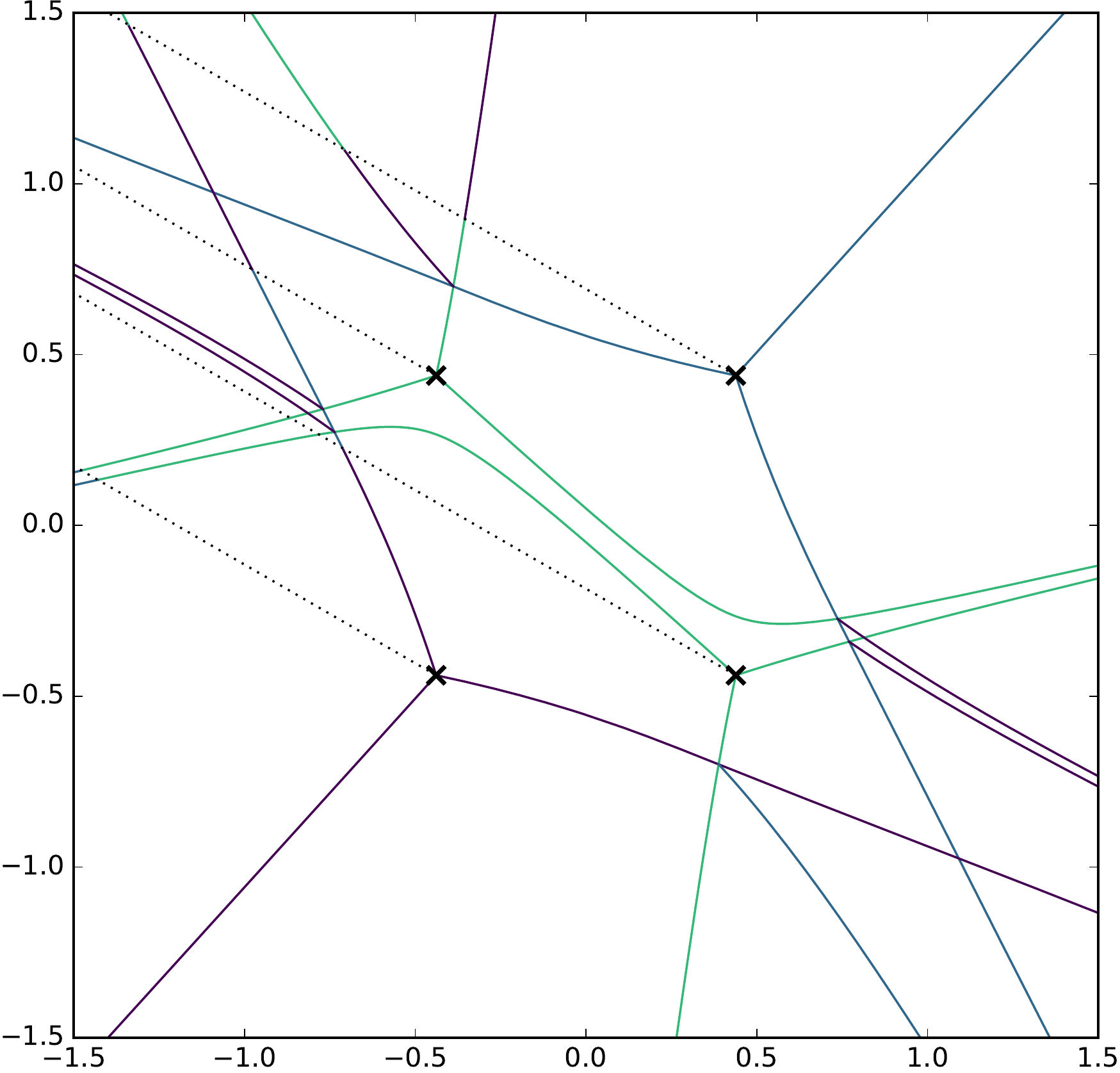}
		\caption{$\vartheta > \vartheta_\mathrm{c}$}		
		\label{fig:W_after_theta_c}
	\end{subfigure}
	\begin{subfigure}[b]{.31\textwidth}
		\includegraphics[width=\textwidth]{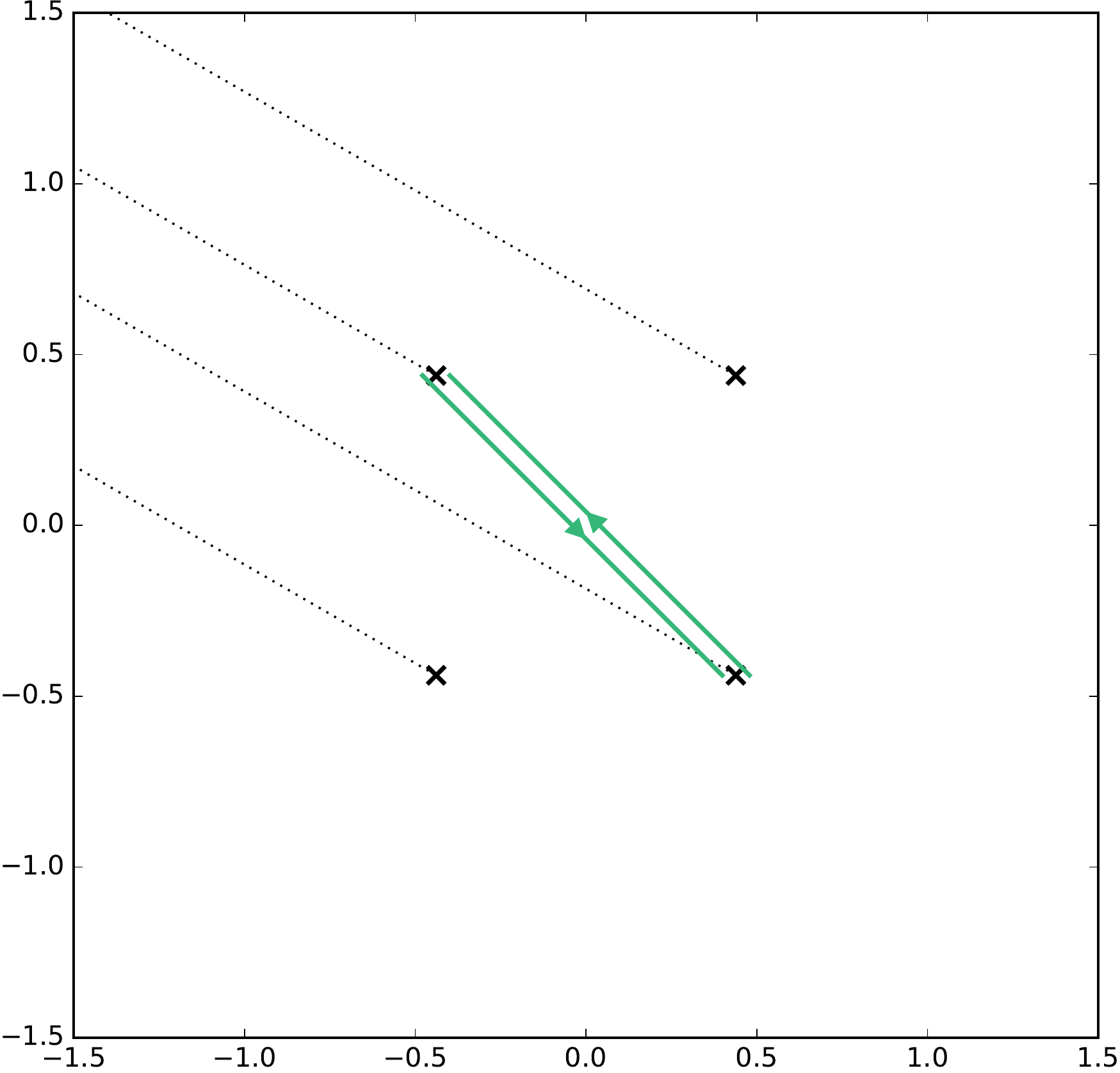}
		\caption{$\mathcal{W}_\mathrm{c}$}		
		\label{fig:W_c}
	\end{subfigure}	
\caption{Spectral networks around $\vartheta = \vartheta_\mathrm{c}$}
\label{fig:critical_spectral_network}
\end{center}
\end{figure}

So far we have been working with a spectral network $\CW$ at a fixed but generic $\vartheta$, which determines the geometry of $\CS$-walls on $C$ via (\ref{eq:geodesic-eq}). 
Now we consider a 1-parameter family of networks $\CW_{\vartheta}$ obtained by varying $\vartheta$. 
As in the standard $\gA_n$-type networks, we find that the formal generating function $F(\wp,\CW_{\vartheta})$ is piecewise constant in $\vartheta$ and exhibits jumps at some values $\vartheta = \vartheta_\mathrm{c}$.\footnote{There are actually two types of jumping behaviors that occur. 
One happens when, as $\vartheta$ changes, some $\CS$-wall crosses an endpoint of $\wp$. The description of this jump is captured by the detour rules described in the previous section, and its physical interpretation is as wall-crossing of framed 2d BPS states, induced by the change in the moduli of the supersymmetric interface $L_{\wp,\vartheta}$ (see \cite{Gaiotto:2011tf, Gaiotto:2012rg} for details). The second type of jump, which is the main subject of this section, is a topological jump of $\CW$ itself.} Figures \ref{fig:W_before_theta_c} and \ref {fig:W_after_theta_c} show spectral networks before and after a jump, respectively. These are from the spectral network of the original Argyres-Douglas fixed point theory from the 4d pure $\SU(3)$ gauge theory \cite{Argyres:1995jj}, whose full spectral networks can be found at \foothref{http://het-math2.physics.rutgers.edu/loom/plot?data=AD_from_pure_SU_3}{this web page}.

Interpreted through the picture of (framed) 2d-4d wall-crossing \cite{Gaiotto:2010be, Gaiotto:2011tf, Gaiotto:2012rg}, these jumps detect 4d BPS states with central charges of phase $\arg(Z) = \vartheta_\mathrm{c}$. Here we explain how to read out the IR charge, the central charge, and the degeneracy of 4d BPS states from a spectral network $\mathcal{W}_{\vartheta_\mathrm{c}}$.

%%%%%%%%%%%%%%%%%%%%%%%%%%%%%%%%%%%
\subsection{Two-way streets}\label{subsubsec:two-way-streets}
%%%%%%%%%%%%%%%%%%%%%%%%%%%%%%%%%%%
\begin{figure}[ht]
\begin{center}
	\begin{subfigure}{.25\textwidth}
		\includegraphics[width=\textwidth]{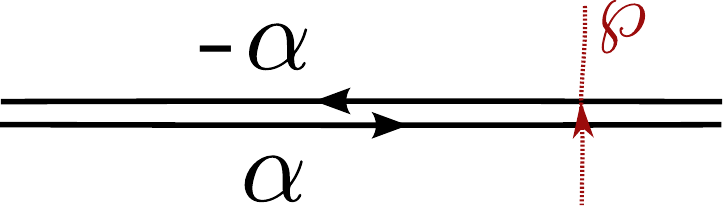}
		\caption*{$\vartheta < \vartheta_\mathrm{c}$}	
		\label{fig:2-way-res-before}
	\end{subfigure}
	\hspace{0.05\textwidth}
	\begin{subfigure}{.25\textwidth}
		\includegraphics[width=\textwidth]{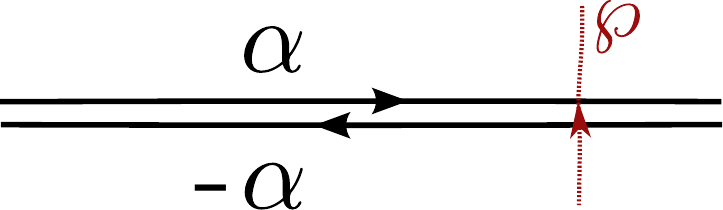}
		\caption*{$\vartheta > \vartheta_\mathrm{c}$}		
		\label{fig:2-way-res-after}
	\end{subfigure}
\caption{Two resolutions of a two-way street. On the left the American resolution, on the right the British one. 
}
\label{fig:2-way-resolutions}
\end{center}
\end{figure}

When $\CW_{\vartheta}$ undergoes a topological jump at a critical phase  $\vartheta_\mathrm{c}$, 
two $\mathcal{S}$-walls of opposite root types, $\CS_{\alpha}$ and $\CS_{-\alpha}$, come to overlap along a segment, called a \emph{two-way} street \cite{Gaiotto:2012rg}. Let $\CW_\mathrm{c}$ be a sub-network of $\CW_{\vartheta_\mathrm{c}}$ that consists of two-way streets only, as shown in Figure \ref{fig:W_c}. 
By perturbing $\vartheta$ away from $\vartheta_\mathrm{c}$ one obtains two distinct resolutions of $\CW_\mathrm{c}$, depicted in Figure \ref{fig:2-way-resolutions}. 

Working in either resolution, we wish to study $F(\wp,\CW_{\vartheta^\pm})$ for a path $\wp$ crossing a two-way street $p$. 
We are going to combine the soliton data of each one-way street of $p$ into a new kind of generating function that involves formal variables  $X_{\ehz{{\gamma}}}$, with $\ehz{{\gamma}} \in H_1(\widetilde\Sigma_\rho,\IZ)/(2H\times \ker(Z))$.\footnote{The central charge $Z$ is defined on $H_1(\Sigma_\rho,\IZ)$ strictly speaking. The same technicality was confronted in equation (\ref{eq:soliton_ij_lattice}), identical considerations apply here.} 
The multiplication rules are extended to these variables as\footnote{For certain purposes it may be important to twist these product rules as $X_{\ehz{\gamma}} X_{\ehz{\gamma}'} = (-1)^{\langle\ehz{\gamma},\ehz{\gamma}'\rangle} X_{\ehz{\gamma}+\ehz{\gamma}'}$. The apparent difficulty here is that there seems to be no natural notion of intersection pairing on the quotient lattice we are considering. Below we will argue that it is possible to define a natural pairing.} 
\be
	X_{\ehz{{\gamma}}}X_{a} = X_{\ehz{{\gamma}}+a}\,,\quad X_{\ehz{{\gamma}}}X_{\ehz{\gamma}'}=X_{\ehz{\gamma} + \ehz{\gamma}'}\,.
\ee

For any two pairs $(i,j)$, $(i',j')\in\CP_\alpha$, a computation of $F(\wp,\CW_{\vartheta^\pm})$ involves terms of the following type 
\begin{align}
	Q(p)&:=1+\Xi_{ij}(p)\Xi_{ji}(p) = 1+\sum_{\substack{\overline a\in\Gamma_{ij}(p)\\ \overline b\in\Gamma_{ji}(p)}}\mu(a)\mu(b)\,X_{\mathrm{cl}(ab)}
\end{align}
\begin{align}
	Q'(p)&:=1+\Xi_{i'j'}(p)\Xi_{j'i'}(p) = 1+\sum_{\substack{\overline{a}'\in\Gamma_{i'j'}(p)\\ \overline{b}'\in\Gamma_{j'i'}(p)}}\mu(a')\mu(b')\,X_{\mathrm{cl}(a'b')} \,.
\end{align}

It follows from Proposition \ref{prop:soliton-symmetry} that each term $\mu(a)\mu(b)\,X_{\mathrm{cl}(ab)}$ of $Q(p)$ has a counterpart $\mu(a')\mu(b')\,X_{\mathrm{cl}(a'b')}$ in $Q'(p)$, such that 
\be
	\mu(a')\mu(b') = \mu(a)\mu(b)\quad \text{ and }\quad Z_{\mathrm{cl}(ab)}=Z_{\mathrm{cl}(a'b')}\,.
\ee
Since $X$'s are class functions on the quotient by $\ker(Z)$, it follows that
\be
	X_{\mathrm{cl}(ab)} = X_{\mathrm{cl}(a'b')} =X_{\ehz{{\gamma}}} \,,
\ee
where $\ehz{{\gamma}}$ denotes the corresponding class in $H_1(\widetilde\Sigma_\rho,\IZ)/(2H\times \ker(Z))$. Therefore we find that 
\be
	Q(p) \equiv Q'(p)\,,
\ee
i.e.\ for any choice of pairs $(i,j), (i',j')\in\CP_\alpha$ their contributions to $F(\wp,\CW)$ are identical.

{The isomorphism in (\ref{eq:sub-quotient-iso}) can be lifted to $H_1(\tilde\Sigma_\rho,\IZ)/2H$, by considering the corresponding $\IZ_2$ extension of each side.} Then each class $\ehz{{\gamma}}$ corresponds to a unique element  $\ephz{\gamma} \in \Gamma$. 
One way to define $\ephz{\gamma}$ is by the following relation\footnote{The uniqueness follows from the definition of $\hat{\Gamma}$, given in Section \ref{subsec:prym}, which involves a quotient by $\ker(Z)$. 
The existence of such an element is not obvious a priori: the main issue is the integrality of $\gamma$. This is precisely the reason why there is a factor of $k_\rho$. This is related to the non-idempotency of the operator $P$ defined in  Section \ref{subsec:prym}.}
\be
	Z_{\ephz{\gamma}} = k_\rho Z_{\overline\gamma}\,.
\ee
As an example, for the pure $\SO(6)$ gauge theory that we studied in Section \ref{subsec:explicit-trivialization}, for $\ehz{{\gamma}} = [\eh{\gamma}_{21}]_{\ker(Z)}= [\eh{\gamma}_{65}]_{\ker(Z)}$ we have $\ephz{\gamma} = [A_1]_{\ker(Z)} = [\eh{\gamma}_{21}+\eh{\gamma}_{65}]_{\ker(Z)}$, where we neglected the winding around the circle fiber of $\tilde\Sigma_\rho$ for simplicity.

For a generic value of $u\in\CB$, the set $\Gamma_c\subset\Gamma$ of charges with $Z_\gamma \in e^{i\vartheta_\mathrm{c}}\IR_{-}$ is rank-1, which we assume to be generated by a single charge $\ephz{\gamma_\mathrm{c}}$. 
Therefore the charge $\ephz{\gamma}$ must be proportional to $\ephz{\gamma_\mathrm{c}}$. 
Then to each two-way street of the network $p\in\CW_\mathrm{c}$ we can uniquely associate a set of integers $\alpha_{\ephz{\gamma}}(p)$ by factorizing $Q(p)$ as
\be
	Q(p) = \prod_{n=1}^{\infty}\big(1 + X_{{n\ehz{\gamma}_\mathrm{c}}}\big)^{\alpha_{{n\ephz{\gamma_\mathrm{c}}}}(p)},\ \alpha_{\ephz{\gamma}}(p)\in\IZ\,.\label{eq:Q_p_gamma_c}
\ee

%%%%%%%%%%%%%%%%%%%%%%%%%%%%%%%%%%%
\subsection{4d BPS degeneracies}\label{subsubsec:BPS-index}
%%%%%%%%%%%%%%%%%%%%%%%%%%%%%%%%%%%
Here we present a formula for computing BPS indices from spectral networks. Its derivation is essentialy that of \cite{Gaiotto:2012rg}, with a bit of extra structure.

As a preliminary, let us set a few conventions on two-way streeets.
Let $p$ be a two-way street formed by walls $\CS_{\pm\alpha}$, where $\alpha$ is a positive root. 
We define the \emph{canonical lift} of $p$ as the formal sum
\be
	\pi^{-1}(p) := \sum_{\weight\in\Lambda_\rho}(\weight\cdot\alpha)\, p_{(\weight)}  = \sum_{(i,j)\in\CP_\alpha} p_{(\weight_j)} - p_{(\weight_i)}\,,
\ee
where $p_{(\weight)}$ is the lift of $p$ to a sheet of $\Sigma_\rho$ corresponding to a weight $\weight$, with orientation given by $\CS_\alpha$, and a minus sign denoting an reversal of the orientation.
Note that this canonical lift is very closely related to the projection of gauge charges presented in Section \ref{subsec:prym} --- in fact (\ref{eq:projection-gauge}) may be understood as taking a hypothetical two-way street $p$ stretching between the two branch points in Figure \ref{fig:single-cut}, and constructing $\pi^{-1}(p)$!

Now, when the topology of $\CW$ jumps, so does the soliton data on its streets.
This results in a jump in $F(\wp,\CW)$, whose definition makes use of the soliton data through the detour rule.
We will prove in Section \ref{sec:K-wall-proof} that the jump of the formal parallel transport is described by a universal formula called a $\mathcal{K}$-wall formula,
\be\label{eq:framed-wcf}
	F(\wp,\CW_{\vartheta^{+}_{c}}) = \CK\big(F(\wp,\CW_{\vartheta^{-}_{c}})\big)\,,
\ee
where $\CK$ is the following substitution on all formal variables
\be\label{eq:K-wall-jump}
	\CK(X_{a}) = X_{a}\,\prod_{n=1}^{\infty}\big(1+X_{{n\ehz{\gamma}_\mathrm{c}}}\big)^{\langle L(\ephz{n \gamma_\mathrm{c}}),\, \overline{a}  \rangle}.
\ee
Here $\overline a$ is the pushforward of the soliton charge $a$ onto $\Sigma_{\rho}$ (see discussion around (\ref{eq:pushforward-rel-hom})), while 
$L(\ephz{n \gamma_\mathrm{c}})$ is a formal sum determined by the soliton data of two-way streets
\be
	L(\ephz{\gamma}) :=\sum_{p\in\CW_\mathrm{c}}\alpha_{\ephz{\gamma}}(p)\,\pi^{-1}(p)\,.\label{eq:L_gamma}
\ee
In Section \ref{sec:L_gamma_closed} we will show that $L(\ephz{\gamma})$ is a closed 1-cycle on $\Sigma_\rho$, in the sense that $\partial L(\ephz{\gamma}) = \emptyset$. 
The fact that $L(\ephz{\gamma})$ is closed is necessary to make sense of (\ref{eq:K-wall-jump}), in particular of the intersection pairing $\langle L(\ephz{n\gamma_\mathrm{c}}), \overline a  \rangle$.%
\footnote{To be precise, to make sense of the pairing more information is needed than what is contained in elements of $H_1^\text{rel}(\tilde\Sigma_\rho,\IZ) / \ker(Z)$. In fact we will give a proof of formula (\ref{eq:K-wall-jump}) in Section \ref{sec:K-wall-proof} by using the actual paths coming from lifting the spectral network geometry.
Nevertheless, this detail is of secondary importance for the present discussion, since the BPS index formula (\ref{eq:BPS-index-formula}) only involves differences of open paths, which live in $H_1(\tilde\Sigma_\rho,\IZ) / \ker(Z)$, and we already gave a proper definition of the DSZ pairing for these.}
Moreover, by genericity of $u\in\CB$, the physical charge  corresponding to $L(\gamma)$ (for definiteness, $[L(\gamma)]_{\ker(Z)}\in\Gamma$), must be proportional to $\ephz{\gamma_{c}}$. 

The formula (\ref{eq:K-wall-jump}) was interpreted in \cite{Gaiotto:2012rg} in terms of framed 2d-4d wall-crossing in 2d-4d coupled systems, describing how bound states of 2d BPS solitons on IR surface defects mix with 4d BPS states of the bulk 4d theory.
The same interpretation applies here, so it is natural to identify the {enhanced 2d-4d degeneracies} of \cite{Gaiotto:2011tf, Gaiotto:2012rg} with
\footnote{The intersection pairing is understood to be evaluated, as usual, upon choosing the unique representative for each of the charges in $H_1(\tilde\Sigma_\rho,\IZ)$. See for example the final remarks of Section \ref{subsec:explicit-trivialization}.}
\be
	\omega(\ephz{\gamma},\overline a) = \langle L(\ephz{\gamma}), \overline{a} \rangle\,.
\ee
These enhanced degeneracies are characterized by the property\footnote{Here $\langle\cdot,\cdot\rangle$ denotes the intersection pairing, which differs from the DSZ pairing by a factor of $k_\rho$. 
It would be interesting to understand this factor from the physical perspective of the halo picture of framed 2d-4d wall-crossing.}
\be
	\omega(\ephz{\gamma},\overline{a} + \ephz{\gamma}') = \omega(\ephz{\gamma},\overline{a})+\Omega(\ephz{\gamma})\,\langle \eh{\gamma}, \eh{\gamma}'\rangle\,,
\ee
which leads to the claim
\be\label{eq:BPS-index-formula}
	\Omega(\ephz{\gamma}) = [L(\ephz{\gamma})]_{\ker(Z)}\, /\,\ephz{\gamma}\,.
\ee
Here $\Omega(\ephz{\gamma})$ is the degeneracy of 4d BPS states of charge $\ephz{\gamma}\in\Gamma$, while $[L(\gamma)]_{\ker(Z)} \in \Gamma=P(H_1(\tilde\Sigma_\rho,\IZ))/\ker(Z)$ is the physical charge corresponding to $L(\gamma)$.
In fact, to make sense of (\ref{eq:BPS-index-formula}), it is actually necessary that both $[L(\ephz{\gamma})]_{\ker(Z)}$ and $\ephz{\gamma}$ are elements of the same lattice. 
This is true provided that the homology class $[L(\gamma)] \in P(H_1(\tilde\Sigma_\rho, \mathbb{Z}))$, 
i.e.\ that it falls in the sub-lattice defined by the image of the operator $P$ from section \ref{subsec:prym}.
From the definition of $L(\gamma)$ it seems plausible this condition holds true generally, 
and we assume it to be true generally, but we do not have a rigorous proof. 
We will show that this holds true in explicit examples in Section \ref{sec:examples}.

%%%%%%%%%%%%%%%%%%%%%%%%%%%%%%%%%%%
\subsection{Framed wall-crossing at \texorpdfstring{$\CK$}{K}-walls}\label{sec:K-wall-proof}
%%%%%%%%%%%%%%%%%%%%%%%%%%%%%%%%%%%

%%%%%%%%%%%%%%%%%%%%%%%%%%%%%%%%%%%
\subsubsection*{Proof of $\partial L(\gamma) = \emptyset$}\label{sec:L_gamma_closed}
%%%%%%%%%%%%%%%%%%%%%%%%%%%%%%%%%%%

From its definition it is obvious that $\partial L(\gamma)$ gets contributions only from (lifts of) endpoints of two-way streets of $\CW_\mathrm{c}$, which come in two types: branch points or joints. 
Branch-points are easily seen to give a null contribution to $\partial L(\gamma)$, while joints require a bit more care. 
For the sake of generality, we consider a 6-way joint\footnote{The analysis of any joint with fewer than $6$ two-way streets follows as a special case of the forthcoming analysis.} of two-way streets located at $z\in C$, as depicted in Figure \ref{fig:6-way-joint}.

\begin{figure}[ht]
\begin{center}
\includegraphics[width=0.35\textwidth]{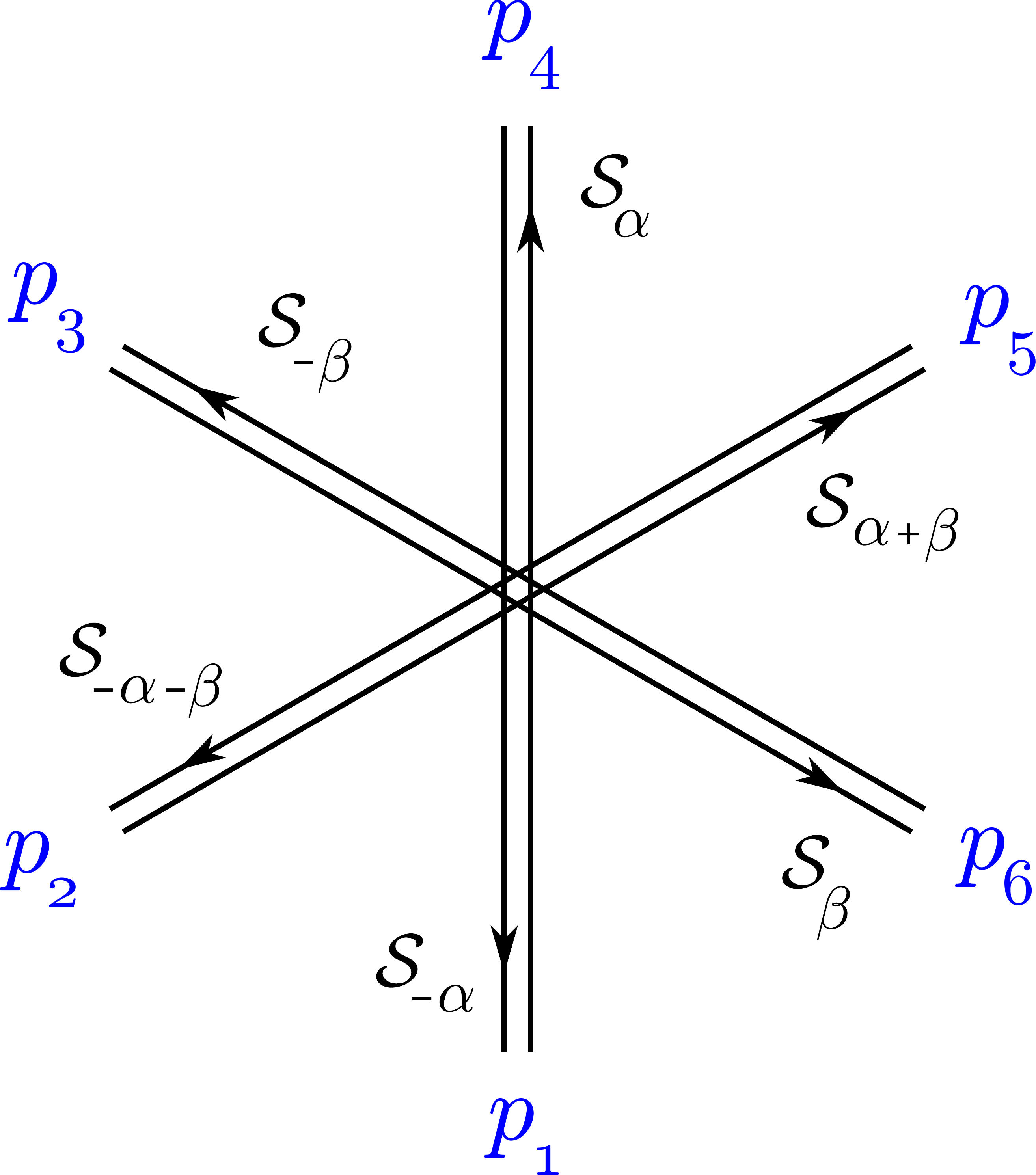}
\caption{}
\label{fig:6-way-joint}
\end{center}
\end{figure}

Each two-way street $p_{i},\,i=1\dots6$ enters the formal sum $L(\gamma)$ as
\be
	L(\gamma) \supset \sum_{i=1}^{6}\alpha_{\gamma}(p_{i})\,\pi^{-1}(p_{i}) \,.
\ee
Denoting by $z$ the location of the joint, we want to compute the contribution to $\partial L(\gamma)$ by the lifts $\{x_\weight(z)\}_{\weight\in\Lambda_\rho}$ of $z$ to the various sheets.
The existence of such a joint implies that $\alpha\measuredangle\beta=2\pi/3$, let us therefore consider the corresponding A$_2$ sub-algebra of $\fg$.
For any $x_\weight \in \partial L$, the weight $\weight$ must belong to a Weyl orbit of $\gA_{2}$, and there are three possible orbits which we collect in Figure \ref{fig:W2-orbits}. 
\begin{figure}[h!]
\begin{center}
\begin{subfigure}{0.15\textwidth}
\includegraphics[width=\textwidth]{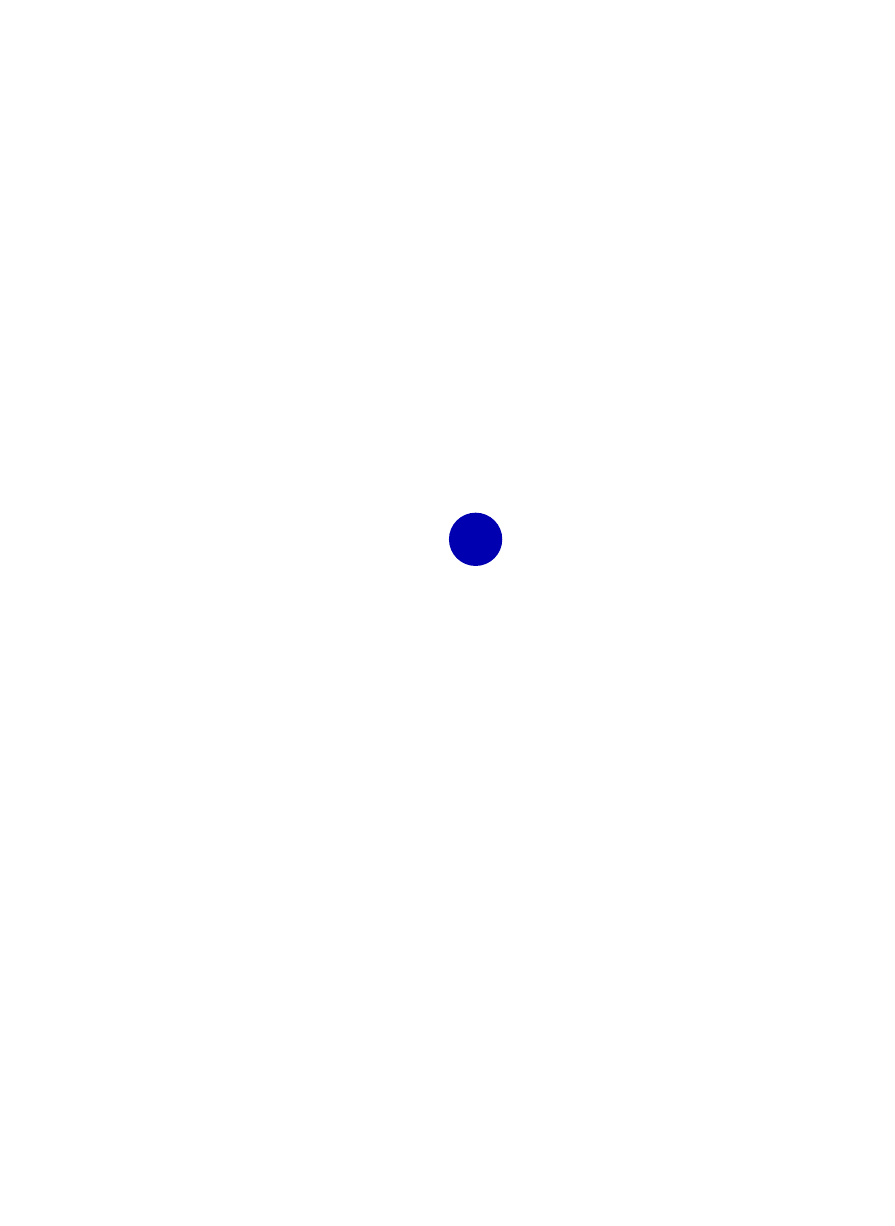}
\caption*{(1)}
\end{subfigure}
\hspace{0.1\textwidth}
\begin{subfigure}{0.15\textwidth}
\includegraphics[width=\textwidth]{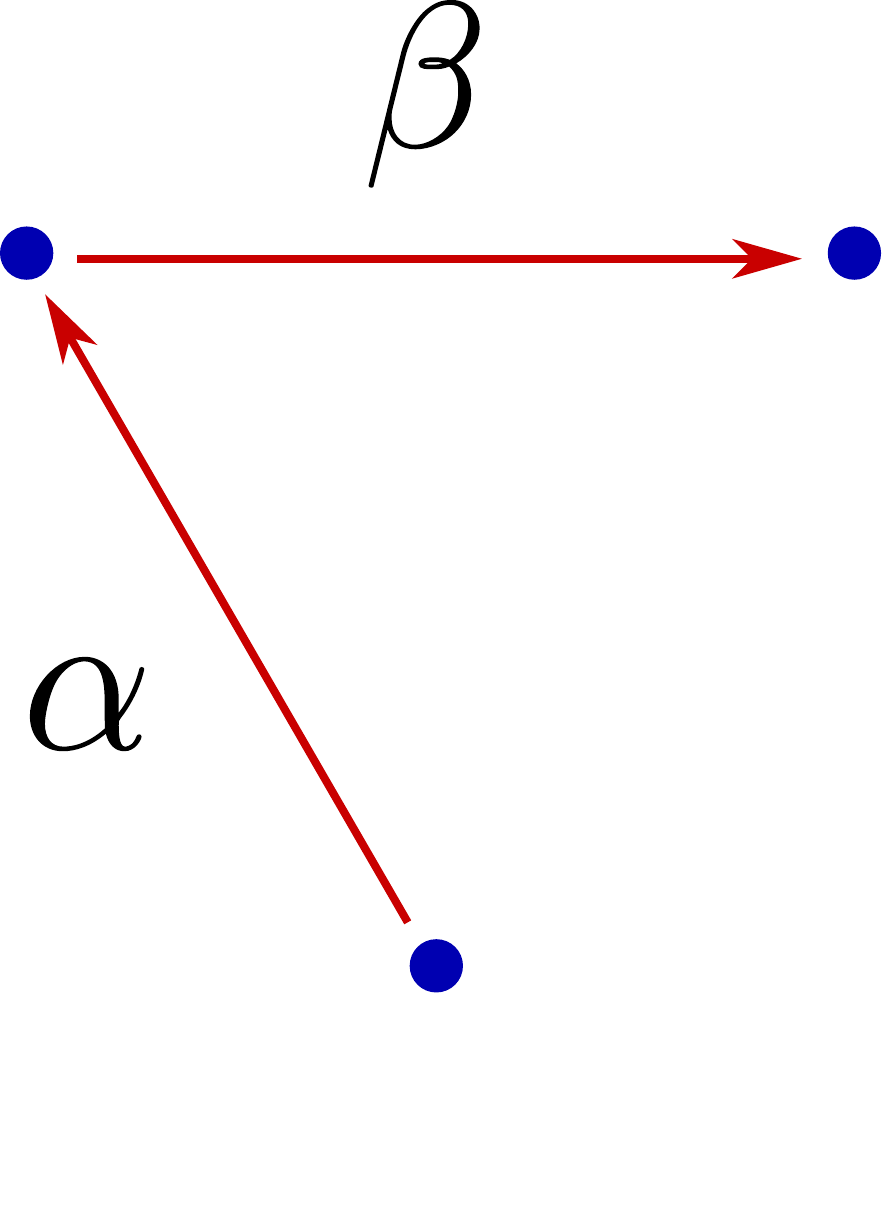}
\caption*{(2)}
\end{subfigure}
\hspace{0.1\textwidth}
\begin{subfigure}{0.15\textwidth}
\includegraphics[width=\textwidth]{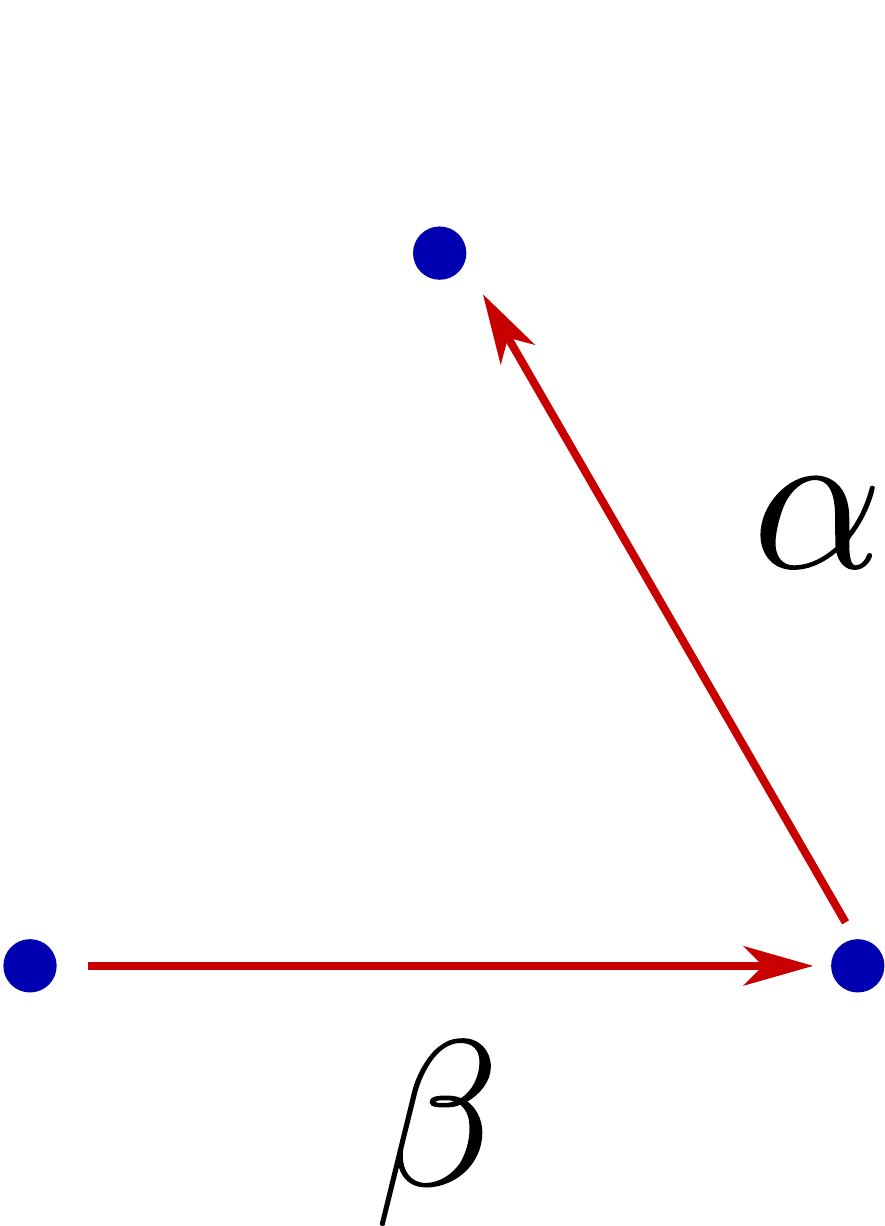}
\caption*{(3)}
\end{subfigure}
\caption{Three Weyl orbits of $\gA_{2}$, only these can occur in a minuscule representation of $\gADE$. 
}
\label{fig:W2-orbits}
\end{center}
\end{figure}
If $\weight_{i}$ belongs to the trivial orbit of case (1), we see that $i\in\CP_{\alpha}^{0}\cap\CP_{\beta}^{0}$ and therefore $x_{i}$ doesn't contribute to $\partial L[\gamma]$.
In case (2), we must consider where $\weight_{i}$ sits in the orbit:
\begin{itemize}
\item[(a)] $i\in\CP_{\alpha}^{-}\cap\CP_{\beta}^{0}$ i.e.\ on the bottom, 
\item[(b)] $i\in\CP_{\alpha}^{+}\cap\CP_{\beta}^{-}$ i.e.\ on the top-left,
\item[(c)] $i\in\CP_{\alpha}^{0}\cap\CP_{\beta}^{+}$ i.e.\ on the top-right. 
\end{itemize}
In case (a), we have 
\be\label{eq:boundary-contribution}
	\partial L(\gamma)\supset x_{i}\big( \alpha_{\gamma}(p_{4}) + \alpha_{\gamma}(p_{5}) - \alpha_{\gamma}(p_{1}) - \alpha_{\gamma}(p_{2}) \big)\,.
\ee
By choosing $\wp,\wp'$ as in Figure \ref{fig:artful-paths} and studying $F(\wp)_{ii}=F(\wp')_{ii}$
we find
\be
		F(\wp)_{ii}  \sim Q(p_{5}) \, Q(p_{4}) \,,\quad 	F(\wp')_{ii} \sim Q(p_{1}) \, Q(p_{2}) \,,
\ee
homotopy invariance then implies
\be
\begin{split}
	&\quad \alpha_{n\gamma_{c}}(p_{5}) + \alpha_{n\gamma_{c}}(p_{4}) = \alpha_{n\gamma_{c}}(p_{1}) + \alpha_{n\gamma_{c}}(p_{2})\quad \forall n\,,
\end{split}
\ee
therefore the sum in the parentheses of (\ref{eq:boundary-contribution}) vanishes as a consequence of the flatness of $F(\wp)$.
\begin{figure}[ht]
\begin{center}
\includegraphics[width=0.35\textwidth]{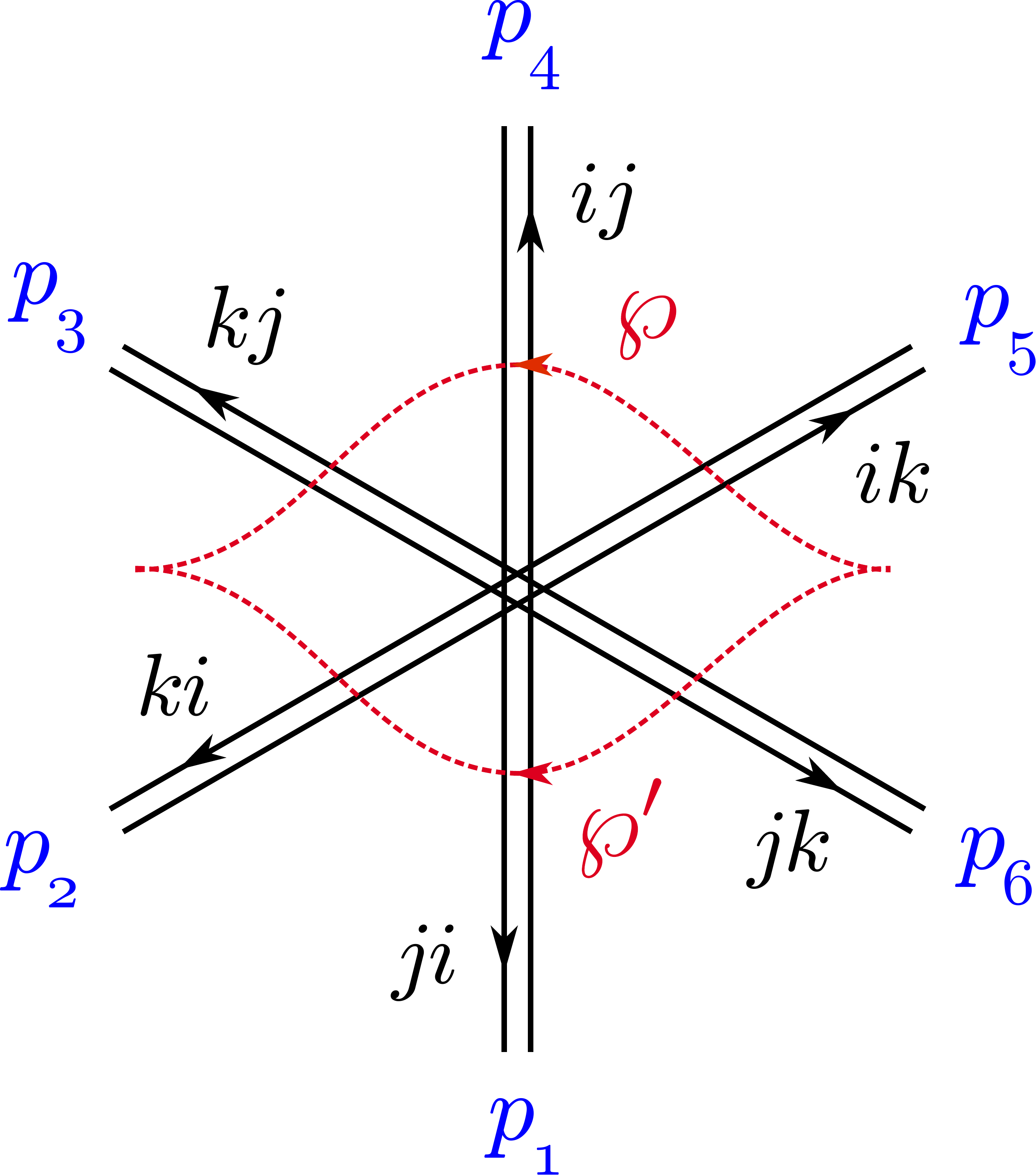}
\caption{A choice of $\wp,\wp'$ suitable for deriving (\ref{eq:boundary-contribution}).}
\label{fig:artful-paths}
\end{center}
\end{figure}

A similar argument shows that the contribution of $x_{i}$ vanishes also in cases (b) and (c), and also for all analogous cases arising for the last type of $\gA_{2}$-orbit in Figure \ref{fig:W2-orbits}.

%%%%%%%%%%%%%%%%%%%%%%%%%%%%%%%%%%%
\subsubsection*{Proof of the $\CK$-wall formula}
%%%%%%%%%%%%%%%%%%%%%%%%%%%%%%%%%%%

We finally proceed to prove that the jump of $F(\wp,\CW)$ is described by equation (\ref{eq:framed-wcf}), which was a crucial assumption in the derivation of the BPS index formula.
The analysis given here is only a mild generalization of that of \cite{Gaiotto:2012rg}, but relies crucially on the analysis of $\CS$-walls and their soliton content carried out in Section \ref{sec:spectral-networks}.

Let $\CW_{\pm}^{}$ be the networks at $\vartheta_\mathrm{c}^{\pm}$ respectively, then we claim that both
\be
\CF(\wp) = F(\wp, \CW_{+})\quad\text{and}\quad \CF(\wp) = \CK\left(F(\wp, \CW_{-})\right)
\ee
have the following properties:
\begin{description}
 \item[P1.] $\CF(\wp)$ is a (twisted) homotopy invariant of $\wp$.
 \item[P2.] If $\wp\cap\CW=\emptyset$, then
\be
	 \CF(\wp) = D(\wp).
\ee
 \item[P3.] If $\wp,\wp'$ have endpoints off $\CW$, then
\be
 \CF(\wp) \CF(\wp') = \CF(\wp \wp').
\ee
 \item[P4.] If $\wp\cap\CW=\{z\}$ on a one-way street $p$ of type ${\alpha}$, then 
\be\label{eq:f-one-way}
 \CF(\wp) = D(\wp_+) \left(1 + \sum_{(i,j)\in\CP_{\alpha}}\sum_{\bar a \in \Gamma_{ij}(p)} \mu(a) X_{a}\right) D(\wp_-),
\ee
for some $\mu(a) \in \IZ$.
\item[P5.] If $\wp\cap\CW=\{z\}$ on a two-way street $p$, and the intersection between $\wp$ and $p$
is positive\footnote{With respect to the orientation of $C$ as a complex curve, and with the orientation of $p$ given by the underlying $\CS_{\alpha}$.},
\be \label{eq:f-two-way}
 \CF(\wp) = D(\wp_+) \left(1 + \sum_{(j,i)\in\CP_{\alpha}}\sum_{\bar a \in \Gamma_{ji}(p)} \mu(a) X_{a}\right) \left(1 +  \sum_{(i,j)\in\CP_{-\alpha}}\sum_{\bar b \in \Gamma_{ij}(p)} \mu(b) X_{b} \right) D(\wp_-)
\ee
for some $\mu(a),\,\mu(b) \in \IZ$.
\end{description}

\noindent For $\CF(\wp) = F(\wp, \CW^+)$ these properties follow from the definitions.  
We then have to prove them for $\CF(\wp) = \CK(F(\wp, \CW^-))$.
The first three are true for $F(\wp, \CW^-)$ and are preserved by $\CK$.
The fourth is also true for $F(\wp, \CW^-)$, and is also preserved by $\CK$:
in fact $\CK$ just multiplies each term by a function of the $X_{\tilde \gamma}$,
thus preserving the form \eqref{eq:f-one-way}.  

The last property requires a bit more work: we begin with the formula 
\be\label{eq:two-way-crossing}
	F(\wp, \CW^-) = D(\wp_{+})\left(1+\Xi_{-\alpha}\right) \left(1+\Xi_{\alpha}\right) D(\wp_{-})
\ee
for a $\wp$ crossing a two-way street in the American resolution (see Figure \ref{fig:2-way-crossing}), going first through $\CS_{-\alpha}$ and then through $\CS_{\alpha}$.
We want to show that $\CK$ transforms it into the form \eqref{eq:f-two-way}.
Expanding out \eqref{eq:two-way-crossing} we find
various classes of terms: for each $(i, j)\in\CP_{-\alpha}$ and each $k\in\CP_{\alpha}^{0}$
\be\label{eq:K-wall-eqs}
\begin{split}
 F(\wp, \vartheta_\mathrm{c}^-)_{ii} &= X_{\wp_+^{(i)}} \left( 1 + \sum_{\substack{\bar a\in\Gamma_{ji}\\ \bar b\in\Gamma_{ij}}}\mu^-(a) \mu^-(b) X_b X_a \right) X_{\wp_-^{(i)}} = X_{\wp^{(i)}} Q(p) , \\
 F(\wp, \vartheta_\mathrm{c}^-)_{ij} &= X_{\wp_+^{(i)}} \left( \sum_{\bar b\in\Gamma_{ij}} \mu^-(b) X_b \right) X_{\wp_-^{(j)}}, \\
 F(\wp, \vartheta_\mathrm{c}^-)_{ji} &= X_{\wp_+^{(j)}} \left( \sum_{\bar a\in\Gamma_{ji}} \mu^-(a) X_a \right) X_{\wp_-^{(i)}}, \\
 F(\wp, \vartheta_\mathrm{c}^-)_{jj} &= X_{\wp^{(j)}}, \\
 F(\wp, \vartheta_\mathrm{c}^-)_{kk} &= X_{\wp^{(k)}} .
\end{split}
\ee

\begin{figure}[ht]
\begin{center}
\includegraphics[width=0.35\textwidth]{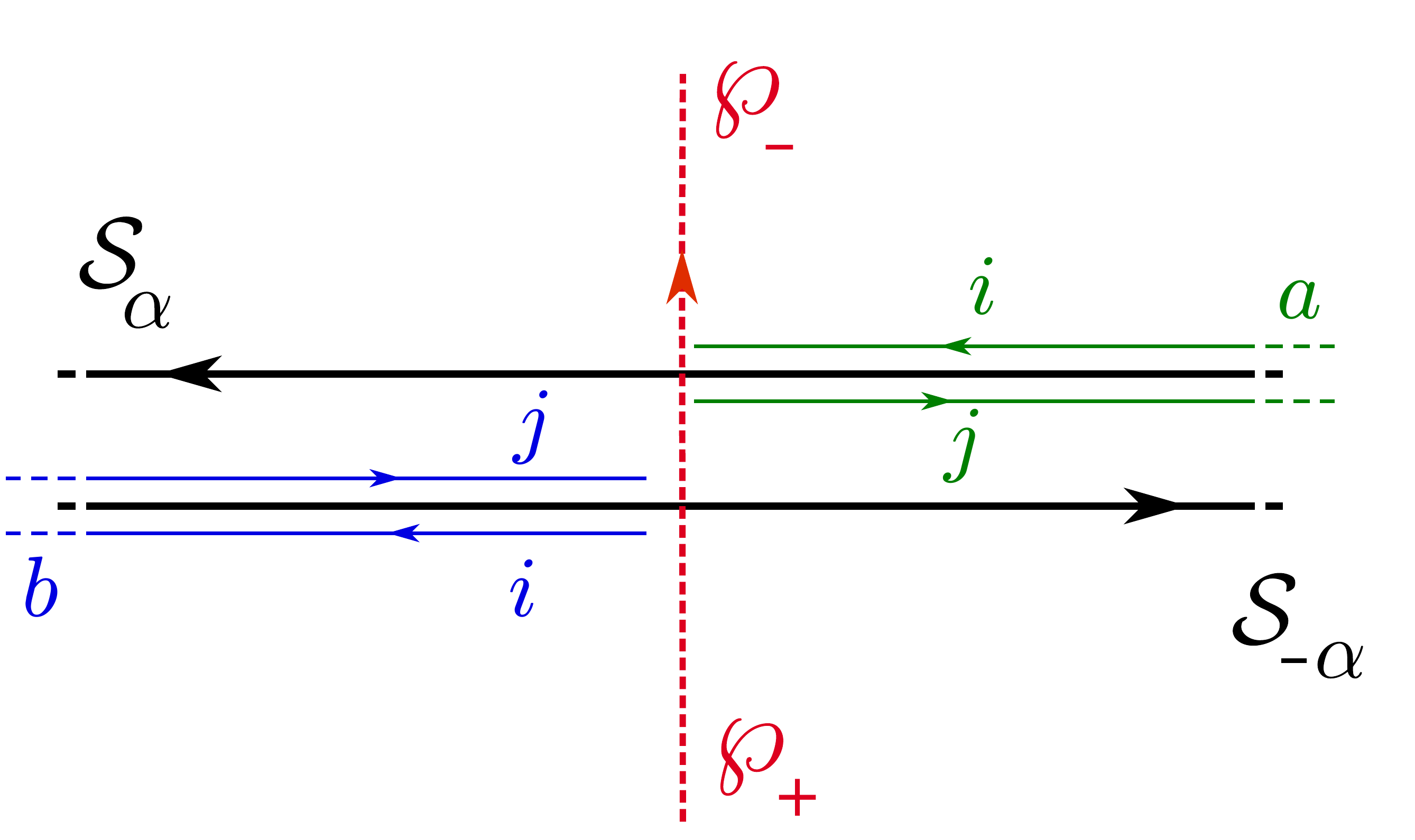}
\caption{The American resolution of a two-way street. of type $\pm\alpha$. Two solitons $a, b$ respectively of types $(i,j)$ and $(j,i)$  are depicted, the arrows denote the direction of the flow on the sheets $i$ and $j$.}
\label{fig:2-way-crossing}
\end{center}
\end{figure}

For each soliton type, such as $(i,i)$,  $(i,j)$, etc., the terms in the sum differ only by powers of $X_{\gamma}$.  
Since $[L(\gamma)]\propto\gamma$, it follows that $\langle{\gamma, L(\gamma)}\rangle = 0$, therefore the action of $\CK$ on each class of terms is independent of the particular term.  
For the $F_{ii}$ terms, $\CK$ acts by multiplication by
\be
	\prod_{n=1}^{\infty} (1 + X_{n \ehz{\gamma}})^{\langle{ L(\gamma) , \wp^{(i)} } \rangle} = \prod_{n=1}^{\infty} (1 + X_{\ehz{\gamma}})^{ -\alpha_\gamma(p)} = Q(p)^{-1}
\ee
Similarly, on $F_{jj}$ terms $\CK$ acts by multiplication by $Q(p)$.
The action on $F_{ij}$ terms is more complicated: it consists of multiplication by a new function 
\be
	H = \prod_{n=1}^{\infty} (1 + X_{\ehz{\gamma}})^{\langle{L(\gamma), \wp_+^{(i)} + \bar b + \wp_-^{(j)}}\rangle}
\ee
But note that
\be
	\langle{L(\gamma), \wp_+^{(i)} + \bar b + \wp_-^{(j)}} \rangle = - \langle{L(\gamma) , \wp_+^{(j)} + \bar a + \wp_-^{(i)}}\rangle,
\ee 
which follows from the fact that $\cl(\bar a+\bar b)$ is proportional to $L(\gamma)$ and
$\langle{L(\gamma), \wp^{(i)}}\rangle = - \langle{L(\gamma), \wp^{(j)}}\rangle$.
This means that $\CK$ acts on the $F_{ji}$ terms by multiplication by $H^{-1}$.
These facts together are sufficient to prove the last property.

The above five properties in fact {determine} $\CF(\wp)$, for the following reasons.
Off-criticality, a network $\CW_{\vartheta}$ only has one-way streets, then in this case properties P1-P4 completely determine the soliton content on all streets, and therefore fix $\CF(\wp)$.
When two-way streets are present, a similar argument applies. A basic ingredient now is a direct computation which, relying only on properties P1-P5, computes the whole \emph{outgoing} soliton content on all two-way streets attached to a joint (for instance, like $p_{1}\dots p_{6}$ in Figure \ref{fig:6-way-joint}), in terms of the \emph{ingoing} soliton content for those same streets.
The explicit formulae for the soliton flow at joints of two-way streets were first given in \cite[App. A]{Gaiotto:2012rg}\footnote{Also see \cite{Galakhov:2013oja} for a small correction of those that takes into account winding of solitons on the circle fiber across the joint.}. Instead of replicating their analysis, we point the reader to the original reference, since the same formulae can be promoted \emph{mutatis mutandis} to ADE networks.
Therefore, all the soliton degeneracies $\mu(a)$ everywhere on $\CW$ are completely determined by the above properties, even when the network contains two-way streets. 
This concludes the proof that the properties completely determine $\CF(\wp)$, which in turn implies (\ref{eq:framed-wcf}).

%%%%%%%%%%%%%%%%%%%%%%%%%%%%%%%%%%
\section{Examples}
\label{sec:examples}
%%%%%%%%%%%%%%%%%%%%%%%%%%%%%%%%%%

For the purpose of computing BPS spectra, one great advantage of the spectral network framework is that the computation is \emph{algorithmic}.  In this section we summarize this procedure and provide several examples of its application for theories with ADE gauge groups. 
General forms of Seiberg-Witten curves presented as spectral covers of $C$ in the vector representation are

\be
\begin{split}
	\gA_{N-1} \, :  \qquad &  \lambda^N + \sum_{w_i} \phi_{w_i} \lambda^{N-w_i} = 0 \\
	\gD_{N} \, :  \qquad &    \lambda^{2N} + \sum_{w_i \neq N} \phi_{w_i} \lambda^{2N-w_i} + \phi_{N}^2= 0 
\end{split}
\ee
and similarly for E-type covers (\ref{eq:generic-curves}).

\subsection{An algorithmic approach to BPS spectra}

While spectral networks involve lots of structure, for the practical purpose of computing BPS spectra, much of the technicalities can be ignored.
In this section we collect a few essential steps that allow us to extract the BPS spectrum from spectral networks. This algorithm was proposed in \cite{Gaiotto:2012rg}, and was used to analyze all the examples presented below.

For a given theory, the first step is to choose a generic point $u$ on its Coulomb branch. This fixes the geometry of the spectral curve $\Sigma_\rho$, and therefore the geometry of $\CS$-walls at phase $\vartheta$.
To study the BPS spectrum, it is necessary to draw a family of spectral networks $\CW_\vartheta$. This step typically requires the use of a computer program. 
Examples presented here are studied using \texttt{loom}, a program designed for this purpose. It is an \foothref{https://github.com/chan-y-park/loom}{open source project}, and a \foothref{http://het-math2.physics.rutgers.edu/loom/config}{web user interface} to a server running \texttt{loom} is available for public use.

Given a family of spectral networks $\CW_\vartheta$, one looks for special values of the phase $\vartheta\in[0,\pi]$ where the network topology degenerates, undergoing a $\CK$-wall jump as described in section \ref{subsec:k-wall-jumps}. Each critical phase $\vartheta_c$ signals the presence of one or several BPS states in the spectrum with central charges of phase $\arg(Z) = \vartheta_{\mathrm{c}}$.
The topology of $\CW_{\vartheta_c}$ can be used to extract their IR gauge and flavor charges $\gamma\in\Gamma$, as well as their BPS indices $\Omega(\gamma,u)$.
The charge $\gamma$ of each BPS state takes value in a rank-1 sub-lattice $\Gamma_\mathrm{c} \subset \Gamma$ fixed by $\vartheta_{\mathrm{c}}$. 

Focusing on the \emph{sub}-network $\CW_{\mathrm{c}}\subset\CW$ composed of two-way streets (see for example Figure \ref{fig:W_c}), one computes $\alpha_{\gamma}(p)$ for each $p\in\CW_{\mathrm{c}}$ by using soliton data on each street and equation (\ref{eq:Q_p_gamma_c}).
The integers $\alpha_{\gamma}(p)$ then define a choice of lift $L(\gamma)$ of $\mathcal{W}_\mathrm{c}$ to $\Sigma_\rho$, defined in (\ref{eq:L_gamma}). 
The BPS degeneracy of $\gamma$ is then obtained from $L(\gamma)$ by applying equation (\ref{eq:BPS-index-formula})

Charges of BPS states are often expressed by choosing an electromagnetic duality frame. While the choice of a frame is to a certain extent arbitrary, invariant physical information is encoded in the DSZ electromagnetic pairing between the charges. In our setup, this can be computed as 
$\langle \gamma, \gamma' \rangle_{\text{DSZ}} = (1/k_\rho) \langle L(\gamma), L(\gamma')  \rangle$ in terms of the intersection pairing between $L(\gamma)$, $L(\gamma')$.

%%%%%%%%%%%%%%%%%%%%%%%%%%%%%%%%%%
\subsection{The hypermultiplet revisited}\label{sec:hyper_revisited}
%%%%%%%%%%%%%%%%%%%%%%%%%%%%%%%%%%
Readers familiar with the literature on spectral networks will recognize that the BPS index formula (\ref{eq:BPS-index-formula}) and the $\mathcal{K}$-wall formula (\ref{eq:K-wall-jump}) look very similar to those first derived in \cite{Gaiotto:2012rg}. 
Nevertheless, this fact conceals a rather intricate interplay among the projection on the physical charge lattice, the symmetry property of 2d solitons, and the $\CK$-wall formula, which appears already in the simplest $\mathcal{K}$-wall jump: the hypermultiplet.

In this subsection we don't specify a specific theory, instead we work on a local patch of $C$ and study a kind of jump of the network that appears quite commonly in many theories.
For this purpose, we fix a Lie algebra $\fg$ of ADE type and a minuscule representation $\rho$, and we take $C$ to be the complex plane. We take $\pi:\Sigma_\rho\to C$ to have two square-root branch points of the same root-type $\alpha$. 

\begin{figure}[h!]
\begin{center}
\includegraphics[width=0.85\textwidth]{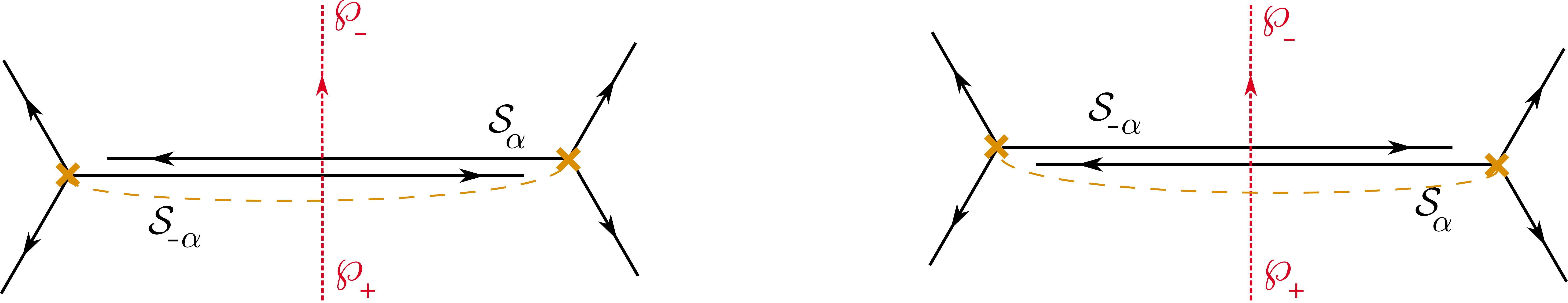}
\caption{The two-way street of type $\alpha$ appearing at the jump of the network $\CW$. On the left the American resolution for $\vartheta_{-}$, on the right the British resolution for $\vartheta_{-}$.}
\label{fig:hypermultiplet}
\end{center}
\end{figure}

We choose a trivialization by drawing a branch cut between the two branch points,  as shown in Figure \ref{fig:hypermultiplet}.
Each branch point sources three $\CS$-walls, we wish to study the jump when a wall of type $\alpha$ forms a two-way street with a wall of type $-\alpha$.
For simplicity, we assume that $k_\rho=2$, although what we will say has a straightforward generalization to higher $k_\rho$. 
There are two soliton types on each $\CS$-wall
\be
	\CP_\alpha = \{(12),(34)\}\qquad \CP_{-\alpha} = \{(21),(43)\}\,.
\ee
The soliton content consists of one simpleton for each soliton type, since each $\CS$-wall is a primary.
Let us denote the simpletons by $a_{12}, a_{34}, b_{21}, b_{43}$ in self-evident notation.
After computing $Q(p)$ we find 
\be
	Q(p) = 1+ X_{a_{12}}X_{b_{21}} = 1+ X_{a_{34}}X_{b_{43}} \,.
\ee
Note that
\be
\left.\begin{array}{c}
\ehz{\gamma}_{12} := \mathrm{cl}(a_{12} b_{21}) \\
\ehz{\gamma}_{34}:= \mathrm{cl}(a_{34} b_{43}) 
\end{array}\right\}\in H_1(\tilde\Sigma_\rho,\IZ)/\ker(Z)
\ee
are really the same equivalence class  $\ehz{\gamma}_{12}=\ehz{\gamma}_{34}=\ehz{\gamma}_\mathrm{c}$ because of (\ref{eq:ij-period}). Therefore we may express $Q(p)$ as
\be
	Q(p) = 1+X_{\ehz{\gamma}_\mathrm{c}}. \label{eq:Q_p_hyper}
\ee
Moreover, assuming $u\in\CB$ to be generic, we can define the unique charge $\ephz{\gamma_\mathrm{c}} \in \Gamma$ such that $Z_{\ephz{\gamma}_c} = k_\rho Z_{\overline\gamma_c}$. $\ephz{\gamma_\mathrm{c}}$ is the equivalence class in $\Gamma$ that contains the $A$-cycle $\eh{\gamma}_\mathrm{c}$ from this branch cut (see the discussion in Section \ref{subsec:prym}): denoting by $\eh{\gamma}_{12}$, $\eh{\gamma}_{34}$ two {distinct} homology classes in $H_1(\tilde \Sigma_\rho,\IZ)$, consisting of cycles around the gluing fixtures above the cut (as in Figure \ref{fig:single-cut}),\footnote{The choice of tangent framing lift to $\tilde\Sigma_\rho$ is understood.}  
\be
	\ephz{\gamma_\mathrm{c}} := \left[{\eh{\gamma}_{12} + \eh{\gamma}_{34}}\right]_{\ker(Z)}
\ee
is the unique element of $\Gamma$ such that
\be
	Z_{\ephz{\gamma}_\mathrm{c}} = Z_{\eh{\gamma}_{12}}+Z_{\eh{\gamma}_{34}} = 2 Z_{\ehz{\gamma}_c}\,.
\ee
From (\ref{eq:Q_p_hyper}) we get $\alpha_{\ephz{\gamma_\mathrm{c}}} = 1$. Then we can construct
\be
	L(\ephz{\gamma_\mathrm{c}}) = p^{(2)}-p^{(1)} + p^{(4)}-p^{(3)}\,.
\ee
The BPS index can be obtained by directly applying equation (\ref{eq:BPS-index-formula}). For simplicity here we consider a lifted (but equivalent) version of that equation, by considering the representative $\eh{\gamma}_c = \eh{\gamma}_{12}+ \eh{\gamma}_{34}$ for $\ephz{\gamma}_c$, which gives 
\begin{align}
	\Omega(\ephz{\gamma_\mathrm{c}}) = [L(\ephz{\gamma_\mathrm{c}})]/ \eh{\gamma}_\mathrm{c} = 1\,.
\end{align}
We thus recover the expected degeneracy for a hypermultiplet from the BPS index formula.

Next we show that this is also compatible with the jump of the formal parallel transport.
The components  $F(\wp,\CW_\vartheta)_{ij}$ of the formal parallel transport can easily be  computed from Figure \ref{fig:hypermultiplet}.
For example, the $12$-component before and after the jump is 
\be
\begin{split}
	& F(\wp,\CW_{\vartheta_-})_{12} = X_{\wp^{(12)}} \\
	& F(\wp,\CW_{\vartheta_+})_{12} = X_{\wp^{(1)}_+} (1+X_{\ehz{\gamma}_{12}}) X_{\wp^{(2)}_-}  = X_{\wp^{(12)}} (1+X_{\ehz{\gamma}_\mathrm{c}})\,.
\end{split}
\ee
where $\wp^{(12)}$ is a path $\wp^{(1)}_+ \, \wp^{(2)}_-$ through the branch cut denoted in Figure \ref{fig:hypermultiplet}.
This confirms that
\be
	F(\wp,\CW_{\vartheta_+})_{12} = \CK( F(\wp,\CW_{\vartheta_-})_{12}  ) = X_{\wp^{(12)}} (1+X_{\ehz{\gamma}_\mathrm{c}})^{\langle L({\ephz{\gamma}_\mathrm{c}}),\wp^{(12)}\rangle}  \,,
\ee
as claimed in (\ref{eq:K-wall-jump}).
It is simple to extend the check to all other components of $F(\wp,\CW_{\vartheta_\pm})$.

%%%%%%%%%%%%%%%%%%%%%%%%%%%%%%%%%%
\subsection{The strong coupling spectra of \texorpdfstring{$\gADE$}{g\_ADE} SYM}
%%%%%%%%%%%%%%%%%%%%%%%%%%%%%%%%%%

The Seiberg-Witten curves for $\CN=2$ pure gauge theory with simply laced gauge groups were given in \cite{Martinec:1995by, Lerche:1991re, Eguchi:2002fc, Klemm:1994qs, Argyres:1994xh,Gorsky:1995zq,Danielsson:1995is}. 
The differentials depend on the invariant Casimirs of the vector-multiplet scalars (here denoted by $u_k$), which parametrize the Coulomb branch, and on the choice of a strong coupling scale (here denoted by $\mu$) \cite{Keller:2011ek}
\begin{align}
	\phi_{k}(z) = u_k \left(\frac{\dd z}{z}\right)^k,\quad 
	\phi_{h^\vee}(z) = \left( \mu^{h^\vee} z + u_{h^\vee} + \frac{\mu^{h^\vee}}{z} \right) \left(\frac{\dd z}{z}\right)^{h^\vee}.
\end{align}
The explicit form of the curves in terms of ${\phi_{k}}$ are given in Appendix \ref{sec:spectral_curves_for_g}.

At generic $u_i$, there are $2r$ square-root branch points, and with a suitable choice of trivialization they can be arranged in pairs $(\fb_{\alpha_i}, \bar\fb_{\alpha_i})$ corresponding  to simple roots of $\fg$.
At the ``origin'' of the Coulomb branch,\footnote{We fix the scale $\mu=1$.} i.e.\ where $\{u_i=0\}_i$, all spectral curves assume a common pattern: there are two branch points, at $z=\pm i$ and irregular singularities at $z=0,\infty$.
The ramification structure at each branch point is given by a Coxeter element of $W$: in this degeneration limit the $\fb_{\alpha_i}$ coalesce together at $z=i$ and the $\bar\fb_{\alpha_i}$ do the same at $z=-i$.
Since the coalescing branch points are all  of distinct simple-root types, there cannot be any 4d BPS states becoming massless at $u=0$
\footnote{Such a BPS state would have a critical network of two-way streets stretching among the $\fb_{\alpha_i}$, subject to the condition that at any joint the sum of the roots vanishes; this is not possible since the simple roots form a basis for $\ft^*$}.
Thus this is a regular point of $\CB$, not a singular one, and the problem of computing the BPS states is well-posed.

\begin{figure}[h!]
\begin{center}
\includegraphics[width=0.30\textwidth]{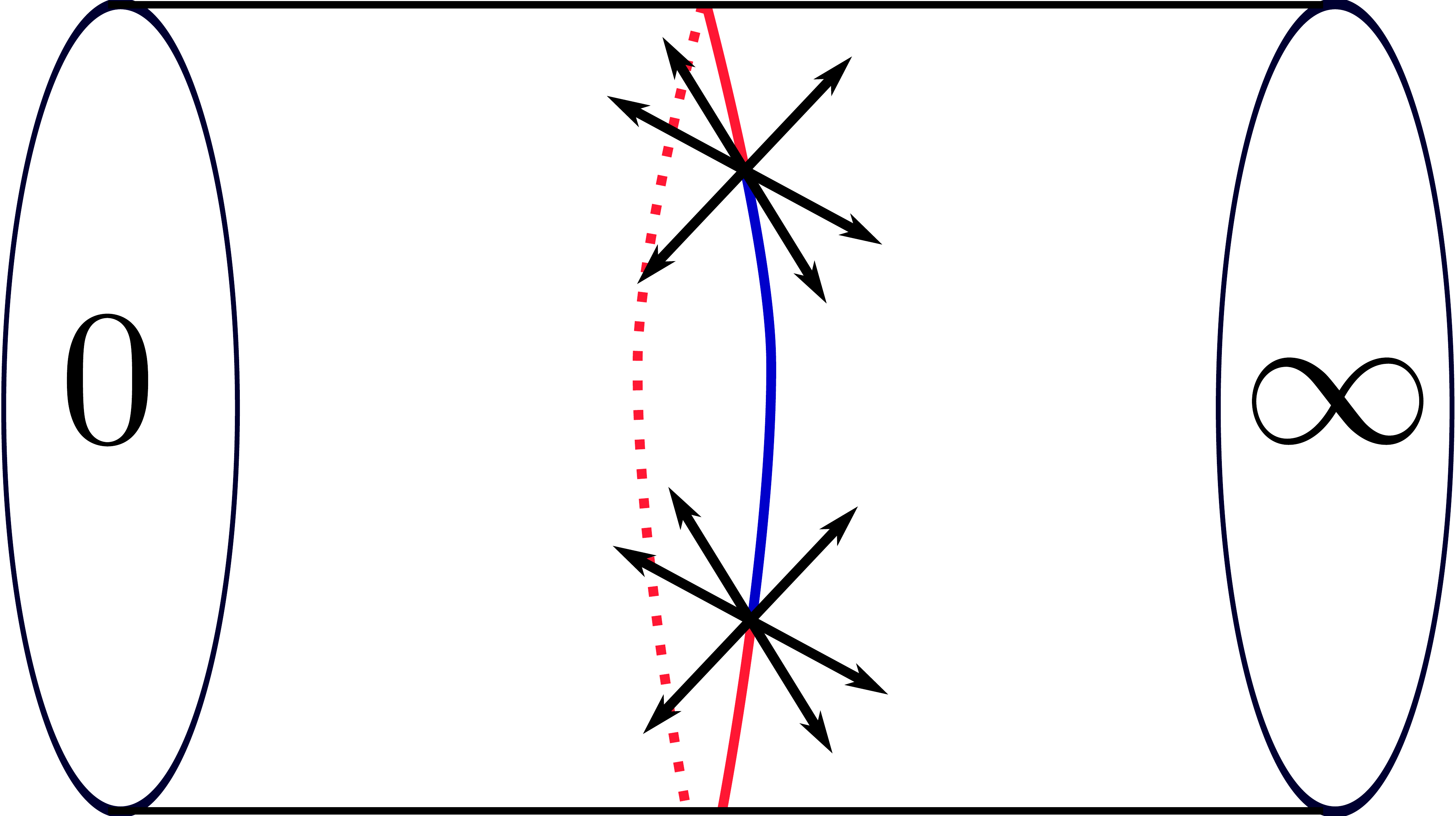}
\includegraphics[width=0.35\textwidth]{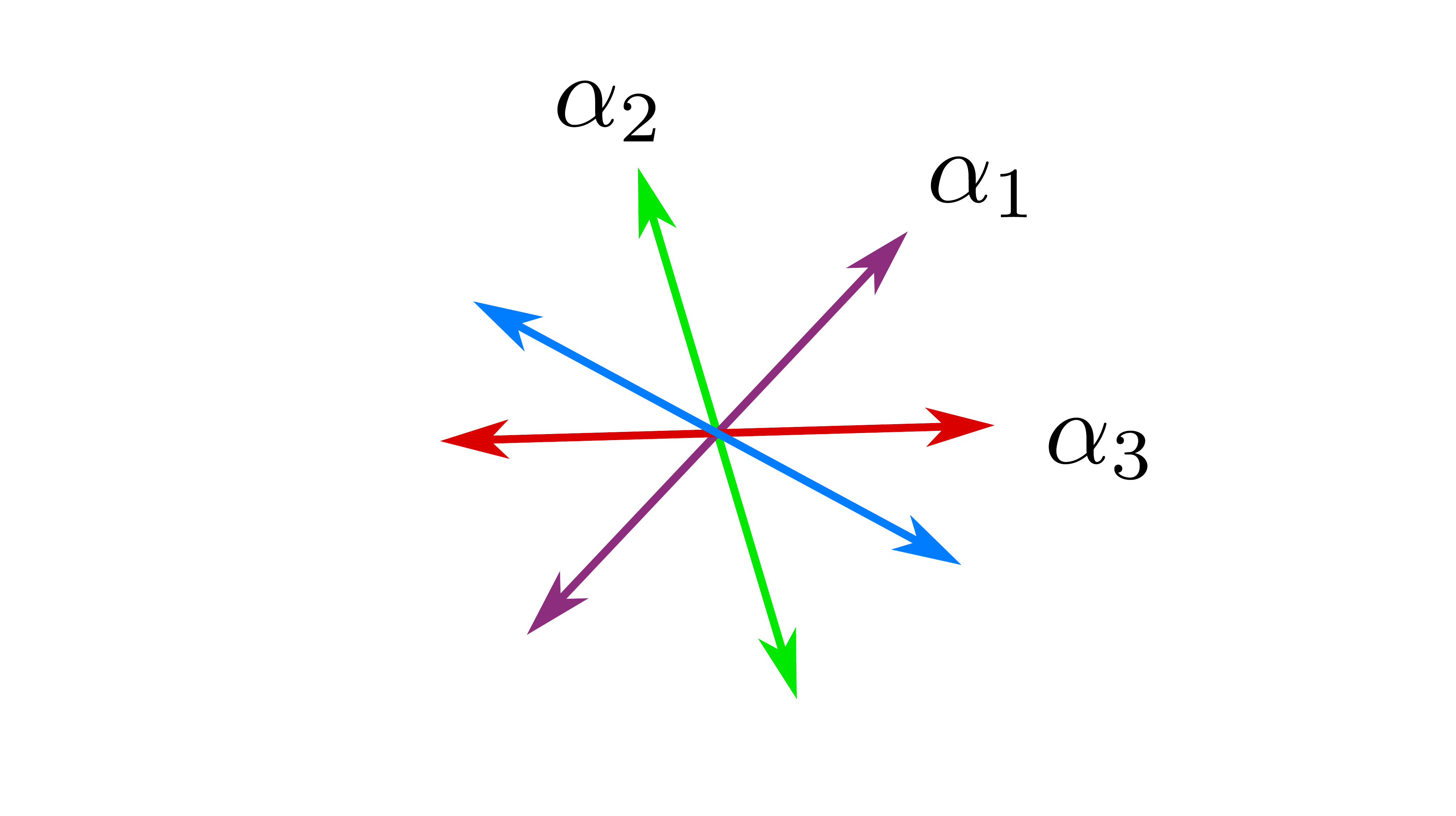}
\caption{A cartoon of the spectral network of SYM with a simply-laced gauge group, for $u=0$. The branch points coalesce into two groups, and the network assumes a highly degenerate structure, with $\CS$-walls of several root types emanating from the same branch point. This situation is highly non-generic, but nevertheless a regular point on $\CB$, no BPS state becomes massless here.
Every jump of the network involves several two-way streets, depicted in red and blue on the left, with both ends on branch points. Several two-way streets may appear simultaneaously at the same phase, even overlapping entirely.}
\label{fig:strong-coupling-SYM}
\end{center}
\end{figure}

The spectral network is particularly simple: from each branch point there will be a number of $\CS$-walls emanating, and there is at least one $\CS$-wall for each root-type, although some of the walls may be overlapping due to the high-degree of discrete symmetry of the curve.
The primary $\mathcal{S}$-walls evolve from the branch points and stretch towards the punctures at $z=0,\infty$.
In so doing, they ocasionally intersect each other, giving rise to joints and secondary walls.

As we vary $\vartheta$, the only $\CK$-wall jumps that appear in this chamber of $\CB$ involve pairs of primary $\CS$-walls of opposite root types running between the two branch points, see Figure \ref{fig:strong-coupling-SYM}.
There will be a sequence of jumps as $\vartheta$ is varied from $0$ to $\pi$.\footnote{This is the half-spectrum of BPS states, the other half being the CPT conjugates}
For any such a jump, the critical network consists of several two-way streets, each of them stretching between the two branch points: $n_a$ of these will run on one side of the singularity at $z=0$, while another $n_b$ will run on the other side. An example is discussed in Figure \ref{fig:jumps-sequence-su3-sym}.

\begin{figure}[ht]
\begin{center}
\includegraphics[width=0.3\textwidth]{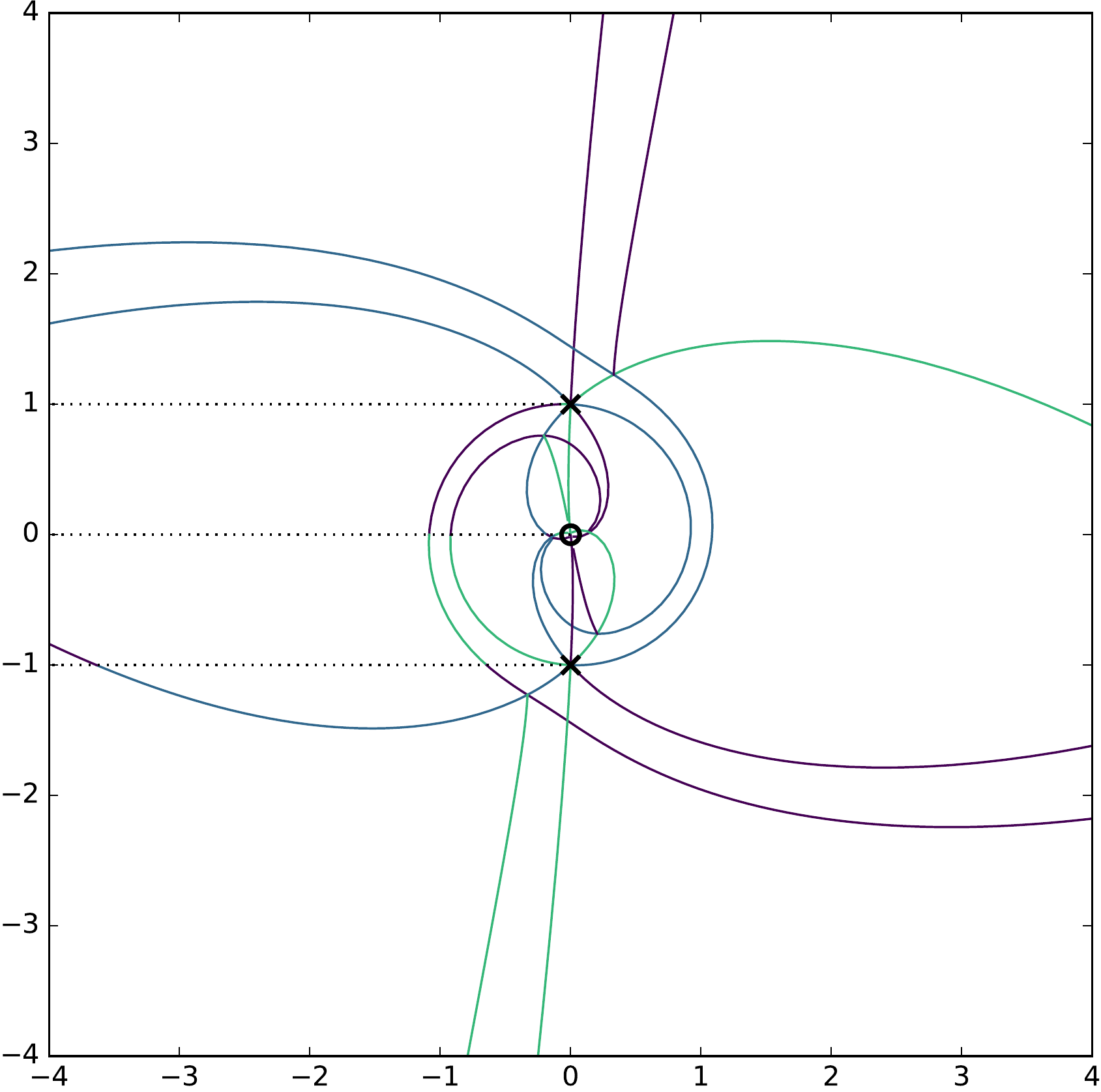}
\includegraphics[width=0.3\textwidth]{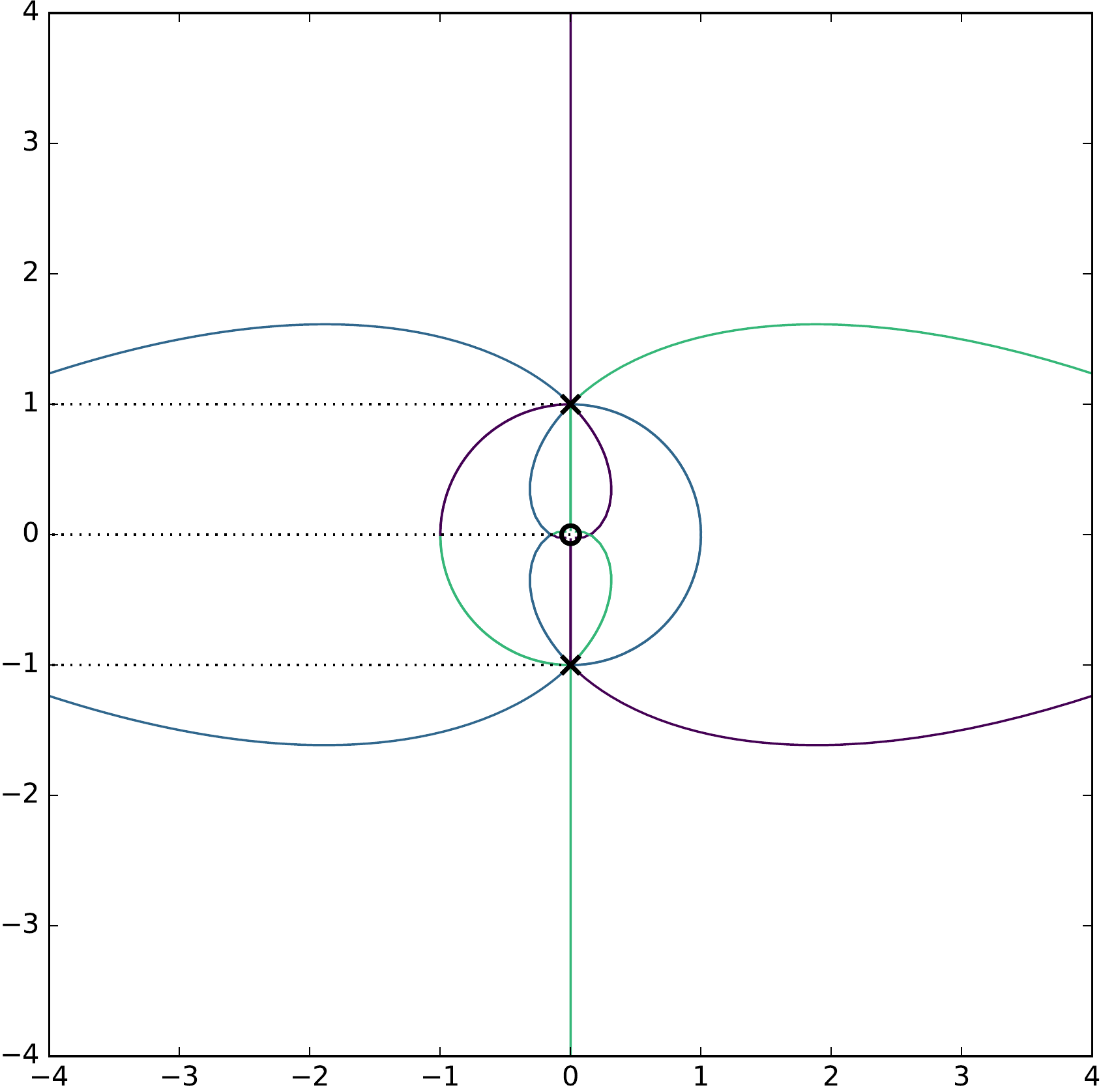}
\includegraphics[width=0.3\textwidth]{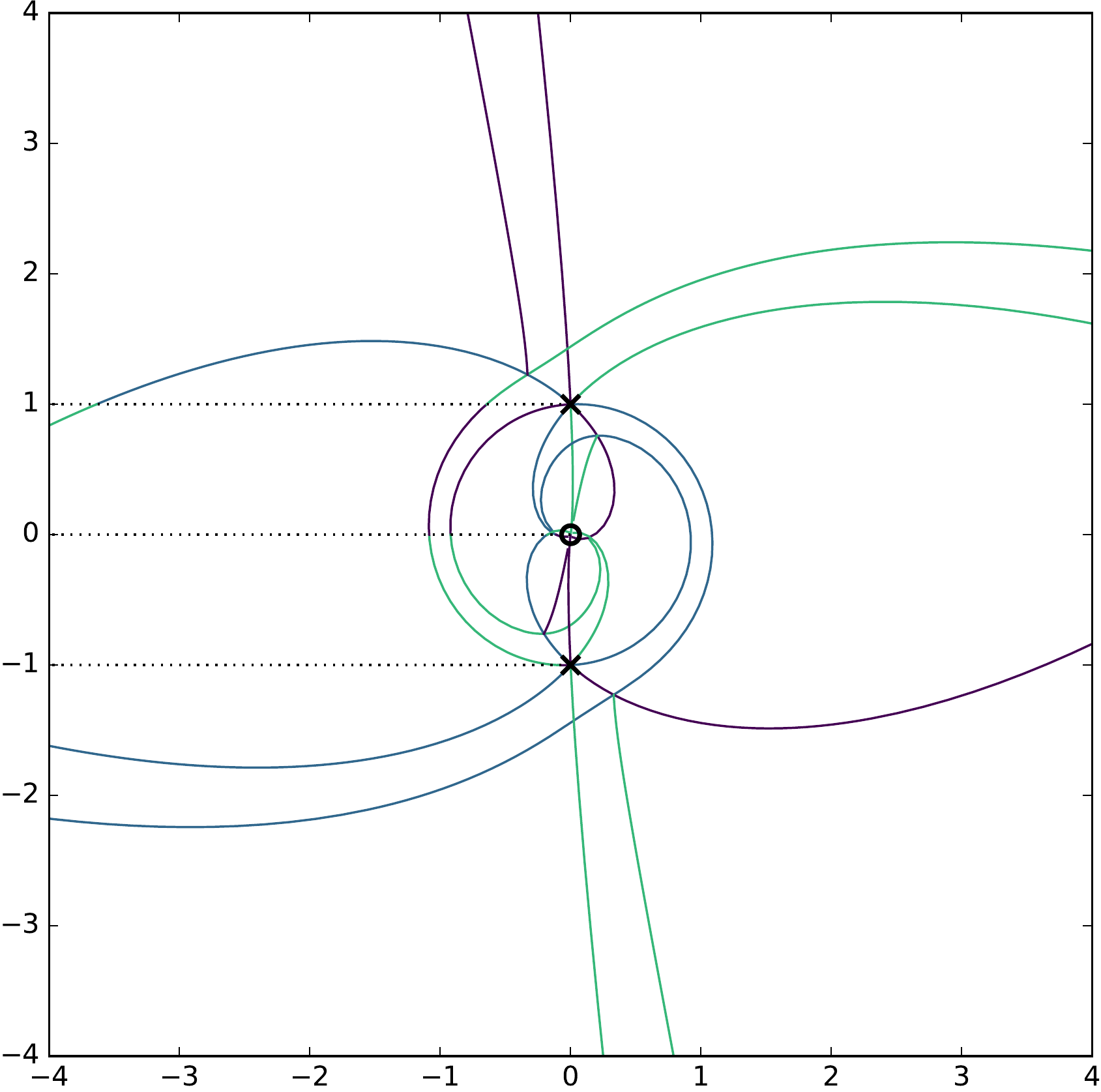}
\caption{A jump for the spectral network of SU(3) SYM. In this case there will be 3 jumps. At each jump $n_a=n_b=1$, i.e.\ two distinct two-way streets appear, one on the left of the singularity and the other on the right of the singularity. In the middle frame these streets lie on the unit circle.
The circle in the middle denotes the irregular singularity at $z=0$, while the branch points are denoted by cross markers. Dashed lines represent branch cuts, and colors denote different root-types along $\CS$-walls.}
\label{fig:jumps-sequence-su3-sym}
\end{center}
\end{figure} 

The soliton content of primary $\mathcal{S}$-walls is particularly simple (see Section \ref{sec:primary-solitons}), and it's easy to see that each jump accounts for a number of BPS states with $\Omega=1$, i.e.\ hypermultiplets.\footnote{Recalling that critical networks with a single two-way street always correspond to hypermultiplets, it's obvious that the jumps such as that of Figure \ref{fig:strong-coupling-SYM} gives us hypermultiplets.}
More precisely, there is one BPS hypermultiplet for each two-way street involved in the jump.
Considering all jumps together gives the full strong coupling spectrum: there is a BPS hypermultiplet for each root of $\fg$. This is in agreement with results from quiver methods \cite{Alim:2011kw}  and provides a nontrivial check of our construction of ADE networks.
We now give some explicit examples.

%%%%%%%%%%%%%%%%%%%%%%%%%%%%%%%%%%
\subsubsection{\texorpdfstring{$\SO(6)$}{SO(6)} SYM in the vector representation}
%%%%%%%%%%%%%%%%%%%%%%%%%%%%%%%%%%

At the origin of the Coulomb branch, the curve in the vector representation is
\be
	\lambda^6 + \frac{\lambda^2}{z^4}\left(z+\frac{1}{z}\right) = 0.
\ee
There is a manifest $\IZ_4$ symmetry on the $x$-plane, the fiber of $T^*C$ at fixed $z$. At generic $z$ the sheets will be arranged in the way depicted in Figure \ref{fig:so6_sheet_symmetry}.

\begin{figure}[ht]
\begin{center}
\includegraphics[width=0.21\textwidth]{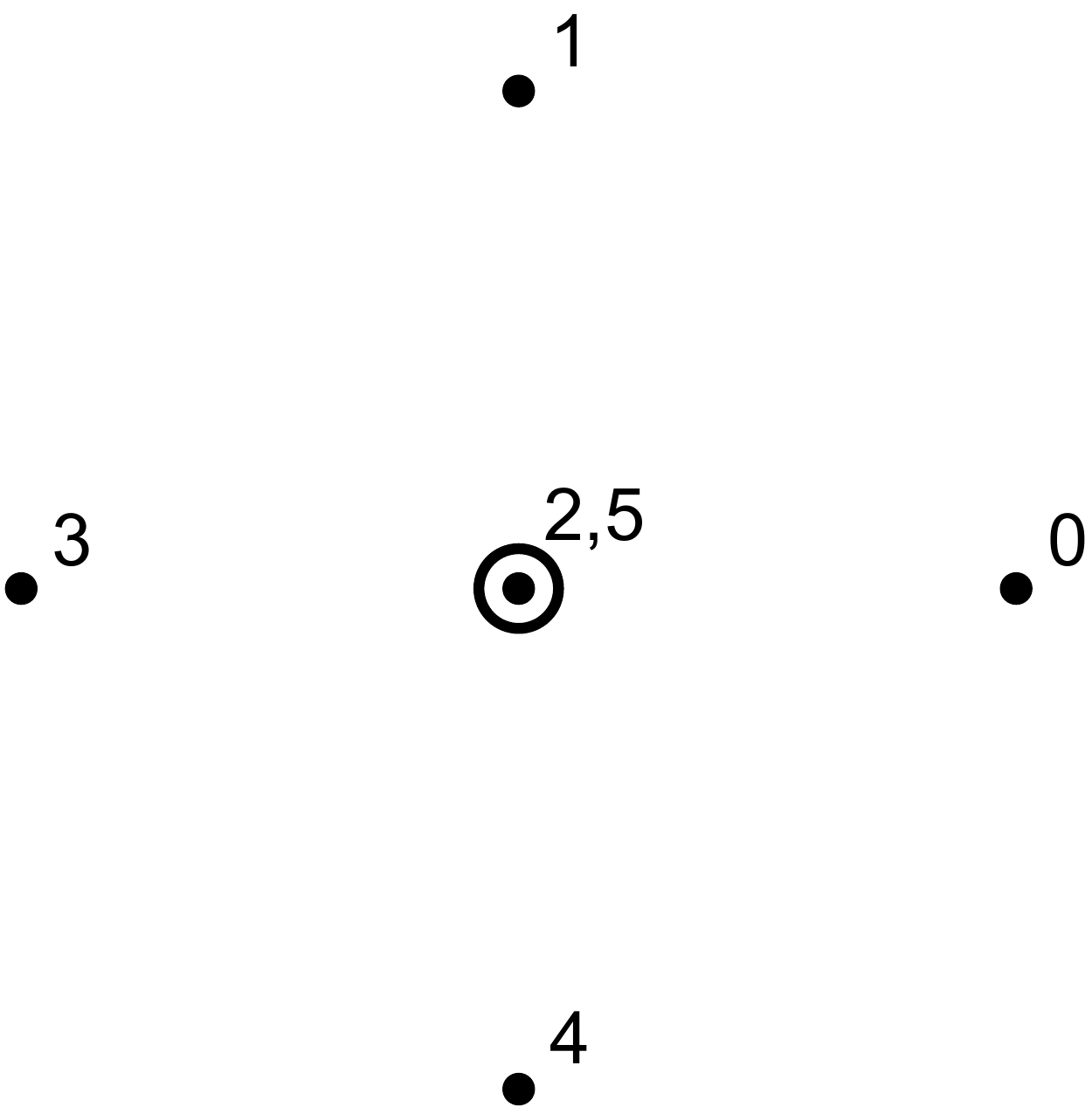}
\caption{Sheets of the vector representation cover of pure SO(6) gauge theory at the origin of $\CB$, shown in the fiber $T_z^*C\simeq \IC$ for generic values of $z$. The label $i$ stands for the sheet corresponding to weight $\weight_i$, its value is $\lambda_i=\langle\weight_i,\varphi(z)\rangle$.}
\label{fig:so6_sheet_symmetry}
\end{center}
\end{figure} 

Let us choose a basis for $\ft^*$, and denote the weights of the vector representation as
\be
\begin{array}{c|ccc|c}
	%\hline
	\weight_0 & (1,0,0) &$\  $& \weight_3 & (-1,0,0) \\
	\weight_1 & (0,1,0) & & \weight_4 & (0,-1,0) \\
	\weight_2 & (0,0,1) & & \weight_5 & (0,0,-1) 
	%\hline
\end{array}
\ee
In this basis, the simple roots are 
\be
	\alpha_1=(1,-1,0)\,,\qquad%
	\alpha_2=(0,1,-1)\,,\qquad%
	\alpha_3=(0,1,1)\,.
\ee
From the diagram of Figure \ref{fig:so6_sheet_symmetry}, together with equation (\ref{eq:geodesic-eq}), it is evident that a wall of root-type $\alpha_{2} = \weight_1-\weight_2 = \weight_5 - \weight_4 = (0,1,-1)$ will be parallel to another wall of different root-type. 
\be
\begin{split}
	\alpha_3 = & \weight_1 - \weight_5 = \weight_2 - \weight_4 = (0,1,1) \,.
\end{split}	
\ee
{To see this, recall that $\lambda_i(z) = \langle \weight_i,\varphi(z)\rangle$, so the quantity $\langle\alpha,\varphi(z)\rangle$ in the figure is represented by the vector stretching between sheets $\lambda_i$ and $\lambda_j$ for each $(i,j)\in\CP_\alpha$.}
Likewise, walls $\CS_{\alpha_1+\alpha_2}$ and $\CS_{\alpha_1+\alpha_3}$ will be parallel.
The primary $\mathcal{S}$-walls of the spectral network emanate from branch points at $z = \pm i$, and there are 15 $\CS$-walls sourced by each branch point.
Due to the degeneracy explained above, however, only 10 distinct trajectories appear, a picture of the network is given in Figure \ref{fig:so6_network}.

\begin{figure}[ht]
\begin{center}
\includegraphics[width=0.65\textwidth]{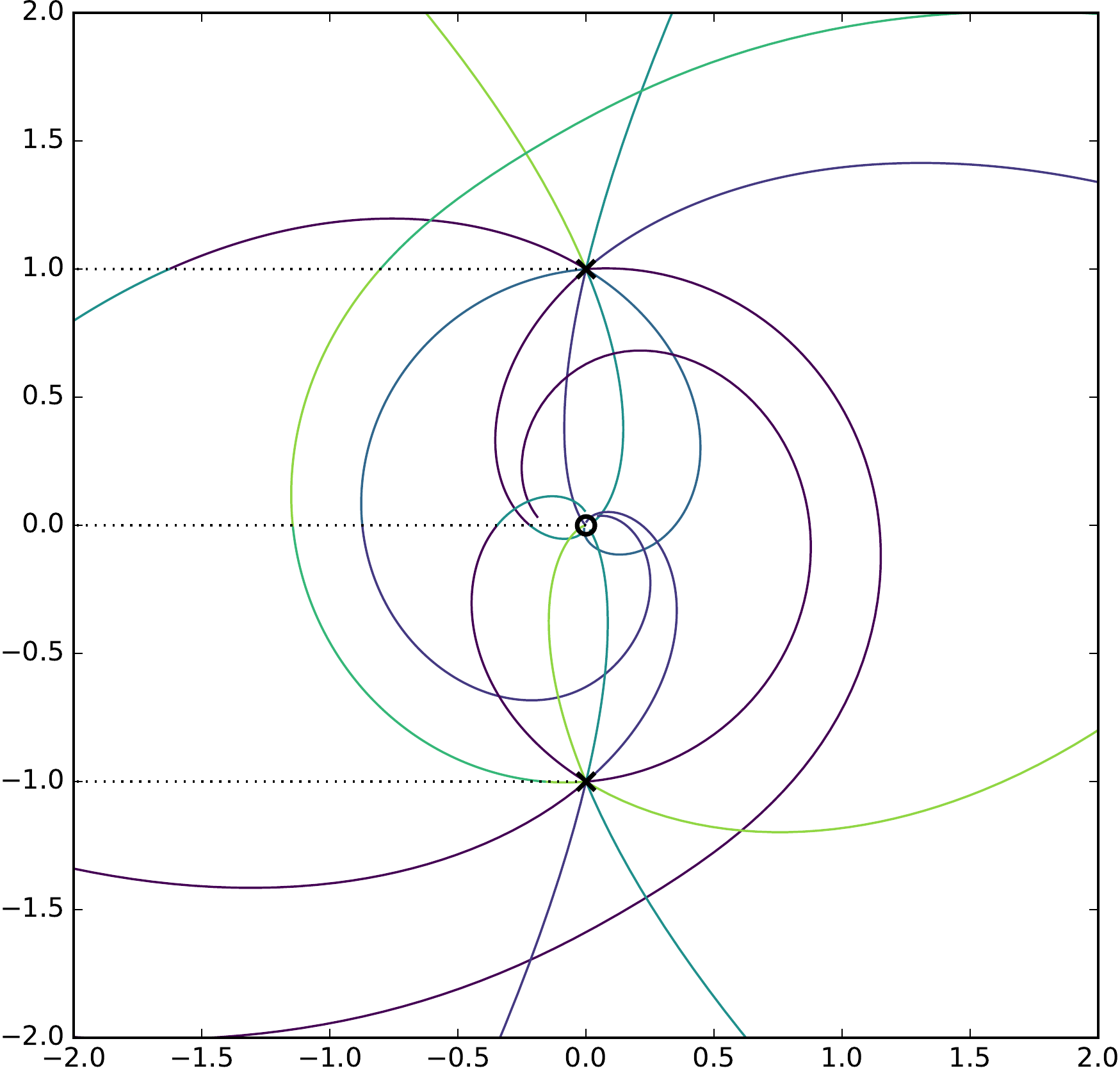}
\caption{The spectral network for a generic phase $\vartheta$, for clarity only the primary $\CS$-walls are shown.
The circle in the middle denotes the irregular singularity at $z=0$, while the branch points are denoted by cross markers. Dashed lines represent branch cuts, and colors denote different root-types along $\CS$-walls.
}
\label{fig:so6_network}
\end{center}
\end{figure}

\noindent The sheet monodromies are as follows: going counter-clockwise around the branching loci at $z=0$, $\pm i$ we have
\be
\begin{split}
	M_0 & : \weight_0 \to \weight_4\to \weight_3\to \weight_1\to \weight_0 \\
	M_1,M_{-1} & : \weight_0 \to \weight_1 \to \weight_3 \to \weight_4 \to \weight_0\,,\ \weight_2\leftrightarrow\weight_5\,.
\end{split}
\ee

There are four jumps $J_1,\dots,J_4$ as $\vartheta$ is varied from $0$ to $\pi$, it is easy to see this from the direct analysis of $\CW_{\vartheta}$, which can be found at \foothref{http://het-math2.physics.rutgers.edu/loom/plot?data=pure_SO_6}{this link}.
For any such jump, the critical network consists of three two-way streets, each of them stretching between the two branch points: two of these will run on one side of the singularity at $z=0$, while the remaining one will run on the other side. 

For example, at the first jump $J_1$ there will be 2 overlapping two-way streets of types $\beta_{1,m}$, $m=1,2$ running around the singularity from the right,
\be
	\beta_{1,1}= (0,1,-1)=\alpha_2\,,\quad \beta_{1,2}= (0,1,1)=\alpha_3\,,
\ee
and a two-way street of type $\alpha_{1,1}$ running around the singularity from the left,
\be
	\alpha_{1,1} =  (1,1,0) = \alpha_1+\alpha_2+\alpha_3 \,.
\ee 
There are then three $L(\gamma_\delta)$ (one for each two-way street of type $\delta$) to construct from the network data. 
To obtain them, note that S-walls $\CS_{\beta_{1,m}}$ and $\CS_{\alpha_{1,1}}$ emanating from the branch point at $z = -i$ will carry \emph{simpletons} of types\footnote{In this equation we pick the convention of labeling 2-way streets by the roots of $\CS$-walls emanating from $z = -i$. More precisely, we use the \emph{local} root type of such $\CS$-walls near the branch point at $z = -i$, in the region between its cut and the cut from $z=0$. With this convention, the $\CS$-walls from $z = -i$ and forming two-way streets to the right of $z=0$ at jump 1 have root types $(0,-1,1)$ and $(0,-1,-1)$ as can be seen by direct inspection.}
\be
\begin{split}
	&  \CP_{\beta_{1,1}} = \{(1,2) , (5,4)\} \qquad \CP_{\beta_{1,2}} = \{(1,5) , (2,4)\} \\
	& \CP_{\alpha_{1,1}} = \{(4,0) , (3,1)\} 
\end{split}
\ee
whereas the walls running into that branch point will carry simpletons of opposite types. Denoting by $p^\pm$ the street on the left and that on the right respectively, we find 
\be
	L(\gamma_{\beta_{1,1}}) = (p^-_2 - p^-_1) + (p^-_4 - p^-_5)
\ee
and so on for $L(\gamma_{\beta_{1,2}}),L(\gamma_{\alpha_{1,1}})$. This gives an explicit characterization of the charges $\gamma_{\beta_{1,m}},\gamma_{\alpha_{1,1}}$ and also proves that the BPS index is $\Omega=1$ for all of them.\footnote{For example, the cycle $L(\gamma_{\beta_{1,1}})$ is primitive, so $\Omega(\gamma_{\beta_{1,1}})$ must be $1$. Also note that the computation of $\Omega$ here is almost identical to the one carried out in Section \ref{sec:hyper_revisited}.}
From the explicit knowledge of the cycles, we may compute their intersection pairing, which turns out to be
\be
	\langle\gamma_{\beta_{1,1}}, \gamma_{\beta_{1,2}}\rangle = 0 \qquad  \langle\gamma_{\alpha_{1,1}}, \gamma_{\beta_{1,m}}\rangle = 0\qquad m=1,2\,.
\ee
The jump $J_1$ therefore captures three mutually local BPS states with $\Omega=1$, i.e.\ hypermultiplets.

The other 3 jumps organize in a similar manner: there will be alternatingly 2 overlapping streets on one side of the singularity, and a single two-way street on the other side.
Having as many as three BPS states appearing at the same jump, hence with multiple central charges at the same phase, a crucial consistency requirement is that all BPS states which appear at the same jump be mutually local.\footnote{I.e. that the intersection pairing of their charges vanishes, otherwise $u$ would be on a wall of marginal stability, and the spectrum would be ill-defined.} We have checked that this is always the case.
{Note that in this example $k_\rho=2$, and by direct inspection one can see that all intersection pairings of populated BPS states are indeed multiples of $2$, guaranteeing integer-valued physical DSZ pairings according to (\ref{eq:intersection-DSZ})}.

The root-types of two-way streets appearing for each jumps are collected in the following table (recall that a two-way street is made of $\CS$-walls of both a root-type and its opposite, we give the root of the $\CS$-wall emanating from the branch point at $z=-i$.)
\be
\begin{array}{c|c|c}
\text{jump}	& \text{street on the left}& \text{street on the right} \\
\hline
J_1 & (1,1,0) & (0,1,-1)\,, \ (0,1,1) \\
J_2 & (0,1,-1)\,,\ (0,1,1) & (1,-1,0) \\
J_3 & (1,-1,0) & (1,0,-1)\,, \ (1,0,1) \\
J_4 & (1,0,-1)\,,\ (1,0,1) & (1,1,0)
\end{array}
\ee
As the table shows, each positive root appears exactly once both on the  street on the left and on the street on the right.

%%%%%%%%%%%%%%%%%%%%%%%%%%%%%%%%%%
\subsubsection{\texorpdfstring{$\SO(6)$}{SO(6)} SYM in a spinor representation}
%%%%%%%%%%%%%%%%%%%%%%%%%%%%%%%%%%

Let us consider again pure $\SO(6)$ gauge theory, this time choosing $\rho$ to be one of the spinor representations.
The shape of the spectral network is not affected by this change, and the root-types of $\CS$-walls are the same as above. The sequence of $\CK$-wall jumps is also identical.
However, the soliton data carried by each $\CS$-wall does change, as does the computation of the 4d BPS degeneracies, although this should not affect the final result.

The weights of the spinor representation are
\be
\begin{array}{c|ccc|c}
	%\hline
	\weight_0 & (1/2,1/2,-1/2) &$\  $& \weight_2 & (-1/2,1/2,1/2) \\
	\weight_1 & (1/2,-1/2,1/2) & & \weight_3 & (-1/2,-1/2,-1/2) \\
	%\hline
\end{array}
\ee
Each root connects now exactly one pair of weights, hence $k_\rho=1$. 
Moreover, by virtue of the linear relation between sheets and weights, the sheets of the spinor representation are arranged in the $x$-plane as shown in Figure \ref{fig:so6_sheet_symmetry_spinor}.

\begin{figure}[ht]
\begin{center}
\includegraphics[width=0.16\textwidth]{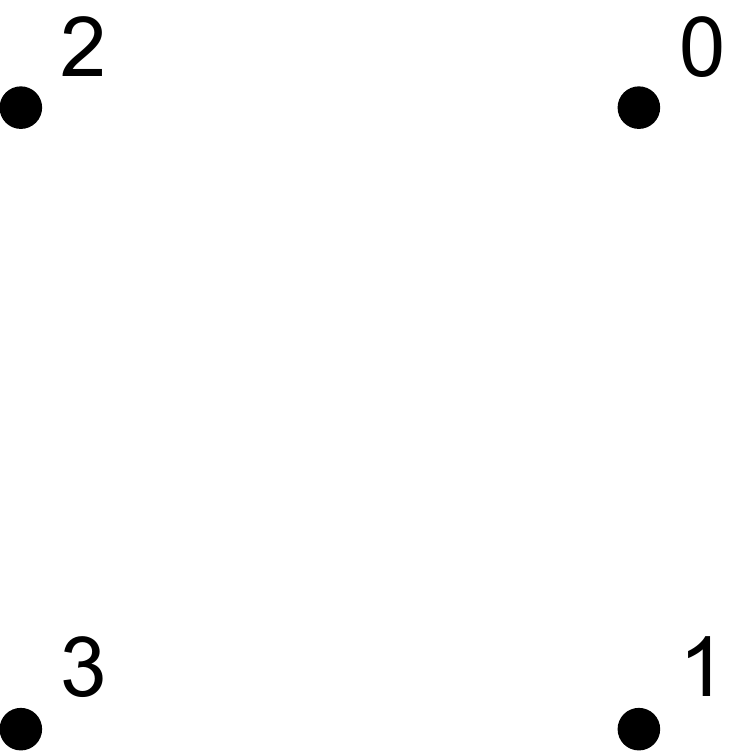}
\caption{Sheets of the spinor representation cover of pure SO(6) gauge theory at the origin of $\CB$, shown in the fiber $T_z^*C\simeq \IC$ for generic values of $z$.}
\label{fig:so6_sheet_symmetry_spinor}
\end{center}
\end{figure} 

From the diagram of Figure \ref{fig:so6_sheet_symmetry_spinor}, we see that a wall of root-type $\alpha_{2}$ will be degenerate with a wall of type $\alpha_3$, and similarly for ${\alpha_1+\alpha_2}$ and ${\alpha_1+\alpha_3}$, as in the case of the vector representation.
At the first $\CK$-wall jump $J_1$, there are three two-way streets in $\CW_c$. These involve three $\CS$-walls emanating from the branch point at $z = -i$, which carry {simpletons} of types
\be
\begin{split}
	&  \CP_{\beta_{1,1}} = \{(0,1)\} \qquad \CP_{\beta_{1,2}} = \{(2,3)\} \qquad \CP_{\alpha_{1,1}} = \{(3,0)\} 
\end{split}
\ee
Therefore we find the following simple expression for the lift of a critical network
\be
	L(\gamma_{\beta_{1,1}}) = (p^-_1 - p^-_0) 
\ee
and so on for $L(\gamma_{\beta_{1,2}}),L(\gamma_{\alpha_{1,1}})$. 
Once again, this gives an explicit characterization of the charges $\gamma_{\beta_{1,m}},\gamma_{\alpha_{1,1}}$ and also proves that the BPS index is $\Omega=1$ for all of them.

A similar analysis for the other jumps recovers the same 4d BPS spectrum as derived in the vector representation.
{With this choice of representation, the physical DSZ pairing coincides with the intersection pairing.}

%%%%%%%%%%%%%%%%%%%%%%%%%%%%%%%%%%
\subsubsection{\texorpdfstring{$\SO(8)$}{SO(8)} SYM in the vector representation}\label{sec:SO_8}
%%%%%%%%%%%%%%%%%%%%%%%%%%%%%%%%%%

At the origin of the Coulomb branch, the curve in the vector representation is
\be
	\lambda^8 + \frac{\lambda^2}{z^6}\left(z+\frac{1}{z}\right) = 0.
\ee
There is a manifest $\IZ_6$ symmetry of the $x$-plane, the fiber of $T^*C$ at fixed $z$. Therefore at  generic $z$ the sheets will be arranged in the way depicted in Figure \ref{fig:so8_sheet_symmetry}.

\begin{figure}[ht]
\begin{center}
\includegraphics[width=0.25\textwidth]{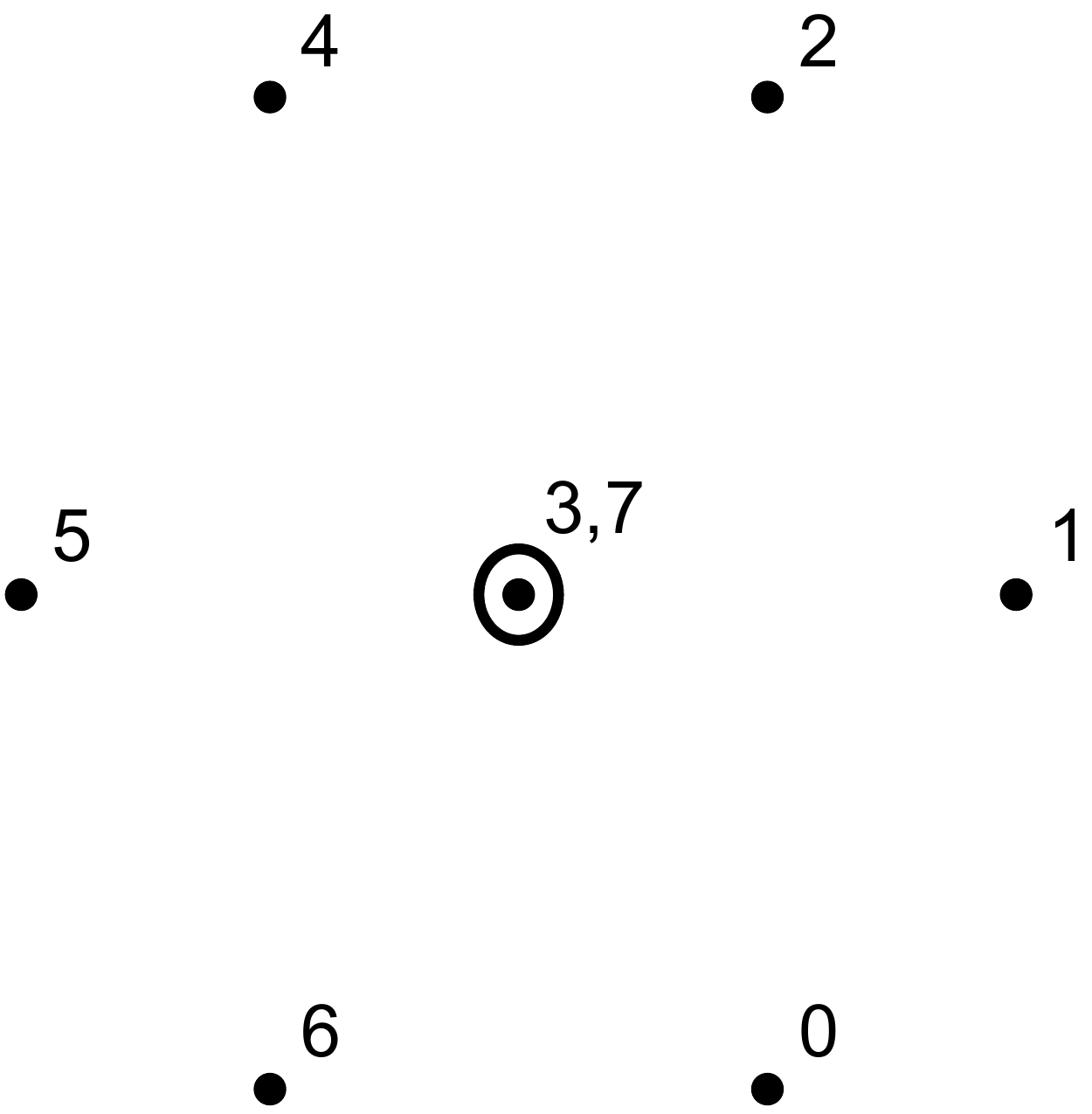}
\caption{Sheets of the vector representation cover of pure SO(8) gauge theory at the origin of $\CB$, shown in the fiber $T_z^*C\simeq \IC$ for generic values of $z$. The label $i$ stands for the sheet corresponding to weight $\weight_i$, its value is $\lambda_i=\langle\weight_i,\varphi(z)\rangle$.}
\label{fig:so8_sheet_symmetry}
\end{center}
\end{figure} 

\noindent Let us choose a basis for $\ft^*$, and denote the weights of the vector representation as
\be
\begin{array}{c|ccc|c}
	%\hline
	\weight_0 & (1,0,0,0) &$\  $& \weight_4 & (-1,0,0,0) \\
	\weight_1 & (0,1,0,0) & & \weight_5 & (0,-1,0,0) \\
	\weight_2 & (0,0,1,0) & & \weight_6 & (0,0,-1,0) \\
	\weight_3 & (0,0,0,1) & & \weight_7 & (0,0,0,-1) \\
	%\hline
\end{array}
\ee
From the diagram of Figure \ref{fig:so8_sheet_symmetry}, it is evident that a wall of root-type $\weight_0-\weight_6 = \weight_2 - \weight_4 = (1,0,1,0)$ will be parallel to walls of different root-types
\be
\begin{split}
	& \weight_1 - \weight_3 = \weight_7 - \weight_5 = (0,1,0,-1) \\ 
	& \weight_1 - \weight_7 = \weight_3 - \weight_5 = (0,1,0,1) 
\end{split}	
\ee
and so on.
The primary $\mathcal{S}$-walls of the spectral network emanate from branch points at $z = \pm i$, and there are 28 walls sourced by each branch point.
Due to the degeneracy explained above, however, only 14 distinct trajectories appear. We give a picture of the network and of the walls emanating from a branch point in Figure \ref{fig:so8_network}.

\begin{figure}[ht]
\begin{center}
\includegraphics[width=0.65\textwidth]{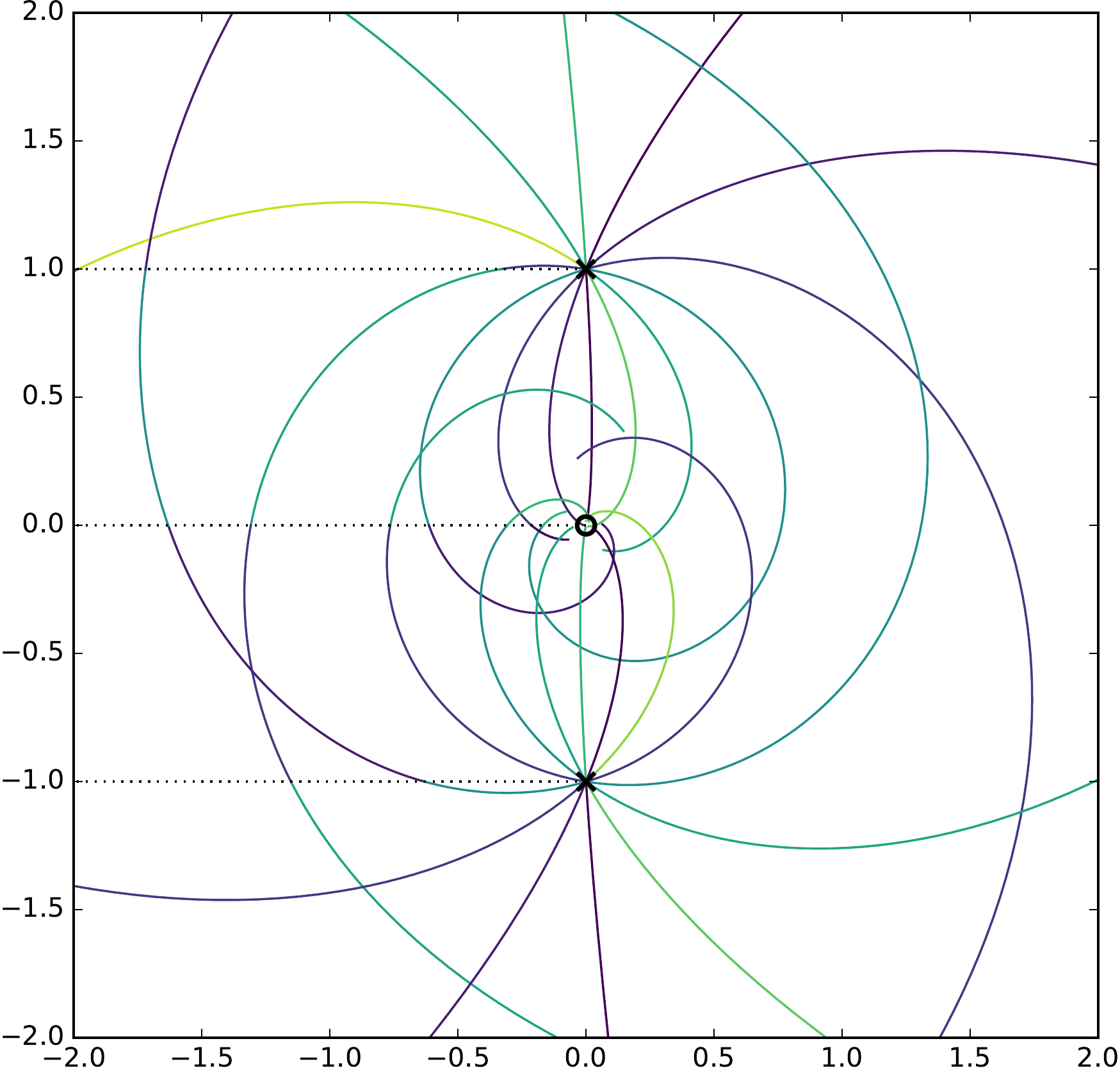}
\caption{The spectral network for a generic phase $\vartheta$, for clarity only the primary streets are drawn. 
The circle in the middle denotes the irregular singularity at $z=0$, while the branch points are denoted by cross markers. Dashed lines represent branch cuts, and colors denote different root-types along $\CS$-walls.
}
\label{fig:so8_network}
\end{center}
\end{figure}

\noindent The sheet monodromies are as follows: going counter-clockwise around the branching loci at $z=0$, $\pm i$ we have
\be
\begin{split}
	M_0 & : \weight_0 \to \weight_6\to \weight_5\to \weight_4\to \weight_2\to \weight_1\to \weight_0 \\
	M_1,M_{-1} & : \weight_0 \to \weight_1 \to \weight_2 \to \weight_4 \to \weight_5 \to \weight_6 \to \weight_0\,,\ \weight_3\leftrightarrow\weight_7\,.
\end{split}
\ee

There are six jumps $J_1,\dots,J_6$ as $\vartheta$ is varied from $0$ to $\pi$, they can be detected by a direct analysis of $\CW_{\vartheta}$, which can be found at \foothref{http://het-math2.physics.rutgers.edu/loom/plot?data=pure_SO_8}{this link}.
For any such a jump, the critical network consists of four two-way streets, each of them stretching between the two branch points: 3 of these will run on one side of the singularity at $z=0$, while the remaining one will run on the other side. 
For example, at the first jump $J_1$ there will be 3 overlapping two-way streets of types $\alpha_{1,i}$, $i=1,2,3$ running around the singularity from the left, 
\be
	\alpha_1= (1,0,1,0)\,,\ \alpha_2=(0,1,0,-1)\,,\ \alpha_3=(0,1,0,1) 
\ee
and a two-way street of type $\beta_{1,1}$ running around the singularity from the right,
\be
	\beta =  (0,-1,-1,0) \,.
\ee
Importantly, $\alpha_{1,i}\cdot\alpha_{1,j}=0$ and therefore solitons from the three overlapping streets will not concatenate. There are then four $L(\gamma_\delta)$ (one for each two-way street of type $\delta$) to construct from the network data. 
To obtain them, note that $\mathcal{S}$-walls $\CS_{\alpha_i}$ and $\CS_\beta$ emanating from the branch point at $z = -i$ will carry \emph{simpletons} of types 
\be
\begin{split}
	&  \CP_{\alpha_1} = \{(6,0) , (4,2)\} \qquad \CP_{\alpha_2} = \{(3,1) , (5,7)\} \\
	& \CP_{\alpha_3} = \{(7,1) , (5,3)\} \qquad \CP_{\beta} = \{(1,6) , (2,5)\}
\end{split}
\ee
whereas the walls running into that branch point will carry simpletons of opposite types. Denoting by $p^\pm$ the left/right streets respectively, we find 
\be
	L(\gamma_{\alpha_1}) = (p^+_0 - p^+_6) + (p^+_2 - p^+_4)
\ee
and so on for $L(\gamma_{\alpha_2}),L(\gamma_{\alpha_3}), L(\gamma_{\beta})$. This gives an explicit characterization of the charges $\gamma_{\alpha_i},\gamma_\beta$ and also proves that the BPS index is $\Omega=1$ for all of them.\footnote{The cycle $L(\gamma_{\alpha_1})$ is primitive, so $\Omega(\gamma_{\alpha_1})$ must be $1$.}
From the explicit knowledge of the cycles, we may compute their intersection pairings, which turn out to be
\be
	\langle\gamma_{\alpha_i}, \gamma_{\alpha_j}\rangle = 0 \qquad  \langle\gamma_{\alpha_i}, \gamma_{\beta}\rangle = 0\,.
\ee
The jump $J_1$ therefore captures four mutually local BPS states with $\Omega=1$, i.e.\ hypermultiplets.

The other 5 jumps organize in a similar manner: there will be alternatingly 3 overlapping streets on one side of the singularity, and a single two-way street on the other side.
Having as many as four BPS states appearing at the same jump, hence with multiple central charges at the same phase, a crucial consistency requirement is that all BPS states which appear at the same jump be mutually local,\footnote{I.e. that the intersection pairing of their charges vanishes, otherwise $u$ would be on a wall of marginal stability, and the spectrum would be ill-defined.} we have checked that this is always the case.

The root-types of two-way streets appearing for each jumps are collected in the following table (recall that a two-way street is made of $\CS$-walls of both a root-type and its opposite, we give the root of the $\CS$-wall emanating from the branch point at $z=-i$.)
\be
\begin{array}{c|c|c}
\text{jump}	& \text{street on the left}& \text{street on the right} \\
\hline
J_1 & (1,0,1,0)\,,\ (0,1,0,-1)\,,\ (0,1,0,1) & (0,-1,-1,0) \\
J_2 & (0,1,1,0) & (0,0,-1,-1)\,,\ (0,0,-1,1)\,,\ (1,-1,0,0) \\
J_3 & (-1,1,0,0)\,,\ (0,0,1,-1)\,,\ (0,0,1,1) & (1,0,-1,0) \\
J_4 & (-1,0,1,0) & (0,1,-1,0)\,,\ (1,0,0,1)\,,\ (1,0,0,-1) \\
J_5 & (0,-1,1,0)\,,\ (-1,0,0,-1)\,,\ (-1,0,0,1) & (1,1,0,0) \\
J_6 & (-1,-1,0,0) & (1,0,1,0)\,,\ (0,1,0,-1)\,,\ (0,1,0,1) 
\end{array}
\ee
As the table shows, each root appears exactly once both on the  street on the left and on the street on the right.

%%%%%%%%%%%%%%%%%%%%%%%%%%%%%%%%%%
\subsubsection{\texorpdfstring{$\gE_6$}{E6} SYM in the \texorpdfstring{$\rho={\mathbf{27}}$}{rho = 27} representation}
%%%%%%%%%%%%%%%%%%%%%%%%%%%%%%%%%%

At the origin of the Coulomb branch, the curve of the $\mathbf{27}$ representation is
\be
	\lambda^{3}\left(\lambda^{24} + 5\,\frac{\lambda^{12}}{z^{12}} \left(z + \frac{1}{z}\right) - \frac{{1}}{108\, z^{24}} \left(z + \frac{1}{z}\right)^{2}\right)=0.
\ee
There is a manifest $\IZ_{12}$ symmetry of the $x$-plane. At  generic $z$ the sheets will be arranged in the way depicted in Figure \ref{fig:e6_sheet_symmetry}.

\begin{figure}[ht]
\begin{center}
\includegraphics[width=0.25\textwidth]{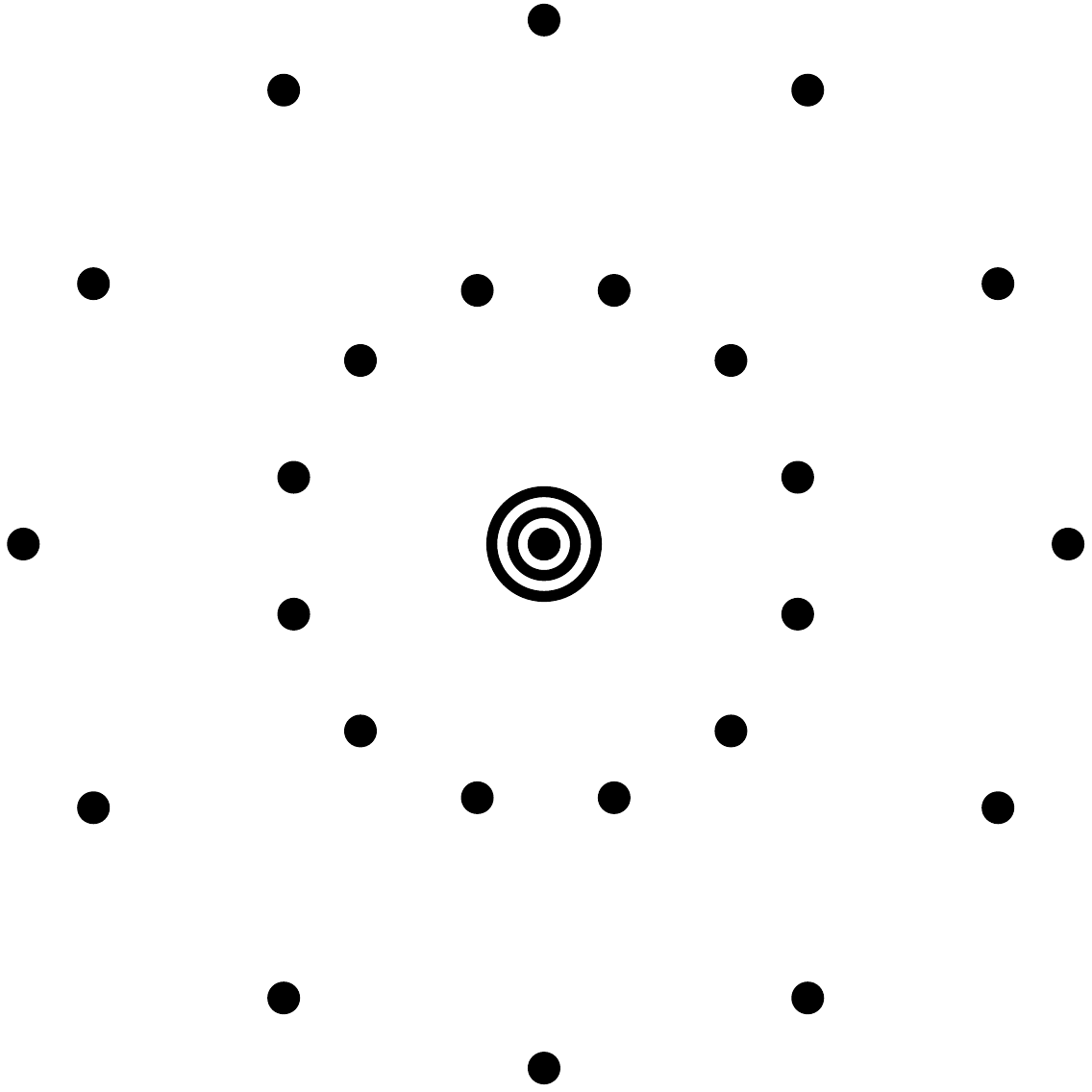}
\caption{Sheets of the $\rho=\mathbf{27}$ cover of pure $E_6$ gauge theory at the origin of $\CB$, shown in the fiber $T_z^*C\simeq \IC$ for generic values of $z$.}
\label{fig:e6_sheet_symmetry}
\end{center}
\end{figure}

From the diagram of Figure \ref{fig:e6_sheet_symmetry}, it is possible to deduce that certain roots have the same phase when considering their projection to the the $x$ plane, 
\be
	\arg(\langle\alpha,\varphi(z)\rangle)=\arg(\langle\alpha',\varphi(z)\rangle)\,,
\ee
resulting in degenerate $\CS$-walls $\CS_\alpha$ and $\CS_{\alpha'}$.
More precisely, each populated phase occurs for precisely three roots. 
This means that every wall emanating from a branch point is 3-fold degenerate.
There are $78$ primary $\CS$-walls sourced by each branch point.
Due to the degeneracy explained above, however, only $26$ distinct trajectories appear. A picture of the network is given in Figure \ref{fig:e6_network}.

\begin{figure}[ht]
\begin{center}
\includegraphics[width=0.65\textwidth]{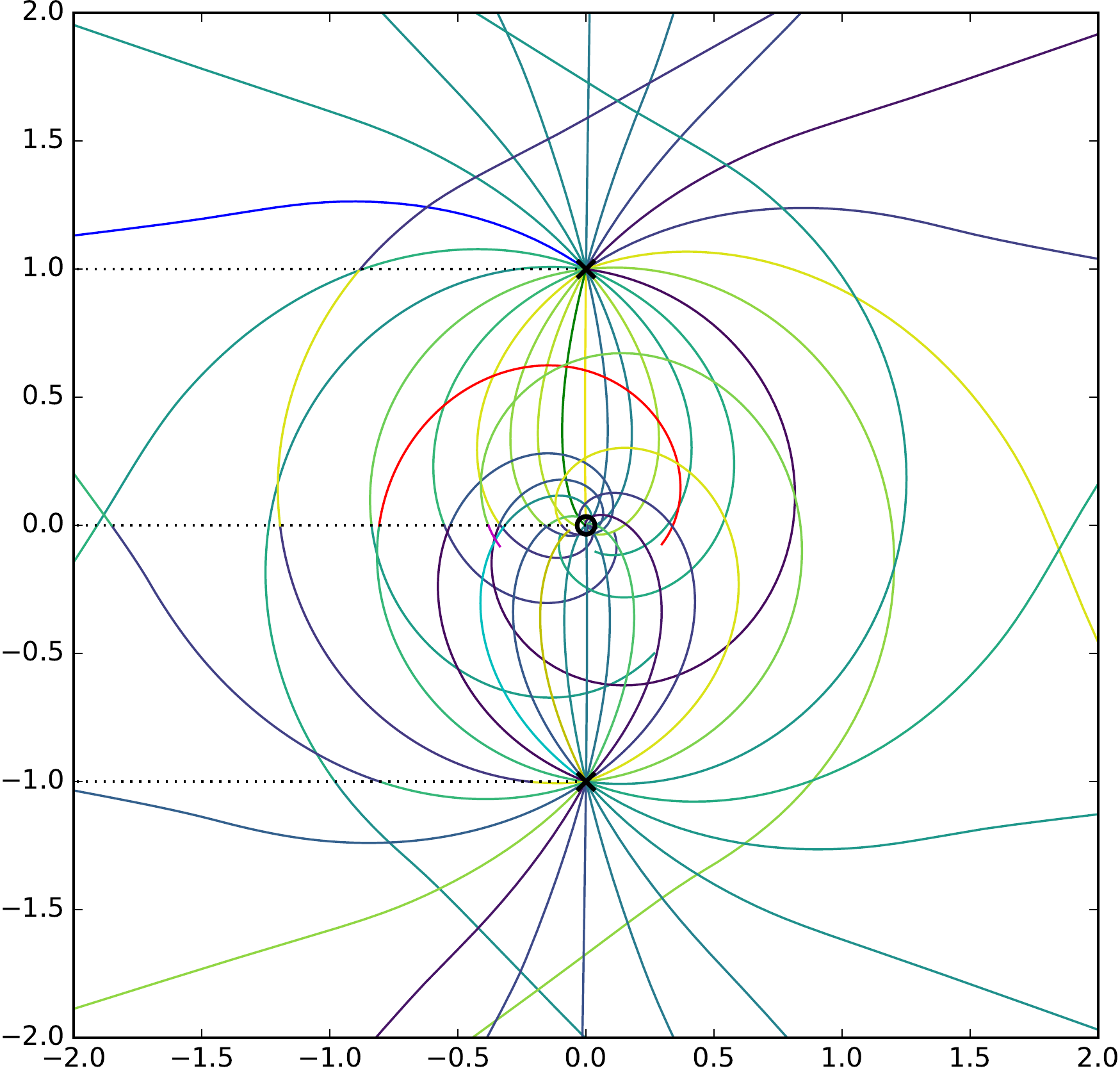}
\caption{The spectral network for a generic phase $\vartheta$, for clarity only the primary streets are shown. 
}
\label{fig:e6_network}
\end{center}
\end{figure} 

As we vary $\vartheta$ from $0$ to $\pi$ we observe 12 $\CK$-wall jumps, as can be seen from direct inspection of the family of networks, which can be found at \foothref{http://het-math2.physics.rutgers.edu/loom/plot?data=pure_E_6}{this link}.
At each of these jumps, the critical network consists of 6 two-way streets, each of them stretching between the two branch points. 3 of these will run on one side of the singularity at $z=0$, while the other 3 will run on the other side. 
There are then 6 $L(\gamma)$'s (one for each two-way street appearing at a given phase) to construct from the network data, for each jump.
To construct them, note that given any root $\alpha$, its set of soliton pairs $\CP_\alpha$ is of order 6. 
For example, for $\alpha_1:=(1, 1, 0, 0, 0, 0, 0, 0)$ we have
\be
	\CP_{\alpha_1} = \{(4, 3), (6, 5), (9, 7), (20, 18), (21, 19), (23, 22)\}
\ee
in the basis of weights given in (\ref{eq:table_E6_weights}).
This means that each $L(\gamma)$ will contain six contributions
\be
	L(\gamma_{\alpha_k}) = \sum_{(i,j)\in\CP_{\alpha_k}} (p^\pm_{j} - p^\pm_{i})\,.
\ee
Like the previous examples, $\CS$-walls emanating from both branch points will carry simpletons. 
Each two-way street of a given root-type will then contribute a hypermultiplet to the spectrum, for a total of 6 for each jump. We find 72 hypermultiplets in total, corresponding to the roots of $E_6$.
It a tedious but straighforward task to construct explicitly the cycles for all these BPS states, but once accomplished this gives an explicit characterization of the charges.

%%%%%%%%%%%%%%%%%%%%%%%%%%%%%%%%%%%%%%%%%%%%%%%%%%%%%%%%%%%%%%%%%%%%%%%%%%%%%%%%
\subsection{Equivalence of Argyres-Douglas fixed points of different 4d \texorpdfstring{$\mathcal{N}=2$}{N=2} theories}
%%%%%%%%%%%%%%%%%%%%%%%%%%%%%%%%%%%%%%%%%%%%%%%%%%%%%%%%%%%%%%%%%%%%%%%%%%%%%%%%

The spectral networks and BPS spectra around the Argyres-Douglas (AD) fixed point of the 4d $\cN=2$ $\SU(N)$, $N_\mathrm{f} = 2$ theory are studied in Section 4.1.1 of \cite{Maruyoshi:2013fwa}. 
At a point in the Coulomb branch moduli space where the number of BPS states is minimal, the BPS spectrum can be conveniently encoded into a quiver.\footnote{To avoid possible confusion, we stress that these are somewhat reminiscent of, but not the same as, the BPS quivers of \cite{Alim:2011kw}. 
The quivers of \cite{Maruyoshi:2013fwa} were employed to classify BPS spectra, but were not endowed with any stability condition, and no application of quiver representation theory was implied.} 
A quiver consists of one node for each populated BPS state, and arrows connecting the nodes pairwise. The number of arrows between two nodes is equal to the DSZ electromagnetic pairing of their charges, whereas the sign of the DSZ pairing determines the orientation of the arrows.
For the class of theories in question, the BPS quiver has the shape of a $\mathrm{D}_{N}$ Dykin diagram, and the fixed point theory (and its deformations) is termed a $\mathrm{D}_{N}$-class theory in \cite{Maruyoshi:2013fwa}.

The AD fixed point of a 4d $\cN=2$ $\SO(2N)$ pure gauge theory is claimed in \cite{Eguchi:1996vu} to be equivalent to the AD fixed point of a 4d $\cN=2$ $\SU(N)$, $N_\mathrm{f} = 2$ theory. Here we study spectral networks of a $\mathrm{D}_{N}$-class theory from the 4d $\cN=2$ $\SO(8)$ pure gauge theory, and show that its BPS spectra across different chambers in the Coulomb branch moduli space are the same as those of a $\mathrm{D}_{N}$-class theory from the 4d $\cN=2$ $\SU(N)$, $N_\mathrm{f} = 2$ theory. This adds evidence to the equivalence of the fixed point theories, and it also shows that the construction of a $\gD$-type spectral network provides physical information that is consistent with previous studies.

The Seiberg-Witten curve of a $\mathrm{D}_{N}$-class theory with $\mathfrak{g} = \gD_{N}$ is obtained by choosing
\begin{align}
	\phi_k &= s_k \left(\dd z \right)^k, \quad (k \neq h^\vee = 2N-2),\\ 
	\phi_{h^\vee} &= \left( s_{h^\vee} + z^2 \right) (\mathrm{d} z)^{h^\vee},
\end{align}
where an irregular singularity is at $z = \infty$. Scaling dimensions of the parameters are $\Delta(s_k) = 2k/N$.

Let us focus on a $\mathrm{D}_{4}$-class theory, whose curve is $\lambda^8 + \phi_2 \lambda^6 + \phi_4 \lambda^4 + \phi_6 \lambda^2 + (\tilde{\phi}_4)^2$ with
\begin{align}
	\phi_{2} = s_{2} (\dd z)^{2},\ \phi_{4} = s_{4} (\dd z)^{4},\ \phi_{6} = \left( s_{6} + z^2 \right) (\dd z)^{6},\ \tilde{\phi}_4 = \tilde{s}_4 (\dd z)^{4}.
\end{align}
By investigating the residue of the irregular singularity at $z = \infty$, we expect that the flavor symmetry is enhanced to $\SU(3)$ when
\begin{align}
	s_4 = \frac{s_2^2}{4},\ \tilde{s}_4 = 0,
\end{align}
and when we fix $s_4$, $\tilde{s}_4$ to satisfy the above relations, the curve equation factorizes to
\begin{align}
	\lambda^2 \left( \lambda^6 + s_2 \lambda^4 + \frac{s_2^2}{4} \lambda^2 + s_6 + z^2 \right).
\end{align}
The two sheets from $\lambda^2 = 0$ correspond to the two zero weights of the vector representation of $\mathfrak{g} = \gD_4$. We can rescale the curve to absorb $s_2$, therefore the theory has a complex 1-dimensional moduli space determined by $s = s_6/(s_2)^3$. By studying the discriminant of the curve, we find that there are two singularities along the real axis of the $s$-plane in the moduli space, one at $s = 0 \eqqcolon s_\mathrm{t}$ and the other at $s = 1/54 \eqqcolon s_\mathrm{s}$. Using spectral networks, we find that at $s = s_\mathrm{t}$ we have three BPS states, which form a triplet of the flavor symmetry $\SU(3)_\mathrm{f}$, becoming massless; and at $s = s_\mathrm{s}$, we have a BPS state without a flavor charge becoming massless.

\begin{figure}[ht]
	\centering
	\begin{tikzpicture}[scale=2]
		\begin{scope}
			\clip (-.5, -.5) rectangle (2.5, .5);
			\fill[black!20] (1, 0) circle (3);
			%\fill[black!30] (1, 0) circle (1);		
		\end{scope}
		
		\begin{scope}
			\clip (-.5, -.5) rectangle (1, .5);		
			\fill[black!30] (2, 0) circle (2);
		\end{scope}

		\begin{scope}
			\clip (1, -.5) rectangle (2.5, .5);		
			\fill[black!30] (0, 0) circle (2);
		\end{scope}

		\draw[->] (-.5, 0) -- (2.5, 0);
		\node[right] at (2.5, 0) {Re$(s)$};
		\draw[->] (0, -0.5) -- (0, .5);
		\node[above] at (0, 0.5) {Im$(s)$};
		
		\draw[fill] (0,0) circle (.025);
		\node (s_t) at (0,0) {};
		\node (s_t_label) at (.25, -.75) {$s_\mathrm{t}$};
		\draw[->, thick] (s_t_label) -- (s_t);

		\draw[fill] (2,0) circle (.025);
		\node (s_s) at (2,0) {};
		\node (s_s_label) at (2 - .25, -.75) {$s_\mathrm{s}$};
		\draw[->, thick] (s_s_label) -- (s_s); 
		
		\node[above] (region_1_label) at (2, .6) {maximal chamber};
		\draw[->, thick] (region_1_label) -- (-.25,.25); 
		\draw[->, thick] (region_1_label) -- (2.25,.25);
				
		\node[above] (region_2_label) at (1, 0.1) {minimal chamber};
		
	\end{tikzpicture}
	\caption{Moduli space of a $\gD_4$-class theory from a pure $\SO(8)$ gauge theory.}
	\label{fig:D_4_s_plane}
\end{figure}
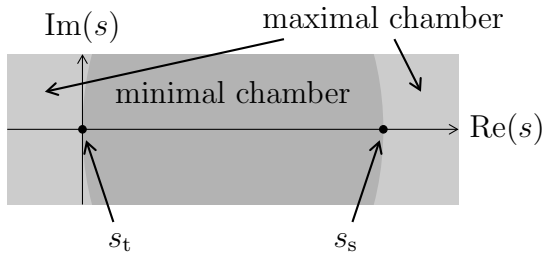

Figure \ref{fig:D_4_s_plane} shows the $s$-plane and the location of the singularities, as well as its BPS chamber structure. There are two chambers that are separated by a boundary containing the two singularities. Across the boundary there is a wall-crossing in the BPS spectrum, and the wall of marginal stability divides the Coulomb branch into a minimal chamber, where we have 4 BPS states, and a maximal chamber, where we have 12 BPS states.

\begin{figure}[!ht]
	\centering
	\begin{subfigure}[c]{.45\textwidth}	
		\includegraphics[width=\textwidth]{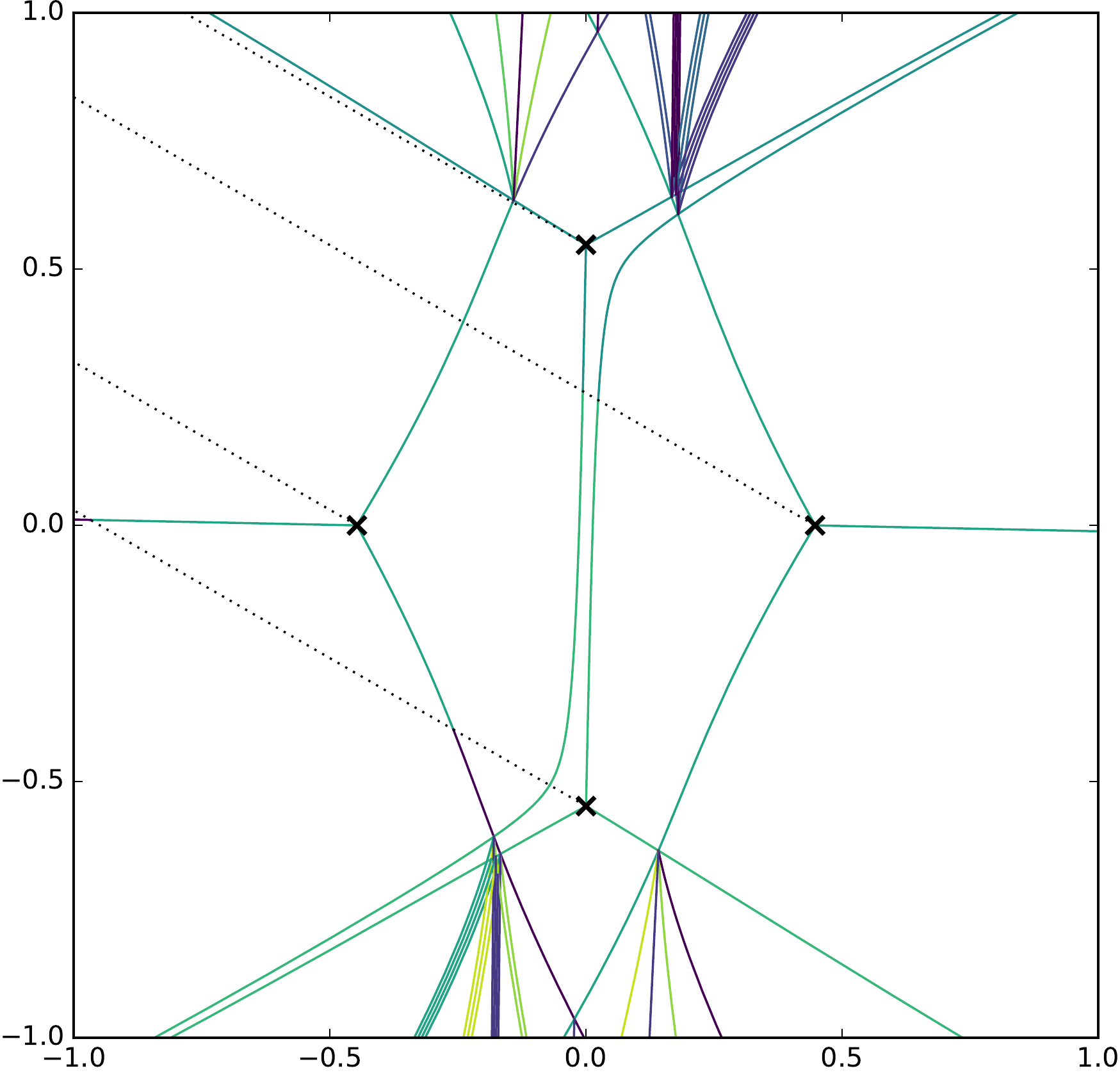}
		%\vspace{-3em}
		\caption{$\theta < \arg \left(Z_{(0,1)}^{\mathbf{3}} \right)$}
		\label{fig:D_4_min_W_before_theta_t}
	\end{subfigure}
	\begin{subfigure}[c]{.005\textwidth}
		\centering
		\begin{tikzpicture}[scale=2]
			\draw[red, thick] (0,2) -- (0, -2);		
		\end{tikzpicture}
		\caption*{}
	\end{subfigure}
	\begin{subfigure}[c]{.45\textwidth}	
		\includegraphics[width=\textwidth]{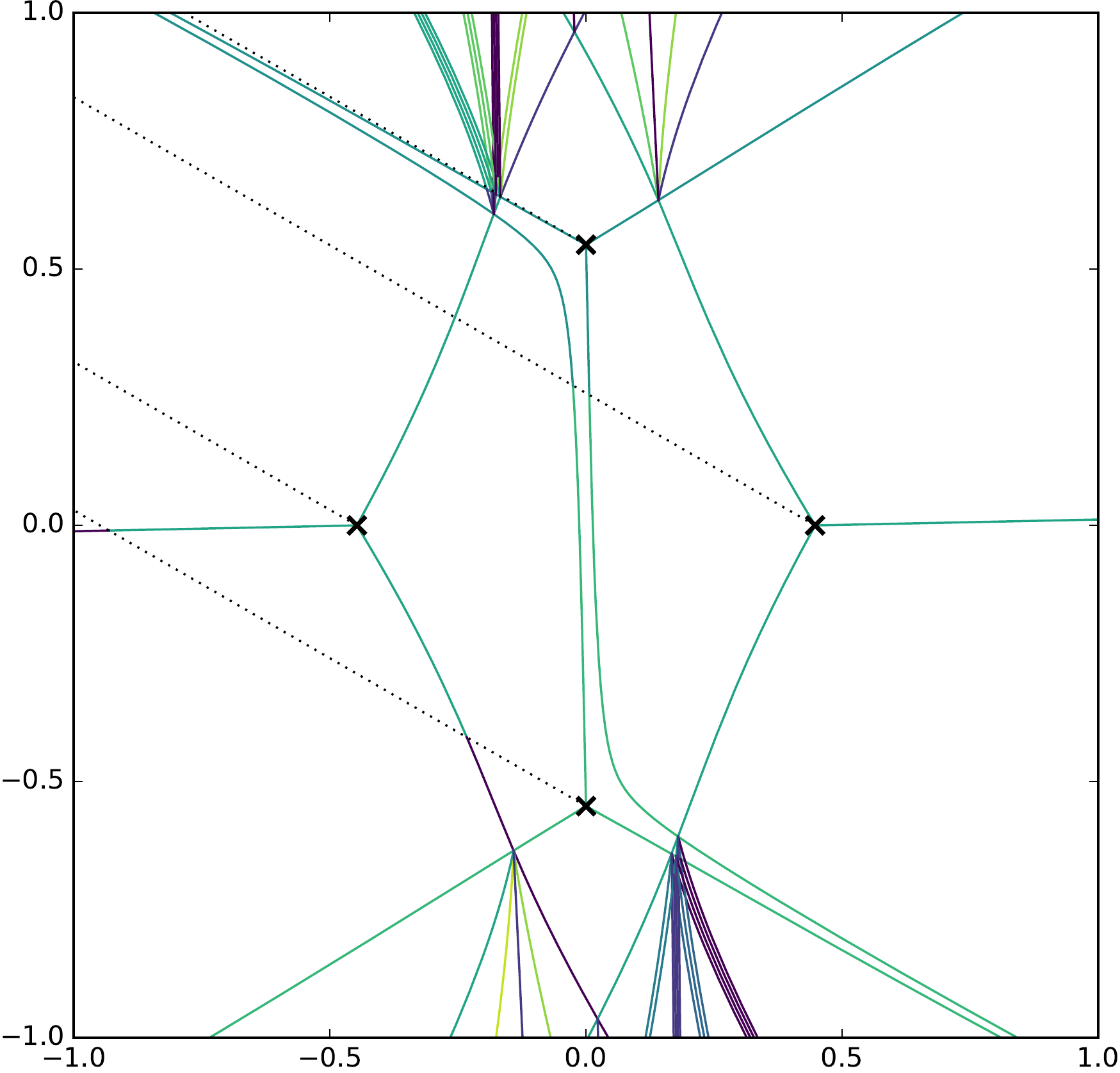}
		%\vspace{-3em}		
		\caption{$\theta > \arg \left(Z_{(0,1)}^{\mathbf{3}} \right)$}		
		\label{fig:D_4_min_W_after_theta_t}
	\end{subfigure}
	
	\begin{subfigure}[c]{.45\textwidth}	
		\includegraphics[width=\textwidth]{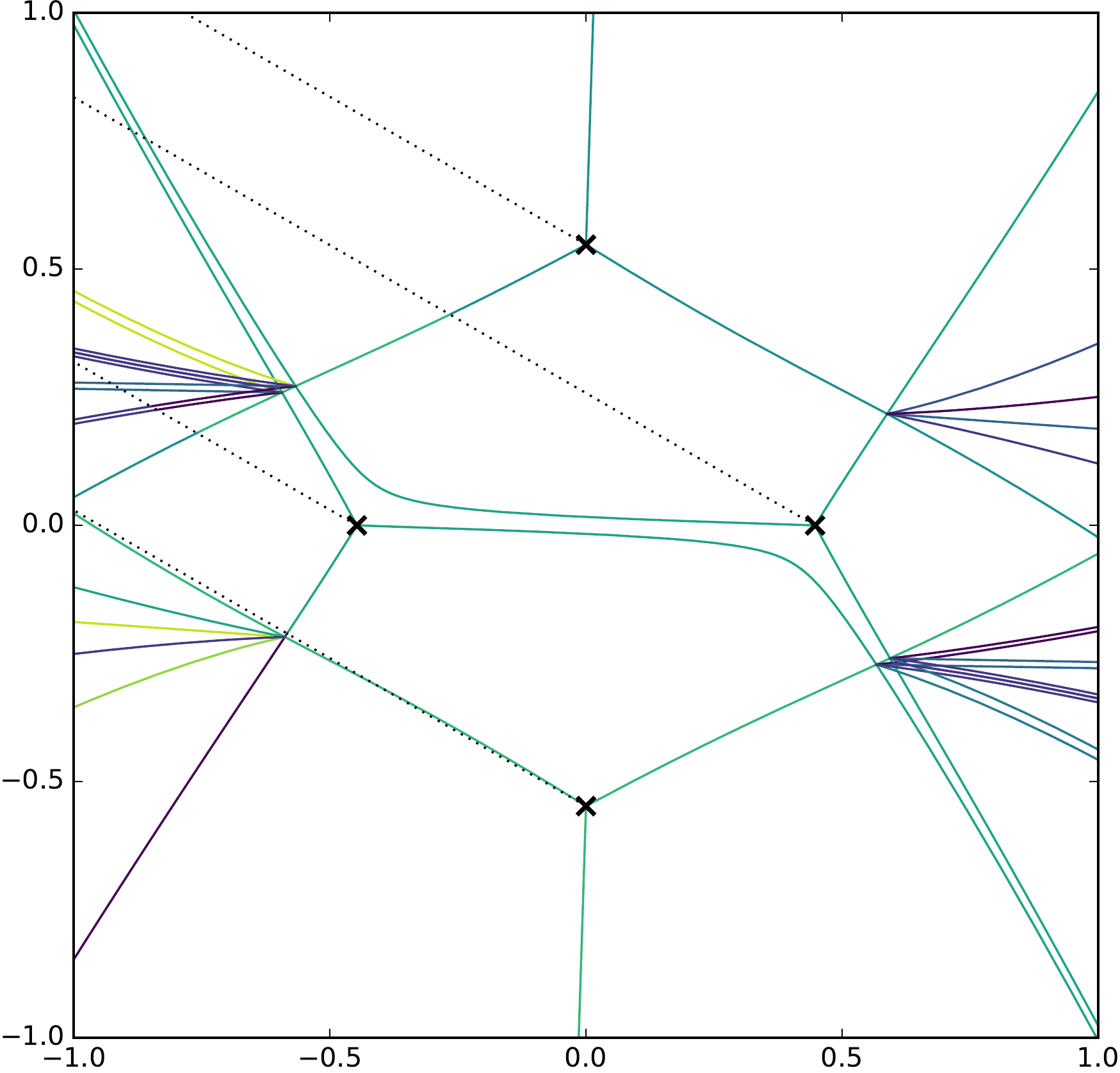}
		%\vspace{-3em}
		\caption{$\theta < \arg \left(Z_{(1,0)} \right)$}		
		\label{fig:D_4_min_W_before_theta_s}
	\end{subfigure}
	\begin{subfigure}[c]{.005\textwidth}
		\centering
		\begin{tikzpicture}[scale=2]
			\draw[blue, thick] (0,2) -- (0, -2);		
		\end{tikzpicture}
		\caption*{}
	\end{subfigure}
	\begin{subfigure}[c]{.45\textwidth}	
		\includegraphics[width=\textwidth]{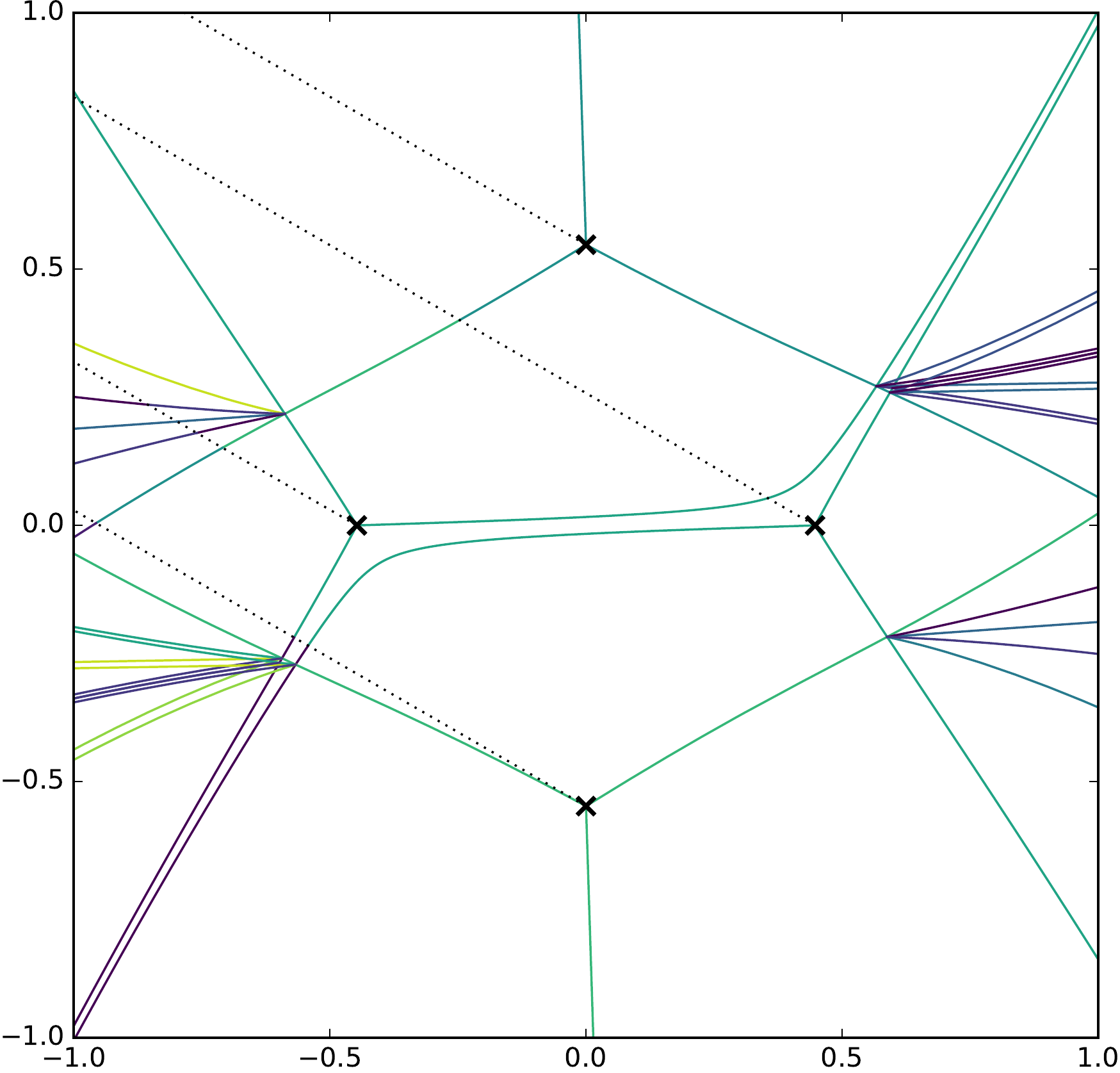}
		%\vspace{-3em}		
		\caption{$\theta > \arg \left(Z_{(1,0)} \right)$}
		\label{fig:D_4_min_W_after_theta_s}
	\end{subfigure}

	\caption{Snapshots of spectral networks of the $\gD_4$-class theory}
	\label{fig:snapshots_of_D_4_min_W}
\end{figure}
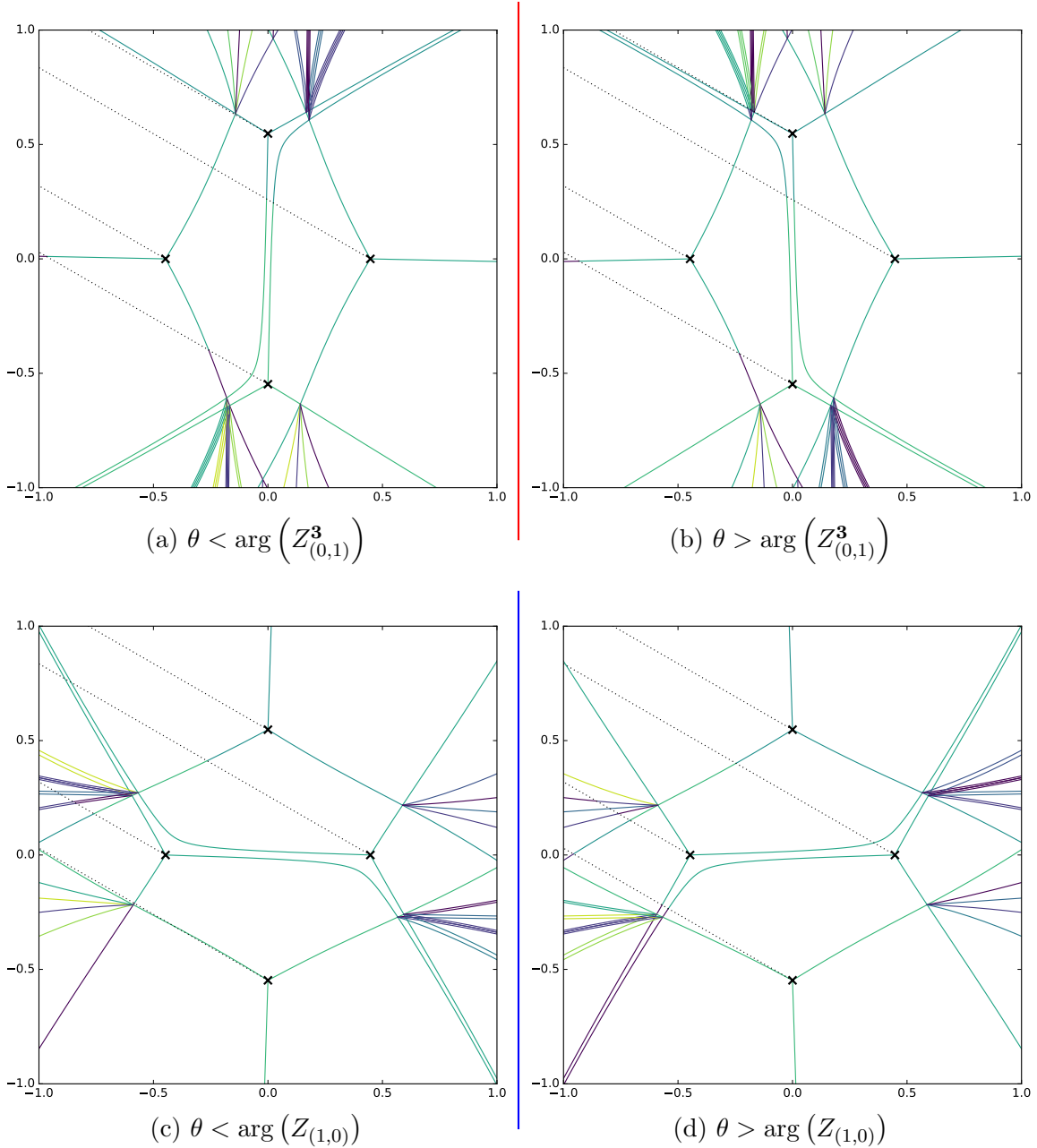

Figure \ref{fig:snapshots_of_D_4_min_W} shows spectral networks around critical phases when the $\gD_4$-class theory is in the minimal chamber. The full family of spectral networks at various phases can be found at \foothref{http://het-math2.physics.rutgers.edu/loom/plot?data=D_4_AD_from_SO_8_min_BPS}{this web page}. At $\vartheta = 0$ there are three two-way streets corresponding to the triplet of $\SU(3)_\mathrm{f}$, showing that the flavor symmetry is enhanced as expected. At $\vartheta = \pi/2$ there is a single two-way street, giving another BPS state. We can calculate the central charges of the 4 BPS states, which is shown in Figure \ref{fig:D_4_minimal_Z}. Because we have a trivialization, we can calculate the DSZ pairings among the BPS states, and using that information we can construct a BPS quiver as shown in \ref{fig:D_4_BPS_quiver}, which is a $\gD_4$ quiver that manifests the $\SU(3)_\mathrm{f}$ as a $S_3$ outer automorphism of the diagram. 

\begin{figure}[!ht]
	\centering	
	\begin{subfigure}[b]{.3\textwidth}	
		\centering	
		\begin{tikzpicture}[scale=1.5]
			\pgfmathsetmacro{\tip}{0.1}
			\begin{scope}[>=latex, very thick]
			\foreach \x in {1, .9, .8}		
				\draw[<->, red] (-1 * \x, 0) -- (1 * \x, 0);				
			\node[right, red] at (1, 0) {$Z_{(0,1)}^{\mathbf{3}}$};
			\draw[<->, blue] (0, -1.5) -- (0, 1.5);
			\node[above, blue] at (0, 1.5) {$Z_{(1,0)}$};
			\end{scope}
		\end{tikzpicture}
		%\vspace{5em}
		\caption{central charges}
		\label{fig:D_4_minimal_Z}
	\end{subfigure}
	\begin{subfigure}[b]{.3\textwidth}	
		\centering	
		\begin{tikzpicture}[scale=1.5]
			\node[W,red] (21) at (0,0) {2};
			\node[W,blue] (1) at (1,0) {1};
			\node[W,red] (22) at (1.5,1.732/2) {2};
			\node[W,red] (23) at (1.5,-1.732/2) {2};
			
			\path (1) edge[->] (21);
			\path (1) edge[->] (22);
			\path (1) edge[->] (23);	
		\end{tikzpicture}
		\vspace{1.5em}
		\caption{BPS quiver}
		\label{fig:D_4_BPS_quiver}
	\end{subfigure}	
	\caption{Minimal BPS spectrum of the $\gD_4$ class theory.}
	\label{fig:D_4_minimal_BPS}
\end{figure}
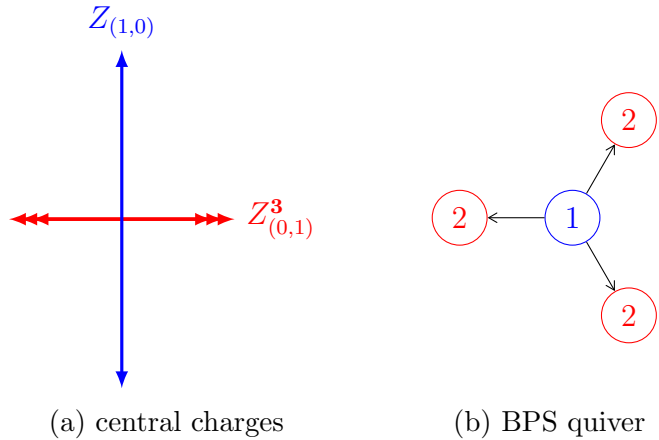

Because of the enhanced flavor symmetry, wall-crossing results in the appearance of 8 additional BPS states in the spectrum. More precisely, the wall of marginal stability on the $s$-plane should be thought as an intersection of multiple walls of marginal stability. 
The Kontsevich-Soibelman identity describing the wall-crossing of the BPS spectrum is
\be\label{eq:AD-wcf}
	\CK_{(1,0)} \CK_{(0,1)}^3 = \CK_{(0,1)}^3  \CK_{(1,3)}  \CK_{(1,2)}^3  \CK_{(2,3)}  \CK_{(1,1)}^3  \CK_{(1,0)} 
\ee
where $(m,n)$ denotes the electric charge of each BPS state, with $\langle(1,0),(0,1)\rangle = 1$.
Flavor charges have been suppressed since they don't play a role in the wall-crossing formula, but are easily recovered as follows. 
In the minimal chamber, the state with charge $(0,1)$ is a triplet of $SU(3)_{\mathrm{f}}$, then in the maximal chamber those states with charge $(m,1)$ are in the $\mathbf{3}$, states with charges $(m,2)$ are in the $\mathbf{\overline{3}}$, and  states with charges $(m,3)$ are singlets. These considerations are reflected in the BPS indices that appear in (\ref{eq:AD-wcf}).

\begin{figure}[!ht]
	\centering
	\begin{subfigure}[c]{.45\textwidth}	
		\includegraphics[width=\textwidth]{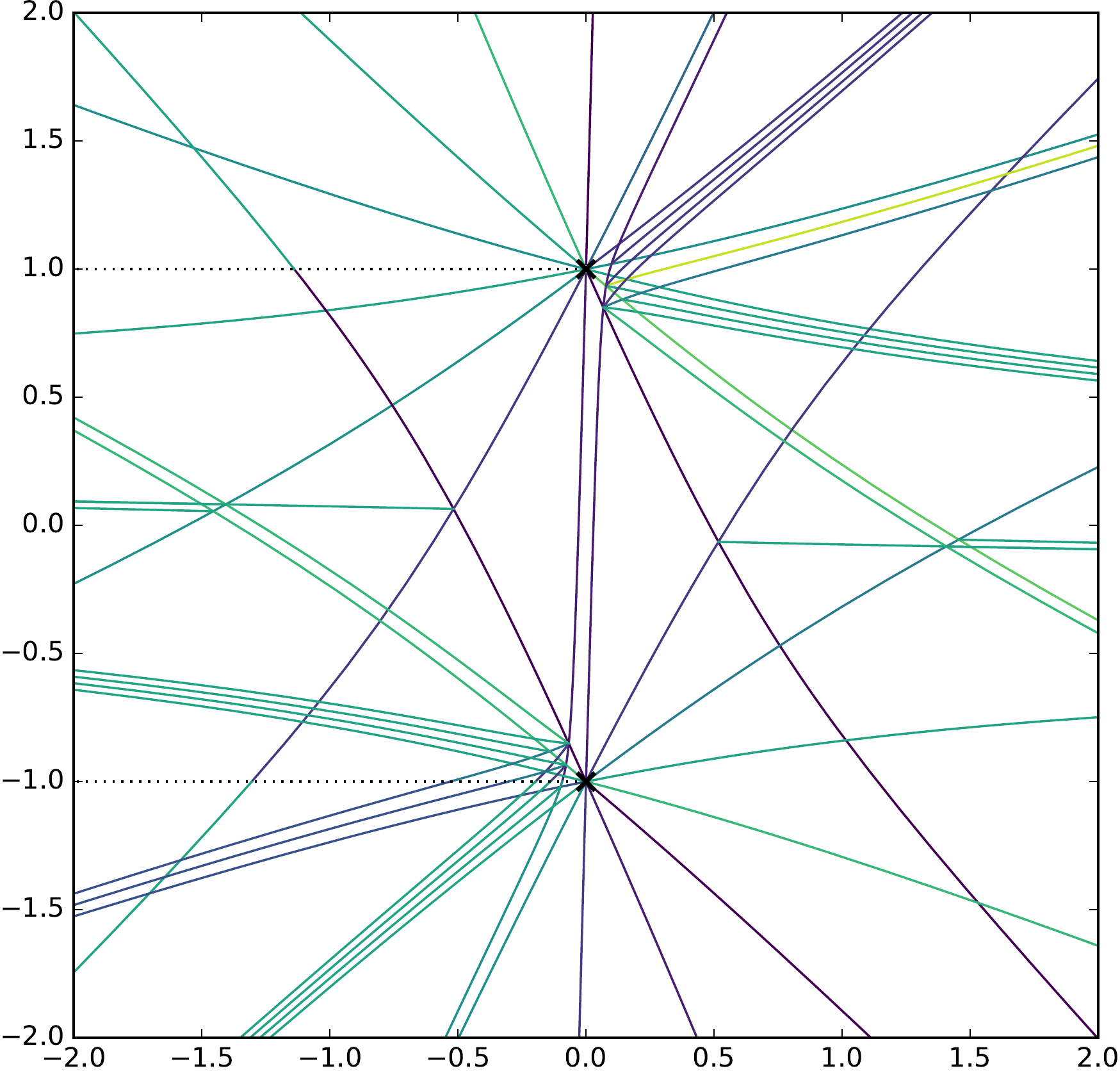}
		%\vspace{-3em}
		\caption{$\theta < \arg \left(-Z_{(0,1)}^{\mathbf{3}} \right)$}
		\label{fig:D_4_max_sym_W_before_theta_t}
	\end{subfigure}
	\begin{subfigure}[c]{.005\textwidth}
		\centering
		\begin{tikzpicture}[scale=2]
			\draw[red, thick] (0,2) -- (0, -2);		
		\end{tikzpicture}
		\caption*{}
	\end{subfigure}
	\begin{subfigure}[c]{.45\textwidth}	
		\includegraphics[width=\textwidth]{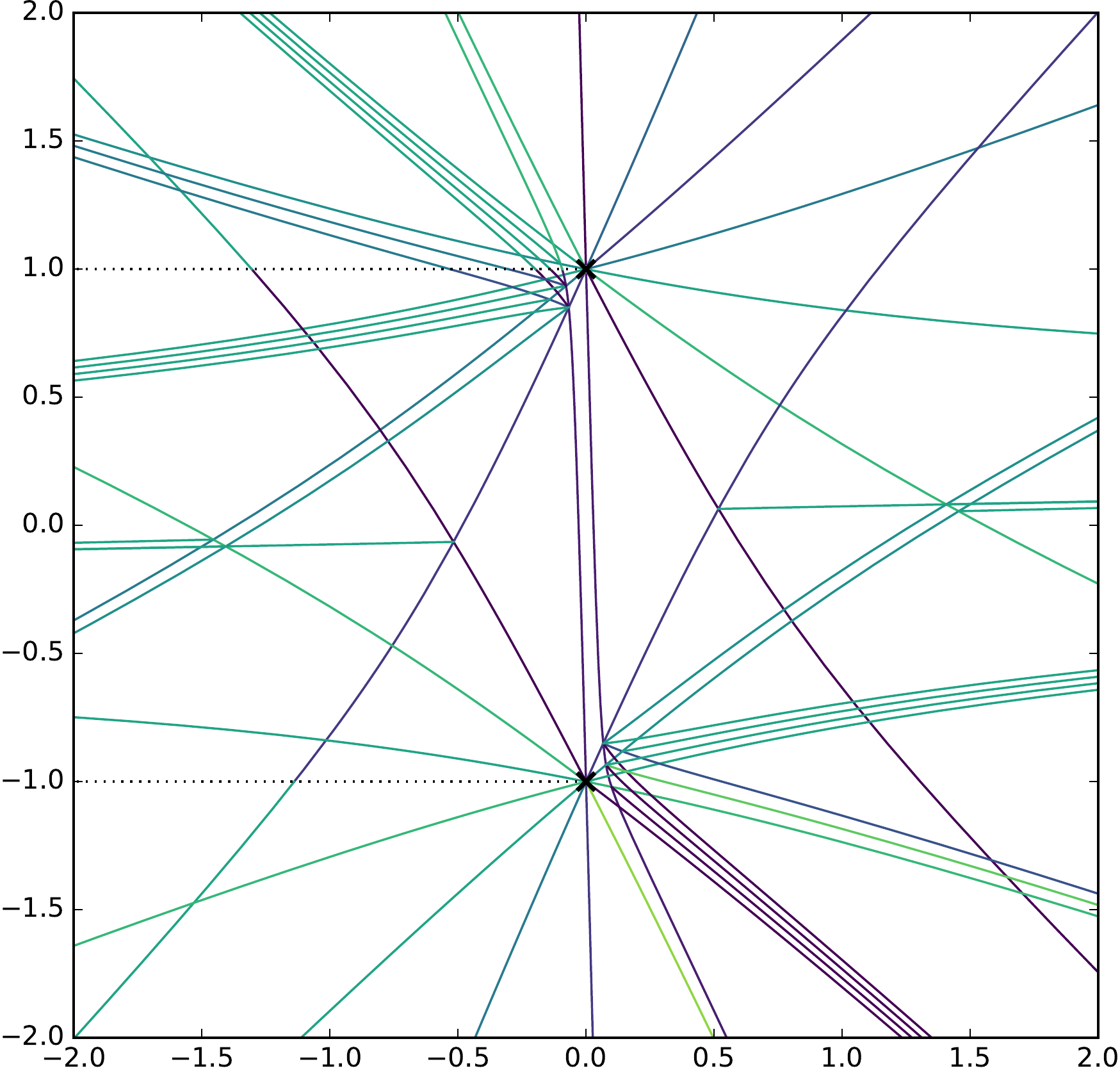}
		%\vspace{-3em}
		\caption{$\theta > \arg \left(-Z_{(0,1)}^{\mathbf{3}} \right)$}
		\label{fig:D_4_max_sym_W_after_theta_t}
	\end{subfigure}
	
	\begin{subfigure}[c]{.45\textwidth}	
		\includegraphics[width=\textwidth]{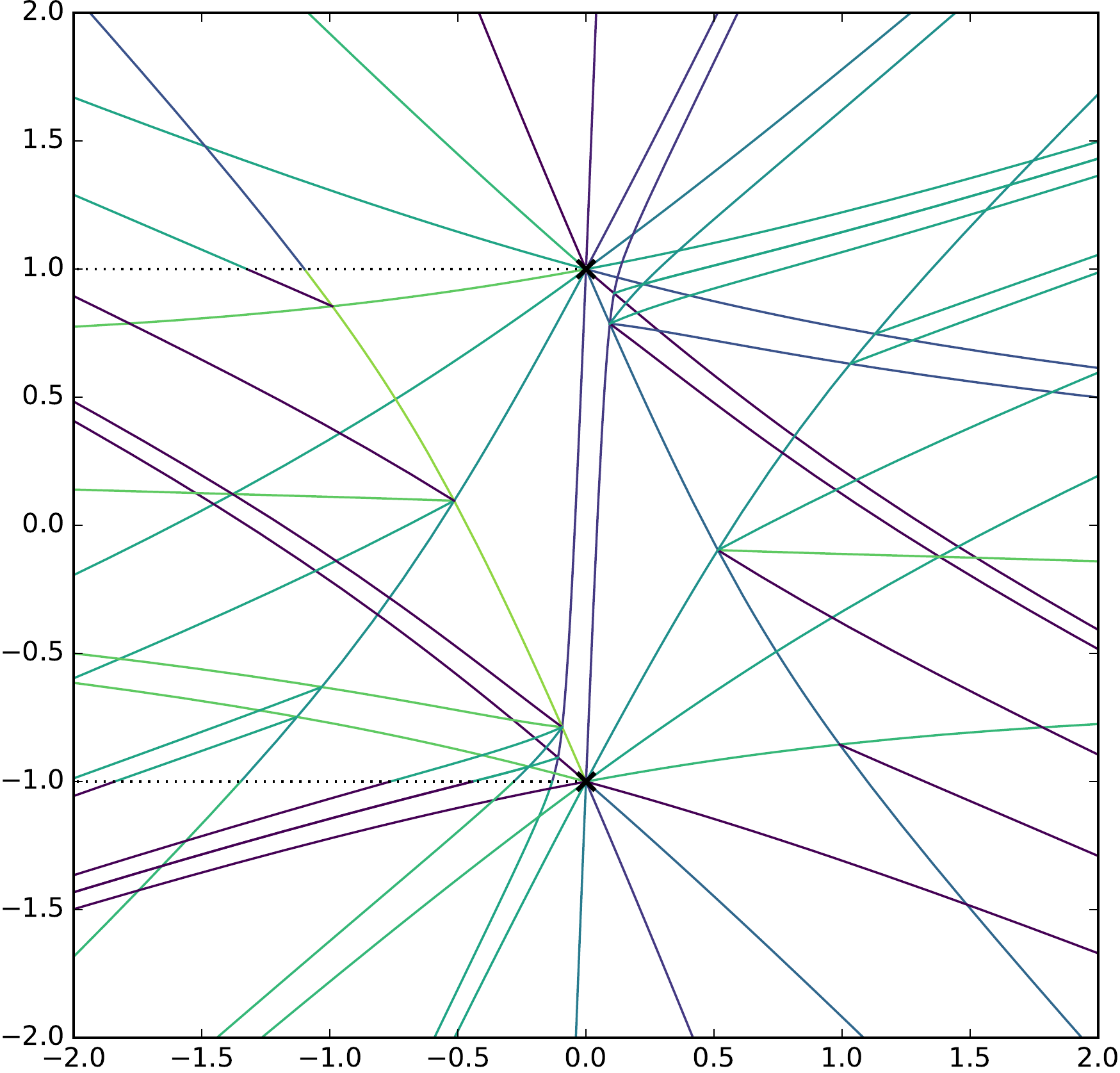}
		%\vspace{-3em}
		\caption{$\theta < \arg \left(Z_{(1,0)}\right)$}
		\label{fig:D_4_max_sym_W_before_theta_s}
	\end{subfigure}
	\begin{subfigure}[c]{.005\textwidth}
		\centering
		\begin{tikzpicture}[scale=2]
			\draw[blue, thick] (0,2) -- (0, -2);		
		\end{tikzpicture}
		\caption*{}
	\end{subfigure}
	\begin{subfigure}[c]{.45\textwidth}	
		\includegraphics[width=\textwidth]{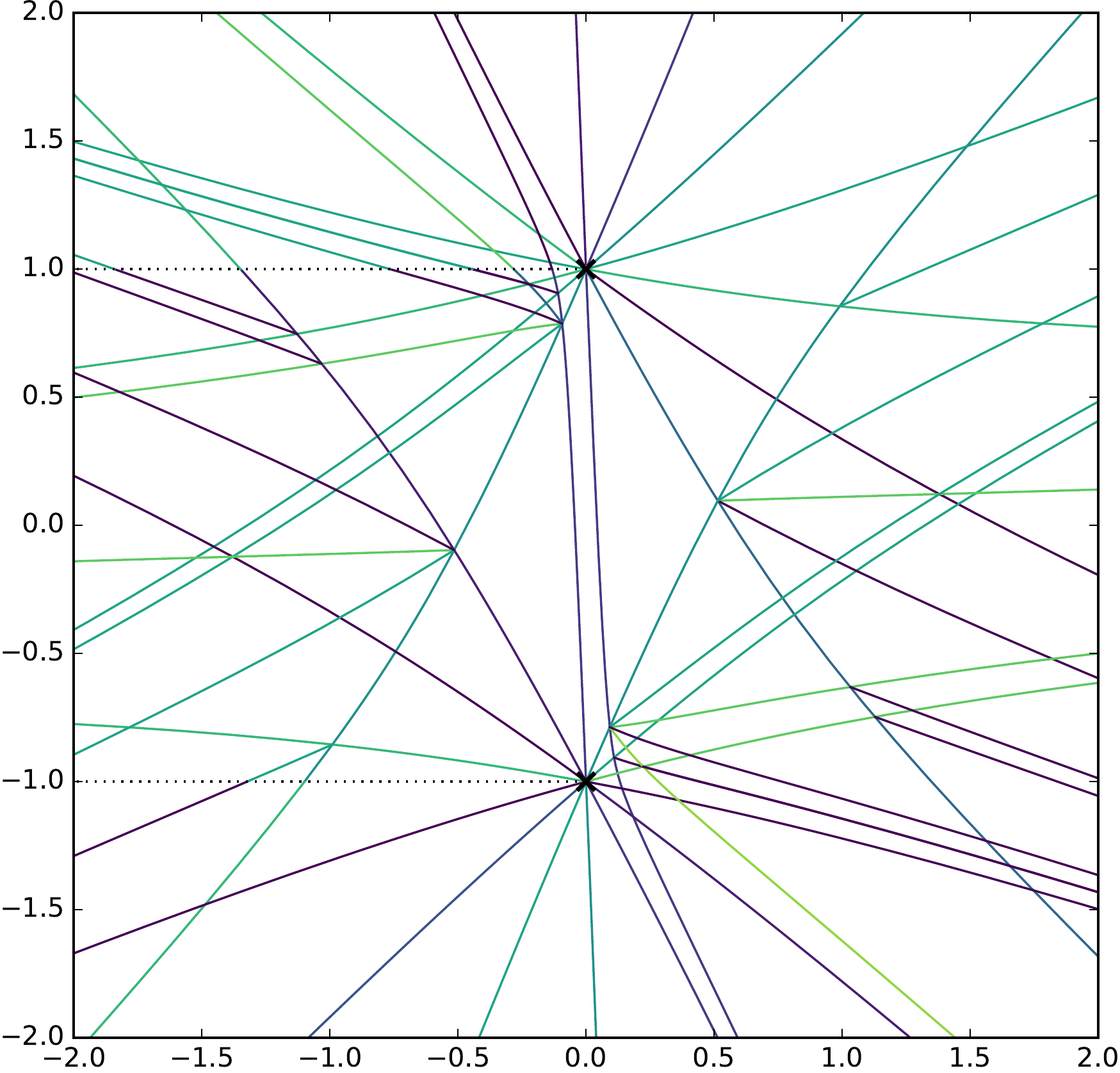}
		%\vspace{-3em}
		\caption{$\theta > \arg \left(Z_{(1,0)}\right)$}
		\label{fig:D_4_max_sym_W_after_theta_s}
	\end{subfigure}

	\caption{Snapshots of spectral networks of the $\gD_4$-class theory}
	\label{fig:snapshots_of_D_4_max_sym_W}
\end{figure}

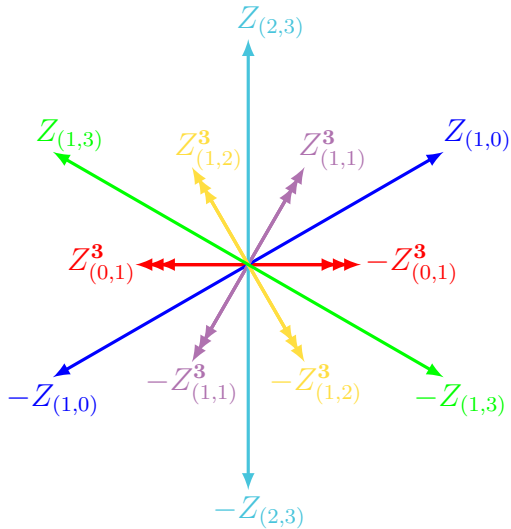
\begin{figure}[!ht]
	\centering	
	\begin{tikzpicture}[scale=1.5]
		\begin{scope}[>=latex, very thick]
			\foreach \r in {1, .9, .8}		
				\draw[<->, red] (-1 * \r, 0) -- (1 * \r, 0);
			\node[red] at (1 + .45, 0) {$-Z_{(0,1)}^{\mathbf{3}}$};
			\node[red] at (-1 - .3, 0) {$Z_{(0,1)}^{\mathbf{3}}$};
			
			\pgfmathsetmacro{\a}{30}
			\pgfmathsetmacro{\tx}{2*cos(\a)}
			\pgfmathsetmacro{\ty}{2*sin(\a)}			
			\draw[<->, blue] (-\tx, -\ty) -- (\tx, \ty);
			\node[blue] at (\tx + .3, \ty + .15) {$Z_{(1,0)}$};	
			\node[blue] at (-\tx , -\ty - .2) {$-Z_{(1,0)}$};	

			\pgfmathsetmacro{\a}{60}
			\pgfmathsetmacro{\tx}{cos(\a)}
			\pgfmathsetmacro{\ty}{sin(\a)}			
			\foreach \r in {1, .9, .8}		
				\draw[<->, Orchid] (-\tx * \r, -\ty * \r) -- (\tx * \r, \ty * \r);			
			\node[Orchid] at (\tx + .25, \ty + .15) {$Z_{(1,1)}^{\mathbf{3}}$};
			\node[Orchid] at (-\tx, -\ty - .2) {$-Z_{(1,1)}^{\mathbf{3}}$};

			\draw[<->, SkyBlue] (0, -2) -- (0, 2);
			\node[SkyBlue] at (0.2, 2.15) {$Z_{(2,3)}$};
			\node[SkyBlue] at (0.1, -2.2) {$-Z_{(2,3)}$};
			
			\pgfmathsetmacro{\a}{120}
			\pgfmathsetmacro{\tx}{cos(\a)}
			\pgfmathsetmacro{\ty}{sin(\a)}			
			\foreach \r in {1, .9, .8}		
				\draw[<->, Goldenrod] (-\tx * \r, -\ty * \r) -- (\tx * \r, \ty * \r);			
			\node[Goldenrod] at (\tx + .15, \ty + .15) {$Z_{(1,2)}^{\mathbf{3}}$};
			\node[Goldenrod] at (-\tx + .1, -\ty - .2) {$-Z_{(1,2)}^{\mathbf{3}}$};
			
			\pgfmathsetmacro{\a}{150}
			\pgfmathsetmacro{\tx}{2*cos(\a)}
			\pgfmathsetmacro{\ty}{2*sin(\a)}
			\draw[<->, green] (-\tx, -\ty) -- (\tx, \ty);
			\node[green] at (\tx + .15, \ty + .15) {$Z_{(1,3)}$};			
			\node[green] at (-\tx + .15 , -\ty - .2) {$-Z_{(1,3)}$};
			
		\end{scope}
	\end{tikzpicture}
	
	\caption{Central charges of the maximal, symmetric BPS spectrum of the $\gD_4$ class theory.}
	\label{fig:D_4_maximal_symmetric_Z}
\end{figure}

At $s = \infty$, the configuration of the spectral network becomes highly symmetric, as  shown in Figure \ref{fig:snapshots_of_D_4_max_sym_W}. The full family of spectral networks are available at \foothref{http://het-math2.physics.rutgers.edu/loom/plot?data=D_4_AD_from_SO_8_max_BPS_sym}{this web page}. Note that the simpleton spectrum around a branch point is the same as that of the spectral network of 4d $\mathcal{N}=2$ pure $\SO(8)$ theory studied in Section \ref{sec:SO_8}. The symmetric configuration enables us to calculate the central charges of the 12 BPS states analytically, and the result is given in Figure \ref{fig:D_4_maximal_symmetric_Z}, where $Z_{(1,0)} / (-Z_{(0,1)}^\mathbf{3}) = \sqrt{3} \exp(i \pi/6)$.

This structure of BPS spectra over the moduli space is exactly same as that of a $\gD_4$-class theory from a 4d $\mathcal{N}=2$ theory with $\mathfrak{g} = \gA_{N-1}$, which is described in \cite{Maruyoshi:2013fwa}. Therefore we conclude that the study of BPS spectra of a $\gD_4$-theory from a 4d $\mathcal{N}=2$ theory with $\mathfrak{g} = \gD_{4}$ provides evidence for the equivalence of the two AD fixed point theories.

%\cleardoublepage

%%%%%%%%%%%%%%%%%%%%%%%%%%%%%%%%%%
\section{Conclusions and future directions}
%%%%%%%%%%%%%%%%%%%%%%%%%%%%%%%%%%

In this paper we defined spectral networks for 4d $\mathcal{N} = 2$ class $\mathcal{S}$ theories of types $\mathfrak{g} = \gA_n$, $\gD_n$, $\gE_6$, or $\gE_7$, for spectral covers determined by a choice of minuscule representation $\rho$. 
Our construction makes contact with the classical theory of Lie algebras in several aspects.
We found a Lie algebraic description for both the soliton content of spectral networks, and the 2d wall-crossing formula for 2d theories on canonical surface defects $\mathbb{S}_{z, \rho}$.
Then we proposed a definition for the physical charge lattice of the 4d $\mathcal{N} = 2$ class $\mathcal{S}$ theory in terms of the homology lattice of a spectral cover $\Sigma_\rho$ in a minuscule representation. 
We also found a formula for the BPS index in terms of ADE spectral networks data, which allows to compute the BPS spectrum in an algorithmic way. 
Finally, we applied our framework to various examples, both to illustrate the use of spectral networks for studying BPS spectra and to put our construction through nontrivial tests.

Our work raises a number of questions and points to future directions for exploration. 

\begin{itemize}

\item
Our definition of minuscule defects $\IS_{z,\rho}$ and the computation of their soliton spectra are rather formal. It is would be very interesting to check our definitions and results by studying a brane construction of the defects and computing the 2d BPS spectra from first principles.
These and related matters are the subjects of ongoing work \cite{Longhi:2016bte}.

\item 
For technical reasons we found it convenient to restrict to minuscule representations for the spectral covers we study. But work of Donagi suggests the possible existence of a more general kind of \emph{Cameral} surface defect \cite{Donagi:1993}.\footnote{This idea was pointed out to us by Greg Moore and Andy Neitzke, in private communications.}
It would be very interesting to understand a further generalization to ``Cameral spectral networks'', and what kind of physics they would describe. 
An immediate payout would be the ability to deal with the one egregious case that we have left behind by restricting to minuscule representations: theories of $\gE_8$ type.\footnote{It seems plausible that $\gE_8$ would require only a mild generalization of our construction, and not necessarily the study of cameral covers. This is because, while lacking minuscule representations, it does admit a \emph{quasi-minuscule} representation, the adjoint.}
We hope to return to these points in the future.

\item
It would be interesting to extend the construction of spectral networks described in this paper to class $\mathcal{S}$ theories associated with non-simply-laced Lie algebras. 
In \cite{Keller:2011ek}, a class $\mathcal{S}$ construction of pure gauge theories of types $\mathfrak{g} = \mathrm{B}_n$, $\mathrm{C}_n$, $\mathrm{F}_4$, or $\mathrm{G}_2$ was proposed by including outer-automorphism twist lines on the class $\mathcal{S}$ curve $C$. %twists around two irregular punctures. % and the corresponding $\mathbb{Z}_r$-twist line connecting the punctures on the UV curve.
We believe that the formalism developed in the present paper is well-suited for constructing spectral networks in presence of such twist lines, which would allow us to study the BPS spectra of the pure gauge theories. 
On the other hand, a fully general construction of class $\mathcal{S}$ theories of non-simply-laced type is not yet available, and is certainly an interesting problem in its own respect.  If such theories admit a construction of non-simply laced spectral networks, these could also be employed to study chiral rings and solitons of coset models based on $\SO(2n+1)$ or $\mathrm{Sp}(2n)$ \cite{Kazama:1988uz, Hori:2013ewa}, which were studied in \cite{Lerche:1989uy} and do not admit a Landau-Ginzburg description.

\item 
In this paper we focus on the basic definition of spectral networks, while a few other extensions already appeared in the literature. 
One of them is the ``lifting'' construction of \cite{Gaiotto:2012db}, which establishes a relation between certain $\gA_K$ theories and $\gA_1$ theories. It is natural to wonder whether there exists a lifting construction to D- and E-type theories, as well as to other choices of $\rho$.
Another extension is the inclusion of spin proposed in \cite{Galakhov:2014xba} for A-type spectral networks. We believe that it should admit a generalization to ADE-types.

\end{itemize}

%%%%%%%%%%%%%%%%%%%%%%%%%%%%%%%%%%
\section*{Acknowledgements}
%%%%%%%%%%%%%%%%%%%%%%%%%%%%%%%%%%

We thank Greg Moore and Andy Neitzke for sharing their notes and ideas on generalizing spectral networks.
We also thank Greg Moore for collaboration during the early stages of this project, as well as for important comments and suggestions.
We further thank Andy Neitzke for discussions at various stages of this project that led to the development of important ideas including the classification of 2d soliton charges, and also for contributions to \texttt{loom}. 
PL is grateful to Vassil Kanev for helpful correspondence. 
CYP wants to thank Tom Mainiero and Jacques Distler for helpful discussions.
We thank Joe Minahan and Dima Galakhov for comments on the manuscript.
CYP thanks the Mathematics Department at the University of Texas at Austin, the organizers of the workshop `Mathematical Aspects of Six-Dimensional Quantum Field Theories' at the University of California, Berkeley, and the organizers of the 2015 Simons Summer Workshop at the Simons Center for Geometry and Physics for hospitality while this work was in progress.
The work of PL is supported by  the Carl Tryggers Stiftelsen. The work of CYP is supported in part by the National Science Foundation under Grant No.\ NSF PHY-140405.

\appendix

%%%%%%%%%%%%%%%%%%%%%%%%%%%%%%%%%%%
\section{Branch points and roots}\label{app:root-bp}
%%%%%%%%%%%%%%%%%%%%%%%%%%%%%%%%%%

In this Section we derive a key property of the branching structure of spectral curves:
square-root branch points are always labeled by roots of the Lie algebra.

We define a square-root branch-point as a locus where $\langle\alpha,\varphi(z_{0})\rangle=0$ for some $\alpha$ living in a (strictly) rank-1 sub-lattice of $\Lambda_\textrm{root}$. 
Let $\Lambda_{\rho}=\{\weight_{i}\}$ be the weight system of a representation $\rho$. 
At $z_{0}$, all sheets whose corresponding weights are separated by multiples of $\alpha$ collapse together. In ${\mathfrak{t}}^{*}$, the sheets collapsing together are aligned on a straight line parallel to $\alpha$
\be
	\weight_{j}-\weight_{i} \propto \alpha \Rightarrow x_{i},\ x_{j}\ \text{sheets collide}\,.
\ee
There will be many such ``lines'', each one being labeled by its intersection with the hyperplane ${\cal H}_{\alpha}$ that is normal to $\alpha$ and goes through the origin.

Consider {one} of these lines $\ell$, the sheet monodromy around the branch point at $z_{0}$ will permute the sheets corresponding to the weights aligned along $\ell$. 
To characterize the monodromy at a branch point, let $\Lambda_{\ell}=\Lambda_{\rho}\cap \ell$, then $w\in W$ must:
\begin{itemize}
\item act linearly, because W is generated by reflections, which act linearly,
\item preserve the line $\ell$ (in particular, preserve $\Lambda_{\ell}$ as a set)
\item preserve norms
\end{itemize}
To analyze these constraints it's convenient to split ${\mathfrak{t}}^{*}\simeq \alpha\IR\oplus {\cal H}_{\alpha}$. Then, choosing a basis in which $\alpha = (|\alpha|,\vec 0)$ any $\nu_{i}\in\Lambda_{\ell}$ has coordinates $\nu_{i}=(x_{i},\vec x^{\perp} )$ with $\vec x^{\perp}\in {\cal H}_{\alpha}$ being the same for all $\nu_{i}\in\Lambda_{\ell}$.

In order to preserve the line $\ell=\{(s,\vec x^{\perp})\, |\, s\in\IR\}$, a linear operator representing $w$ must be of the form (acting from the left)
\be
\left(\begin{array}{c|ccc}
	a_{1} & a_{2} & \dots & a_{n} \\
	\hline
	0 & 1 & & \\
	\vdots & & \ddots & \\
	0 & & & 1
\end{array}\right)
\ee
The isometric property reads (here $a,x$ are $n$-vectors: $a=(a_{1}\dots a_{n})$ and $x=(x_{||},\vec x_{\perp})$)
\be
	(a\cdot x)^{2} + \vec x_{\perp}^{2} = |w(x)|^{2} = |x|^{2} = x_{||}^{2} + \vec x_{\perp}^{2}
\ee
thus fixing $a_{1}=\pm1$ and $a_{i}=0$ for $i=2,\dots,n$.
The only nontrivial action corresponds to $a_{1}=-1$, so we fix it that way. 

We thus find that $w$ must be an involution, in fact it has to be a reflection about the hyperplane ${\cal H}_{\alpha}$ perpendicular to $\alpha$. Note that this is intimately tied to our definition of square-root branch point, in fact if we had let $\langle\alpha,\varphi\rangle=0$ for any $\alpha$ from a higher-rank sub-lattice of $\Lambda_\text{root}$, then we would not be talking about a line $\ell$, but a higher-dimensional affine space, and the above reasoning wouldn't necessarily lead to an involution.

So $w$ must be a reflection about ${\cal H}_{\alpha}$  for $\alpha\in\Lambda_\text{root}$, and it must also be Weyl. 
It is a well-known result\footnote{See for example \cite[\S 1.14]{Humpreys}} that any reflection in $W$ corresponds to a root, so $\alpha$ (suitably normalized) can indeed be identified with a root.

%%%%%%%%%%%%%%%%%%%%%%%%%%%%%%%%%%%%%%%%
\section{Identifying sheets with weights} \label{app:sheet-weight-identification}
%%%%%%%%%%%%%%%%%%%%%%%%%%%%%%%%%%%%%%%%

In this section we present an algorithm 
for identifying the weights of a representation $\rho$ with 
the sheets $\{x_i\}_i$ of a spectral cover 
$\Sigma_\rho$, after branch cuts have been specified. 

This kind of problem arises immediately in concrete approaches 
to the study of these covers, where one makes a choice of 
branch cuts, and wants to explicitly trivialize the cover.
The input data is given simply by a set of complex numbers $\{x_{i}\}$,
marking the positions of the sheets in the fiber $\IC$ above 
$z$, the task is then to identify consistently each of these
with a weight. 
Once this is accomplished for a single point (the basepoint 
of the trivialization, in particular), the choice of cuts then
makes sure that all points of a sheet corresponds to the same weight.

We will discuss algorithms for accomplishing this task for
all representations of ADE Lie groups 
(with the exception of $\gE_{8}$, although out methods 
could conceivably generalize to that group as well).
A basic fact which holds in general is that the positions
of the sheets are linear functions of the weights.
\be
  x_i = \langle\weight_i,\varphi(z)\rangle
\ee
This entails a great simplification: for each Lie algebra,
it is sufficient to solve the problem for any (nontrivial)
representation, the solution will then extend to all other
representations by linearity.

A bit more concretely, if we solve the problem of identifying
the weights of a certain rep $\rho$ with the sheets of the 
cover $\Sigma_\rho$, we may then pick a basis for $\ft^*$ 
among the $\{\weight_i\}_i$, and expand all the weights of any
other rep $\rho'$ in that basis
\be
  \weight'_j = \sum_i c_{ij} \weight_i
\ee
this data being easily recovered from the knowledge of the 
weight systems themselves.
By linearity, this gives the desired identification of the
sheets of $\Sigma_\rho'$ as
\be
  x'_j = \sum_i c_{ij} x_i
\ee
where $x_i$ are the sheets of $\Sigma_\rho$.

%%%%%%%%%%%%%%%%%%%%%%%%%%%%%%%
\subsection{\texorpdfstring{$\gA_n$}{A\_n}}
%%%%%%%%%%%%%%%%%%%%%%%%%%%%%%%

We choose the first fundamental representation, 
and label the weights $\weight_1\dots \weight_{n+1}$.
The residual gauge freedom after diagonalizing $\varphi(z)$
consists in the action of the Weyl group $W=S_{n+1},$ which permutes
all sheets $x_i\mapsto x_{\sigma(i)}$.
Given this freedom, we can choose to assign any weight to
any sheet, as long as no repetitions are made.

%%%%%%%%%%%%%%%%%%%%%%%%%%%%%%%
\subsection{\texorpdfstring{$\gD_n$}{D\_n}}
%%%%%%%%%%%%%%%%%%%%%%%%%%%%%%%
We choose the vector representation, which is the same as the first fundamental representation for $n > 2$.
This consists of $2n$ weights, subject to the 
relations 
\be
  \weight_i + \weight_{i+n} = 0\qquad (i \in \IZ_{2n})\,.
\ee
The Weyl group is $H_{n-1}\rtimes S_n$, of order $2^{n-1} n!$.
Here $H_{n-1}\subset \IZ_2^n$ is the kernel of the product 
homomorphism 
$\{\epsilon_1,\dots,\epsilon_n\}\to\epsilon_1\dots\epsilon_n$ 
where $\epsilon_i=\pm1$. In other words, $H_{n-1}$ 
is the subgroup of elements with an even number of $-1$'s.

The action of $W$ is as follows: it permutes all the pairs
\be
  \{\weight_i, \weight_{i+n}\}\mapsto\{\weight_{\sigma(i)},\weight_{\sigma(i)+n}\}\qquad \sigma\in S_n
\ee
and independently switches an \emph{even} number of signs 
\be
  \weight_i \mapsto -\weight_i = \weight_{i+n}\,.
\ee

The identification of sheets and weights can be carried out as 
follows: first identify the $n$ pairs of opposite sheets 
such that $x + x' = 0$.
Then, the permutation symmetry $S_n$ allows us to match any 
pair of opposite sheets with any pair of opposite weights.
The problem thus boils down to identifying consistently a 
'positive' and a 'negative' sheet in \emph{each} pair.

The choices are not independent: the $H_{n-1}$ freedom however
allows us to choose freely the positive sheet from the first 
$n-1$ pairs. Only the choice of positive vs negative within 
the last pair 
\be
\begin{array}{c}
 x_n \leftrightarrow \weight_n \\
 x_{2n} \leftrightarrow \weight_{2n}
\end{array}\qquad \text{vs}\quad %
\begin{array}{c}
 x_n \leftrightarrow \weight_{2n} \\
 x_{2n} \leftrightarrow \weight_{n}
\end{array}
\ee
is constrained.
This can be seen as follows: suppose we knew a 'reference' 
weight-sheet identification that works, 
then we can compare it to ours by first permuting the sheet
pairs suitably, then by 'flipping' the positive/negative
role within each sheet pair, for the first $n-1$ pairs.
If we had to perform an even number of switches, then 
we should make the same positive/negative identification
for the last pair as in the reference one; 
if instead we had to perform an odd number of switches,
we should invert the last identification with respect to the 
reference one.

Given however that we don't have a reference identification
to compare with, we cannot determine how to make the 
positive/negative identification in the last pair.
All we can do is evaluate the consequence of a wrong choice.
As it turns out,the wrong choice corresponds to acting with 
the outer automorphism which exchanges the spinor weights
in the $D_n$ Dynkin diagram.
We can therefore choose the identification for the last pair at will, 
the price for such will be that we cannot 
tell whether a given spinor cover $\Sigma_{\rho}$ corresponds 
to one spinor rep, or the other.

%%%%%%%%%%%%%%%%%%%%%%%%%%%%%%%
\subsection{\texorpdfstring{$\gE_6$}{E\_6}}
%%%%%%%%%%%%%%%%%%%%%%%%%%%%%%%

The Cartan Matrix is
\be
\left(\begin{array}{cccccc}
2 &  0 &  -1 &   0 &   0 &   0 \\
0 &   2 &   0 &  -1 &   0 &   0 \\
-1 &   0 &   2 &  -1 &   0 &   0 \\
0 &  -1 &  -1 &   2 &  -1 &   0 \\
0 &   0 &   0 &  -1 &   2 &  -1\\
0 &   0 &   0 &   0 &  -1 &   2 
\end{array}\right)
\ee

Thus the Dynkin diagram and the corresponding simple roots read\\
\begin{center}
\includegraphics[width=0.3\textwidth]{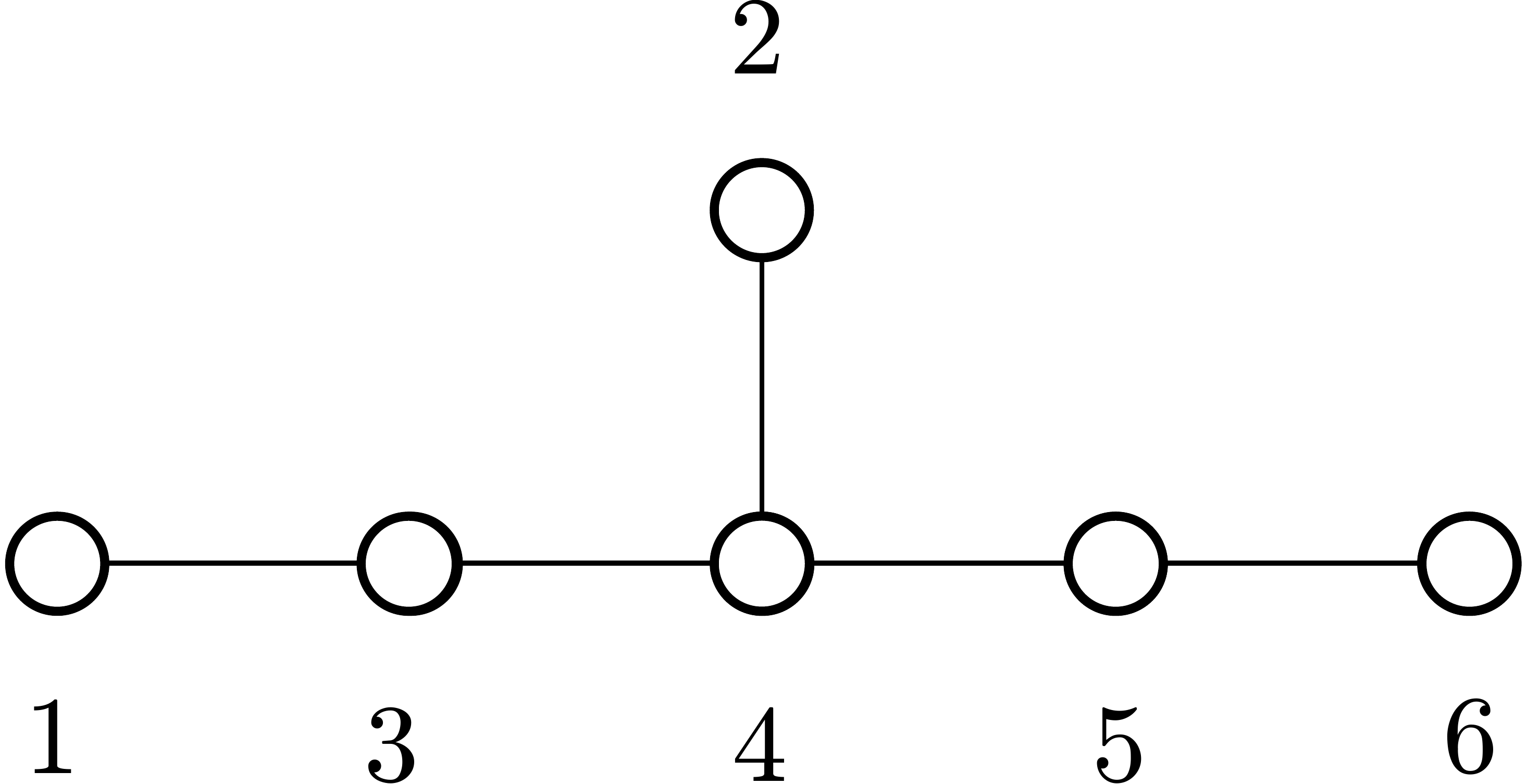}
\end{center}

The weights of the representation we study read, explicitly
{\footnotesize
\be\label{eq:table_E6_weights}
\begin{array}{c|l}
0 & (0, 0, 0, 0, 0, -2/3, -2/3, 2/3)\\
1 & (-1/2, 1/2, 1/2, 1/2, 1/2, -1/6, -1/6, 1/6)\\
 2 & (1/2, -1/2, 1/2, 1/2, 1/2, -1/6, -1/6, 1/6)\\
 3 & (1/2, 1/2, -1/2, 1/2, 1/2, -1/6, -1/6, 1/6)\\
 4 & (-1/2, -1/2, -1/2, 1/2, 1/2, -1/6, -1/6, 1/6)\\
 5 & (1/2, 1/2, 1/2, -1/2, 1/2, -1/6, -1/6, 1/6)\\
 6 & (-1/2, -1/2, 1/2, -1/2, 1/2, -1/6, -1/6, 1/6)\\
 7 & (1/2, 1/2, 1/2, 1/2, -1/2, -1/6, -1/6, 1/6)\\
 8 & (-1/2, 1/2, -1/2, -1/2, 1/2, -1/6, -1/6, 1/6)\\
 9 & (-1/2, -1/2, 1/2, 1/2, -1/2, -1/6, -1/6, 1/6)\\
 10 & (1/2, -1/2, -1/2, -1/2, 1/2, -1/6, -1/6, 1/6)\\
 11 & (-1/2, 1/2, -1/2, 1/2, -1/2, -1/6, -1/6, 1/6)\\
 12 & (0, 0, 0, 0, 1, 1/3, 1/3, -1/3)\\
 13 & (1/2, -1/2, -1/2, 1/2, -1/2, -1/6, -1/6, 1/6)\\
 14 & (-1/2, 1/2, 1/2, -1/2, -1/2, -1/6, -1/6, 1/6)\\
 15 & (0, 0, 0, 1, 0, 1/3, 1/3, -1/3)\\
 16 & (1/2, -1/2, 1/2, -1/2, -1/2, -1/6, -1/6, 1/6)\\
 17 & (0, 0, 1, 0, 0, 1/3, 1/3, -1/3)\\
 18 & (1/2, 1/2, -1/2, -1/2, -1/2, -1/6, -1/6, 1/6)\\
 19 & (0, 1, 0, 0, 0, 1/3, 1/3, -1/3)\\
 20 & (-1/2, -1/2, -1/2, -1/2, -1/2, -1/6, -1/6, 1/6)\\
 21 & (-1, 0, 0, 0, 0, 1/3, 1/3, -1/3)\\
 22 & (1, 0, 0, 0, 0, 1/3, 1/3, -1/3)\\
 23 & (0, -1, 0, 0, 0, 1/3, 1/3, -1/3)\\
 24 & (0, 0, -1, 0, 0, 1/3, 1/3, -1/3)\\
 25 & (0, 0, 0, -1, 0, 1/3, 1/3, -1/3)\\
 26 & (0, 0, 0, 0, -1, 1/3, 1/3, -1/3)
\end{array}
\ee
}

We choose to work with the $\rho_1$ representation, 
with Dynkin indices $(1,0,0,0,0,0)$, of dimension $27$.
The Weyl group is of order $51840\ll 27!$.
The precise goal will be to ``label'' all the sheets 
$x_{0}\dots x_{26}$ so that $x_{i}=\langle \weight_{i},\varphi(z)\rangle$.

As a preliminary, we state the following empirical
observation.
The $27$ weights of $\rho_1$ can be arranged into 
\emph{null triples}, obeying
\be
  \weight_i + \weight_j + \weight_k = 0
\ee
There are $45$ such triples (up to $S_3^{\times 45}$ 
permutations within each triple), 
each weight appears in exactly $5$ of them.
Similarly, by linearity, the sheets must also organize into 
$45$ null triples, and each sheet will feature in exactly 
$5$ of them.

We can now start identifying sheets and weights.
Choose any sheet, then by using part of the Weyl freedom 
(to be quantified more precisely below) we can label it by 
$x_{0}$, thus identifying $\weight_0\leftrightarrow x_0$.
Then consider the quintet $Q_0$ of null triples to which $\weight_0$ 
belongs, they are:
\be
  \begin{array}{c|c|c|c|c}
   t_1^{(0)} & t_2^{(0)} & t_3^{(0)} & t_4^{(0)} & t_5^{(0)} \\
   \hline
   (0, 12, 26) & (0, 15, 26) & (0, 17, 24) & (0, 19, 23) & (0, 21, 22)
  \end{array}
\ee
In notation where $(i,j,k)$ stands for $(\weight_i,\weight_j,\weight_k)$; 
here we stick to weight-labels as they are given by the standard 
ordering employed by \texttt{SAGE}. We provide explicit expressions in (\ref{eq:table_E6_weights}).

Our choice to identify $x_0,\weight_{0}$ cost us part of the Weyl
freedom, but we still have the remaining freedom given by the 
stabilizer subgroup 
\be
  W_0 := \{w \in W \,|\, w(\weight_0) = \weight_0\}
\ee
Now $W_{0}\simeq W(D_{5})$, as follows from the fact that $\weight_{0}$ can be taken to be the first 
fundamental weight $\weight_{0}=\omega_{1}$, and that the stabilizer 
of a dominant weight $\weight = \sum_{i}m_{i}\omega_{i}$ in $W$ is 
generated by the simple reflections for which $m_{i}=0$
\cite[p.324]{Procesi}.

The order of $W_0$ is $1920$. 
The action of $W_0$ on the 
\emph{quintet of ordered triples} $(t_1^{(0)},\dots, t_5^{(0)})$
gives $1920$ distinct quintets of ordered triples.
For clarity, two quintets 
\be
\begin{split}
  (t_1^{(0)},t_2^{(0)},t_3^{(0)},t_4^{(0)},t_5^{(0)})\\
  (t_2^{(0)},t_1^{(0)},t_3^{(0)},t_4^{(0)},t_5^{(0)})
\end{split}
\ee
are considered different. Moreover, also two quintets
differing by a triplet changing from $(0, 12, 26)$ to 
$(0, 26, 12)$ are considered different in our count.
Note that the index $0$ is left invariant by $W_0$, 
by definition.

In fact $1920 = 5! \cdot 2^4$ is the number of such ordered
quintets obtained by considering all their permutations by 
$S_5$, and by flipping an \emph{even} number of pairs 
within each triple (in accordance with the fact that 
$W_0=H_4 \rtimes S_5\simeq W(D_{5})$).

This freedom means the following. Consider the null triples of 
\emph{sheets} in which $x_0$ features. This singles out $5$ 
pairs of sheets $(x_i, x_j)$ such that $x_i+x_j=-x_0$.
We choose to label the sheets of those five triples as follows:
\be
    \tilde Q_0 := \ \big\{(x_0, x_{12}, x_{26}) \ (x_0, x_{15}, x_{26}) \ (x_0, x_{17}, x_{24}) \ (x_0, x_{19}, x_{23}) \ (x_0, x_{21}, x_{22})\big\}
\ee
where the identification of $x_0$ is fixed by our initial choice, 
and the $W_0$ gauge freedom accounts for possible permutations
among the five pairs of other sheets, as well as for an 
\emph{even} number of 'flips' of the pairs of sheets 
$x_i\leftrightarrow x_j$ within each triple $(x_0, x_i, x_j)$.
More precisely, we have the freedom to choose who is 
$x_i$ vs $x_j$ (where we associate $x_{i,j}\to \weight_{i,j}$), 
in 4 out of 5 pairs, but the last choice is constrained by the
first four ones (the reasoning is the same as for $D_n$,
and likewise a wrong choice on the last pair corresponds 
to the $\IZ_2$ outer automorphism of the Lie algebra).

Having identified the first $11$ sheets, the subgroup of $W$ 
which stabilizes all of our choices is precisely of order $1$. 
Thus, the Weyl freedom has been used exhaustively at this point.

How to identify the remaining 16 sheets?
We have identified the $11$ sheets 
\be
  x_0, x_{12}, x_{26}, x_{15}, x_{25}, x_{17}, x_{24}, x_{19}, x_{23}, x_{21}, x_{22}
\ee
with the corresponding weights ($x_i\leftrightarrow \weight_i$).\footnote{More precisely, we have done so up to a $\IZ_2$ ambiguity. One cannot tell at this point between $x_0, x_{12}, x_{26}, x_{15}, x_{25}, x_{17}, x_{24}, x_{19}, x_{23}, x_{21}, x_{22}$ and $x_0, x_{12}, x_{26}, x_{15}, x_{25}, x_{17}, x_{24}, x_{19}, x_{23}, x_{22}, x_{21}$, because the `flip' of a single pair (or an odd nomber of them) is not accounted for by the Weyl symmetry. This means that the following procedure must be actually carried out twice, once for each possibility. Only one of the two cases will yield a successfull identification of sheets and weights. In the other case, one won't be able to match all the sheets with the weights.}
We can therefore construct their quintets:
\be
  \tilde Q_0, \tilde Q_{12}, \tilde Q_{26}, \tilde Q_{15}, \tilde Q_{25}, \tilde Q_{17}, \tilde Q_{24}, \tilde Q_{19}, \tilde Q_{23}, \tilde Q_{21}, \tilde Q_{22}
\ee
Again, each quintet will contain $5$ triples of the form 
$(x_i, x_j, x_k)$ with $x_i$ being one of the sheets we 
identified, and $x_j, x_k$ being generally sheets that 
we have yet to identify with a weight.

Now, from the weight data, we know that each one of the missing 
weights has a unique pattern of whether it belongs or not
to each of these quintets.
The same must be true for the weights, by linearity,
and therefore the problem of identifying the remaining weights 
and sheets is completely solved. We give below the data table: 
each row tells whether one of the missing weights belongs or not
to the quintets which label the columns. A '$0$' stands for 
'not contained', while a '$1$' stands for 'contained'.
The same analysis can be carried out for the remaining sheets,
since we know their coordinates, and we know the sheet quintets
corresponding to the weight quintets $Q_i\leftrightarrow\tilde Q_i$. 
The requirement that the pattern in table of $(x_i, \tilde Q_j)$ 
matches with the pattern of the table $( \weight_i, Q_j)$
uniquely fixes all the remaining $x_i$'s.
\be
\begin{array}{c|ccccccccccc}
 & Q_0 & Q_{12} & Q_{26}& Q_{15}& Q_{25}& Q_{17}& Q_{24}& Q_{19}& Q_{23}& Q_{21}& Q_{22}\\
 \hline
  \weight_1 & 0 & 0 & 1 & 0 & 1 & 0 & 1 & 0 & 1 & 0 & 1 \\
 \weight_2 & 0 & 0 & 1 & 0 & 1 & 0 & 1 & 1 & 0 & 1 & 0 \\
 \weight_3 & 0 & 0 & 1 & 0 & 1 & 1 & 0 & 0 & 1 & 1 & 0 \\
 \weight_4 & 0 & 0 & 1 & 0 & 1 & 1 & 0 & 1 & 0 & 0 & 1 \\
 \weight_5 & 0 & 0 & 1 & 1 & 0 & 0 & 1 & 0 & 1 & 1 & 0 \\
 \weight_6 & 0 & 0 & 1 & 1 & 0 & 0 & 1 & 1 & 0 & 0 & 1 \\
 \weight_7 & 0 & 1 & 0 & 0 & 1 & 0 & 1 & 0 & 1 & 1 & 0 \\
 \weight_8 & 0 & 0 & 1 & 1 & 0 & 1 & 0 & 0 & 1 & 0 & 1 \\
 \weight_9 & 0 & 1 & 0 & 0 & 1 & 0 & 1 & 1 & 0 & 0 & 1 \\
 \weight_{10} & 0 & 0 & 1 & 1 & 0 & 1 & 0 & 1 & 0 & 1 & 0 \\
 \weight_{11} & 0 & 1 & 0 & 0 & 1 & 1 & 0 & 0 & 1 & 0 & 1 \\
 \weight_{13} & 0 & 1 & 0 & 0 & 1 & 1 & 0 & 1 & 0 & 1 & 0 \\
 \weight_{14} & 0 & 1 & 0 & 1 & 0 & 0 & 1 & 0 & 1 & 0 & 1 \\
 \weight_{16} & 0 & 1 & 0 & 1 & 0 & 0 & 1 & 1 & 0 & 1 & 0 \\
 \weight_{18} & 0 & 1 & 0 & 1 & 0 & 1 & 0 & 0 & 1 & 1 & 0 \\
 \weight_{20} & 0 & 1 & 0 & 1 & 0 & 1 & 0 & 1 & 0 & 0 & 1
\end{array}
\ee

%%%%%%%%%%%%%%%%%%%%%%%%%%%%%%%
\subsection{\texorpdfstring{$\gE_7$}{E\_7}}
%%%%%%%%%%%%%%%%%%%%%%%%%%%%%%%

The Cartan Matrix is
\be
\left(\begin{array}{ccccccc}
2 &  0 &  -1 &   0 &   0 &   0 &   0\\
0 &   2 &   0 &  -1 &   0 &   0 &   0\\
-1 &   0 &   2 &  -1 &   0 &   0 &   0\\
0 &  -1 &  -1 &   2 &  -1 &   0 &   0\\
0 &   0 &   0 &  -1 &   2 &  -1&   0\\
0 &   0 &   0 &   0 &  -1 &   2 &   -1 \\
0 & 0 &   0 &   0 &   0 &  -1 &   2 
\end{array}\right)
\ee

Thus the Dynkin diagram and the corresponding simple roots read\\
\begin{center}
\includegraphics[width=0.35\textwidth]{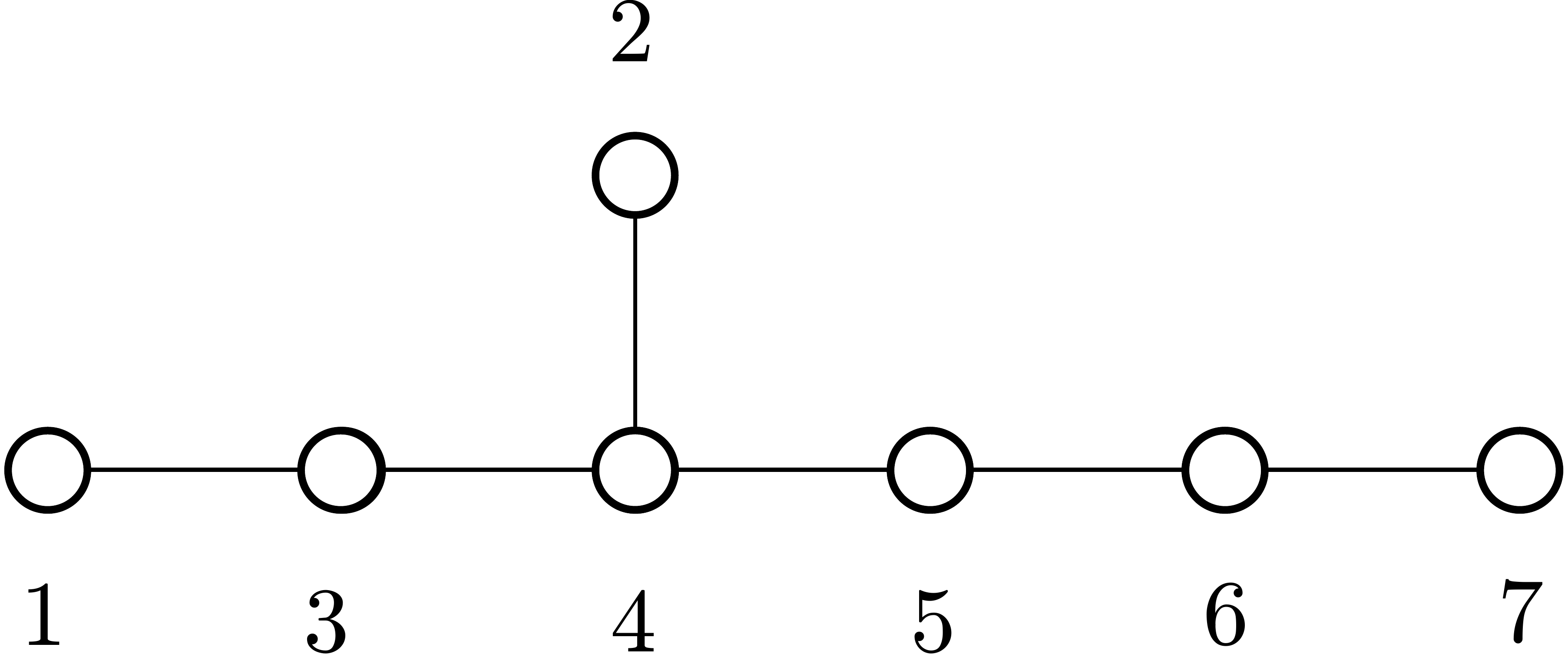}
\end{center}

The weights of the representation we study read, explicitly
{\footnotesize
\be
\begin{array}{c|l}
0 & (-1/2, -1/2, 1/2, 1/2, -1/2, -1/2, 0, 0)\\
1 & (-1/2, 1/2, 1/2, 1/2, -1/2, 1/2, 0, 0)\\
2 & (0, 0, 0, 0, 0, 1, 1/2, -1/2)\\
3 & (-1/2, 1/2, 1/2, -1/2, 1/2, 1/2, 0, 0)\\
4 & (0, 0, 0, 0, -1, 0, -1/2, 1/2)\\
5 &  (-1/2, 1/2, -1/2, 1/2, -1/2, -1/2, 0, 0)\\
6 &  (0, 0, 1, 0, 0, 0, 1/2, -1/2)\\
7 &  (-1/2, -1/2, -1/2, -1/2, 1/2, 1/2, 0, 0)\\
8 &  (1/2, 1/2, 1/2, 1/2, -1/2, -1/2, 0, 0)\\
9 &  (1/2, -1/2, 1/2, -1/2, -1/2, -1/2, 0, 0)\\
10 &  (0, 0, 0, 1, 0, 0, 1/2, -1/2)\\
11 &  (0, 0, 0, 0, 0, -1, -1/2, 1/2)\\
12 &  (-1/2, 1/2, 1/2, 1/2, 1/2, -1/2, 0, 0)\\
13 &  (1/2, -1/2, -1/2, -1/2, -1/2, 1/2, 0, 0)\\
14 &  (0, 0, 1, 0, 0, 0, -1/2, 1/2)\\
15 &  (1/2, -1/2, 1/2, -1/2, 1/2, 1/2, 0, 0)\\
16 &  (-1/2, -1/2, -1/2, -1/2, -1/2, -1/2, 0, 0)\\
17 &  (-1, 0, 0, 0, 0, 0, 1/2, -1/2)\\
18 &  (0, 0, -1, 0, 0, 0, -1/2, 1/2)\\
19 &  (1, 0, 0, 0, 0, 0, 1/2, -1/2)\\
20 &  (0, 1, 0, 0, 0, 0, -1/2, 1/2)\\
21 &  (1/2, 1/2, -1/2, 1/2, -1/2, 1/2, 0, 0)\\
22 &  (0, 0, 0, -1, 0, 0, 1/2, -1/2)\\
23 &  (1/2, -1/2, -1/2, 1/2, -1/2, -1/2, 0, 0)\\
24 &  (-1/2, -1/2, -1/2, 1/2, 1/2, -1/2, 0, 0)\\
25 &  (-1/2, 1/2, -1/2, -1/2, 1/2, -1/2, 0, 0)\\
26 & (0, -1, 0, 0, 0, 0, 1/2, -1/2)\\
27 &  (1/2, 1/2, 1/2, -1/2, -1/2, 1/2, 0, 0)\\
\end{array}
\begin{array}{c|l}
28 & (1/2, -1/2, 1/2, 1/2, -1/2, 1/2, 0, 0) \\
29 &  (0, 0, -1, 0, 0, 0, 1/2, -1/2)\\
30 &  (1, 0, 0, 0, 0, 0, -1/2, 1/2)\\
31 &  (-1/2, -1/2, -1/2, 1/2, -1/2, 1/2, 0, 0)\\
32 &  (1/2, -1/2, -1/2, 1/2, 1/2, 1/2, 0, 0)\\
33 &  (-1/2, -1/2, 1/2, -1/2, -1/2, 1/2, 0, 0)\\
34 &  (1/2, -1/2, -1/2, -1/2, 1/2, -1/2, 0, 0)\\
35 &  (-1/2, -1/2, 1/2, 1/2, 1/2, 1/2, 0, 0)\\
36 &  (0, -1, 0, 0, 0, 0, -1/2, 1/2)\\
37 &  (0, 0, 0, 1, 0, 0, -1/2, 1/2)\\
38 &  (0, 0, 0, 0, 0, -1, 1/2, -1/2)\\
39 &  (-1/2, -1/2, 1/2, -1/2, 1/2, -1/2, 0, 0)\\
40 &  (1/2, 1/2, -1/2, -1/2, -1/2, -1/2, 0, 0)\\
41 &  (-1, 0, 0, 0, 0, 0, -1/2, 1/2)\\
42 &  (-1/2, 1/2, -1/2, -1/2, -1/2, 1/2, 0, 0)\\
43 &  (0, 1, 0, 0, 0, 0, 1/2, -1/2)\\
44 &  (0, 0, 0, 0, -1, 0, 1/2, -1/2)\\
45 &  (1/2, -1/2, 1/2, 1/2, 1/2, -1/2, 0, 0)\\
46 &  (1/2, 1/2, -1/2, -1/2, 1/2, 1/2, 0, 0)\\
47 &  (-1/2, 1/2, 1/2, -1/2, -1/2, -1/2, 0, 0)\\
48 &  (1/2, 1/2, -1/2, 1/2, 1/2, -1/2, 0, 0)\\
49 &  (-1/2, 1/2, -1/2, 1/2, 1/2, 1/2, 0, 0)\\
50 &  (0, 0, 0, 0, 1, 0, -1/2, 1/2)\\
51 &  (1/2, 1/2, 1/2, 1/2, 1/2, 1/2, 0, 0)\\
52 &  (0, 0, 0, 0, 0, 1, -1/2, 1/2)\\
53 &  (0, 0, 0, 0, 1, 0, 1/2, -1/2)\\
54 &  (1/2, 1/2, 1/2, -1/2, 1/2, -1/2, 0, 0)\\
55 &  (0, 0, 0, -1, 0, 0, -1/2, 1/2)
\end{array}
\ee
}

We choose to work with the $\rho_7$ representation, 
with Dynkin indices $(0,0,0,0,0,0,1)$, of dimension $56$.
The Weyl group is of order $2903040\ll 56!$.

The $56$ weights of $\rho_7$ can be arranged into 
$28$ \emph{null pairs}, obeying
\be
   \weight_i +  \weight_j = 0
\ee
The $56$ weights of $\rho_7$ can be also arranged into 
\emph{null quartets}, obeying
\be
   \weight_i +  \weight_j +  \weight_k +  \weight_l = 0
\ee
There are $1008$ such quartets, 
each weight appears in exactly $72$ of them.
But, if we exclude the $378$ quartets obtained
from combining null pairs, we are left with $630$
genuine null quartets, and each weight
appears in $45$ of them.

Similarly, by linearity, the sheets must also organize into $28$ 
null pairs, as well as $630$ ``genuine'' null quartets, 
and each sheet will feature in exactly $45$ of them.

Choose a labeling for the sheets $x_0, \dots, x_{55}$, and 
make an ansatz by identifying $ \weight_0\leftrightarrow x_0$.
The Weyl group is
\be
	W(E_{7}) = \IZ_{2}\times PSp_{6}(2)
\ee
where the $\IZ_{2}$ takes
$ \weight\to- \weight\,,\ \forall \weight\in\ft^{*}$.

We consider the $45$-plet $Q_0$ of null quartets to which $ \weight_0$ 
belongs, they are:
\be
  Q_{0}\ :\quad %
  \begin{array}{c|c|c}
   q_1^{(0)} & \dots & q_{45}^{(0)} \\
   \hline
   (0,2,18,54) & \dots & (0,40,52,53)
  \end{array}
\ee
In notation where $(i,j,k,l)$ stands for $( \weight_i, \weight_j, \weight_k, \weight_l)$.
It turns out that all the weights appearing in $Q_{0}$ come from all the $28$ distinct null pairs.
To be explicit, these 28 weights are:
\be
\begin{split}
\CW_{0}= \{ & 0, 2, 18, 54, 20, 34, 25, 30, 40, 50, 48, 55, 3, 13,  \\
	& 19, 21, 29, 32, 7, 27, 43, 51, 53, 49, 15, 42, 22, 52\}
\end{split}
\ee
On the side of sheets, having fixed the correspondence $ \weight_{0}\leftrightarrow x_{0}$, we will obtain by the same procedure an unordered set of $28-1 = 27$ sheets, which we call $\CS_{0}$. The next task is to understand how to identify each sheet in $\CS_{0}$ with a weight from $\CW_{0}-\{0\}$.

From the $45$ quartets in $Q_{0}$ we can extract $45$ triplets, in the obvious way: these will be $(i,j,k)$ such that $ \weight_{i}+ \weight_{j}+ \weight_{k} = - \weight_{0}$.
Now, taken any $i\in\CW_{0}$, with $i\neq 0$, it turns out that it appears in \emph{exactly} $5$ triplets.
The same story, by linearity, must be true of the sheets $x_{i}$: the sheets in $\CS_{0}$ must arrange in triples $(x_{i}, x_{j}, x_{k})$ such that $x_{i}+x_{j}+x_{k}=-x_{0}$; moreover there will be 45 triples, and each $x_{i}$ will appear in precisely 5 of them.

The similarity with the previous section is now evident: in fact, the stabilizer of $ \weight_0$ is $W_0\simeq W(E_{6})$, following a reasoning analogous to that employed in the previous section (where the stabilizer was found to be $W(D_{5})$).

So we know from the $E_{6}$ case how to proceed now: the residual Weyl symmetry can be used to uniquely fix all the $27$ weights/sheet pairs in $\CW_{0}-\{0\}$, precisely following the algorithm devised in the previous section, i.e.\ choose the subgroup of $W_{0}$ which stabilizes, say, $ \weight_{2}$, it will be of order $1920$, and so on. 
The remaining $28$ sheets are related to these by the $ \weight\to- \weight$ $\IZ_{2}$ symmetry, the same relation extends by linearity to the sheets $x\to -x$. This completely fixes the sheet/weight identification.

%%%%%%%%%%%%%%%%%%%%%%%%%%%%%%%%%%%%%%%%
\section{Simpletons from twisted homotopy invariance} \label{app:simpletons}
%%%%%%%%%%%%%%%%%%%%%%%%%%%%%%%%%%%%%%%%

In this section we repeat the analysis of the soliton content of primary $\CS$-walls, this time we carry it out in the language of parallel transport and homotopy invariance. 
Consider $\wp,\wp'$ as in Figure \ref{fig:branch-point-homotopy},

\begin{figure}[!ht]
\begin{center}
\includegraphics[width=0.45\textwidth]{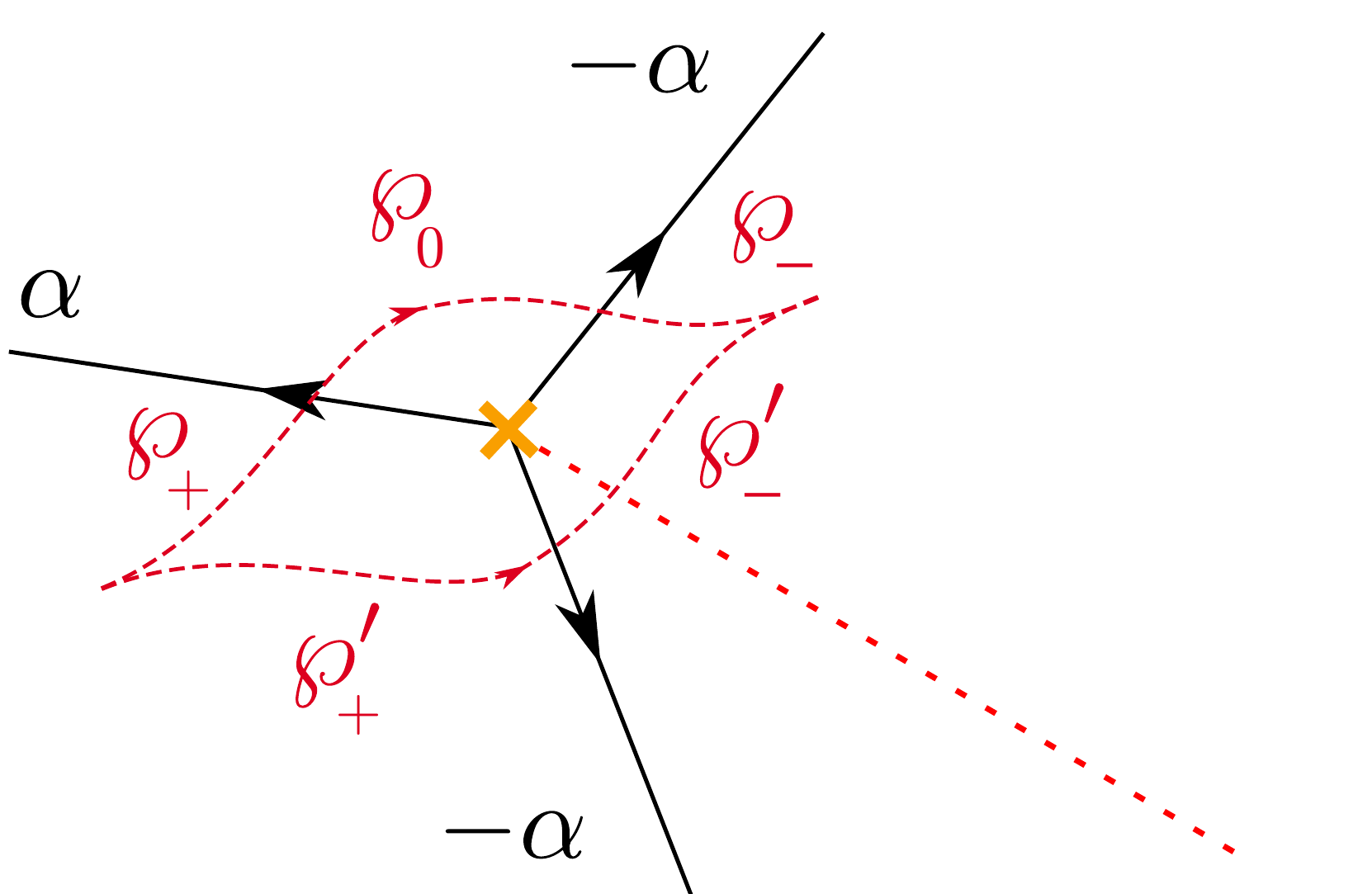}
\caption{Two paths $\wp,\wp'$ in the same relative homotopy class, across a branch point.}
\label{fig:branch-point-homotopy}
\end{center}
\end{figure}

\be
\begin{split}
	F(\wp) & = %
	D(\wp_{+})\,%
	\left(1+\sum_{(ij)\in\CP_{\alpha}}\sum_{a\in\Gamma_{ij}}\mu(a)X_{a}\right)\,%
	D(\wp_{0})\,%
	\left(1+\sum_{(i'j')\in\CP_{-\alpha}}\sum_{b\in\Gamma_{i'j'}}\mu(b)X_{b}\right)\,%
	D(\wp_{-})\\
	& = D(\wp) %
	+ \sum_{(ij)\in\CP_{\alpha}}\sum_{a\in\Gamma_{ij}}\mu(a)X_{\wp_{+}^{(i)}a\wp_{0}^{(j)}\wp_{-}^{(j)}} %
	+ \sum_{(i'j')\in\CP_{-\alpha}}\sum_{b\in\Gamma_{i'j'}}\mu(b)X_{\wp_{+}^{(i')}\wp_{0}^{(i')}b\wp_{-}^{(j')}} \\%
	& + \sum_{(ij)\in\CP_{\alpha}}\sum_{a\in\Gamma_{ij}} \sum_{(i'j')\in\CP_{-\alpha}}\sum_{b\in\Gamma_{i'j'}} \mu(a)\mu(b)  X_{\wp_{+}^{(i)}a\wp_{0}^{(j)} b \wp_{-}^{(j')}} \delta_{ji'}
\end{split}
\ee
in the last line the $\delta_{ji'}$ is enforced by the $X_{a}$ algebra, and together with our considerations about minuscule representations, it also entails that $j'\equiv i$. So the last piece simplifies to
\be
\begin{split}
	& \sum_{(ij)\in\CP_{\alpha}}\sum_{a\in\Gamma_{ij}} \sum_{b\in\Gamma_{ji}} \mu(a)\mu(b)  X_{\wp_{+}^{(i)}a\wp_{0}^{(j)} b \wp_{-}^{(i)}}
\end{split}
\ee
On the other hand
\be
\begin{split}
	F(\wp') & = %
	D(\wp'_{+})\,%
	\left(1+\sum_{(i''j'')\in\CP_{-\alpha}}\sum_{c\in\Gamma_{i''j''}}\mu(c)X_{c}\right)\,%
	\widetilde D(\wp'_{-}) \\
	& = 	D(\wp'_{+})\,%
	\left(1+\sum_{(i''j'')\in\CP_{-\alpha}}\sum_{c\in\Gamma_{i''j''}}\mu(c)X_{c}\right)\,%
	\left(\sum_{i=1}^{d} X_{{\wp'_{-}}^{(i\bar i)}} \right) \\
	& = 	\sum_{i=1}^{d} X_{{\wp'}^{(i\bar i)}} \,%
	+ \sum_{(i''j'')\in\CP_{-\alpha}}\sum_{c\in\Gamma_{i''j''}}\mu(c)X_{{\wp'_{+}}^{(i'')}c{\wp'_{-}}^{(j'',\overline{j''})}}
\end{split}
\ee
but notice that, because of the peculiar structure of $\CP_{-\alpha}$ for minuscule representations, we must have $i''\equiv \overline{j''}$. So we rewrite the last line as
\be
	F(\wp') =\sum_{i=1}^{d} X_{{\wp'}^{(i\bar i)}} \,%
	+ \sum_{(i''j'')\in\CP_{-\alpha}}\sum_{c\in\Gamma_{i''j''}}\mu(c)X_{{\wp'_{+}}^{(i'')}c{\wp'_{-}}^{(j''i'')}}
\ee
We now turn to examine the equation $F(\wp)=F(\wp')$. 

\medskip

\noindent{\bf\underline{Off-diagonal piece}}
We focus on two distinct pieces: first $(ij)$ paths with $i\in\CP_{\alpha}^{-}$:
\be
	\sum_{(ij)\in\CP_{\alpha}}\sum_{a\in\Gamma_{ij}}\mu(a)X_{\wp_{+}^{(i)}a\wp_{0}^{(j)}\wp_{-}^{(j)}} %
	= \sum_{i\in\CP_{\alpha}^{-}} X_{{\wp'}^{(i\bar i)}} 
\ee
again, for minuscule $\rho$ we precisely must have $j=\bar i$ on the LHS of the above equation. So we can recast it as
\be
	\sum_{(i\bar i)\in\CP_{\alpha}}\sum_{a\in\Gamma_{i\bar i}}\mu(a)X_{\wp_{+}^{(i)}a\wp_{0}^{(\bar i)}\wp_{-}^{(\bar i)}} %
	= \sum_{i\in\CP_{\alpha}^{-}} X_{{\wp'}^{(i\bar i)}} 
\ee
moreover, and this is very important, for minuscule reps $\CP_{\alpha}$ is 1-1 with $\CP_{\alpha}^{\pm}$, so we can further turn the above into
\be
	\sum_{(i)\in\CP_{\alpha}^{-}}\sum_{a\in\Gamma_{i\bar i}}\mu(a)X_{\wp_{+}^{(i)}a\wp_{0}^{(\bar i)}\wp_{-}^{(\bar i)}} %
	= \sum_{i\in\CP_{\alpha}^{-}} X_{{\wp'}^{(i\bar i)}} 
\ee
It is now clear that $\mu(a)=1$ for each $i\bar i$ simpleton, and zero for everybody else.

Analogous considerations fix $\mu(b)$ in the same way.\\

\medskip

\noindent{\bf\underline{Diagonal piece}}
This reads
\be
\begin{split}
	& D(\wp) %
	+ \sum_{(ij)\in\CP_{\alpha}}\sum_{a\in\Gamma_{ij}} \sum_{b\in\Gamma_{ji}} \mu(a)\mu(b)  X_{\wp_{+}^{(i)}a\wp_{0}^{(j)} b \wp_{-}^{(i)}} \\
	&\qquad \qquad \qquad = \sum_{i | \weight_{i}\in\Lambda_{\rho}^{(0)}} X_{{\wp'}^{(i\bar i)}} \,%
	+ \sum_{(i''j'')\in\CP_{-\alpha}}\sum_{c\in\Gamma_{i''j''}}\mu(c)X_{{\wp'_{+}}^{(i'')}c{\wp'_{-}}^{(j''i'')}}
\end{split}
\ee
where it is understood that if $\weight_{i}\in\Lambda_{\rho}^{(0)}$ then $i=\bar i$. Let us in fact split the above equation: first consider the piece where $i$ is such that $\weight_{i}\in\Lambda_{\rho}^{(0)}$. Then equality is trivially satisfied. So we just need to focus on the other piece, namely when $i\in \CP_{\alpha}^{\pm}$. This reads
\be
\begin{split}
	& \sum_{i\in\CP_{\alpha}^{+}\sqcup\CP_{\alpha}^{-}} X_{\wp^{(i)}}%
	+ \sum_{(ij)\in\CP_{\alpha}}\sum_{a\in\Gamma_{ij}} \sum_{b\in\Gamma_{ji}} \mu(a)\mu(b)  X_{\wp_{+}^{(i)}a\wp_{0}^{(j)} b \wp_{-}^{(i)}} \\
	& \qquad \qquad \qquad= \sum_{(i''j'')\in\CP_{-\alpha}}\sum_{c\in\Gamma_{i''j''}}\mu(c)X_{{\wp'_{+}}^{(i'')}c{\wp'_{-}}^{(j''i'')}}
\end{split}
\ee
now note that for each $i\in\CP_{\alpha}^{-}$ we have on the LHS both a term $X_{\wp^{(i)}}$ and a term 
\be
\sum_{a\in\Gamma_{ij}} \sum_{b\in\Gamma_{ji}} \mu(a)\mu(b)  X_{\wp_{+}^{(i)}a\wp_{0}^{(j)} b \wp_{-}^{(i)}}\,.
\ee 
Since we already fixed the simpleton degeneracies $\mu(a),\,\mu(b)$  the sum of these really reads
\be
	X_{\wp^{(i)}} + X_{\wp_{+}^{(i)}a\wp_{0}^{(j)} b \wp_{-}^{(i)}} = 0
\ee
where the cancellation is afforded by the relative minus sign coming from the extra winding of the concatenation of simpletons around the circle fiber (cf. Figure 21 of \cite{Gaiotto:2012rg}, the two diagrams in the third column).
So we are left with 
\be
\begin{split}
	& \sum_{i\in\CP_{\alpha}^{+}} X_{\wp^{(i)}}%
	= \sum_{(i''j'')\in\CP_{-\alpha}}\sum_{c\in\Gamma_{i''j''}}\mu(c)X_{{\wp'_{+}}^{(i'')}c{\wp'_{-}}^{(j''i'')}}
\end{split}	
\ee
this is consistent, since $i''$ must clearly belong to $\CP_{\alpha}^{+}\equiv\CP_{-\alpha}^{-}$. So this last requirement fixes the simpleton degeneracies for the last wall: $\mu(c)=1$ for the simpleton in $\Gamma_{i'' j''}$, for every $(i''j'')\in\CP_{-\alpha}$, and zero for every other class.

%%%%%%%%%%%%%%%%%%%%%%%%%%%%%%%%%%%%%%%%%%%%%%%%%%%%%%%%%%%%%%%%%%%%%%%%%%%%%%%%
\section{Joint polygons of minuscule representations}
%%%%%%%%%%%%%%%%%%%%%%%%%%%%%%%%%%%%%%%%%%%%%%%%%%%%%%%%%%%%%%%%%%%%%%%%%%%%%%%%
\label{app:no-hexagons}
Here we show that ``hexagonal'' joints do not occur in $\gA_{n}, \gD_{n}$ minuscule representations. For minuscule representations of $\gE_n$ we prove this property by brute force: there are only three cases to consider, they are the two 27-dimensional irreps of $E_{6}$, and the $56$-dimensional irrep of $E_{7}$, and the results of their analyses can be found at \foothref{http://het-math2.physics.rutgers.edu/loom/E6_E7_data}{this link}.
More precisely, given any minuscule weight system $\Lambda_{\rho}$ and any two roots $\alpha,\beta$ with $\alpha\measuredangle\beta=2\pi/3$, we want to show that $\Lambda_{\rho}\big|_{\weight^{\perp}}$ is always a single point (the trivial representation ${\bf 1}$ of $\gA_{2}$) or a triangle (the ${\bf 3}$ or the ${\bf \overline 3}$ of $\gA_{2}$) in the $\IR^{2}_{\weight^{\perp}}$ plane.

Given a generic weight $\weight\in\Lambda_{\rho}$, we can express its components with respect to the decomposition $\ft^{*}\simeq (\alpha\IR\oplus\beta\IR)\oplus \IR^{r-2}$
\be
\begin{split}
	\weight^{||} &= c_{\alpha}\alpha+c_{\beta}\beta \\
	\weight^{\perp} &= \weight - \weight^{||} \\
	c_{\alpha} & = \frac{(\weight\cdot\alpha)\beta^{2}-(\weight\cdot\beta)(\alpha\cdot\beta)}{\alpha^{2}\beta^{2}-(\alpha\cdot\beta)^{2}} \\
	c_{\beta} & = (\alpha\leftrightarrow\beta) 
\end{split}
\ee
We recall then that the definition of the weight sub-system $\Lambda_{\rho}\big|_{\weight^{\perp}}$ is simply
\be
	\{\weight_{i}\in\Lambda_{\rho} \,|\, \weight_{i}^{\perp}=\weight^{\perp}\}
\ee
%%%%%%%%%%%%%%%%%%%%%%%%%%%%%%%%%%%%%%%%%%%%%%%%%%%%%%%%%%%%%%%%%%%%%%%%%%%%%%%%
\subsection{\texorpdfstring{$\gA_{n}$}{A\_n}}
%%%%%%%%%%%%%%%%%%%%%%%%%%%%%%%%%%%%%%%%%%%%%%%%%%%%%%%%%%%%%%%%%%%%%%%%%%%%%%%%
For $\gA_{n}$ we may choose to work in $\IR^{n+1}$ with orthonormal basis vectors denoted by $e_{i}\ (i=1,\dots,n+1)$.
Then the positive roots are 
\be
	\Phi^{+} = \{e_{i}-e_{j}\}_{i<j},
\ee
and the simple roots are
\be
	\alpha_{i} = e_{i}-e_{i+1},\quad i=1,\dots,n\,.
\ee
The fundamental weights are
\be
	\weight_{i} = \sum_{j=1}^{i} e_{i}
\ee
satisfying 
\be
\begin{split}
	& \alpha_{i}\cdot\alpha_{j} = \left\{%
	\begin{array}{lr}
	2 & i=j, \\
	-1 & i=j\pm1, \\
	0 & \text{otherwise}.
	\end{array}
	\right. \\
	& \weight_{i}\cdot\alpha_{j} = \delta_{ij}.
\end{split}
\ee
The minuscule representations are the $p$-th antisymmetric powers of the defining representation. 
Given that the weights of the defining representation are simply $\{e_{i}\}$, the highest weight of a minuscule representation are
\be
	\weight = \nu_{1}+\nu_{2}+\dots+\nu_{p}=e_{1}+\dots+e_{p}\equiv\weight_{p}
\ee
thus, the minuscule representations are precisely the fundamental representations.

Fix a $p$, then the weight system of $\rho_{p}$ is 
\be
	\Lambda_{\rho_{p}} = \{ \nu_{i_{1}}+\nu_{i_{2}}+\dots+ \nu_{i_{p}}\,|\, i_{1}<i_{2}\dots<i_{p} \},
\ee
whose dimension is $\binom{n+1}{p}$.

Pick two roots $\alpha$, $\beta$ such that $\alpha\measuredangle\beta=2\pi/3$, which requires
\be
	\alpha = e_{i}-e_{j}\qquad \beta = e_{j}-e_{k}\,.
\ee
The Weyl group $S_{n}$ acts by permuting the $e_{i}$'s so it doesn't matter which $i,\ j,\ k$ we choose in particular, and we choose $i = 1$, $j = 2$, $k = 3$. Then we have $\alpha=\alpha_{1},\,\beta=\alpha_{2}$, the first two simple roots.

Next we want to choose a generic weight $\weight_{(i)}$  where $(i)=\{i_{1},\dots,i_{p}\}$, and compute its projection.
We have to distinguish among several cases:
\be
\begin{split}
& \weight_{(i)}\cdot\alpha = %
\left\{\begin{array}{lr}
	0 & 1,2\notin (i)\\
	1 & 1\in(i),\,2\notin(i)\\
	-1 & 1\notin(i),\, 2\in(i) \\
	0 & 1,2\in(i)
\end{array} \right.\\
& \weight_{(i)}\cdot\beta = %
\left\{\begin{array}{lr}
	0 & 2,3\notin (i)\\
	1 & 2\in(i),\,3\notin(i)\\
	-1 & 2\notin(i),\, 3\in(i) \\
	0 & 2,3\in(i)
\end{array}\right.
\end{split}
\ee
So all those $\weight_{(i)}$ such that $(i)$ does not contain $1,2,3$ are orthogonal to both $\alpha,\beta$; there are $\binom{n-2}{p}$ such weights.
Likewise all those $\weight_{(i)}$ such that $(i)$ does contain $1,2,3$ are orthogonal to both $\alpha,\beta$; there are $\binom{n-2}{p-3}$ such weights.
All of these $\binom{n-2}{p}+\binom{n-2}{p-3}$ weights have $\weight_{(i)}=\weight_{(i)}^{\perp}$ therefore $\Lambda_{\rho}\big|_{\weight^{\perp}}$ is a single point, and no solitons begin or end on them.

The remaining cases are summarized in the following table
\be\label{eq:W2-orbits-A-type}
\begin{array}{c|cccccc}
(i)\text{ contains exclusively} & 1 & 2 & 3 & 1,2 & 2,3 & 1,3 \\
\hline\hline
\weight_{(i)}\cdot\alpha & 1 & -1 & 0 & 0 & -1 & 1 \\
\weight_{(i)}\cdot\beta & 0 & 1 & -1 & 1 & 0 & -1 \\
\hline
\text{\# cases} &\binom{n-2}{p-1}&\binom{n-2}{p-1}&\binom{n-2}{p-1}&\binom{n-2}{p-2}&\binom{n-2}{p-2} & \binom{n-2}{p-2}\\
\hline
c_{\alpha} & \frac{2}{3} & -\frac{1}{3} & -\frac{1}{3} & \frac{1}{3} & -\frac{2}{3} & \frac{1}{3} \\
c_{\beta} &  \frac{1}{3} & \frac{1}{3} & -\frac{2}{3} & \frac{2}{3} & -\frac{1}{3} & -\frac{1}{3} \\
\hline
c_{\alpha} (1,-\sqrt{3}) + c_{\beta}(1,\sqrt{3}) & \left(1,-\frac{1}{\sqrt{3}}\right) & \left(0,\frac{2}{\sqrt{3}}\right) & \left(-1,-\frac{1}{\sqrt{3}}\right) & \left(1,\frac{1}{\sqrt{3}}\right) & \left(-1,\frac{1}{\sqrt{3}}\right) & \left(0,-\frac{2}{\sqrt{3}}\right) 
\end{array}
\ee
The last line of the table shows the projection of $\weight_{(i)}$ to the plane $\IR^{2}_{\weight^{\perp}_{(i)}}$ in a basis where $\alpha=(1,-\sqrt{3})$ and $\beta=(1,\sqrt{3})$.
We see that the first three columns correspond to the three weights of the antifundamental (Dynkin labels $(0,1)$) of $\gA_{2}$, and there are $\binom{n-2}{p-1}$ copies of it (possibly at different $\weight^{\perp}$).
The last three columns correspond to the three weights of the fundamental (Dynkin labels $(1,0)$) of $\gA_{2}$, and there are $\binom{n-2}{p-2}$ copies of it (possibly at different $\weight^{\perp}$).

Overall we have shown that only triangular joint diagrams appear for minuscule representations of $\gA_{n}$.

%%%%%%%%%%%%%%%%%%%%%%%%%%%%%%%%%%%%%%%%%%%%%%%%%%%%%%%%%%%%%%%%%%%%%%%%%%%%%%%%
\subsection{\texorpdfstring{$\gD_{n}$}{D\_n}}
%%%%%%%%%%%%%%%%%%%%%%%%%%%%%%%%%%%%%%%%%%%%%%%%%%%%%%%%%%%%%%%%%%%%%%%%%%%%%%%%
Here $\ft^{*}\simeq \IR^{n}$. The fundamental weights are
\be
	\weight_{1} = e_{1},\,\quad \weight_{2} = e_{1}+e_{2}\quad\dots\quad  \weight_{n}=\frac{1}{2}(e_{1}+\dots+e_{n-1}-e_{n}),\quad  \weight_{n}=\frac{1}{2}(e_{1}+\dots+e_{n-1}+e_{n})
\ee
The simple roots are
\begin{align}
	\alpha_{i} &=e_{i}-e_{i+1}, \quad i=1,\dots, n-1,\\
	\alpha_{n} &=e_{n-1}+e_{n}.
\end{align}
All roots are of the form
\be
	\pm e_{j} \pm e_{k}, \quad j\neq k,
\ee
therefore any two roots $\alpha$, $\beta$ with $\alpha\measuredangle\beta=2\pi/3$ must be of the form
\be
	\alpha = \pm e_{i} - e_{j}  \equiv  \sigma_{\alpha}\, e_{i} - e_{j}, \quad \beta = \pm e_{k} + e_{j} \equiv  \sigma_{\beta}\, e_{k} + e_{j}, 
\ee
with $i$, $j$, $k$ all different.
There are three minuscule representations: the vector $\rho_{1}$ and the two spinors $\rho_{n-1}$, $\rho_{n}$.

%%%%%%%%%%%%%%%%%%%%%%%%%%%%%%%%%%%%%%
\subsubsection*{Vector $\rho_{1}$}
%%%%%%%%%%%%%%%%%%%%%%%%%%%%%%%%%%%%%%
The $n$ weights are
\be
	\nu_{l}^{\pm}=\pm e_{l}	\qquad l=1,\dots, n
\ee
There are four possible cases for the subset of the weights in the $\alpha,\beta$ planes
\be\label{eq:W2-orbits-D-type-vector}
	\begin{array}{c|ccc|c}
		 & l=i & l=j & l=k & l\neq i,j,k \\
		 \hline\hline
		 \nu_{l}^{\sigma}\cdot\alpha & \sigma \sigma_{\alpha} & -\sigma & 0 & 0 \\
		 \nu_{l}^{\sigma}\cdot\beta & 0 & \sigma & \sigma \sigma_{\beta}  & 0 \\
		 \hline
		 c_{\alpha} & 2\sigma\sigma_{\alpha}/3 & -\sigma/3 & \sigma\sigma_{\beta}/3 & 0 \\
		 c_{\beta} & \sigma\sigma_{\alpha}/3 & \sigma/3 & 2\sigma\sigma_{\beta}/3 & 0 \\
		 \hline
		c_{\alpha} (1,-\sqrt{3}) + c_{\beta}(1,\sqrt{3}) %
		& \sigma\sigma_{\alpha}(1,-1/\sqrt{3}) & (0, 2/\sqrt{3}) & \sigma\sigma_{\beta}(1,1/\sqrt{3}) & (0,0)
	\end{array}
\ee
since the $\sigma,\sigma_{\alpha},\sigma_{\beta}$ are all $\pm1$, we conclude that all the projections of the weights to $\IR^{2}_{\weight^{\perp}}$ planes fall indeed into either the fundamental or the antifundamental of $\gA_{2}$.

%%%%%%%%%%%%%%%%%%%%%%%%%%%%%%%%%%%%%%
\subsubsection*{Spinor $\rho_{n-1}$}
%%%%%%%%%%%%%%%%%%%%%%%%%%%%%%%%%%%%%%
The $2^{n-1}$ weights are
\be
	\frac{1}{2}\sum_{j=1}^{n}\eta_{j} e_{j}, \quad\eta_{j}=\pm1, \quad \prod_{j=1}^{n}\eta_{j}=-1.
\ee
Denoting $\weight_{(\eta)}$ a weight corresponding to $(\eta)=\{\eta_{i}\}_{i=1}^{n}$, we can easily compute
\be
\begin{split}
	& \weight_{(\eta)}\cdot\alpha = \frac{1}{2}(\sigma_{\alpha}\eta_{i}-\eta_{j}) \\
	& \weight_{(\eta)}\cdot\beta = \frac{1}{2}(\sigma_{\beta}\eta_{k}+\eta_{j}) 
\end{split}
\ee
hence
\be
\begin{split}
	& c_{\alpha} = \frac{1}{6}(2\sigma_{\alpha}\eta_{i}+\sigma_{\beta}\eta_{j}-\eta_{j}) \\
	& c_{\beta} = \frac{1}{6}(2\sigma_{\beta}\eta_{k}+\sigma_{\alpha}\eta_{i}+\eta_{j}) 
\end{split}
\ee
and therefore the projections of the $\weight_{(\eta)}$ to their $\IR^{2}_{\weight_{(\eta)}^{(\perp)}}$ planes are
\be
	c_{\alpha}(1,-\sqrt{3})+c_{\beta}(1,\sqrt{3}) = \left( \frac{1}{2}(\sigma_{\alpha}\eta_{i}+\sigma_{\beta}\eta_{k}),\,\frac{1}{2\sqrt{3}}(\sigma_{\beta}\eta_{k}-\sigma_{\alpha}\eta_{i} +2\eta_{j}) \right)
\ee
accounting for all possible signs of $\sigma_{\alpha}$, $\sigma_{\beta}$, $\eta_{i}$, $\eta_{j}$, $\eta_{k}$ we find exactly seven possibilities:
\be\label{eq:W2-orbits-D-type-spinor}
\begin{array}{l}
 (0,0) \\
 \left(1,\frac{1}{\sqrt{3}}\right) 
 \left(0,-\frac{2}{\sqrt{3}}\right) 
 \left(-1,\frac{1}{\sqrt{3}}\right) \\
  \left(1,-\frac{1}{\sqrt{3}}\right)
 \left(-1,-\frac{1}{\sqrt{3}}\right)
 \left(0,\frac{2}{\sqrt{3}}\right) 
\end{array}
\ee
where we recognize the trivial representation of $\gA_{2}$ as well as the $(1,0)$ and the $(0,1)$.

%%%%%%%%%%%%%%%%%%%%%%%%%%%%%%%%%%%%%%
\subsubsection*{Spinor $\rho_{n}$}
%%%%%%%%%%%%%%%%%%%%%%%%%%%%%%%%%%%%%%
The $2^{n-1}$ weights are
\be
	\frac{1}{2}\sum_{j=1}^{n}\eta_{j} e_{j}, \quad \eta_{j}=\pm1, \quad \prod_{j=1}^{n}\eta_{j}=1.
\ee
The argument goes the same as for $\rho_{n-1}$.

%%%%%%%%%%%%%%%%%%%%%%%%%%%%%%%%%%%%%%%%%%%%%%%%%%%%%%%%%%%%%%%%%%%%%%%%%%%%%%%%
\section{\texorpdfstring{$k_\rho$}{k\_rho} and the Cartan matrix}
\label{sec:k_rho_cartan}
%%%%%%%%%%%%%%%%%%%%%%%%%%%%%%%%%%%%%%%%%%%%%%%%%%%%%%%%%%%%%%%%%%%%%%%%%%%%%%%%
In this section we prove the identity
\be\label{eq:cartan-identity}
	\tilde C_{\alpha_i,\alpha_j} :=  \sum_{\nu\in\Lambda_\rho} (\alpha_i \cdot\nu)(\nu \cdot\alpha_j) = k_\rho \, C_{\alpha_i,\alpha_j}
\ee
where $\rho$ is a minuscule representation for $\fg$ of $\gA,\gD,\gE$ type, while $\alpha_i,\alpha_j$ are any two simple roots, and $C_{\alpha_i,\alpha_j}$ is the Cartan matrix element for them.

With our conventions the Cartan matrix is given directly by the inner products of simple roots as $C_{\alpha_i,\alpha_j} = \alpha_i \cdot\alpha_j$. 
To compute the left hand side of (\ref{eq:cartan-identity}), let us start from the diagonal elements i.e.\ when $\alpha_i = \alpha_j$. In this case $(\nu\cdot\alpha_i)^2$ is $1$ for all $\nu\in\CP_{\alpha_i}^{\pm}$ and zero for all other weights ($\CP_\alpha^\pm$ is defined in (\ref{eq:soliton-charge-pairs-pm})). Therefore for diagonal elements we find $2 |\CP_{\alpha_i}|$, we take this to define $k_\rho \equiv |\CP_{\alpha_i}|$.

For off-diagonal elements, let us distinguish between two cases: $\alpha_i\measuredangle\alpha_j = \pi/2$ or $2\pi/3$.
In the first case, we can split the sum over $\Lambda_\rho$ into slices as shown in Figure \ref{fig:angle-90}. The first four types of slice give no contribution because $ (\alpha_i \cdot\nu)(\nu \cdot\alpha_j) = 0$ for each $\nu$.
The last type of slice, resembling a square, gives instead four contributions which sum up to zero.
So overall the left-hand side of (\ref{eq:cartan-identity}) is zero, and this is also true of the Cartan matrix element $C_{\alpha_i,\alpha_j}$ on the right-hand side.
The last case that remains to check is when  $\alpha_i\measuredangle\alpha_j = 2\pi/3$. Now the sum over $\Lambda_\rho$ may be split into slices of the type shown in Figure \ref{fig:A2-joint} (there are also slices made of a single point, but these contain weights that are orthogonal to both $\alpha_i,\alpha_j$ and therefore do not contribute to the sum).
In the notation of Figure \ref{fig:A2-joint} let us identify $\alpha,\beta$ with $\alpha_i,\alpha_j$ respectively. Then of the three weights shown in the left frame, only $\weight_2$ gives a non-zero contribution, because $\weight_3\cdot\alpha = 0 = \weight_1\cdot\beta$. Similarly for the weights depicted in the right frame, only the weight in the bottom-right will give a non-zero contribution. 
We find that, for each element of  $\CP_{\alpha_i}$ there is a contribution of $-1$, so overall we find $\tilde C_{\alpha_i,\alpha_j} = - k_\rho$ which again matches the right hand side of (\ref{eq:cartan-identity}).

This concludes the proof of (\ref{eq:cartan-identity}), and also gives an intepretation of $k_\rho$ as 
\be
	k_\rho = |\CP_\alpha|	\,,
\ee
which does not depend on the choice of root $\alpha$, but only on $\rho$. It is worthwhile to stress that the proof involves special properties of minuscule representaitons, and the identity is not expected to hold for non-minuscule ones.

%%%%%%%%%%%%%%%%%%%%%%%%%%%%%%%%%%%%%%%%%%%%%%%%%%%%%%%%%%%%%%%%%%%%%%%%%%%%%%%%
\section{Spectral curves for \texorpdfstring{$\mathfrak{g} = \mathrm{A}_{N-1}, \mathrm{D}_{N}$}{g = A\_N-1, D\_N}, and \texorpdfstring{$\mathrm{E}_N$}{E\_N}}
\label{sec:spectral_curves_for_g}
%%%%%%%%%%%%%%%%%%%%%%%%%%%%%%%%%%%%%%%%%%%%%%%%%%%%%%%%%%%%%%%%%%%%%%%%%%%%%%%%

In \cite{Martinec:1995by} the Seiberg-Witten curve of a general gauge group $G$ is described in terms of a spectral curve in a general representation $\rho$ of $G$. Its modern description is given in \cite{Keller:2011ek} like the following. We compactify 6d $\cN=(2,0)$ theory of type $\mathfrak{g}$ on $C$ parametrized by $z$. The 6d theory has world-volume field $\phi_{w_i}(z)$ on $C$, transforming as degree $w_i$ multi-differentials, where $w_i$ is the degree of the Casimir invariants of $\mathfrak{g}$.
\begin{itemize}
	\item $\mathfrak{g} = \mathrm{A}_{N-1}$: $h^{\vee} = N$, $w_i = 2, 3, \ldots, N-1, N$.
	\item $\mathfrak{g} = \mathrm{D}_{N}$: $h^{\vee} = 2N-2$, $w_i = 2, 4, \ldots, 2N-2; N$.
	\item $\mathfrak{g} = \mathrm{E}_6$: $h^{\vee} = 12$, $w_i = 2,5,6,8,9,12$.
	\item $\mathfrak{g} = \mathrm{E}_7$: $h^{\vee} = 18$, $w_i = 2,6,8,10,12,14,18$.
\end{itemize}
\begin{align}\label{eq:generic-curves}
	\phi_{w_i}(z) &= f_{w_i}(z) \left(\mathrm{d}z\right)^{w_i}.
\end{align}

When $(\rho, V)$ is the vector representation, the spectral covers are 
\begin{itemize}

	\item $\mathfrak{g} = \mathrm{A}_{N-1}$: $\Sigma_\rho = \lambda^N + \sum_{w_i} \phi_{w_i} \lambda^{N-w_i}$,

	\item $\mathfrak{g} = \mathrm{D}_{N}$ 
		: $\Sigma_\rho = \lambda^{2N} + \sum_{w_i \neq N} \phi_{w_i} \lambda^{2N-w_i} + \phi_{2N},\ 
		\phi_{2N} = (\tilde{\phi}_{N})^2$,	

	\item $\mathfrak{g} = \mathrm{E}_6$ \cite{Lerche:1996an}: $\Sigma_\rho = \frac{1}{2} \lambda^3 \phi_{12}^2 - q_1 \phi_{12} + q_2 = 0$, where
        \begin{align}
        	q_1 &= 270 \lambda^{15} + 342 \phi_2 \lambda^{13} + 162 \phi_2^2 \lambda^{11} - 252 \phi_5 \lambda^{10}  +  (26 \phi_2^3 + 18 \phi_6) \lambda^9 \cr
        		&\qquad -162 \phi_2 \phi_5 \lambda^8 + (6 \phi_2 \phi_6 - 27 \phi_8) \lambda^7 - (30 \phi_2^2 \phi_5 - 36 \phi_9) \lambda^6 \cr
        		&\qquad + (27 \phi_5^2 - 9 \phi_2 \phi_8) \lambda^5 -  (3 \phi_5 \phi_6 - 6 \phi_2 \phi_9) \lambda^4 - 3 \phi_2 \phi_5^2 \lambda^3 \cr
        		&\qquad - 3 \phi_5 \phi_9 \lambda - \phi_5^3, \cr
        	q_2 &= \frac{1}{2 \lambda^3}  \left(q_1^2 - p_1^2 p_2 \right), \cr
        	p_1 &= 78 \lambda^{10} + 60 \phi_2 \lambda^8 + 14 \phi_2^2 \lambda^6 - 33 \phi_5 \lambda^5 + 2 \phi_6 \lambda^4 - 5 \phi_2 \phi_5 \lambda^3 - \phi_8 \lambda^2 \cr
        		&\qquad - \phi_9 \lambda - \phi_5^2 \cr
        	p_2 &= 12 \lambda^{10} + 12 \phi_2 \lambda^8 + 4 \phi_2^2 \lambda^6 - 12 \phi_5 \lambda^5 + \phi_6 \lambda^4 - 4 \phi_2 \phi_5 \lambda^3 - 2 \phi_8 \lambda^2 \cr
        		&\qquad + 4 \phi_9 \lambda  + \phi_5^2. \nonumber
        \end{align}

	\item $\mathfrak{g} = \mathrm{E}_7$:  
	$\Sigma_\rho = -\frac{1}{3^6} \lambda^2 \left(\phi_{18}^3+ A_2\, \phi_{18}^2 + 
	A_1\, \phi_{18} +A_0\right)$,\footnote{This is obtained from the characteristic polynomical in Appendix A.1. of \cite{Eguchi:2001fm}. Note that there is a typo in the expression for $A_2$ in the reference.} where
        \begin{align*}
            A_2 &= \frac{9}{16\, \lambda^2}\, (6\, q\, p_1-3\, p_3)  \,,  \\
            A_1 &= \left(\frac{9}{16\, \lambda^2}\right)^2 \, (9\, q^2\, p_1^2-6\, r\, p_1\, p_2
            	-12\, q\, p_1\, p_3+3\, q\, p_2^2+3\, p_3^2)  \,, \\
            A_0 &=  -\left(\frac{9}{16\, \lambda^2}\right)^3 \, (4\, r^2\, p_1^3 +6\, q\, r\, p_1^2\, p_2
                +9\, q^2\, p_1^2\, p_3-6\, r\, p_1\, p_2\, p_3 -6\, q\, p_1\, p_3^2   \\
            	&\hspace{90pt}  + 2\, r\, p_2^3+3\, q\, p_2^2\, p_3+p_3^3)  \,,\\
            p_1 &= 1596\, \lambda^{10}+88\,\phi_2^2\, \lambda^6+7\, \phi_6\, \lambda^4+660\, \phi_2\, \lambda^8
                      -2\,\phi_8\, \lambda^2 -\phi_{10}\,,   \\
            p_2 &= 16872\, \lambda^{15}+11368\, \phi_2\, \lambda^{13}+2952\, \phi_2^2\, \lambda^{11}
                      +(176\, \phi_6+264\, \phi_2^3)\, \lambda^9  \\
                &\qquad +\left(-100\, \phi_8+\frac{100}{3}\phi_2\, \phi_6\right) \lambda^7
                 +\left(-\frac{68}{3}\phi_2\, \phi_8+68\, \phi_{10}\right) \lambda^5  \\
                &\qquad +\left(\frac{2}{9}\phi_6^2 -\frac{4}{3}\phi_{12}\right) \lambda^3
                 +\left(\frac{218}{8613}\phi_2\, \phi_6^2-\frac{4}{3}\phi_{14}\right) \lambda   \,,  \\
            p_3 &= 44560\, \lambda^{20}+41568\, \phi_2\lambda^{18}+16080\, \phi_2^2\, \lambda^{16}
                 +\left(2880\, \phi_2^3+\frac{2216}{3}\phi_6\right) \lambda^{14}  \\
                &\qquad +\left(312\, \phi_2\, \phi_6+192\, \phi_2^4-\frac{1552}{3}\phi_8\right) \lambda^{12}
                	+\left(32\, \phi_2^2\, \phi_6-40\, \phi_2\, \phi_{10}-\frac{64}{3}\phi_{12}
                	+\frac{11}{3}\phi_6^2\right) \lambda^8   \\
                &\qquad +\left(-\frac{416}{3}\phi_{14}-16\, \phi_2^2\, \phi_{10}-\frac{4}{9}\phi_6\, \phi_8
                	-\frac{32}{9}\phi_2\, \phi_{12}+\frac{27776}{8613}\phi_2\, \phi_6^2\right) \lambda^6 \\
                &\qquad +\left(\frac{3488}{8613}\phi_2^2\, \phi_6^2+\frac{4}{9}\phi_8^2
                	-\frac{64}{3}\phi_2\, \phi_{14}-\frac{2}{3}\phi_6\, \phi_{10}\right) \lambda^4
                    +\frac{4}{3}\phi_8\, \phi_{10}\, \lambda^2+\phi_{10}^2 \,, \\
             q 	&= -28\, \lambda^{10}-\frac{44}{3} \phi_2\, \lambda^8-\frac{8}{3} \phi_2^2\, \lambda^6
             	-\frac{1}{3} \phi_6\, \lambda^4+\frac{2}{9} \phi_8\, \lambda^2-\frac{1}{3} \phi_{10} \,, \\
             r 	&= 148\, \lambda^{15}+116\, \phi_2\, \lambda^{13}+36\, \phi_2^2\, \lambda^{11}
             	+ \left(4\, \phi_2^3+\frac{8}{3} \phi_6\right) \lambda^9 
				+ \left(\frac{2}{3} \phi_2\phi_6-2\, \phi_8 \right) \lambda^7  \\
             	&\qquad + \left(2\, \phi_{10}-\frac{2}{3} \phi_2\phi_8 \right) \lambda^5
            	+\left(\frac{1}{81} \phi_6^2-\frac{2}{27} \phi_{12}\right) \lambda^3
                +\left(\frac{109}{8613} \phi_2\phi_6^2-\frac{2}{3} \phi_{14}\right) \lambda \,.
        \end{align*}	
\end{itemize}

The $n$-th fundamental cover of $\mathfrak{g} = \mathrm{A}_{N-1}$ is of order $\binom{N}{n}$ in $\lambda$. The spectral cover of $\mathfrak{g} = \mathrm{D}_{N}$ in the spinor representation is of order $2^{N-1}$ in $\lambda$. The coefficients of these covers are functions of $\phi_{w_i}$, which can be found case by case by considering the relations between the weights in the first fundamental representation and those in the representation of the cover, but there is no known general expressions for such covers.

%%%%%%%%%%%%%%%%%%%%%%%%%%%%%%%%%%%
\section{Summary of conventions for \texorpdfstring{$\SO(6)$}{SO(6)} SYM}\label{app:so6}
%%%%%%%%%%%%%%%%%%%%%%%%%%%%%%%%%%%

%%%%%%%%%%%%%%%%%%%%%%%%%%%%%%%%%%%
\subsubsection*{Roots and weights}
%%%%%%%%%%%%%%%%%%%%%%%%%%%%%%%%%%%
The algebra is rank 3, and in the standard choice of basis $e_{1}, e_{2}, e_{3}$ for $\IR^{3}\simeq\mathfrak{t^{*}}$, the roots are
\be
	\Phi = \{\pm e_{i}\pm e_{j}, i\neq j\}
\ee
with signs chosen independently.
The positive roots are chosen to be
\be
	\Phi^{+} = \{e_{i}\pm e_{j}, i\neq j\},
\ee
while simple roots are 
\be
	\alpha_{1} = e_{1}-e_{2}\,,\qquad %
	\alpha_{2} = e_{2}-e_{3}\,,\qquad %
	\alpha_{3} = e_{2}+e_{3}\,.
\ee
with Cartan matrix\footnote{The Cartan matrix would actually be $2\frac{\alpha_{i}\cdot\alpha_{j}}{\alpha_{j}\cdot\alpha_{j}}$, but for simply laced groups $\alpha_{j}^{2}=2$. We use conventions: $e_{i}\cdot e_{j}=\delta_{ij}$.} and Dynkin diagram are
\be
	\alpha_{i}\cdot\alpha_{j} = \left(%
	\begin{array}{ccc}
	2 & -1 & -1 \\
	-1 & 2 & 0 \\
	-1 & 0 & 2
	\end{array}%
	\right)_{i,j} %
	\qquad %
	\includegraphics[width=0.15\textwidth]{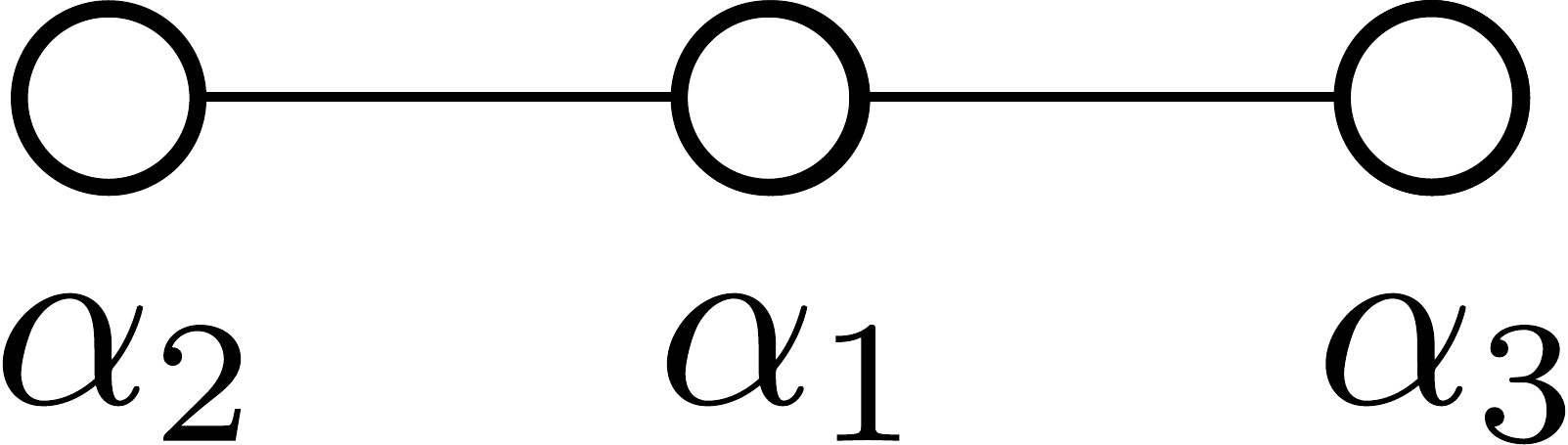}
\ee
The three fundamental weights are then
\be
	\omega_{1} = e_{1}\,,\qquad %
	\omega_{2} = \frac{1}{2}(e_{1}+e_{2}-e_{3})\,,\qquad %
	\omega_{3} = \frac{1}{2}(e_{1}+e_{2}+e_{3})\,.
\ee

\begin{figure}[!ht]
\begin{center}
\includegraphics[width=0.3\textwidth]{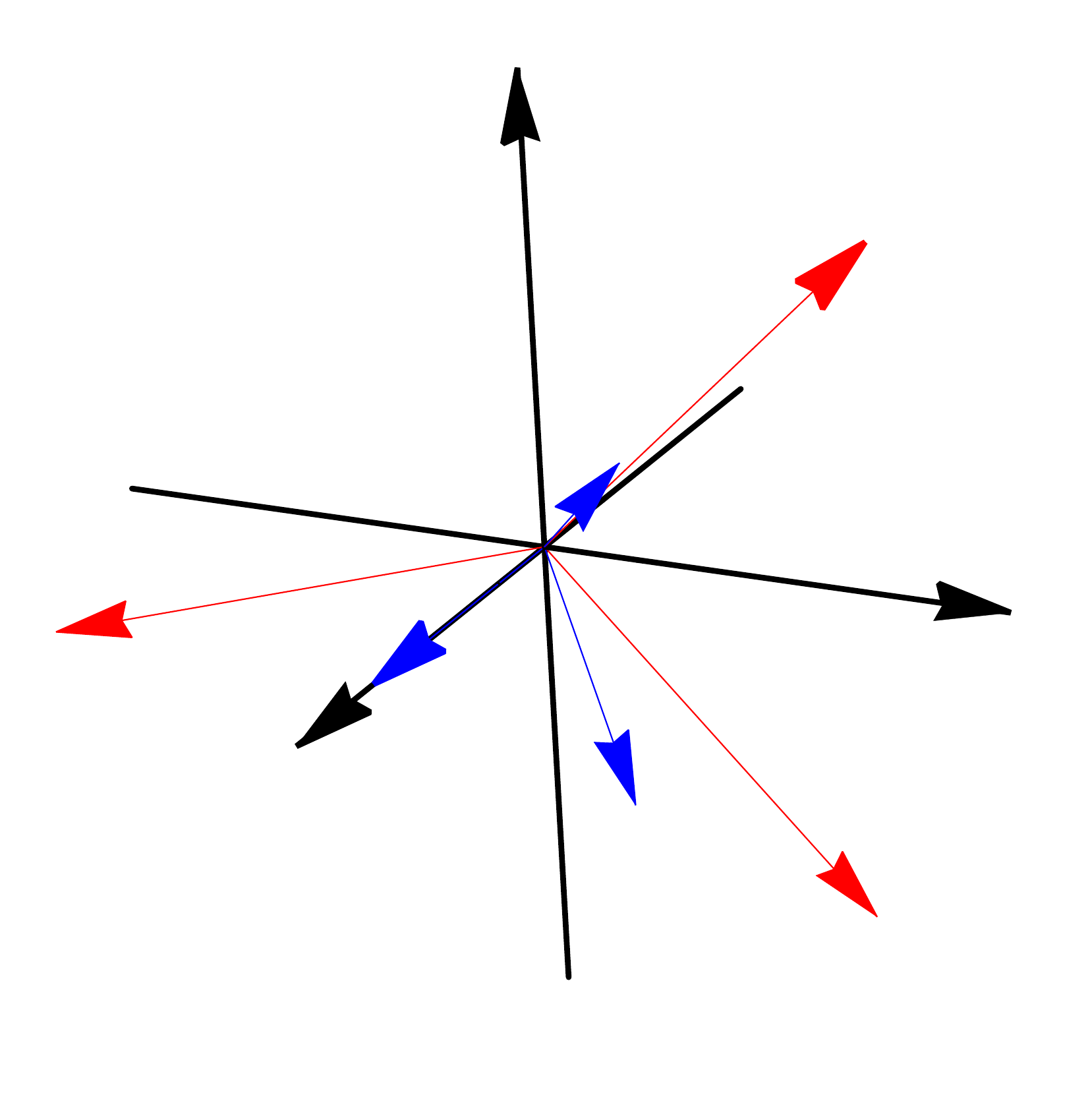}
\caption{Simple roots in red, fundamental weights in blue. The black frame is the $e_{1}, e_{2}, e_{3}$ basis.}
\label{fig:roots-weights}
\end{center}
\end{figure}

%%%%%%%%%%%%%%%%%%%%%%%%%%%%%%%%%%%
\subsubsection*{Vector representation}
%%%%%%%%%%%%%%%%%%%%%%%%%%%%%%%%%%%
The vector representation has highest weight $\weight_{1}$, it is six-dimensional, with weights given by:
\be
\begin{split}
	& \weight_{1}=\omega_{1} \qquad\qquad \weight_{2}=\weight_{1}-\alpha_{1} \qquad\qquad \weight_{3}=\weight_{2}-\alpha_{2} \\
	& \weight_{4}=\weight_{2}-\alpha_{3} \qquad \weight_{5}=\weight_{2}-\alpha_{2}-\alpha_{3} \qquad \weight_{6}=\weight_{2}-\alpha_{1}-\alpha_{2}-\alpha_{3} \,.
\end{split}
\ee

\begin{figure}[!ht]
\begin{center}
\includegraphics[width=0.35\textwidth]{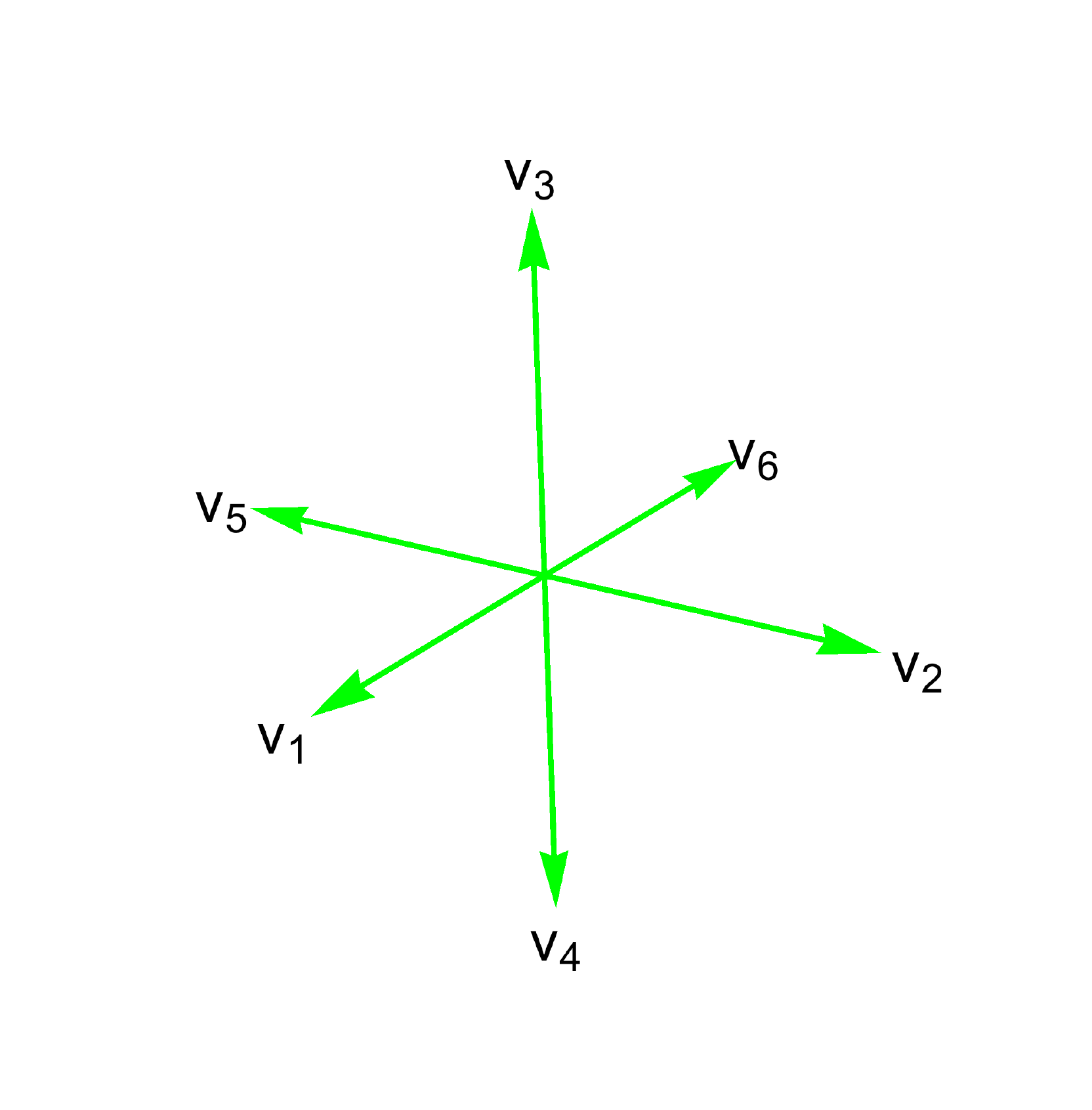}
\includegraphics[width=0.35\textwidth]{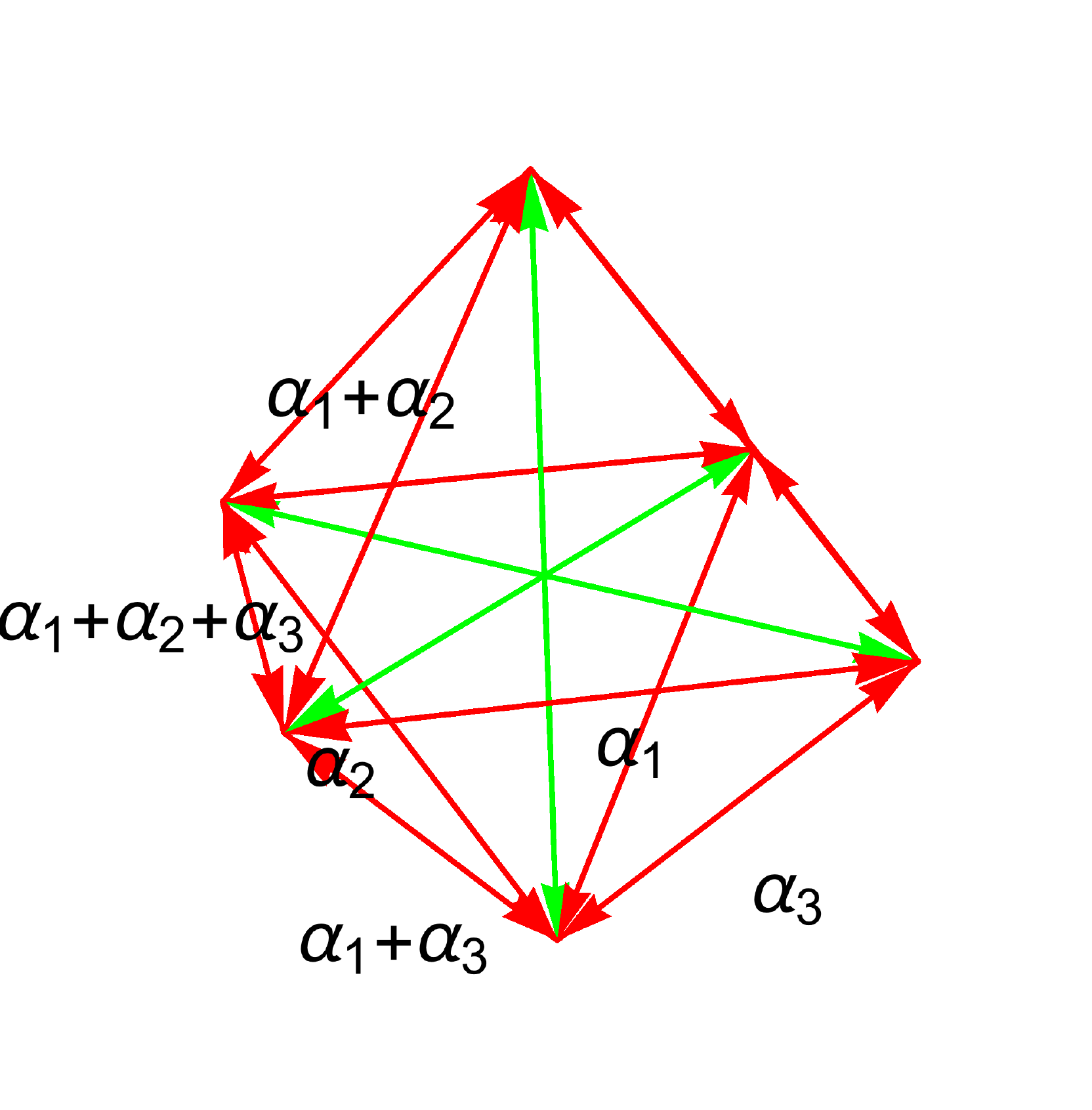}
\caption{The fundamental representation. Weights in green, the connecting roots in red.}
\label{fig:fundamental}
\end{center}
\end{figure}

%%%%%%%%%%%%%%%%%%%%%%%%%%%%%%%%%%%
\subsubsection*{Weyl group}\label{app:so6-weyl-group}
%%%%%%%%%%%%%%%%%%%%%%%%%%%%%%%%%%%
$W\simeq S_{4}$ is generated by $w_{1}, w_{2}, w_{3}$ defined as reflections about the planes orthogonal to the three simple roots. In our basis for $\mathfrak{t^{*}}$ they read
\be
	w_{1} = \left(\begin{array}{ccc}
		0 & 1 & 0 \\
		1 & 0 & 0 \\
		0 & 0 & 1
	\end{array}\right) \,,\qquad %
	w_{2} = \left(\begin{array}{ccc}
		1 & 0 & 0 \\
		0 & 0 & 1 \\
		0 & 1 & 0
	\end{array}\right) \,,\qquad %
	w_{3} = \left(\begin{array}{ccc}
		1 & 0 & 0 \\
		0 & 0 & -1 \\
		0 & -1 & 0
	\end{array}\right) \,. %	
\ee
The Weyl chambers are easily determined: the plane orthogonal to $\alpha_{1}$ is the span of $(e_{1}+e_{2}, e_{3})$; the plane orthogonal to $\alpha_{2}$ is the span of $(e_{1},+e_{2}+ e_{3})$; the plane orthogonal to $\alpha_{3}$ is the span of $(e_{1},+e_{2}- e_{3})$. The corresponding intersection lines are thus 
\be
	\ell_{12} = (e_{1}+e_{2}+e_{3})\,\IR\,,\qquad %
	\ell_{13} = (e_{1}+e_{2}-e_{3})\,\IR\,,\qquad %
	\ell_{23} = e_{1}\,\IR\,,\qquad %
\ee
they are parallel to the fundamental weights, as expected by construction.

\begin{figure}[h!]
\begin{center}
\includegraphics[width=0.45\textwidth]{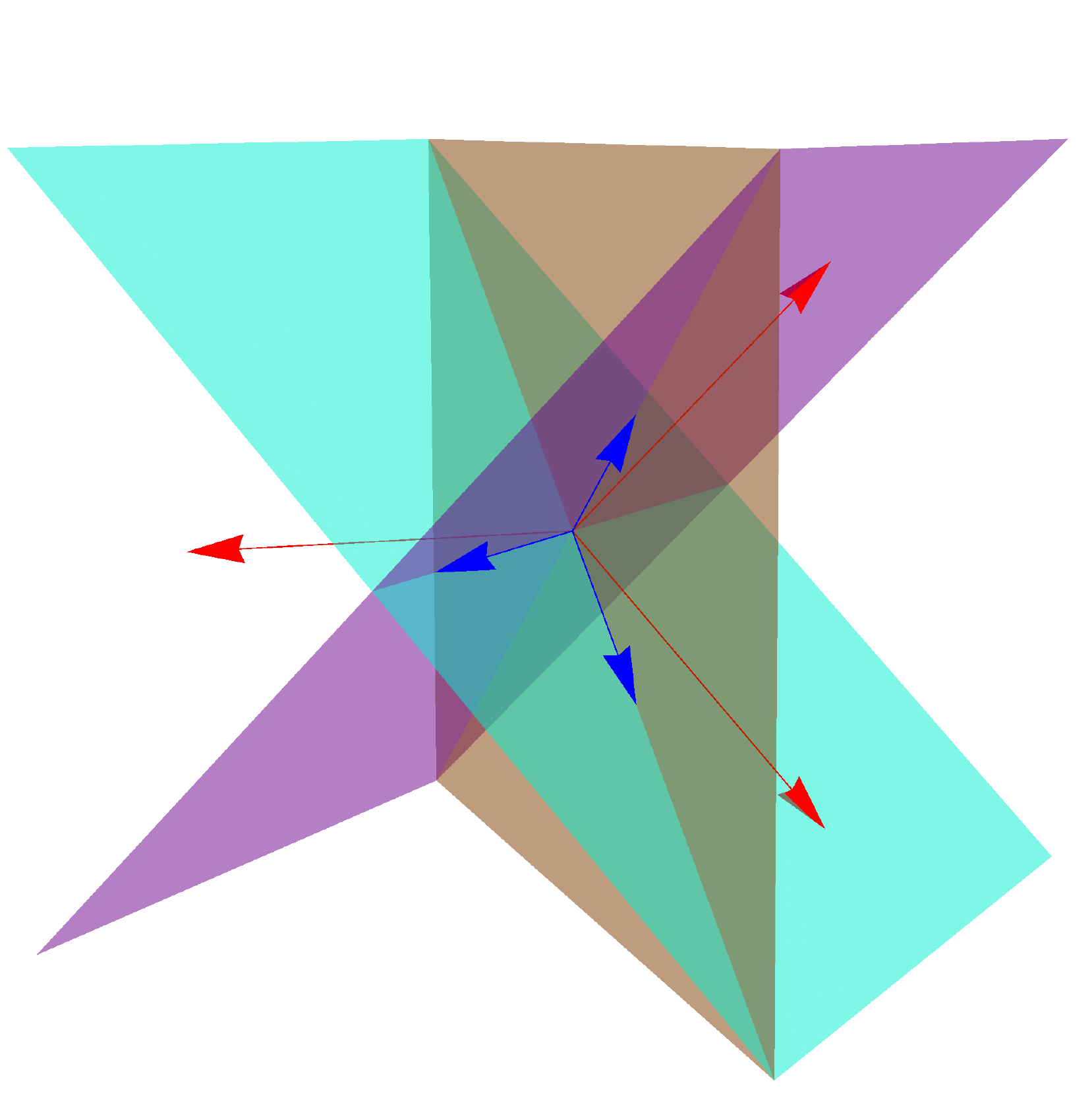}\hspace{0.1\textwidth}\includegraphics[width=0.45\textwidth]{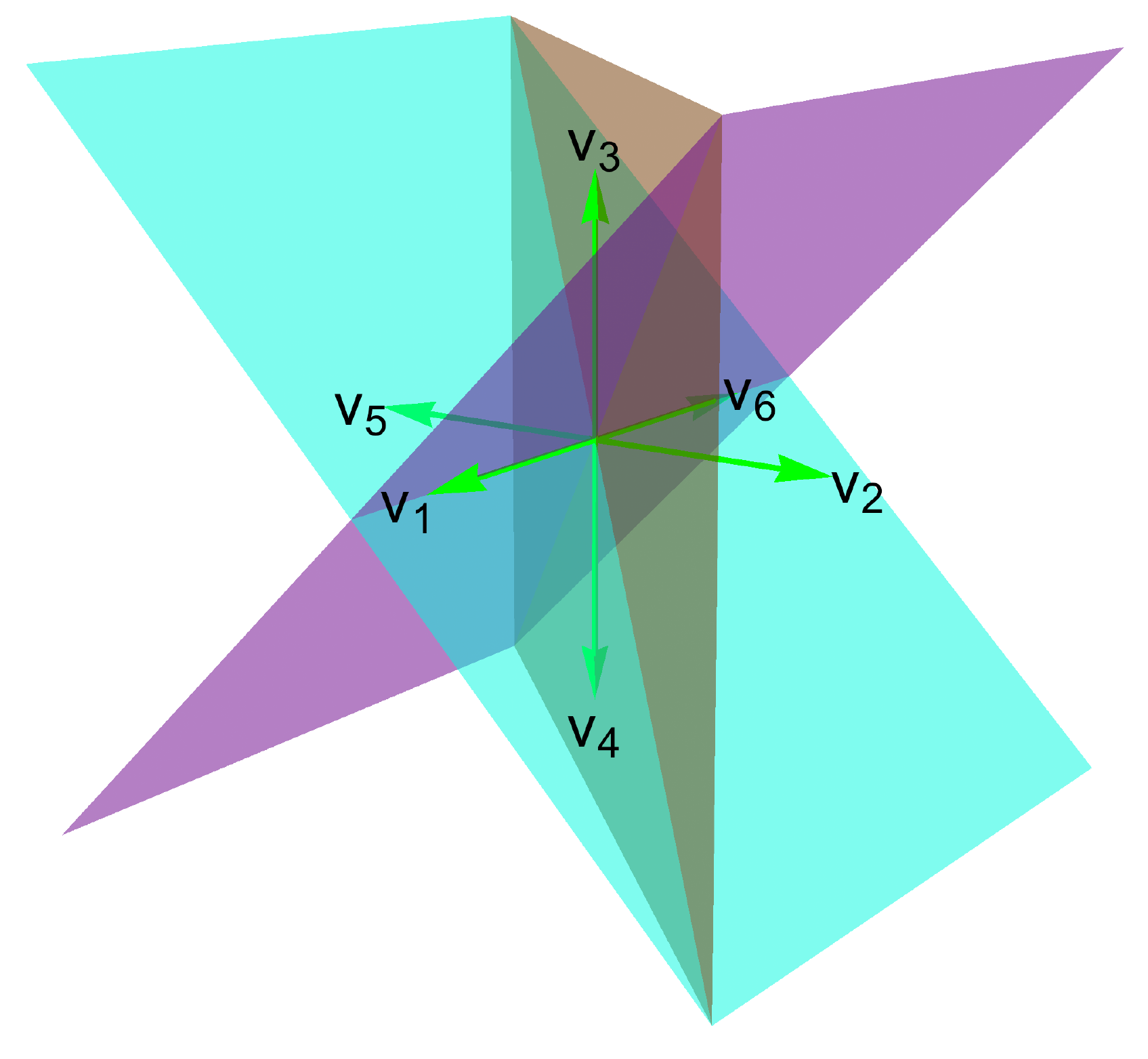}
\caption{On the left, the Weyl planes: fundamental weights are in blue, simple roots in red; the brown plane corresponds to $w_{1}$, the fuchsia plane corresponds to $w_{2}$, and the cyan plane to $w_{3}$. On the right, the Weyl planes and the weights of the fundamental representation.}
\label{fig:weyl-planes}
\end{center}
\end{figure}

The action of the Weyl group on the weights of the fundamental representation has the following form
\be\label{eq:weyl-fundamental}
\begin{split}
	w_{1}: \qquad %
	\weight_{1}\leftrightarrow \weight_{2}\,,\quad %
	\weight_{5}\leftrightarrow \weight_{6}\,,\quad 
	\weight_{3}\,, \weight_{4}\ \text{fixed} \\
	w_{2}: \qquad %
	\weight_{2}\leftrightarrow \weight_{3}\,,\quad %
	\weight_{4}\leftrightarrow \weight_{5}\,,\quad 
	\weight_{1}\,, \weight_{6}\ \text{fixed} \\
	w_{3}: \qquad %
	\weight_{2}\leftrightarrow \weight_{4}\,,\quad %
	\weight_{3}\leftrightarrow \weight_{5}\,,\quad 
	\weight_{1}\,, \weight_{6}\ \text{fixed} 
\end{split}
\ee
as can be checked graphically in Figure \ref{fig:weyl-planes}.
The vector representation clearly consists of single Weyl orbit.

%%%%%%%%%%%%%%%%%%%%%%%%%%%%%%%%%%%%%%%%
\section{Standard trivialization for \texorpdfstring{$\SO(6)$}{SO(6)} SYM in the vector representation} \label{app:canonical-example}
%%%%%%%%%%%%%%%%%%%%%%%%%%%%%%%%%%%%%%%%

In this section we show how to obtain the explicit trivialization of Section \ref{subsec:explicit-trivialization}, presented into the standard form. {We fix $\mu=5-5i$, $u_{2}=1+i$, $u_{4}=3+i$, $u_{3}^{2}=(1+i)^{3}$.} We start by choosing a basepoint $z_{0}$, and branch cuts emanating to infinity,
the sheets can then be labeled by weights \footnote{The labeling of sheets by weights is not entirely arbitrary. In appendix \ref{app:sheet-weight-identification} we discuss an algorithmic way of carrying out the identification.} 
$\weight_{1},\dots,\weight_{6}$ above $z_{0}$. 
We then choose paths from $z_{0}$ to each of the branch points, such that these paths  do not intersect the branch cuts (a ``Lefshetz spider''), and determine which sheets collide at each branch-point. 
Concretely, we need to keep track of the solutions of the spectral cover equation: an example of this tracking is provided in Figure \ref{fig:so6-fund-branch-points}.
\begin{figure}[!ht]
\begin{center}
\includegraphics[width=0.85\textwidth]{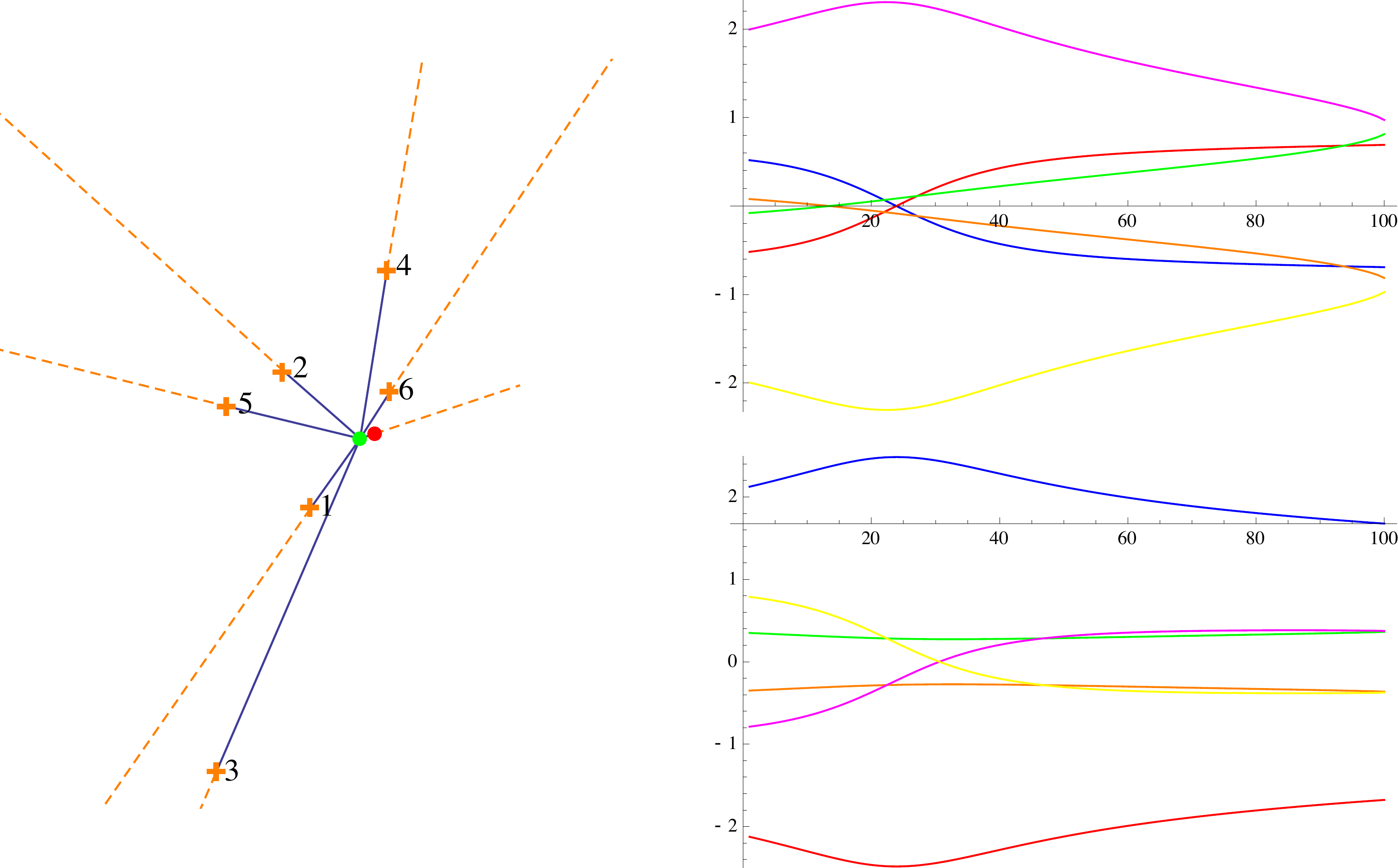}
\caption{On the left: the Lefshetz spider: the base-point is marked in green, branch-points are crosses and the irregular puncture at $z=0$ is a red dot. On the right: the tracking of real and imaginary parts of the values of $x(z)$ along the path from the basepoint to branch-point $6$.
Colors correspond to the following weights: fuchsia = $\weight_{1}$; red = $\weight_{2}$; orange = $\weight_{3}$; green=$\weight_{4}$; blue = $\weight_{5}$; yellow = $\weight_{6}$.}
\label{fig:so6-fund-branch-points}
\end{center}
\end{figure}
Likewise, for the monodromy around $z=0$ we may simply choose a path circling the puncture while avoiding cuts, and keep track of roots, as shown in Figure \ref{fig:zero-monodromy}. The overall results of this analysis are summarized in table (\ref{eq:monodromy-analysis}), they contain all the explicit details on the trivialization.
Away from the branch cuts, there is a 1-1 correspondence between sheets of the cover and weights of the representation.

\begin{figure}[!ht]
\begin{center}
\includegraphics[width=0.85\textwidth]{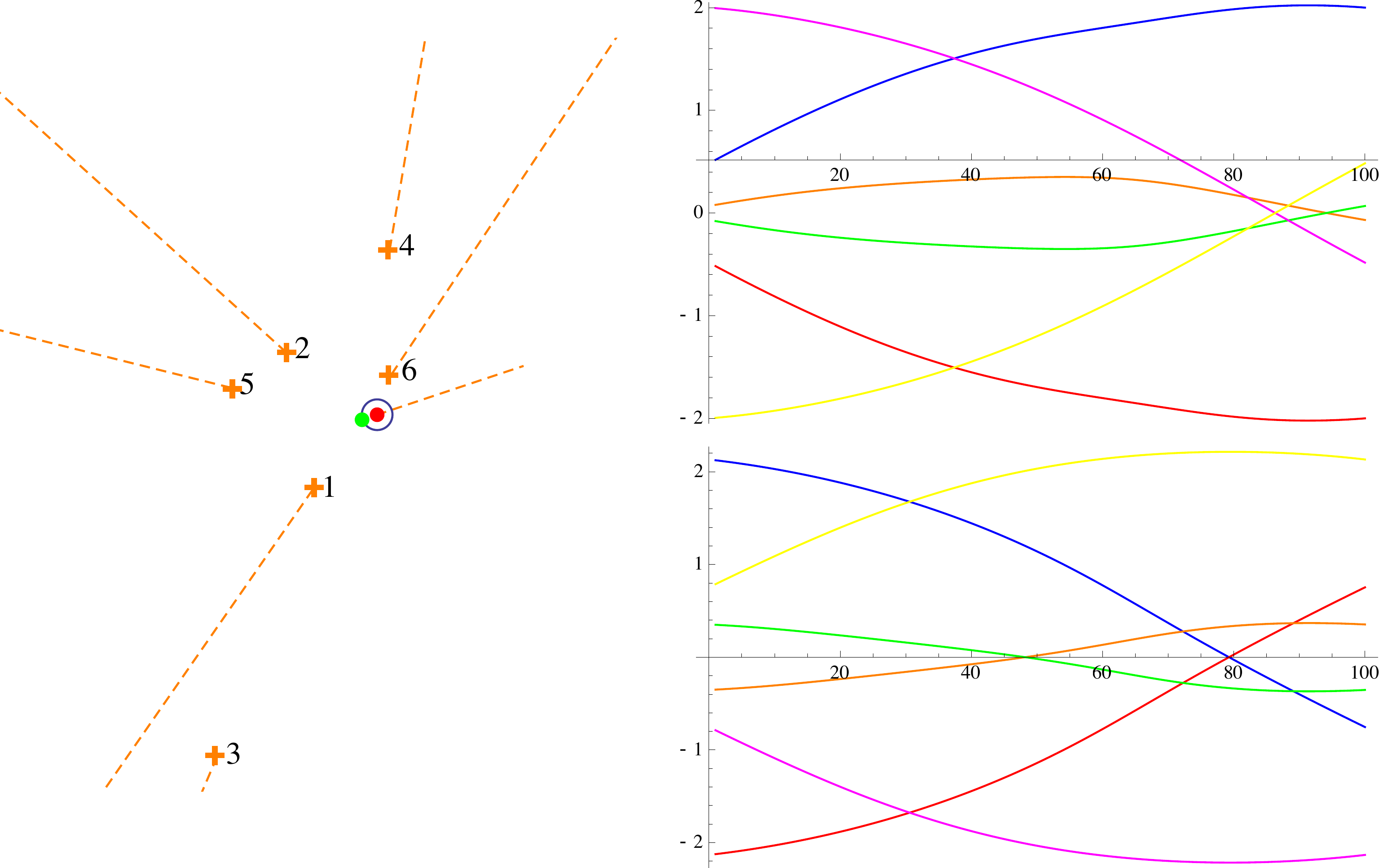}\\
\caption{On the left: the choice of {counter-clockwise} monodromy path around $z=0$. On the right: the tracking of real and imaginary parts of the values of $x(z)$ along the path. The color-weight dictionary of Figure \ref{fig:so6-fund-branch-points} applies.}
\label{fig:zero-monodromy}
\end{center}
\end{figure}

\begin{figure}[!ht]

\be\label{eq:monodromy-analysis}
\begin{array}{c|c|c|c|c}
	\text{branch point} & \text{partition} & \text{Weyl element} & & \text{positive root}\\
	\hline
	\hline
	1 & %
	(\weight_{1},\weight_{3}) % 
	(\weight_{4},\weight_{6}) % 
	(\weight_{2}) %
	(\weight_{5}) %
	& w_{1}w_{2}w_{1} %
	&\left(\begin{array}{ccc}
		0 & 0 & 1 \\
		0 & 1 & 0 \\
		1 & 0 & 0
	\end{array}\right)%
	& \alpha_{1}+\alpha_{2} = e_{1}-e_{3}\\
	\hline
	2 & %
	(\weight_{2}, \weight_{3}) %
	(\weight_{4}, \weight_{5}) %
	(\weight_{1}) %
	(\weight_{6})%
	& w_{2} %
	&\left(\begin{array}{ccc}
		1 & 0 & 0 \\
		0 & 0 & 1 \\
		0 & 1 & 0
	\end{array}\right)%
	& \alpha_{2} = e_{2}-e_{3}\\
	\hline
	3 & %
	(\weight_{2}, \weight_{4}) % 
	(\weight_{3}, \weight_{5}) %
	(\weight_{1}) %
	(\weight_{6}) %
	& w_{3} %
	&\left(\begin{array}{ccc}
		1 & 0 & 0 \\
		0 & 0 & -1 \\
		0 & -1 & 0
	\end{array}\right)%
	& \alpha_{3} = e_{2}+e_{3}\\
	\hline
	4 & %
	(\weight_{1},\weight_{2}) % 
	(\weight_{5},\weight_{6}) % 
	(\weight_{3}) %
	(\weight_{4}) %
	& w_{1} %
	& \left(\begin{array}{ccc}
		0 & 1 & 0 \\
		1 & 0 & 0 \\
		0 & 0 & 1
	\end{array}\right)%
	& \alpha_{1} = e_{1}-e_{2}\\
	\hline
	5 & %
	(\weight_{2},\weight_{3}) % 
	(\weight_{4},\weight_{5}) % 
	(\weight_{1}) %
	(\weight_{6}) %
	& w_{2} %
	&\left(\begin{array}{ccc}
		1 & 0 & 0 \\
		0 & 0 & 1 \\
		0 & 1 & 0
	\end{array}\right)%
	& \alpha_{2} = e_{2}-e_{3}\\
	\hline
	6 & %
	(\weight_{1},\weight_{4}) % 
	(\weight_{3},\weight_{6}) % 
	(\weight_{2}) %
	(\weight_{5}) %
	& w_{1}w_{3}w_{1} %
	&\left(\begin{array}{ccc}
		0 & 0 & -1 \\
		0 & 1 & 0 \\
		-1 & 0 & 0
	\end{array}\right)%
	& \alpha_{1}+\alpha_{3} = e_{1}+e_{3}\\
	\hline\hline
	{z=0} & %
	(\weight_{1},\weight_{2},\weight_{6},\weight_{5}) % 
	(\weight_{3},\weight_{4}) % 
	& w_{1}w_{3}w_{2} %
	&\left(\begin{array}{ccc}
		0 & -1 & 0 \\
		1 & 0 & 0 \\
		0 & 0 & -1
	\end{array}\right)%
	& \text{N/A}
\end{array}
\ee
\end{figure}

Note that this choice of trivialization is not of the type described in Section \ref{subsec:trivialization}, since the square root branch points do not correspond to hyperplanes bounding a unique Weyl chamber. 
However, by a few simple moves we can bring the trivialization into the desired form.

The basic move consists of rotating the branch cuts that all land at the irregular singularity at $z=\infty$, when a cut associated with sheet monodromy $w$ sweeps across a branch point of type $w'$ clockwise as shown in Figure \ref{fig:branch-cut-move}, the branch point becomes of type
\be
	w'' = w\cdot w'\cdot w^{-1}
\ee
\begin{figure}[!ht]
\begin{center}
\includegraphics[width=0.3\textwidth]{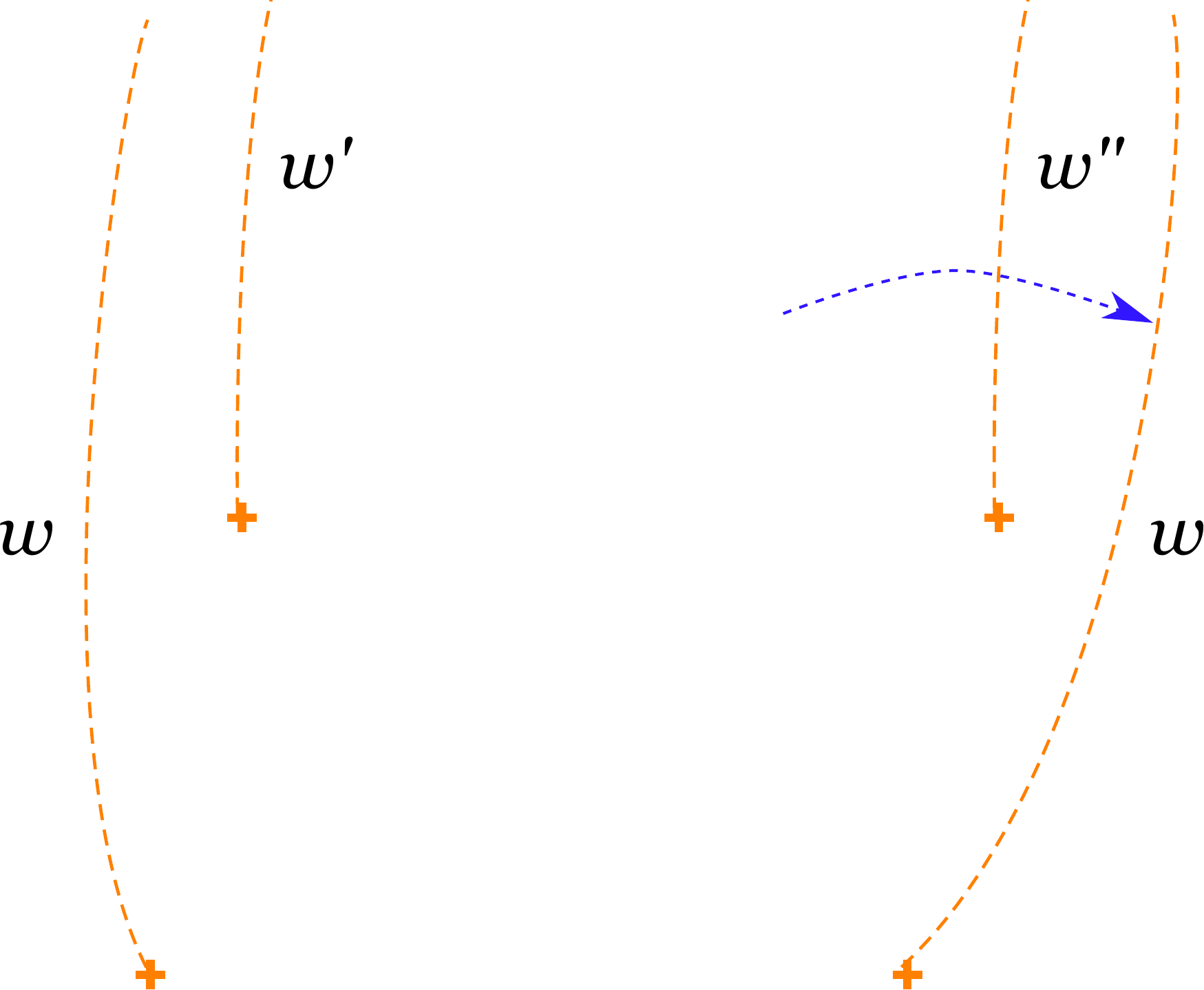}
\caption{Rotating a branch cut of type $w$ across a branch point of type $w'$.}
\label{fig:branch-cut-move}
\end{center}
\end{figure}

Employing the explicit expressions of Weyl group elements given in Section \ref{app:so6}, it the easy to see that the sequence of moves shown in Figure \ref{fig:so6-cut-moves} brings to the desired standard form of the trivialization. 
\begin{figure}[!ht]
\begin{center}
\includegraphics[width=0.3\textwidth]{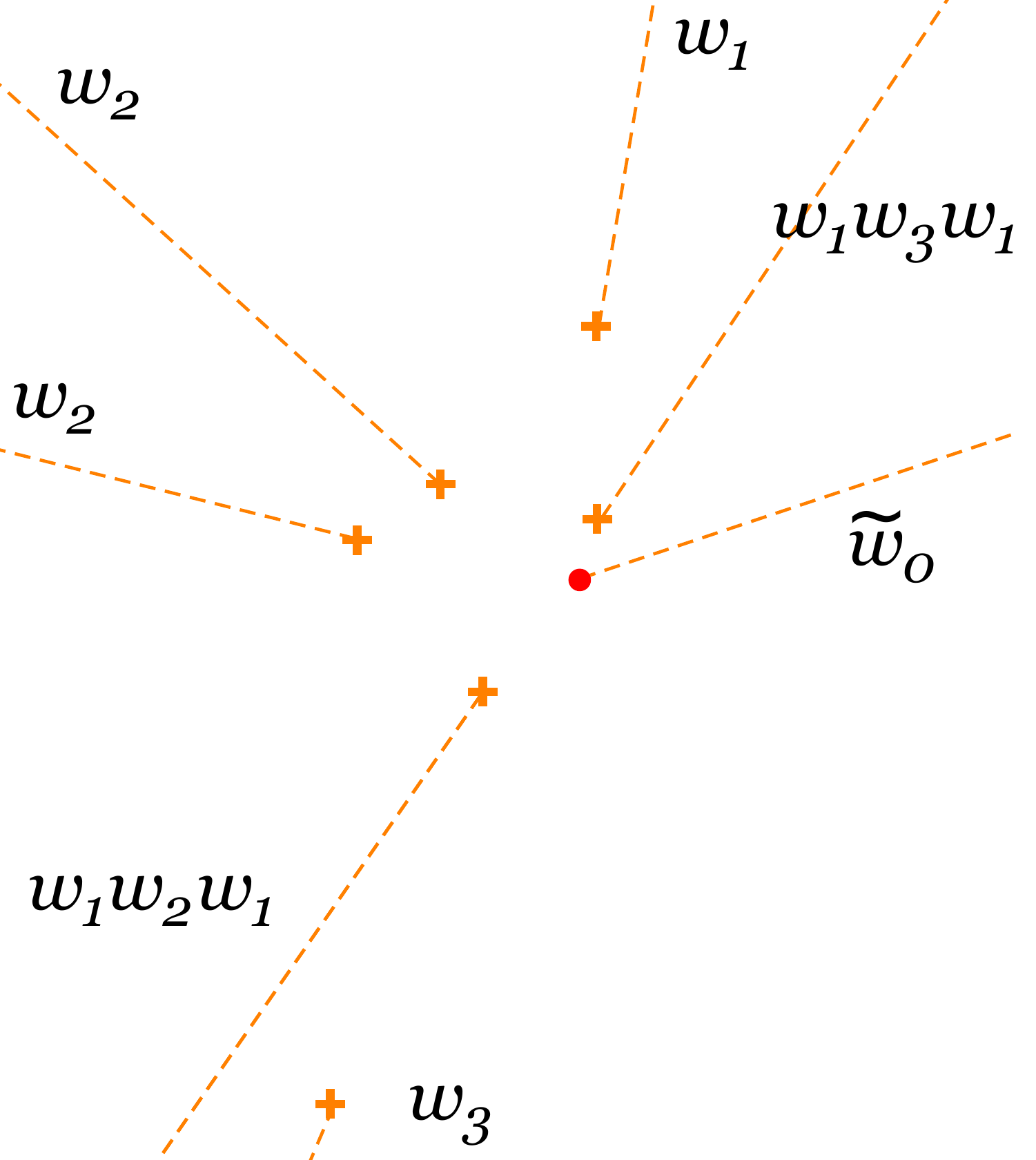}\hfill
\includegraphics[width=0.3\textwidth]{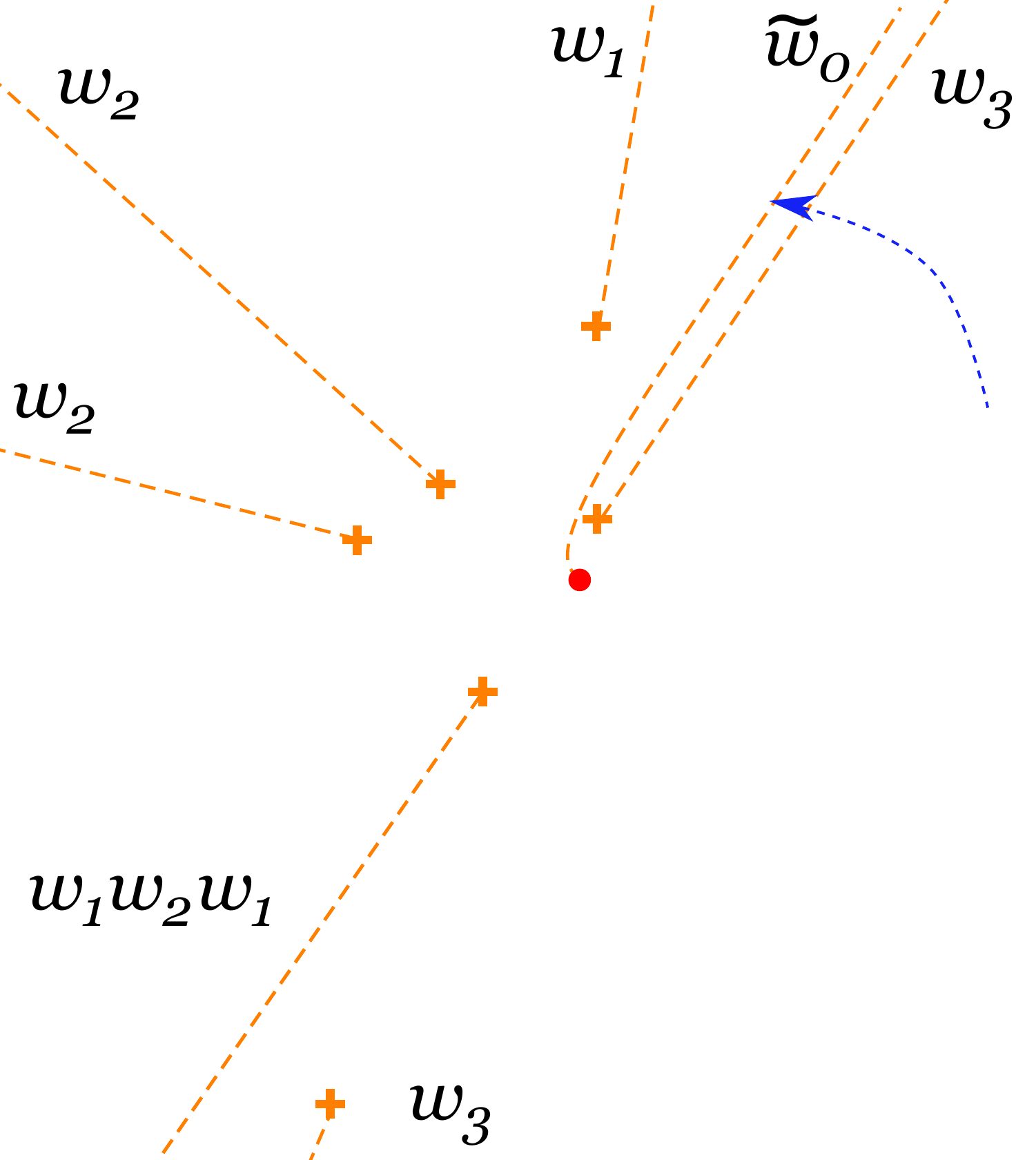}\hfill 
\includegraphics[width=0.3\textwidth]{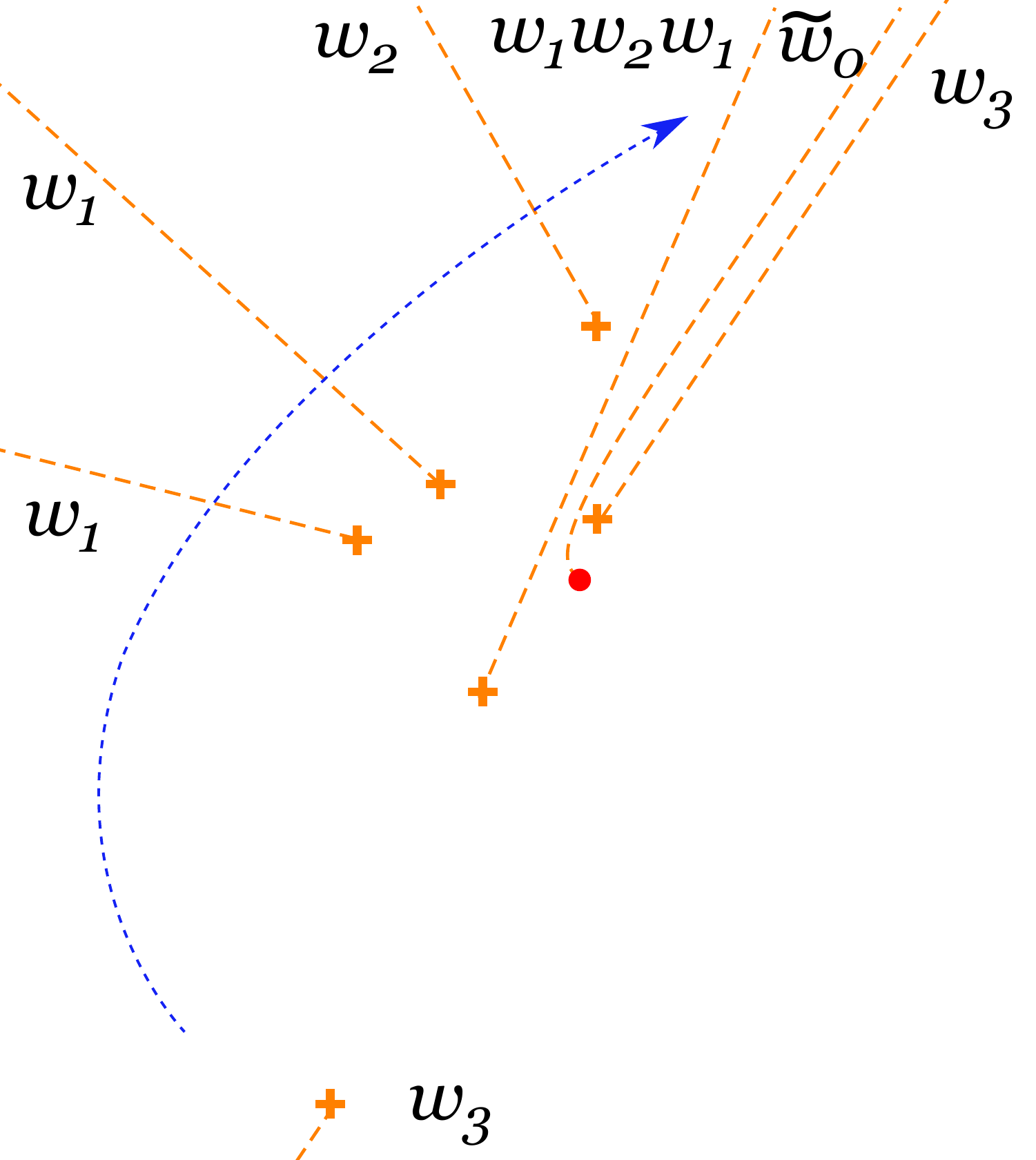}\\ \vspace{20pt}
\includegraphics[width=0.3\textwidth]{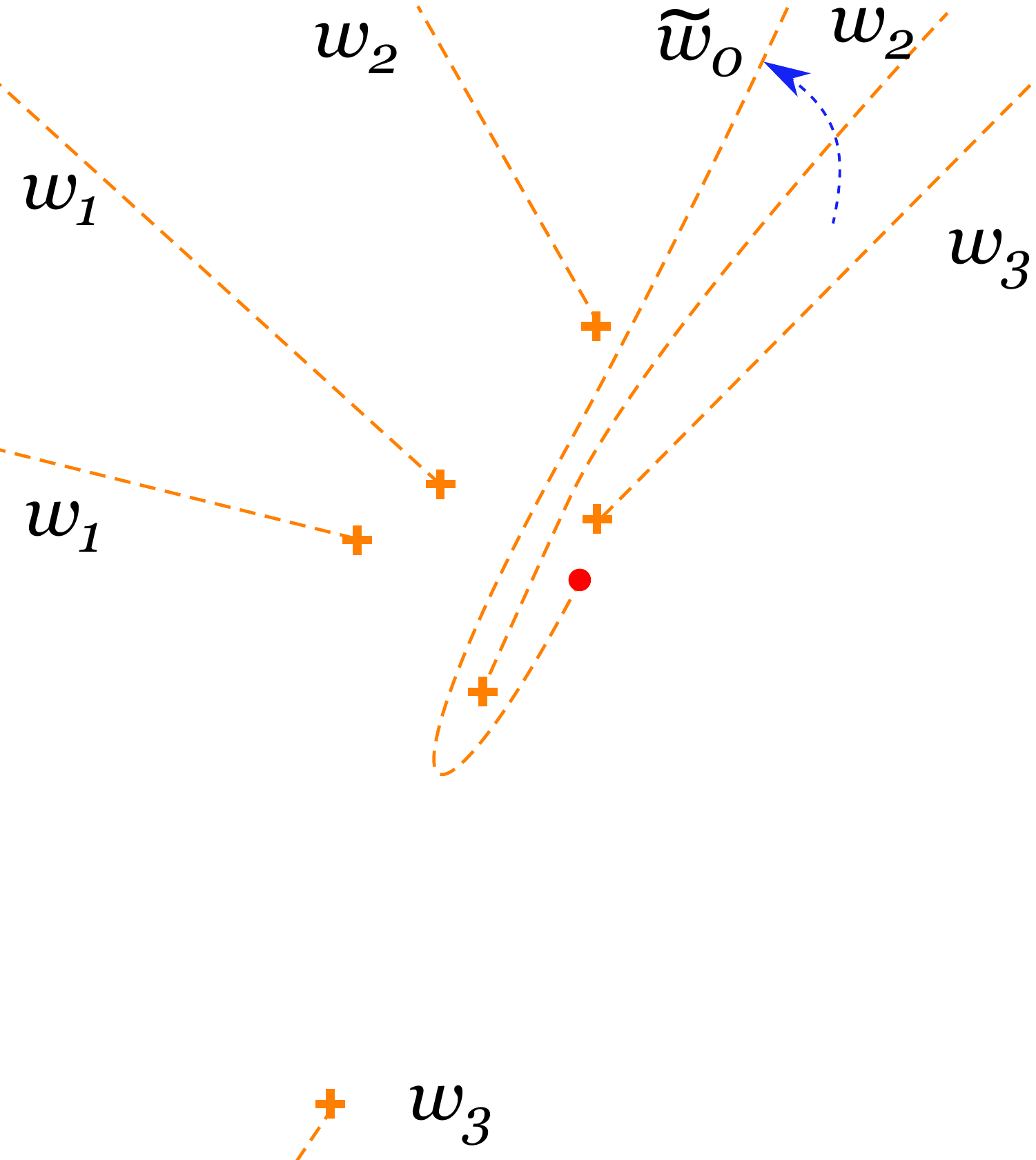}\hfill%\hspace{40pt}
\includegraphics[width=0.3\textwidth]{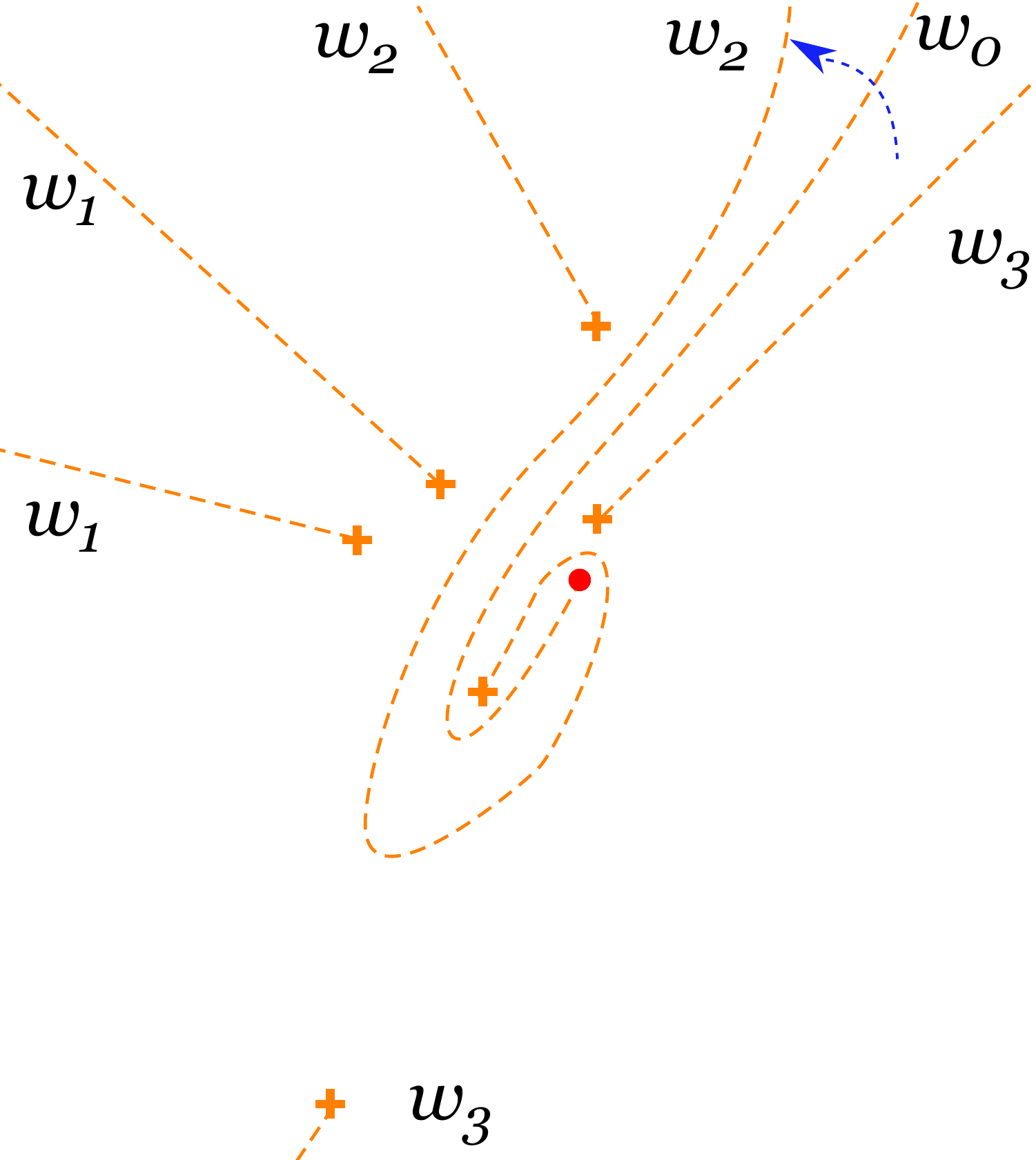}\hfill
\includegraphics[width=0.3\textwidth]{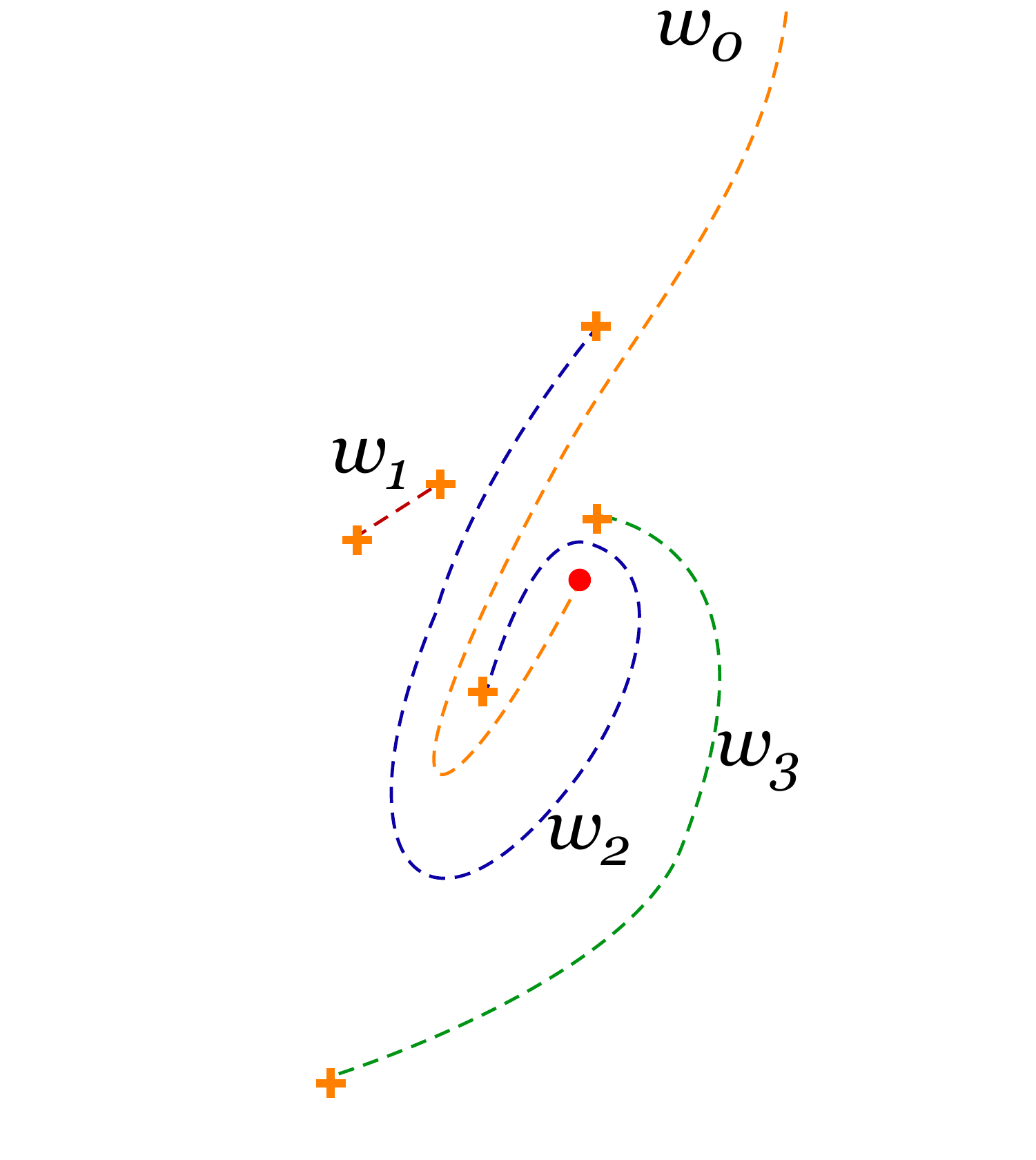}
\caption{A sequence of moves that brings the trivialization for the $\SO(6)$ cover into standard form. The initial monodromy of the cut at $z=0$ has been denoted by $\widetilde w_{0}=w_{1}w_{3}w_{2}$, see Table (\ref{eq:monodromy-analysis}). The final monodromy at $z=0$ is $w_0 = w_2 w_{1} w_{3}$.}
\label{fig:so6-cut-moves}
\end{center}
\end{figure}

All square-root cuts are of simple-root type, compatibly with the claimed standard form of the trivialization. 
The higher-degree cut with monodromy $w_{1}w_3 w_2$ can likewise be identified with a {corner} of the fundamental Weyl chamber $\CC_{0}$. In fact $w_{1} w_3 w_2$ cyclically permutes $(\weight_{1},\weight_{2},\weight_{6},\weight_{5})$ and $(\weight_{3},\weight_{4})$. But note that $\weight_{1} = -\weight_{6},\, \weight_{2} = -\weight_{5},\, \weight_{3}=-\weight_{4}$, moreover $\weight_{1}, \weight_{2}, \weight_{3}$ form a basis for $\ft^{*}$. Since at the ramification point we have
\be
	\langle \weight_{1}-\weight_{6},\varphi\rangle = 2 \langle \weight_{1},\varphi\rangle = 0
\ee
and similarly for $\weight_{2}, \weight_{3}$, this means that $w_{1} w_3 w_2$ is the monodromy associated with crossing from $\CC_{0}$ into another Weyl chamber by passing through the origin in $\ft^{*}$.

\clearpage

\bibliographystyle{JHEP}
\bibliography{ade_networks}

\providecommand{\href}[2]{#2}\begingroup\raggedright\begin{thebibliography}{10}

\bibitem{Gaiotto:2009hg}
D.~Gaiotto, G.~W. Moore, and A.~Neitzke, {\it {Wall-crossing, Hitchin Systems,
  and the WKB Approximation}},  \href{http://xxx.lanl.gov/abs/0907.3987}{{\tt
  arXiv:0907.3987}}.

\bibitem{Gaiotto:2009we}
D.~Gaiotto, {\it {N=2 dualities}},  {\em JHEP} {\bf 08} (2012) 034,
  [\href{http://xxx.lanl.gov/abs/0904.2715}{{\tt arXiv:0904.2715}}].

\bibitem{Tachikawa:2009rb}
Y.~Tachikawa, {\it {Six-dimensional D(N) theory and four-dimensional SO-USp
  quivers}},  {\em JHEP} {\bf 0907} (2009) 067,
  [\href{http://xxx.lanl.gov/abs/0905.4074}{{\tt arXiv:0905.4074}}].

\bibitem{Seiberg:1994rs}
N.~Seiberg and E.~Witten, {\it {Electric - magnetic duality, monopole
  condensation, and confinement in N=2 supersymmetric Yang-Mills theory}},
  {\em Nucl. Phys.} {\bf B426} (1994) 19--52,
  [\href{http://xxx.lanl.gov/abs/hep-th/9407087}{{\tt hep-th/9407087}}].
  [Erratum: Nucl. Phys.B430,485(1994)].

\bibitem{Seiberg:1994aj}
N.~Seiberg and E.~Witten, {\it {Monopoles, duality and chiral symmetry breaking
  in N=2 supersymmetric QCD}},  {\em Nucl. Phys.} {\bf B431} (1994) 484--550,
  [\href{http://xxx.lanl.gov/abs/hep-th/9408099}{{\tt hep-th/9408099}}].

\bibitem{Seiberg:1996nz}
N.~Seiberg and E.~Witten, {\it {Gauge dynamics and compactification to
  three-dimensions}},  in {\em {The mathematical beauty of physics: A memorial
  volume for Claude Itzykson}}, 1996.
\newblock \href{http://xxx.lanl.gov/abs/hep-th/9607163}{{\tt hep-th/9607163}}.

\bibitem{Ferrari:1996sv}
F.~Ferrari and A.~Bilal, {\it {The Strong coupling spectrum of the
  Seiberg-Witten theory}},  {\em Nucl. Phys.} {\bf B469} (1996) 387--402,
  [\href{http://xxx.lanl.gov/abs/hep-th/9602082}{{\tt hep-th/9602082}}].

\bibitem{Gaiotto:2008cd}
D.~Gaiotto, G.~W. Moore, and A.~Neitzke, {\it {Four-dimensional wall-crossing
  via three-dimensional field theory}},  {\em Commun. Math. Phys.} {\bf 299}
  (2010) 163--224, [\href{http://xxx.lanl.gov/abs/0807.4723}{{\tt
  arXiv:0807.4723}}].

\bibitem{Gaiotto:2012rg}
D.~Gaiotto, G.~W. Moore, and A.~Neitzke, {\it {Spectral networks}},  {\em
  Annales Henri Poincare} {\bf 14} (2013) 1643--1731,
  [\href{http://xxx.lanl.gov/abs/1204.4824}{{\tt arXiv:1204.4824}}].

\bibitem{Gukov:2006jk}
S.~Gukov and E.~Witten, {\it {Gauge Theory, Ramification, And The Geometric
  Langlands Program}},  \href{http://xxx.lanl.gov/abs/hep-th/0612073}{{\tt
  hep-th/0612073}}.

\bibitem{Gaiotto:2010be}
D.~Gaiotto, G.~W. Moore, and A.~Neitzke, {\it {Framed BPS States}},  {\em Adv.
  Theor. Math. Phys.} {\bf 17} (2013) 241--397,
  [\href{http://xxx.lanl.gov/abs/1006.0146}{{\tt arXiv:1006.0146}}].

\bibitem{Gaiotto:2011tf}
D.~Gaiotto, G.~W. Moore, and A.~Neitzke, {\it {Wall-Crossing in Coupled 2d-4d
  Systems}},  {\em JHEP} {\bf 12} (2012) 082,
  [\href{http://xxx.lanl.gov/abs/1103.2598}{{\tt arXiv:1103.2598}}].

\bibitem{Gaiotto:2013sma}
D.~Gaiotto, S.~Gukov, and N.~Seiberg, {\it {Surface Defects and Resolvents}},
  {\em JHEP} {\bf 09} (2013) 070,
  [\href{http://xxx.lanl.gov/abs/1307.2578}{{\tt arXiv:1307.2578}}].

\bibitem{Galakhov:2013oja}
D.~Galakhov, P.~Longhi, T.~Mainiero, G.~W. Moore, and A.~Neitzke, {\it {Wild
  Wall Crossing and BPS Giants}},  {\em JHEP} {\bf 1311} (2013) 046,
  [\href{http://xxx.lanl.gov/abs/1305.5454}{{\tt arXiv:1305.5454}}].

\bibitem{Maruyoshi:2013fwa}
K.~Maruyoshi, C.~Y. Park, and W.~Yan, {\it {BPS spectrum of Argyres-Douglas
  theory via spectral network}},  {\em JHEP} {\bf 1312} (2013) 092,
  [\href{http://xxx.lanl.gov/abs/1309.3050}{{\tt arXiv:1309.3050}}].

\bibitem{Galakhov:2014xba}
D.~Galakhov, P.~Longhi, and G.~W. Moore, {\it {Spectral Networks with Spin}},
  \href{http://xxx.lanl.gov/abs/1408.0207}{{\tt arXiv:1408.0207}}.

\bibitem{Hitchin}
N.~Hitchin, {\it {The self-duality equations on a Riemann surface}},  {\em
  Proc. London Math. Soc.} (1987).

\bibitem{FelixKlein}
G.~W. Moore, {\it {2012 Felix Klein Lecture Notes}},  {\em \rm
  \url{http://www.physics.rutgers.edu/~gmoore/}}.

\bibitem{AndyPisa}
A.~Neitzke, {\it Spectral networks and their uses, talk in pisa, june 2014},
  {\em \rm
  \url{http://www.ma.utexas.edu/users/neitzke/talks/spectral-networks-pisa.pdf}}.

\bibitem{AndyKITP}
A.~Neitzke, {\it Spectral networks, talk at kitp, august 2011},  {\em \rm
  \url{http://www.ma.utexas.edu/users/neitzke/talks/spectral-networks-kitp.pdf}}.

\bibitem{Cecotti:1992rm}
S.~Cecotti and C.~Vafa, {\it {On classification of N=2 supersymmetric
  theories}},  {\em Commun. Math. Phys.} {\bf 158} (1993) 569--644,
  [\href{http://xxx.lanl.gov/abs/hep-th/9211097}{{\tt hep-th/9211097}}].

\bibitem{Hori:2013ewa}
K.~Hori, C.~Y. Park, and Y.~Tachikawa, {\it {2d SCFTs from M2-branes}},  {\em
  JHEP} {\bf 1311} (2013) 147, [\href{http://xxx.lanl.gov/abs/1309.3036}{{\tt
  arXiv:1309.3036}}].

\bibitem{Donagi:1993}
R.~Donagi, {\it Decomposition of spectral covers},  {\em Journ\'ees de
  Geometrie Alg\'ebrique d'Orsay} {\bf 218} (1993) 145--175.

\bibitem{Martinec:1995by}
E.~J. Martinec and N.~P. Warner, {\it {Integrable systems and supersymmetric
  gauge theory}},  {\em Nucl.Phys.} {\bf B459} (1996) 97--112,
  [\href{http://xxx.lanl.gov/abs/hep-th/9509161}{{\tt hep-th/9509161}}].

\bibitem{Adler1980267}
M.~Adler and P.~van Moerbeke, {\it Completely integrable systems, euclidean lie
  algebras, and curves},  {\em Advances in Mathematics} {\bf 38} (1980), no.~3
  267 -- 317.

\bibitem{ADLER1980318}
M.~Adler and P.~van Moerbeke, {\it Linearization of hamiltonian systems, jacobi
  varieties and representation theory},  {\em Advances in Mathematics} {\bf 38}
  (1980), no.~3 318 -- 379.

\bibitem{mcdaniel1988}
A.~McDaniel, {\it Representations of $\mathrm {sl}(n,\mathbb{C})$ and the toda
  lattice},  {\em Duke Math. J.} {\bf 56} (02, 1988) 47--99.

\bibitem{Kanev:1989}
V.~Kanev, {\it Spectral curves, simple lie algebras, and prym-tjurin
  varieties},  {\em Proc. Symp. Pure Math.} {\bf 49} (1989) Part I, 627--645.

\bibitem{Kanev2}
V.~Kanev, {\it Spectral curves and prym-tjurin varieties},  {\em Abelian
  Varieties: Proceedings of the International Conference, Held in Egloffstein,
  Germany, October 3-8, 1993} (1995) 151--197.

\bibitem{mcdaniel1992}
A.~McDaniel and L.~Smolinsky, {\it A lie-theoretic galois theory for the
  spectral curves of an integrable system. i},  {\em Comm. Math. Phys.} {\bf
  149} (1992), no.~1 127--148.

\bibitem{mcdaniel1997}
A.~McDaniel and L.~Smolinsky, {\it A lie-theoretic galois theory for the
  spectral curves of an integrable system. ii},  {\em Trans. Amer. Math. Soc.,}
  {\bf 349} (1997) 713--746.

\bibitem{mcdaniel1998}
A.~McDaniel and L.~Smolinsky, {\it Lax equations, weight lattices, and
  prym-tjurin varieties},  {\em Acta Mathematica} {\bf 181} (1998), no.~2
  283--305.

\bibitem{Hollowood:1997pp}
T.~J. Hollowood, {\it {Strong coupling N=2 gauge theory with arbitrary gauge
  group}},  {\em Adv.Theor.Math.Phys.} {\bf 2} (1998) 335--355,
  [\href{http://xxx.lanl.gov/abs/hep-th/9710073}{{\tt hep-th/9710073}}].

\bibitem{Shifman:2014lba}
M.~Shifman and A.~Yung, {\it {Quantum Deformation of the Effective Theory on
  Non-Abelian string and 2D-4D correspondence}},  {\em Phys. Rev.} {\bf D89}
  (2014), no.~6 065035, [\href{http://xxx.lanl.gov/abs/1401.1455}{{\tt
  arXiv:1401.1455}}].

\bibitem{Longhi:2016bte}
P.~Longhi and C.~Y. Park, {\it {ADE Spectral Networks and Decoupling Limits of
  Surface Defects}},  \href{http://xxx.lanl.gov/abs/1611.0940}{{\tt
  arXiv:1611.0940}}.

\bibitem{Argyres:1995jj}
P.~C. Argyres and M.~R. Douglas, {\it {New phenomena in SU(3) supersymmetric
  gauge theory}},  {\em Nucl. Phys.} {\bf B448} (1995) 93--126,
  [\href{http://xxx.lanl.gov/abs/hep-th/9505062}{{\tt hep-th/9505062}}].

\bibitem{Lerche:1991re}
W.~Lerche and N.~Warner, {\it {Polytopes and solitons in integrable, N=2
  supersymmetric Landau-Ginzburg theories}},  {\em Nucl.Phys.} {\bf B358}
  (1991) 571--599.

\bibitem{Eguchi:2002fc}
T.~Eguchi and K.~Sakai, {\it {Seiberg-Witten curve for the E string theory}},
  {\em JHEP} {\bf 05} (2002) 058,
  [\href{http://xxx.lanl.gov/abs/hep-th/0203025}{{\tt hep-th/0203025}}].

\bibitem{Klemm:1994qs}
A.~Klemm, W.~Lerche, S.~Yankielowicz, and S.~Theisen, {\it {Simple
  singularities and N=2 supersymmetric Yang-Mills theory}},  {\em Phys. Lett.}
  {\bf B344} (1995) 169--175,
  [\href{http://xxx.lanl.gov/abs/hep-th/9411048}{{\tt hep-th/9411048}}].

\bibitem{Argyres:1994xh}
P.~C. Argyres and A.~E. Faraggi, {\it {The vacuum structure and spectrum of N=2
  supersymmetric SU(n) gauge theory}},  {\em Phys. Rev. Lett.} {\bf 74} (1995)
  3931--3934, [\href{http://xxx.lanl.gov/abs/hep-th/9411057}{{\tt
  hep-th/9411057}}].

\bibitem{Gorsky:1995zq}
A.~Gorsky, I.~Krichever, A.~Marshakov, A.~Mironov, and A.~Morozov, {\it
  {Integrability and Seiberg-Witten exact solution}},  {\em Phys. Lett.} {\bf
  B355} (1995) 466--474, [\href{http://xxx.lanl.gov/abs/hep-th/9505035}{{\tt
  hep-th/9505035}}].

\bibitem{Danielsson:1995is}
U.~H. Danielsson and B.~Sundborg, {\it {The Moduli space and monodromies of N=2
  supersymmetric SO(2r+1) Yang-Mills theory}},  {\em Phys. Lett.} {\bf B358}
  (1995) 273--280, [\href{http://xxx.lanl.gov/abs/hep-th/9504102}{{\tt
  hep-th/9504102}}].

\bibitem{Keller:2011ek}
C.~A. Keller, N.~Mekareeya, J.~Song, and Y.~Tachikawa, {\it {The ABCDEFG of
  Instantons and W-algebras}},  {\em JHEP} {\bf 1203} (2012) 045,
  [\href{http://xxx.lanl.gov/abs/1111.5624}{{\tt arXiv:1111.5624}}].

\bibitem{Alim:2011kw}
M.~Alim, S.~Cecotti, C.~Cordova, S.~Espahbodi, A.~Rastogi, {\em et.~al.}, {\it
  {N=2 Quantum Field Theories and Their BPS Quivers}},
  \href{http://xxx.lanl.gov/abs/1112.3984}{{\tt arXiv:1112.3984}}.

\bibitem{Eguchi:1996vu}
T.~Eguchi, K.~Hori, K.~Ito, and S.-K. Yang, {\it {Study of N=2 superconformal
  field theories in four-dimensions}},  {\em Nucl.Phys.} {\bf B471} (1996)
  430--444, [\href{http://xxx.lanl.gov/abs/hep-th/9603002}{{\tt
  hep-th/9603002}}].

\bibitem{Kazama:1988uz}
Y.~Kazama and H.~Suzuki, {\it {Characterization of N=2 Superconformal Models
  Generated by Coset Space Method}},  {\em Phys.Lett.} {\bf B216} (1989) 112.

\bibitem{Lerche:1989uy}
W.~Lerche, C.~Vafa, and N.~P. Warner, {\it {Chiral Rings in N=2 Superconformal
  Theories}},  {\em Nucl.Phys.} {\bf B324} (1989) 427.

\bibitem{Gaiotto:2012db}
D.~Gaiotto, G.~W. Moore, and A.~Neitzke, {\it {Spectral Networks and Snakes}},
  {\em Annales Henri Poincare} {\bf 15} (2014) 61--141,
  [\href{http://xxx.lanl.gov/abs/1209.0866}{{\tt arXiv:1209.0866}}].

\bibitem{Humpreys}
J.~Humphreys, {\em Reflection Groups and Coxeter Groups}.
\newblock Cambridge Studies in Advanced Mathematics. Cambridge University
  Press, 1992.

\bibitem{Procesi}
C.~Procesi, {\em {Lie Groups: An Approach through Invariants and
  Representations}}.

\bibitem{Lerche:1996an}
W.~Lerche and N.~P. Warner, {\it {Exceptional SW geometry from ALE
  fibrations}},  {\em Phys. Lett.} {\bf B423} (1998) 79--86,
  [\href{http://xxx.lanl.gov/abs/hep-th/9608183}{{\tt hep-th/9608183}}].

\bibitem{Eguchi:2001fm}
T.~Eguchi, N.~P. Warner, and S.-K. Yang, {\it {ADE singularities and coset
  models}},  {\em Nucl.Phys.} {\bf B607} (2001) 3--37,
  [\href{http://xxx.lanl.gov/abs/hep-th/0105194}{{\tt hep-th/0105194}}].

\end{thebibliography}\endgroup

\end{document}